\documentclass[a4paper,11pt]{article}

\usepackage[utf8]{inputenc}
\usepackage[T1]{fontenc}
\usepackage{amsmath, amssymb, amsthm, verbatim, mathtools, thm-restate}
\usepackage{IEEEtrantools}
\usepackage[english]{babel}
\usepackage[margin=1.2in]{geometry}
\usepackage[colorlinks, allcolors=blue]{hyperref}
\usepackage{mathrsfs, dsfont}
\usepackage{enumitem}
\usepackage[ruled, vlined]{algorithm2e}
\usepackage{tcolorbox}
\usepackage{breakcites}
\usepackage{subcaption}
\usepackage{wrapfig}

\usepackage{todonotes}

\usepackage{tikz}
\usetikzlibrary{automata,positioning,backgrounds}

\numberwithin{equation}{section}
\numberwithin{figure}{section}

\title{Mixing Any Cocktail with Limited Ingredients: On the Structure of Payoff Sets in Multi-Objective POMDPs and its Impact on Randomised Strategies\thanks{James C. A. Main is an F.R.S.-FNRS Research Fellow and Mickael Randour is an F.R.S.-FNRS Research Associate. Both authors are members of the TRAIL institute. This work has been supported by the Fonds de la Recherche Scientifique – FNRS under Grant n° T.0188.23 (PDR ControlleRS).}}
\author{James {C.~A.} Main \and Mickael Randour \and F.R.S.-FNRS \& UMONS -- Université de Mons}
\date{}

\theoremstyle{definition} \newtheorem{example}{Example}[section]
\theoremstyle{definition} \newtheorem{theorem}{Theorem}[section]
\theoremstyle{definition} 
\theoremstyle{definition} 
\theoremstyle{definition} \newtheorem{lemma}[theorem]{Lemma}
\theoremstyle{definition} \newtheorem{corollary}[theorem]{Corollary}
\theoremstyle{definition} \newtheorem{remark}[theorem]{Remark}

\tikzset{
  >=stealth,
  left sided/.style={
    draw=none,
    append after command={
      [shorten <= -0.5\pgflinewidth]
      (\tikzlastnode.north west) edge[dashed](\tikzlastnode.south west)
    }
  },
  two sided/.style={
    draw=none,
    append after command={
      [shorten <= -0.5\pgflinewidth]
      (\tikzlastnode.north west) edge[dashed](\tikzlastnode.south west)
      (\tikzlastnode.north east) edge[dashed](\tikzlastnode.south east)
    }
  },
  right sided/.style={
    draw=none,
    append after command={
      [shorten <= -0.5\pgflinewidth]
      (\tikzlastnode.north east) edge[dashed](\tikzlastnode.south east)
    }
  },
  every state/.style={circle, minimum size=1cm},
  every path/.style={thick},
  initial text=,
  node distance=1cm
}

\tikzstyle{stochasticc} = [fill, circle, minimum size=0.1cm, inner sep=0.05cm, outer sep=0cm]
\tikzstyle{stochastics} = [fill, rectangle, minimum size=0.1cm, inner sep=0.05cm, outer sep=0cm]

\newcommand{\ud}{\mathrm{d}}

\newcommand{\bigo}{\mathcal{O}}
\newcommand{\image}[1]{\mathsf{Im}(#1)}

\newcommand{\leLex}{\mathrel{\leq_{\mathsf{lex}}}}
\newcommand{\geLex}{\mathrel{\geq_{\mathsf{lex}}}}
\newcommand{\lLex}{\mathrel{<_{\mathsf{lex}}}}

\newcommand{\vect}{\mathbf{v}}
\newcommand{\vectComp}{v} \newcommand{\vectB}{\mathbf{w}}
\newcommand{\vectCompB}{w} \newcommand{\scalar}{\alpha} \newcommand{\scalarB}{\beta}
\newcommand{\scalarC}{\gamma}

\newcommand{\ooInt}[2]{\left]#1, #2\right[}
\newcommand{\ccInt}[2]{\left[#1, #2\right]}
\newcommand{\ocInt}[2]{\left]#1, #2\right]}
\newcommand{\coInt}[2]{\left[#1, #2\right[}

\newcommand{\IR}{\mathbb{R}}
\newcommand{\IRbar}{\bar{\mathbb{R}}}

\newcommand{\IN}{\mathbb{N}}

\newcommand{\subsets}[1]{2^{#1}}
\newcommand{\proba}{\mathbb{P}}
\newcommand{\expectation}{\mathbb{E}}
\newcommand{\expectancy}{\expectation}
\newcommand{\dist}[1]{\mathcal{D}(#1)}
\newcommand{\distMeasure}{\mu} \newcommand{\cyl}[1]{\mathsf{Cyl}(#1)}
\newcommand{\supp}[1]{\mathsf{supp}(#1)}

\newcommand{\sigmaAlgebra}{\mathcal{A}}
\newcommand{\probaEvent}{E}
\newcommand{\rv}{X}
\newcommand{\rvVect}{\mathcal{\rv}}
\newcommand{\rvB}{Y}
\newcommand{\rvBVect}{\mathcal{\rvB}}
\newcommand{\rvC}{V}
\newcommand{\rvCVect}{\mathcal{\rvC}}

\newcommand{\indic}[1]{\mathds{1}_{#1}}

\newcommand{\concurStateSpace}{S}
\newcommand{\concurState}{s}

\newcommand{\pomdp}{\mathfrak{P}}
\newcommand{\obsFun}{\mathsf{Obs}}
\newcommand{\observation}{z}
\newcommand{\obsSpace}{\mathcal{Z}}
\newcommand{\pomdpTuple}{(\mdpStateSpace, \mdpActionSpace, \mdpTrans, \obsSpace, \obsFun)}
\newcommand{\belief}{B}
\newcommand{\beliefUpdate}{\mathsf{Bel}_{\mathsf{up}}}
\newcommand{\beliefHistory}{\mathsf{Bel}_{\mathsf{hist}}}

\newcommand{\mdp}{\mathfrak{M}}
\newcommand{\mdpStateSpace}{S}
\newcommand{\mdpState}{s}
\newcommand{\mdpStateB}{t}
\newcommand{\mdpTrans}{\delta}
\newcommand{\mdpActionSpace}{A}
\newcommand{\mdpAction}{a}
\newcommand{\mdpActionB}{b}
\newcommand{\mdpActionC}{c}
\newcommand{\mdpTuple}{(\mdpStateSpace, \mdpActionSpace, \mdpTrans)}

\newcommand{\scc}{C}

\newcommand{\playSet}[1]{\mathsf{Plays}(#1)}
\newcommand{\play}{\pi}
\newcommand{\histSet}[1]{\mathsf{Hist}(#1)}
\newcommand{\playPrefix}[2]{#1_{\leq#2}}
\newcommand{\playSuffix}[2]{#1_{\geq#2}}
\newcommand{\hist}{h}

\newcommand{\histConcat}[2]{#1\cdot#2}

\newcommand{\histPart}{\mathcal{H}}

\newcommand{\pomdpSigmaAlgebra}{\sigmaAlgebra_{\pomdp}}

\newcommand{\stratSigmaAlgebra}{\sigmaAlgebra_{\mathsf{pure}}}

\newcommand{\last}[1]{\mathsf{last}(#1)}

\newcommand{\buchi}[1]{\mathsf{B\ddot{u}chi}(#1)}

\newcommand{\reach}[1]{\mathsf{Reach}(#1)}
\newcommand{\reachBounded}[2]{\mathsf{Reach}^{\leq{}#2}(#1)}
\newcommand{\reachExact}[2]{\mathsf{Reach}^{=#2}(#1)}
\newcommand{\target}{T}

\newcommand{\infSet}[1]{\mathsf{inf}(#1)}

\newcommand{\discSum}[2]{\mathsf{DSum}^{#1}_{#2}} \newcommand{\discSumG}[2]{\mathsf{GDSum}^{#1}_{#2}} \newcommand{\discFactor}{\lambda}
\newcommand{\spath}[2]{\mathsf{SPath}^{#1}_{#2}} \newcommand{\totrew}[1]{\mathsf{TRew}_{#1}}

\newcommand{\indexPosition}{\ell}
\newcommand{\indexLast}{r}
\newcommand{\indexEC}{i}
\newcommand{\indexSequence}{n} \newcommand{\indexSequenceB}{m} 

\newcommand{\lengthBound}{k}
\newcommand{\numEC}{k}

\newcommand{\word}{w}
\newcommand{\wordOmega}{u}

\newcommand{\strat}[1]{\sigma_{#1}}
\newcommand{\stratB}[1]{\tau_{#1}}

\newcommand{\stratMDP}{{\strat{}}}
\newcommand{\stratBMDP}{{\stratB{}}}

\newcommand{\mixedStrat}{\mu}

\newcommand{\stratClass}{\Sigma}
\newcommand{\stratClassPure}[1]{\stratClass_{\mathsf{pure}}(#1)}
\newcommand{\stratClassAll}[1]{\stratClass(#1)}

\newcommand{\weight}{w}
\newcommand{\weightMax}{W}

\newcommand{\objective}{\Omega}

\newcommand{\payoff}{f}
\newcommand{\payoffB}{g}
\newcommand{\payoffTuple}{\bar{\payoff}}
\newcommand{\payoffTupleB}{\bar{\payoffB}}

\newcommand{\achSet}[2]{\mathsf{Ach}_{#2}(#1)} \newcommand{\achSetPure}[2]{\mathsf{Ach}_{#2}^{\mathsf{pure}}(#1)}

\newcommand{\achSetClass}[3]{\mathsf{Ach}_{#2}^{#3}(#1)}
\newcommand{\paySet}[2]{\mathsf{Pay}_{#2}(#1)} \newcommand{\paySetPure}[2]{\mathsf{Pay}_{#2}^{\mathsf{pure}}(#1)}

\newcommand{\paySetClass}[3]{\mathsf{Pay}_{#2}^{#3}(#1)}

\newcommand{\spanAff}[1]{\mathsf{aff}(#1)}
\newcommand{\convex}[1]{\mathsf{conv}(#1)}
\renewcommand{\ker}[1]{\mathsf{ker}(#1)}
\newcommand{\down}[1]{\mathsf{down}(#1)}
\newcommand{\corners}[1]{\mathsf{extr}(#1)}

\newcommand{\zeroVect}{\mathbf{0}}
\newcommand{\zeroVectDim}[1]{\zeroVect_{#1}}
\newcommand{\oneVect}{\mathbf{1}}
\newcommand{\oneVectDim}[1]{\oneVect_{#1}}
\newcommand{\vectSpace}{V}

\newcommand{\linMap}{L}
\newcommand{\linMapB}{\mathcal{L}}
\newcommand{\linForm}{x^*}
\newcommand{\linFormB}{y^*}

\newcommand{\hplane}{H}

\newcommand{\proj}[1]{\mathsf{proj}_{#1}}

\newcommand{\scalarProd}[2]{\langle{}#1, #2\rangle{}}

\newcommand{\interior}[1]{\mathsf{int}(#1)}
\newcommand{\relInt}[1]{\mathsf{ri}(#1)}
\newcommand{\closure}[1]{\mathsf{cl}(#1)}
\newcommand{\border}[1]{\mathsf{bd}(#1)}
\newcommand{\ball}[2]{B(#1, #2)}

\newcommand{\topSpace}{X}
\newcommand{\topSpaceB}{Y}
\newcommand{\openSet}{U}
\newcommand{\closedSet}{F}
\newcommand{\nhood}{N}
\newcommand{\topology}{\mathcal{T}}
\newcommand{\topologyB}{\mathcal{T}'}
\newcommand{\topBasis}{\mathcal{B}}

\newcommand{\topologyDis}{\mathcal{T}_{\mathsf{dis}}}
\newcommand{\topologyInd}{\mathcal{T}_{\mathsf{ind}}}
\newcommand{\prodTopology}{\mathcal{T}_\Pi}

\newcommand{\metric}{\mathsf{dist}}

\newcommand{\discMetric}{\metric_{\mathsf{disc}}}
\newcommand{\distMetric}{\metric_{\mathsf{proba}}}
\newcommand{\playMetric}{\metric_{\mathsf{play}}}

\newcommand{\numObj}{d}

\newcommand{\payoffComp}{q}
\newcommand{\payoffVect}{\mathbf{\payoffComp}}
\newcommand{\payoffVectVerbose}{(\payoffComp_\indexPayoff)_{1\leq\indexPayoff\leq\numObj}}
\newcommand{\payoffCompB}{p}
\newcommand{\payoffVectB}{\mathbf{\payoffCompB}}
\newcommand{\payoffVectBVerbose}{(\payoffCompB_\indexPayoff)_{1\leq\indexPayoff\leq\numObj}}
\newcommand{\indexPayoff}{j}

\newcommand\vartextvisiblespace[1][.5em]{\makebox[#1]{\kern.07em
    \vrule height.3ex
    \hrulefill
    \vrule height.3ex
    \kern.07em
  }}

\makeatletter
\def\squarecorner#1{
\pgf@x=\the\wd\pgfnodeparttextbox \pgfmathsetlength\pgf@xc{\pgfkeysvalueof{/pgf/inner xsep}}\advance\pgf@x by 2\pgf@xc \pgfmathsetlength\pgf@xb{\pgfkeysvalueof{/pgf/minimum width}}\ifdim\pgf@x<\pgf@xb \pgf@x=\pgf@xb \fi \pgf@y=\ht\pgfnodeparttextbox \advance\pgf@y by\dp\pgfnodeparttextbox \pgfmathsetlength\pgf@yc{\pgfkeysvalueof{/pgf/inner ysep}}\advance\pgf@y by 2\pgf@yc \pgfmathsetlength\pgf@yb{\pgfkeysvalueof{/pgf/minimum height}}\ifdim\pgf@y<\pgf@yb \pgf@y=\pgf@yb \fi \ifdim\pgf@x<\pgf@y \pgf@x=\pgf@y \else
        \pgf@y=\pgf@x \fi
\pgf@x=#1.5\pgf@x \advance\pgf@x by.5\wd\pgfnodeparttextbox \pgfmathsetlength\pgf@xa{\pgfkeysvalueof{/pgf/outer xsep}}\advance\pgf@x by#1\pgf@xa \pgf@y=#1.5\pgf@y \advance\pgf@y by-.5\dp\pgfnodeparttextbox \advance\pgf@y by.5\ht\pgfnodeparttextbox \pgfmathsetlength\pgf@ya{\pgfkeysvalueof{/pgf/outer ysep}}\advance\pgf@y by#1\pgf@ya }
\makeatother

\pgfdeclareshape{square}{
    \savedanchor\northeast{\squarecorner{}}
    \savedanchor\southwest{\squarecorner{-}}

    \foreach \x in {east,west} \foreach \y in {north,mid,base,south} {
        \inheritanchor[from=rectangle]{\y\space\x}
    }
    \foreach \x in {east,west,north,mid,base,south,center,text} {
        \inheritanchor[from=rectangle]{\x}
    }
    \inheritanchorborder[from=rectangle]
    \inheritbackgroundpath[from=rectangle]
}

\begin{document}

\maketitle

\begin{abstract}
We consider multi-dimensional payoff functions in partially observable Markov decision processes. We study the structure of the set of expected payoff vectors of all strategies (policies) and study what kind are needed to achieve a given expected payoff vector. In general, pure strategies (i.e., not resorting to randomisation) do not suffice for this problem.

We prove that for any payoff for which the expectation is well-defined under all strategies, it is sufficient to mix (i.e., randomly select a pure strategy at the start of a play and committing to it for the rest of the play) \textit{finitely many} pure strategies to approximate any expected payoff vector up to any precision. Furthermore, for any payoff for which the expected payoff is finite under all strategies, any expected payoff can be obtained exactly by mixing finitely many strategies.
\end{abstract}

\section{Introduction}\label{section:intro}
\subparagraph*{The model.}
\textit{Markov decision processes} (MDPs) are a classical framework for decision making in uncertain environments.
MDPs are notably used in the fields of formal methods (e.g.,~\cite{BK08,gog23}) and reinforcement learning (e.g.,~\cite{SuttonB18}).
An MDP models the interaction of a controllable system with a stochastic environment: at each step of the process, the system selects an action and the state of the process is updated following a probability distribution that depends only on the current state and chosen action.
This interaction goes on forever and yields a \textit{play} (i.e., an execution) of the process.
The quality of a play is quantified by a \textit{payoff function}, which assigns a numeric value to all plays.
The goal is then to compute an \textit{optimal strategy} (policy), i.e., a decision-making procedure for the system that maximises the expected payoff.
Such a strategy constitutes a \textit{formal blueprint} for a controller of the modelled system~\cite{rECCS}.

\subparagraph*{Imperfect information.}

In an MDP, the system is \textit{perfectly informed} of the current state in each step.
\textit{Partially observable MDPs} (POMDPs, e.g.~\cite{DBLP:journals/ai/KaelblingLC98}) generalise MDPs by relaxing this assumption.
In a POMDP, the system perceives an observation at each step from the current state, which may be shared with other states, and strategies must make decisions based \textit{only on observations}.
POMDPs are harder to analyse than MDPs in general~\cite{DBLP:conf/mfcs/ChatterjeeDH10}.

\subparagraph*{Multi-objective (PO)MDPs.}
More complex specifications, such as simultaneous constraints on the response time and energy consumption of a system, are usually modelled using  \textit{multi-dimensional payoff functions}.
In this setting, some expected payoff vectors may be incomparable; an analysis of trade-offs between the different dimensions may be necessary.

The goal is generally to determine, given a vector, whether it is \textit{achievable}, i.e., whether there exists a strategy whose expected payoff is greater than or equal to the vector in all components (e.g.,~\cite{DBLP:journals/lmcs/EtessamiKVY08,DBLP:conf/lpar/ChatterjeeFW13,DBLP:journals/fmsd/RandourRS17}).
A related problem is to compute or approximate the \textit{Pareto curve} of the set of expected payoffs, i.e., the expected payoffs that are \textit{Pareto-optimal} (e.g.,~\cite{DBLP:conf/atva/ForejtKP12,DBLP:conf/ijcai/RoijersWO15,DBLP:journals/lmcs/ChatterjeeKK17,DBLP:conf/tacas/QuatmannK21,DBLP:conf/atal/ReymondBN22}).
Intuitively, an expected payoff is Pareto-optimal if there is no strategy whose expected payoff is as good on all dimensions and strictly better on one dimension.
Alternatively, one can look for strategies with expected payoffs that are optimal for the \textit{lexicographic} order over vectors (e.g.,~\cite{DBLP:conf/ijcai/WrayZ15,DBLP:conf/fm/HahnPSSTW21,CKMWW23,DBLP:conf/atva/BusattoGastonCMMPR23}).
For instance, in a two-dimensional setting, this equates to finding strategies that maximise the expected payoff on the second dimension among the strategies that maximise the expected payoff on the first dimension.

\subparagraph*{Strategy complexity.}
For controller design, simpler strategies, i.e., using limited memory and randomisation, are preferable, notably for their explainability.
In MDPs, for many classical (one-dimensional) specifications, there exist optimal strategies that use no memory and no randomisation.
For instance, such strategies suffice to maximise the probability of a parity objective~\cite{DBLP:conf/soda/ChatterjeeJH04} or the expectation of a mean~\cite{bierth1987expected} or discounted-sum payoff~\cite{Sha53}.
For one-dimensional specifications in POMDPs, although memory is usually necessary to play optimally, randomisation is not~\cite{DBLP:journals/iandc/Chatterjee0GH15}.

In contrast to the one-dimensional case, Pareto-optimal strategies in multi-objective MDPs and POMDPs may require both \textit{memory} and \textit{randomisation} (see, e.g.,~\cite{DBLP:journals/fmsd/RandourRS17,DBLP:conf/fsttcs/BrihayeGMR23}).
For this reason, a major problem is understanding what are the simplest relevant strategies, e.g., whether randomisation is necessary, or what is the minimum amount of memory required for the application at hand.
A related approach consists in directly searching for simple strategies that satisfy the desired constraints (e.g.,~\cite{DBLP:conf/tacas/DelgrangeKQR20}).

There is an extensive body of work aimed at understanding memory requirements for strategies in various game-based models.
These works provide sufficient conditions on payoffs for which memoryless or finite-memory strategies suffice, or even complete characterisations of such payoffs.
For instance, see~\cite{DBLP:conf/fsttcs/0001PR18,DBLP:journals/lmcs/BouyerLORV22,DBLP:journals/theoretics/BouyerRV23,DBLP:conf/icalp/CasaresO23} for games on deterministic graphs,~\cite{DBLP:journals/lmcs/BouyerORV23,DBLP:journals/ijgt/GimbertK23} for stochastic games and~\cite{DBLP:conf/stacs/Gimbert07} for MDPs.
Most of these works focus on pure strategies.

In this work, our focus is on \textit{randomisation requirements} in countable-state multi-objective POMDPs.
Instead of identifying bounds on memory, we study whether we can achieve vectors while using strategies with \textit{limited randomisation power}.
In general, it is not possible to jointly minimise memory and randomisation requirements: there can be a trade-off between memory and randomisation, even in the one-dimensional case~\cite{DBLP:conf/qest/ChatterjeeAH04,DBLP:conf/stacs/Horn09,DBLP:journals/acta/ChatterjeeRR14,DBLP:conf/concur/MonmegePR20}.

\subparagraph*{Randomised strategies.}
A strategy that does not use randomisation is called pure.
Formally, a \textit{pure strategy} is a function that assigns actions to histories (i.e., sequences of observations).
Randomised strategies can be defined in two ways.
On the one hand, a \textit{behavioural strategy} is a function that assigns distributions over actions to histories.
In other words, when following a behavioural strategy, the random choices of actions are made at each step of the process.
On the other hand, a \textit{mixed strategy} is a distribution over the set of pure strategies.
When playing according to a mixed strategy, a pure strategy is randomly chosen at the start of the process and is followed for the entire play.

In POMDPs, behavioural and mixed strategies are equivalent whenever actions are observable; any strategy from one class can be emulated by a strategy of the other.
This result is known as \textit{Kuhn's theorem}~\cite{Kuhn53,Aumann64}.
We note that distributions over behavioural strategies are no more general than mixed and behavioural strategies~\cite{DBLP:journals/jacm/BertrandGG17}.

When limiting ourselves to \textit{finite-memory strategies}, the natural analogues of mixed strategies are strictly weaker than the analogues of behavioural strategies~\cite{DBLP:journals/iandc/MainR24}.
In fact, the emulation of a behavioural strategy that flips a coin at each round requires a mixed strategy that randomises over \textit{uncountably} many pure strategies, infinitely many of which require infinite memory to account for all patterns.
It follows that mixed strategies with a finite support, i.e., that mix finitely many pure strategies, are a \textit{strict sub-class} of the set of randomised strategies (regardless of memory).
Arguably, this is one of the simplest ways of integrating randomness in strategies.
Considering such strategies instead of general behavioural (or mixed) strategies amounts to \textit{limiting the randomisation power of strategies}.

\subparagraph*{A simple example.}
For the sake of illustration, let us consider the MDP depicted in Figure~\ref{figure:introduction:mdp}.
This MDP models a situation where a person wants to go to work.
They must choose between riding their bicycle or taking the train.
However, the train may be delayed with high probability due to an ongoing strike.
The goal of the commuter is twofold: maximise the likelihood of reaching work within $40$ time units (to reach work on time) and do so as fast as possible on average.

\begin{figure*}
  \centering
  \begin{subfigure}[t]{0.48\textwidth}
    \centering
      \begin{tikzpicture}
      \node[state, align=center] (h) {\textsf{home}};
      \node[right = of h] (ref) {};
      \node[state,above = of ref] (r) {\textsf{ride}};
      \node[state, accepting, right = of ref] (w) {\textsf{work}};

      \path[->] (h) edge[bend left] node[midway,stochasticc] (mt) {} (r);
      \path[->] (r) edge[bend left] node[midway,stochasticc] (mr) {} (w);
      \path[->] (h) edge[bend right] node[midway,stochasticc] (mb) {} (w);

      \begin{scope}[on background layer]
        \path[-,white] (h) edge[bend right] node[xshift=-1mm,yshift=-1.5mm,text=black,align=center] {$\mathsf{bike}\mid 30$} (mb);
        \path[-,white] (mb) edge[bend right] node[text=black,align=center] {$1$} (w);
        \path[-,white] (h) edge[bend left] node[left, text=black,align=center,xshift=2mm,yshift=1mm] {$\mathsf{train}\mid5$} (mt);
        \path[-,white] (mt) edge[bend left] node[above, text=black,align=center,yshift=-2mm] {$\frac{1}{4}$} (r);
        \path[-,white] (r) edge[bend left] node[above, text=black,align=center,xshift=3mm,yshift=-1mm] {$\mathsf{train}\mid 5$} (mr);
        \path[-,white] (mr) edge[bend left] node[right, text=black,align=center] {$1$} (w);
      \end{scope}
      \path[->] (mt) edge[bend left] node[right, text=black,align=center] {$\frac{3}{4}$} (h);
      
      \path[->] (w) edge[loop below] node[midway,stochasticc,yshift=0.5mm] {} node[right,yshift=3mm,xshift=2mm,align=center] {$\mathsf{meeting}\mid1$}  node[left,yshift=3mm,xshift=-2mm] {$1$} (w);
    \end{tikzpicture}
    \caption{Transitions are labelled by actions and probabilities.
      Weights represent the time taken by each action.
      The target is doubly circled.}\label{figure:introduction:mdp}
  \end{subfigure}
  \hfill
  \begin{subfigure}[t]{0.48\textwidth}
    \centering
    \scalebox{0.8}{
      \begin{tikzpicture}[scale=0.8]
        \draw[-stealth] (-0.15,0) -- (5.5,0) node[yshift=-23] {$\proba(\spath{}{\textsf{work}}\leq{40})$}; 
        \draw[-stealth] (0,-0.15) -- (0,5.5) node[yshift=5] {$\expectancy(\spath{}{\textsf{work}})$};
        \coordinate (oo) at (0, 0);
        \coordinate (1x) at (1,0);
        \coordinate (2x) at (2,0);
        \coordinate (3x) at (3,0);
        \coordinate (4x) at (4,0);
        \coordinate (5x) at (5,0);
        
        \coordinate (1y) at (0,1);
        \coordinate (2y) at (0,2);
        \coordinate (3y) at (0,3);
        \coordinate (4y) at (0,4);
        \coordinate (5y) at (0,5);

        \node[xshift=-10] at (oo) {$25$};
        \node[stochastics] at (1y) (qy){};
        \node[xshift=-10] at (1y) {$26$};
        \node[stochastics] at (2y) (qy){};
        \node[xshift=-10] at (2y) {$27$};
        \node[stochastics] at (3y) (qy){};
        \node[xshift=-10] at (3y) {$28$};
        \node[stochastics] at (4y) (qy){};
        \node[xshift=-10] at (4y) {$29$};
        \node[stochastics] at (5y) (qy){};
        \node[xshift=-10] at (5y) {$30$};
        
        \node[yshift=-10] at (oo) {$0.5$};
        \node[stochastics] at (1x) (qx){};
        \node[yshift=-10] at (1x) {$0.6$};
        \node[stochastics] at (2x) (qx){};
        \node[yshift=-10] at (2x) {$0.7$};
        \node[stochastics] at (3x) (qx){};
        \node[yshift=-10] at (3x) {$0.8$};
        \node[stochastics] at (4x) (qx){};
        \node[yshift=-10] at (4x) {$0.9$};
        \node[stochastics] at (5x) (qx){};
        \node[yshift=-10] at (5x) {$1$};

\coordinate (v0) at (3.6651611328125, 0.0);
        \coordinate (v1) at (5.0, 2.8125); \coordinate (v2) at (5.0, 5); \coordinate (v3) at (0.78125, 2.109375);
        \coordinate (v4) at (3.6651611328125, 0.0);

\coordinate (p1) at (5.0, 3.75);
\coordinate (p4) at (1.8359375, 1.58203125);
        \coordinate (p5) at (2.626953125, 1.1865234375);
        \coordinate (p6) at (3.22021484375, 0.889892578125);
        \coordinate (p7) at (3.6651611328125, 0.66741943359375);
        \coordinate (p8) at (3.6651611328125, 0.5005645751953125);
        \coordinate (p9) at (3.6651611328125, 0.3754234313964844);

        \coordinate (mix) at (4.5728515625, 1.9124999999999979); 

        \fill[red!20] (v0) --(v1) -- (v2) -- (v3) -- (v0);
        \draw[red] (v0) --(v1) -- (v2) -- (v3) -- (v0);
        
        \begin{scope}[red]
          \node[stochasticc] at (v0) {};
          \node[xshift=18, yshift=6] at (v0) {$\stratMDP_{\texttt{train}}$};
          \node[stochasticc] at (v2) {};
          \node[yshift=10,xshift=5,red] at (v2) {$\stratMDP_{\texttt{bike}}$};
          
          \node[stochasticc] at (p1) {};
          \node[xshift=15] at (p1) {$\stratMDP_{1\texttt{t+b}}$};
          \node[stochasticc] at (v1) {};
          \node[xshift=15] at (v1) {$\stratMDP_{2\texttt{t+b}}$};
          \node[stochasticc] at (v3) {};
          \node[yshift=10,xshift=-5] at (v3) {$\stratMDP_{3\texttt{t+b}}$};
          \node[stochasticc] at (p4) {};
          \node[yshift=7] at (p4) {$\stratMDP_{4\texttt{t+b}}$};
          
          \node[stochasticc] at (p5) {};
          \node[yshift=7] at (p5) {$\stratMDP_{5\texttt{t+b}}$};
\end{scope}

        \begin{scope}[blue]
          \node[stochasticc] at (mix) {};
          \node[xshift=-13, yshift=7] at (mix) {$\stratMDP_{\texttt{mix}}$};
        \end{scope}
      \end{tikzpicture}
    }
    \caption{A set of expected payoffs.
      The probability of reaching work before $40$ time units is given on the first dimension.
      The other dimension represents the expected time to reach work.}\label{figure:introduction:payoffs}
  \end{subfigure}
  \caption{An MDP with two payoffs and its associated set of expected payoffs.}\label{figure:introduction}
\end{figure*}
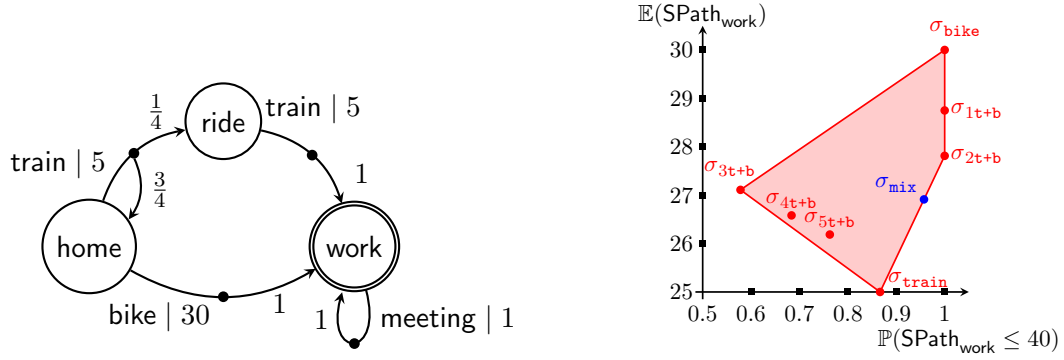
We illustrate the set of expected payoffs for this situation in Figure~\ref{figure:introduction:payoffs}.
We label expected payoffs with their corresponding strategies.
On the one hand, $\stratMDP_{\mathsf{train}}$ and $\stratMDP_{\mathsf{bike}}$ denote the pure strategies that always choose the action in the subscript.
On the other hand, we denote by $\stratMDP_{\indexPosition\mathsf{t+b}}$ the strategy that attempts to take the train $\indexPosition$ times before choosing the bicycle.

This simple example highlights the need for randomisation and memory in multi-objective MDPs and POMDPs.
When limited to pure strategies, for instance, it is not possible to reach work on time with probability exceeding $90\%$ while guaranteeing an expected commute time lower than $27$: any pure strategy whose payoff is not represented will reach work on time with a smaller probability than $\stratMDP_{\mathsf{train}}$, and thus will not be satisfactory.
However, as suggested by the point labelled by $\stratMDP_{\mathsf{mix}}$ on the illustration, these constraints can be satisfied by mixing $\stratMDP_{\mathsf{train}}$ with $\stratMDP_{2\mathsf{t+b}}$.
In fact, here, all expected payoffs can be obtained by finite-support mixed strategies.
We explain below how this property generalises.
We remark however that, in general, the set of expected payoffs need not be a convex polytope like in this example even in a finite MDP (see Section~\ref{section:running}).

\subparagraph*{Our contributions.}
We study the structure of sets of expected payoffs in countable-state multi-objective POMDPs, focusing on the relationship between what can be obtained with pure and with randomised strategies.
Our goal is to provide \textit{general results}, i.e., that apply to a broad class of payoffs.
We only consider \textit{universally unambiguously integrable} payoffs (Section~\ref{section:payoffs}), i.e., payoffs that have a well-defined (possibly infinite) expectation no matter the strategy, a natural requirement when aiming to reason about expectations.
We obtain finer results for \textit{universally integrable} payoffs, i.e., payoffs whose expectation is \textit{finite} under all strategies.
We show that our results for universally integrable payoffs cannot be extended to universally unambiguously integrable payoffs via \textit{finite MDPs}.

First, we prove that for universally integrable (multi-dimensional) payoffs, for all strategies, there exists a \textit{pure strategy} with a greater or equal expected payoff in the lexicographic sense (Theorem~\ref{thm:lexico:pure}).

This first result serves as a building block to one of our main results: any expected payoff vector of a universally integrable (multi-dimensional) payoff is a \textit{convex combination} of expected payoffs of pure strategies (Theorem~\ref{thm:mixing:exact}).
From a strategic perspective, this means that in multi-objective POMDPs, \textit{finite-support mixed strategies} suffice to exactly obtain any (Pareto-optimal) expected payoff vector.
As a corollary, we obtain that any extreme point of a set of expected payoffs of a universally integrable payoff can be obtained by a pure strategy (Corollary~\ref{cor:mixing:exact:extreme}).
These results generalise known properties for classical combinations of objectives for which the set of expected payoffs is a convex polytope (e.g.,~combinations of $\omega$-regular specifications~\cite{DBLP:journals/lmcs/EtessamiKVY08} and universally integrable total-reward payoffs~\cite{DBLP:conf/atva/ForejtKP12}).
While none of the previous properties generalise to the whole class of universally unambiguously integrable payoffs (Examples~\ref{ex:lexico:general} and~\ref{ex:mixing:approx}), we prove that, for such payoffs, convex combinations of expected payoffs of pure strategies can be used to \textit{approximate} any expected payoff (Theorem~\ref{thm:mixing:approx}).

In both cases, we can bound the number of strategies that are mixed depending on the number of dimensions $\numObj$.
Depending on the setting, we can match or approximate expected payoffs by mixing no more than $\numObj+1$ strategies (Theorem~\ref{thm:mixing:support:all}).

We also provide a sufficient condition to guarantee that the set of expected payoffs is closed.
We show that for continuous universally \textit{square integrable} payoffs (which generalise \textit{real-valued continuous payoffs}), the set of expected payoffs is closed (Theorem~\ref{thm:continuous:closed}).

To prove our results, we borrow techniques from several fields.
We mainly rely on topology, properties of convex sets (separating and supporting hyperplanes) and the theory of the Lebesgue integral.
We assume familiarity with basic properties of the Lebesgue integral, and otherwise endeavour to keep the text self-contained by recalling most required notions, notably through an extensive appendix.

\subparagraph*{Applicability.} The class of \textit{universally integrable} payoffs is large: it contains most classical payoffs. Indeed, all \textit{bounded} payoffs are de facto in this class. For example, all indicators of objectives are in it. This means that all settings where one considers the \textit{probabilities} of sets of plays (i.e., either an inherently qualitative objective or one arising from fixing a threshold for a quantitative payoff) are in it, for \textit{any} definition of such sets of plays. Classical examples include $\omega$-regular objectives~\cite{DBLP:journals/lmcs/EtessamiKVY08}, window objectives~\cite{DBLP:journals/lmcs/BrihayeDOR20} or percentiles queries~\cite{DBLP:journals/fmsd/RandourRS17}. Bounded payoffs also encompass discounted-sum~\cite{Sha53} and mean-payoff~\cite{lmcs:1156,DBLP:journals/lmcs/ChatterjeeKK17} functions. Heterogeneous combinations, such as, e.g., combinations of energy and mean-payoff~\cite{DBLP:conf/concur/BruyereHRR19} also fit under this umbrella.

Payoffs that are completely out of the scope of our results---i.e., not even universally unambiguously integrable---include total payoff~\cite{gog23}, average-energy~\cite{DBLP:journals/acta/BouyerMRLL18} and shortest path~\cite{DBLP:journals/fmsd/RandourRS17}, all with both positive and negative weights: the expectation is ill-defined when plays of positive and negative infinite payoffs have non-zero measure. Let us focus on shortest-path objectives. When only positive weights are used, a classical restriction~\cite{DBLP:journals/fmsd/RandourRS17}, the associated payoff is universally unambiguously integrable (but not universally integrable as plays of infinite payoffs may exist). Finally, when the state and action spaces are finite and targets are almost-surely visited (as in the above example), the payoff is continuous and universally square integrable (Appendix~\ref{appendix:square:shortest path}).

For many payoffs that are at least universally unambiguously integrable, extreme points of the set of expected payoffs can be obtained via pure finite-memory strategies (e.g., $\omega$-regular objectives~\cite{DBLP:journals/lmcs/EtessamiKVY08} or mean-payoff~\cite{lmcs:1156}). It follows that for these objectives, mixing over \textit{finitely many pure finite-memory strategies} is sufficient to fulfil any achievability requirement.
In other words, one of the least expressive models of randomised finite-memory strategies (see the classification of~\cite{DBLP:journals/iandc/MainR24}) suffices.
Furthermore, the blow-up in memory from this mixing argument is small: it suffices to mix, at most, one more strategy than the number of payoffs.

Our results also highlight the role of randomisation in strategies for multi-objective POMDPs: it is useful \textit{only} to balance the payoffs on different dimensions in many cases.

\subparagraph*{Related work.}
We provide a few references, in addition to the main references cited previously.
In our proof of Theorem~\ref{thm:mixing:exact}, we invoke the hyperplane separation theorem.
This theorem also plays a role in approximation schemes of the set of achievable vectors: see, e.g.,~\cite{DBLP:conf/atva/ForejtKP12,DBLP:conf/tacas/QuatmannK21}; a unifying approach is presented in~\cite{DBLP:phd/dnb/Quatmann23}.

Specifications with multiple objectives have also been considered in the context of two-player (non-stochastic) games on graphs (e.g.,~\cite{FH13,DBLP:journals/acta/ChatterjeeRR14}) and two-player stochastic games (e.g.,~\cite{DBLP:conf/mfcs/ChenFKSW13,DBLP:conf/lics/AshokCKWW20}).
Closely related to multi-objective specifications are approaches that provide guarantees simultaneously in the worst case and the expected case~\cite{DBLP:journals/iandc/BruyereFRR17,DBLP:conf/aaai/Chatterjee0PRZ17} or expectation optimisation among strategies that satisfy some probabilistic constraint~\cite{DBLP:conf/ijcai/ChatterjeeE0R18}.

Regarding randomisation in strategies,~\cite{DBLP:journals/iandc/Chatterjee0GH15} studies when randomisation is helpful in strategies or in transitions of games.
In particular, the authors show that if there exists an optimal strategy to maximise the probability of an event in an POMDP, then there exists a pure optimal strategy.
This property is generalised by our result on lexicographic POMDPs.

\subparagraph*{Outline.}
Section~\ref{section:prelim} introduces notation and preliminary notions.
We introduce payoff functions in Section~\ref{section:payoffs} and establish relationships between expected payoffs of general strategies and expected payoffs of pure strategies.
In Section~\ref{section:running}, we provide an example highlighting the potential complexity of sets of expected payoffs in multi-objective POMDPs.
We study randomisation requirements for the lexicographic optimisation of several objectives in POMDPs in Section~\ref{section:lexico}.
Section~\ref{section:achievable} presents our results regarding expected payoff sets and achievable sets.
Finally, we study continuous payoff functions in Section~\ref{section:continuous}.

Appendices~\ref{appendix:prelim},~\ref{appendix:plays:topology} and~\ref{appendix:continuous payoffs} complement Sections~\ref{section:prelim} and~\ref{section:payoffs}.
Appendix~\ref{appendix:square:shortest path} shows that the results of Section~\ref{section:continuous} apply to universally integrable continuous shortest-path payoffs.

\section{Preliminaries}\label{section:prelim}

This section defines the objects occurring in the main part of this paper.
Additional definitions and results have been deferred to Appendix~\ref{appendix:prelim}.

\subparagraph*{Set-theoretic notation.}
We let $\IN$ and $\IR$ respectively denote the sets of natural and real numbers.
We let $\IN_0=\IN\setminus\{0\}$.
We denote the extended real line by $\IRbar = \IR\cup\{-\infty,+\infty\}$.

Let $A'\subseteq A$ and $B'\subseteq B$ be sets and $f\colon A\to B$ be a function.
We let $\indic{A'}\colon A\to \{0, 1\}$ denote the indicator function of $A'$.
We let $\image{f}$ denote the image of $f$.
We let $f^{-1}(B') = \{a\in A\mid f(a)\in B'\}$ denote the inverse image of $B'$ by $f$.
For any $b\in B$, we write $f^{-1}(b)$ instead of $f^{-1}(\{b\})$ to lighten notation.
We let $|A|$ denote the cardinal of $A$.

\subparagraph*{Probability.}
Let $A$ be a countable set.
We write $\dist{A}$ for the set of distributions over $A$, i.e., the set of functions $\distMeasure\colon A\to\ccInt{0}{1}$ such that $\sum_{a\in A}\distMeasure(a) = 1$.
The support of a distribution $\distMeasure\in\dist{A}$ is $\supp{\distMeasure} = \{a\in A\mid \distMeasure(a)> 0\}$.

Given a set $B$ and a sigma-algebra $\sigmaAlgebra$ over $B$, we denote by $\dist{B, \sigmaAlgebra}$ the set of probability distributions over the measurable space $(B, \sigmaAlgebra)$.
Let $\mu\in\dist{B, \sigmaAlgebra}$ and $\payoff\colon B\to\IRbar$ be a measurable function.
We say that $\payoff$ is $\mu$-integrable if it is integrable with respect to $\mu$, i.e., if $\int_{B}|\payoff|\ud\mu\in\IR$.
We extend the Lebesgue integral to non-positive functions in the following way: if $\payoff$ is non-positive, we let $\int_{B}\payoff\ud\mu=-\int_{B}-\payoff\ud\mu$.
If $\payoff$ is non-negative, non-positive or $\mu$-integrable, we say that $\int_{B}\payoff\ud\mu$ is the $\mu$-integral of $\payoff$.

\subparagraph*{Vector spaces.}
Throughout this text, vectors are written in boldface to distinguish them from scalars.
Let $\numObj\in\IN_0$.
We let $\zeroVectDim{\numObj}$ and $\oneVectDim{\numObj}\in\IR^\numObj$ be the vectors of $\IR^\numObj$ where all components are zero and one respectively.
We omit the dimension subscript when the dimension is clear from the context.

Given $\vect = (\vectComp_\indexPayoff)_{1\leq\indexPayoff\leq\numObj}, \vectB = (\vectCompB_\indexPayoff)_{1\leq\indexPayoff\leq\numObj}\in\IR^\numObj$, we let $\scalarProd{\vect}{\vectB} = \sum_{\indexPayoff=1}^\numObj \vectComp_\indexPayoff\vectCompB_\indexPayoff$ denote the scalar product of $\vect$ and $\vectB$.
We let $\|\!\cdot\!\|_2$ denote the Euclidean norm  on $\IR^\numObj$, defined by $\|\vect\|_2 = \sqrt{\scalarProd{\vect}{\vect}}$ for all $\vect\in\IR^\numObj$.

The affine span of a set $D\subseteq\IR^\numObj$, which we denote by $\spanAff{D}$, is the smallest affine set (i.e., translation of a vector subspace of $\IR^\numObj$) in which $D$ is included.

Given a linear map $\linMap\colon\IR^\numObj\to\IR^{\numObj'}$ (where $\numObj'\in\IN_0$), we let $\ker{\linMap}$ denote the kernel of $\linMap$, i.e., the set $\{\vect\in\IR^\numObj\mid\linMap(\vect) = \zeroVectDim{\numObj'}\}$.
A \textit{linear form} is a linear map whose co-domain is $\IR$.

\subparagraph*{Convexity.}
A \textit{convex combination} of vectors $\vect_1$, \ldots, $\vect_\indexSequence\in\IR^\numObj$ is a linear combination $\sum_{\indexSequenceB=1}^\indexSequence\scalar_\indexSequenceB\cdot\vect_\indexSequenceB$ such that $\scalar_1$, \ldots, $\scalar_\indexSequence\in\ccInt{0}{1}$ and $\sum_{\indexSequenceB=1}^\indexSequence\scalar_\indexSequenceB = 1$.
We refer to a sequence of coefficients $\scalar_1$, \ldots, $\scalar_\indexSequence\in\ccInt{0}{1}$ such that $\sum_{\indexSequenceB=1}^\indexSequence\scalar_\indexSequenceB = 1$ as \textit{convex combination coefficients}.
Given $\vect$, $\vectB\in\IR^\numObj$, we let $\ccInt{\vect}{\vectB} = \{\scalar\cdot\vect + (1-\scalar)\vectB\mid\scalar\in\ccInt{0}{1}\}$ denote the (closed) segment from $\vect$ to $\vectB$; it is the set of convex combinations of $\vect$ and $\vectB$.
Open and half-open segments are defined analogously.

Let $D\subseteq\IR^\numObj$.
The \textit{convex hull} of $D$, denoted by $\convex{D}$, is the set of all convex combinations of elements of $D$.
The set $D$ is \textit{convex} if for all $\vect$, $\vectB\in D$, $\ccInt{\vect}{\vectB}\subseteq D$, or, equivalently, if $D = \convex{D}$.
If $D$ is convex, we say that $\payoffVect\in D$ is an \textit{extreme point} of $D$ if $\payoffVect\notin\convex{D\setminus\{\payoffVect\}}$, i.e., if $\payoffVect$ is not a convex combination of elements of $D$ other than $\payoffVect$ and we let $\corners{D}$ denote the set of extreme points of $D$.
Extreme points generalise the notion of vertices of polytopes.

The definition of a convex combination does not bound the number of involved vectors.
However, in $\IR^\numObj$, it is sufficient to only consider convex combinations involving no more than $\numObj+1$ vectors.
This is formalised by the following theorem. \begin{theorem}[Carath\'{e}odory's theorem for convex hulls (e.g.,~{\cite[Thm.~17.1]{DBLP:books/degruyter/Rockafellar70}})]\label{theorem:caratheodory:convex}
  Let $D\subseteq\IR^\numObj$ and $\payoffVect\in\convex{D}$.
  There exists $D'\subseteq D$ such that $|D'|\leq\numObj+1$ and $\payoffVect\in\convex{D'}$.
\end{theorem}

Convexity does not generalise to $\IRbar^\numObj$.
Most notably, convex combinations of vectors of $\IRbar^\numObj$ (defined in the same way as above) may be ill-defined.
Although we consider convex combinations of elements of $\IRbar^\numObj$ in the sequel, these are always guaranteed to be well-defined, for reasons explained below.

\subparagraph*{Topology.}
We only provide notation in this section.
For convenience, we recall some definitions in Appendix~\ref{appendix:prelim:topology}, including the definitions of the product topology, continuity, compactness and the usual topology of $\IRbar$.

Let $(\topSpace, \topology)$ be a Hausdorff topological space.
For all $D\subseteq\topSpace$, we let $\closure{D}$ and $\interior{D}$ denote the closure and interior of $D$.
The boundary of $D\subseteq\topSpace$ is the set $\border{D}= \closure{D}\setminus\interior{D}$.

In addition to the three notions above, we introduce an additional one for subsets of $\IR^\numObj$ (with its usual topology).
The \textit{relative interior} of a set $D\subseteq\IR^\numObj$, denoted by $\relInt{D}$, is the interior of $D$ as a subset of $\spanAff{D}$ (with the induced topology).
The relative interior of a subset of $\IR^\numObj$ includes its interior and these sets may differ.
For instance, the segment $\ccInt{\zeroVectDim{2}}{\oneVectDim{2}}\subseteq\IR^2$ has empty interior.
However, we have $\relInt{\ccInt{\zeroVectDim{2}}{\oneVectDim{2}}} = \ooInt{\zeroVectDim{2}}{\oneVectDim{2}}$ (because $\spanAff{\ccInt{\zeroVectDim{2}}{\oneVectDim{2}}}$ is the line of equation $x = y$).

\subparagraph*{Hyperplane separation.}
A \textit{hyperplane} $\hplane$ of $\IR^\numObj$ is a set of the form $\{\vect\in\IR^\numObj\mid \linForm(\vect) = \scalar\}$ for some non-zero linear form $\linForm$ and $\scalar\in\IR$.
Let $D_1$ and $D_2\subseteq\IR^\numObj$.
The sets $D_1$ and $D_2$ are \textit{strongly separated} by a hyperplane if there exists a non-zero linear form $\linForm$ such that $\inf_{\payoffVect\in D_1}\linForm(\payoffVect)>\sup_{\payoffVectB\in D_2}\linForm(\payoffVectB)$.
A convex set $D\subseteq\IR^\numObj$ is \textit{supported} by a hyperplane at $\payoffVect\in D$ if there exists a non-zero linear form $\linForm$ such that, for all $\payoffVectB\in D$, $\linForm(\payoffVectB)\leq\linForm(\payoffVect)$ (a \textit{supporting hyperplane} in this case is $\hplane = \{\vect\in\IR^\numObj\mid\linForm(\vect) = \linForm(\payoffVect)\}$).
We provide an illustration of the notions of separating and supporting hyperplanes in Figure~\ref{figure:prelim:separating and supporting}.

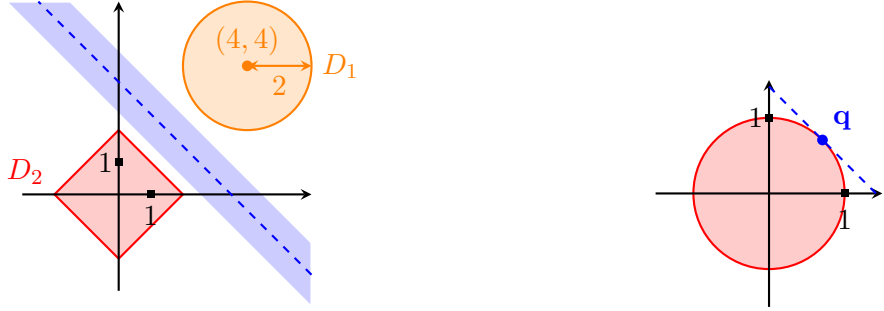
\begin{figure}
  \begin{subfigure}[t]{0.48\textwidth}
    \centering
    \begin{tikzpicture}[scale = 0.85]
      \draw[-stealth] (-1.5,0) -- (3, 0); \draw[-stealth] (0,-1.5) -- (0,3); 

      \coordinate (og) at (0,0);
      \coordinate (x) at (0.5,0);
      \coordinate (y) at (0,0.5);
      \coordinate (py) at (0, 1);
      \coordinate (px) at (1, 0);
      \coordinate (mx) at (-1, 0);
      \coordinate (my) at (0, -1);
      \coordinate (ox) at (1,0);
      \coordinate (m) at (3, 2);
      \coordinate (c) at (2, 2);

      \node[right, orange] at (m) {$D_1$};
      \node[above left, red] at (mx) {$D_2$};

      \node[stochasticc, orange] at (c) {};
      \node[above, orange] at (c) {$(4, 4)$};
      \draw[stealth-stealth, orange] (c) -- (m) node[midway, below] {$2$};

      \node[xshift=-5] at (y) {$1$};
      \node[yshift=-8] at (x) {$1$};
      \node[stochastics] at (x) (qx){};
      \node[stochastics] at (y) (qy){};
      
      \coordinate (h1) at (3, -1.25); \coordinate (h2) at (-1.25, 3);
      \coordinate (a1) at (3, -0.75); \coordinate (a2) at (-0.75, 3); 
      \coordinate (b1) at (3, -1.75); \coordinate (b2) at (-1.75, 3);
      
      \draw[-, thick, dashed, color=blue]  (h1) -- (h2);

      \begin{scope}[on background layer]
        \filldraw[red, fill=red!20,thick] (px) -- (py) -- (mx) -- (my) -- cycle;
        \filldraw[white, fill=blue!20] (a1) -- (a2) -- (b2) -- (b1) -- cycle;
        \filldraw[orange, fill=orange!20,thick] (m) arc[start angle=0, end angle=360, radius = 1cm] -- cycle;
      \end{scope}
      
    \end{tikzpicture}
    \caption{The dashed blue line ($x + y = \frac{7}{2}$) strongly separates $D_1$ and $D_2$: they lie on a different side of the line and there is a strip around the line that is disjoint from $D_1$ and $D_2$.}\label{figure:prelim:separating hyperplane}
  \end{subfigure}\hfill \begin{subfigure}[t]{0.48\textwidth}
    \centering
    \scalebox{1}{
      \begin{tikzpicture}
        \draw[-stealth] (-1.5,0) -- (1.5,0); \draw[-stealth] (0,-1.5) -- (0,1.5); \coordinate (og) at (0, 0);
        \coordinate (x) at (1, 0);
        \coordinate (y) at (0, 1);
        \coordinate (m) at (45:1);

        \node[xshift=-5] at (y) {$1$};
        \node[yshift=-10] at (x) {$1$};
        \node[stochastics] at (x) (qx){};
        \node[stochastics] at (y) (qy){};

        \begin{scope}[on background layer]
          \filldraw[red, fill=red!20,thick] (x) arc[start angle=0, end angle=360, x radius=1cm, y radius = 1cm] -- cycle;
        \end{scope}
        
        \draw[dashed, blue] (0, {sqrt(2)}) -- ({sqrt(2)}, 0);
        \node[stochasticc, blue] at (m) {};
        \node[above right, blue] at (m) {{$\payoffVect$}};
      \end{tikzpicture}
    }
    \caption{The dashed blue line ($x + y = \sqrt{2}$)  supports the unit ball for the Euclidean norm in $\IR^2$ at $\payoffVect = \frac{\sqrt{2}}{2}\oneVect$.}\label{figure:prelim:supporting hyperplane}
  \end{subfigure}
  \caption{Illustration of separating and supporting hyperplanes.}\label{figure:prelim:separating and supporting}
\end{figure}

We recall a variant of the hyperplane separation theorem and the supporting hyperplane theorem.
We first outline a sufficient condition such that two disjoint convex sets can be strongly separated.
\begin{theorem}[Hyperplane separation theorem~{\cite[Cor.~11.4.2]{DBLP:books/degruyter/Rockafellar70}}]\label{thm:hyperplane:separation}
  Let $D_1$ and $D_2$ be two convex subsets of $\IR^\numObj$.
  If $\closure{D_1}\cap\closure{D_2}=\emptyset$ and $D_1$ or $D_2$ is bounded, then there exists a hyperplane strongly separating $D_1$ and $D_2$.
\end{theorem}
Figure~\ref{figure:prelim:separating hyperplane} illustrates a setup in which we can apply the theorem.
We often apply Theorem~\ref{thm:hyperplane:separation} in the case where $D_1$ is a singleton set $\{\payoffVect\}$ (it is bounded) and $\payoffVect\notin\closure{D_2}$.
The next theorem provides a sufficient condition for the existence of a supporting hyperplane at a given point of a convex set.
\begin{theorem}[Supporting hyperplane theorem~{\cite[Thm.~11.6]{DBLP:books/degruyter/Rockafellar70}}]\label{thm:hyperplane:supporting}
  Let $D\subseteq\IR^\numObj$ be convex and $\payoffVect\in D$.
  If $\payoffVect\notin\relInt{D}$, then there exists a hyperplane $\hplane$ supporting $D$ at $\payoffVect$ such that $D\nsubseteq\hplane$.
\end{theorem}

\subparagraph*{Ordering vectors.}
We consider two order relations on $\IRbar^\numObj$: the component-wise order and the lexicographic order.
Let $\payoffVect = \payoffVectVerbose$ and $\payoffVectB = \payoffVectBVerbose\in\IRbar^\numObj$.
For the component-wise order, we write $\payoffVect\leq \payoffVectB$ if and only if $\payoffComp_\indexPayoff\leq\payoffCompB_\indexPayoff$ for all $1\leq \indexPayoff\leq\numObj$.
For the lexicographic ordering over $\IR^\numObj$, we write $\payoffVect\leLex \payoffVectB$ if and only if $\payoffVect = \payoffVectB$ or $\payoffComp_\indexPayoff\leq\payoffCompB_\indexPayoff$ where $\indexPayoff = \min\{\indexPayoff'\leq\numObj\mid \payoffComp_{\indexPayoff'}\neq\payoffCompB_{\indexPayoff'}\}$.
We write $\payoffVect\lLex \payoffVectB$ if $\payoffVect\leLex\payoffVectB$ and $\payoffVect\neq\payoffVectB$.
We recall that the component-wise order is partial, whereas the lexicographic order is a total order.

Let $D\subseteq\IRbar^\numObj$.
We say that $\payoffVect\in D$ is a \textit{Pareto-optimal} element of $D$ if it is maximal for the component-wise order, i.e., if there does not exist $\payoffVectB\in D$ such that $\payoffVect\leq\payoffVectB$ and $\payoffVect\neq\payoffVectB$.
We say that $D$ is \textit{downward-closed} if for all $\payoffVect\in D$ and $\payoffVectB\in\IRbar^\numObj$, $\payoffVectB\leq\payoffVect$ implies $\payoffVectB\in D$.
We let $\down{D}$ denote the \textit{downward closure} of $D$, which is defined as the smallest (with respect to set inclusion) downward-closed set in which $D$ is included.
A set and its downward closure have the same set of Pareto-optimal elements.

\subparagraph*{Partially observable Markov decision processes.}
A \textit{partially observable Markov decision process} (POMDP) models the interaction of a system that is imperfectly informed of the state of the world with a stochastic environment.
At each step of the process, the system receives an observation of the current state of the process then selects an action, and the state of the process is updated in a stochastic fashion, where the relevant distribution depends only on the current state and chosen action.
We consider systems with \textit{countable} state and action spaces.

Formally, a \textit{POMDP} is a tuple $\pomdp = \pomdpTuple$ where $\mdpStateSpace$ is a countable non-empty set of states, $\mdpActionSpace$ is a countable non-empty set of actions, $\mdpTrans\colon \mdpStateSpace\times\mdpActionSpace\to\dist{\mdpStateSpace}$ is a partial probabilistic transition function, $\obsSpace$ is a countable set of observations and $\obsFun\colon\mdpStateSpace\to\obsSpace$ is an observation function.
We say that $\pomdp$ is finite if $\mdpStateSpace$ and $\mdpActionSpace$ are finite.
For $\mdpState\in\mdpStateSpace$, we let $\mdpActionSpace(\mdpState)$ denote the set of actions $\mdpAction\in\mdpActionSpace$ such that $\mdpTrans(\mdpState, \mdpAction)$ is defined.
If $\mdpAction\in\mdpActionSpace(\mdpState)$, we say that $\mdpAction$ is \textit{enabled} in $\mdpState$.
We assume that MDPs have no deadlocks, i.e., there is an action enabled in each state.
We also require that any two states that share the same observation have the same enabled actions, i.e., for all $\mdpState$, $\mdpStateB\in\mdpStateSpace$, $\obsFun(\mdpState) = \obsFun(\mdpStateB)$ implies $\mdpActionSpace(\mdpState) = \mdpActionSpace(\mdpStateB)$.

A \textit{play} of $\pomdp$ is an infinite sequence $\play = \mdpState_0\mdpAction_0\mdpState_1\mdpAction_1\ldots\in(\mdpStateSpace\mdpActionSpace)^\omega$ such that $\mdpAction_\indexPosition\in \mdpActionSpace(\mdpState_\indexPosition)$ and $\mdpState_{\indexPosition+1}\in\supp{\mdpTrans(\mdpState_{\indexPosition}, \mdpAction_{\indexPosition})}$ for all $\indexPosition \in\IN$.
A \textit{history} is a finite prefix of a play ending in a state.
Let $\play = \mdpState_0\mdpAction_0\mdpState_1\mdpAction_1\ldots$ be a play.
We write $\playPrefix{\play}{\indexLast}$ for the history $\mdpState_0\mdpAction_0\mdpState_1\ldots \mdpAction_{\indexLast-1}\mdpState_\indexLast$.
We also use this notation for prefixes of histories.
We denote the suffix $\mdpState_\indexLast\mdpAction_\indexLast\mdpState_{\indexLast+1}\mdpAction_{\indexLast+1}\ldots$ of $\play$ by $\playSuffix{\play}{\indexLast}$.
For any history $\hist = \mdpState_0\mdpAction_0\mdpState_1\ldots \mdpState_\indexLast$, we let $\last{\hist} = \mdpState_\indexLast$.
We write $\playSet{\mdp}$ and $\histSet{\mdp}$ for the sets of plays and histories of $\mdp$.
We extend the observation function $\obsFun$ to histories as follows: for all $\hist = \mdpState_0\mdpAction_0\mdpState_1\ldots \mdpState_\indexLast$, we let $\obsFun(\hist) = \obsFun(\mdpState_0)\mdpAction_0\obsFun(\mdpState_1)\ldots\obsFun(\mdpState_\indexLast)$.

A \textit{Markov decision process} (MDP) is a POMDP in which the observation function is the identity function.
We denote MDPs as tuples $\mdp = \mdpTuple$, i.e., we drop the observation space and observation function from the notation.

We fix a POMDP $\pomdp = \pomdpTuple$ for the rest of the section.

\subparagraph*{Behavioural strategies.}
A strategy is a function that describes how to select actions based on the observation history of the ongoing play.
A strategy may resort to randomisation.
Formally, a \textit{behavioural strategy} of $\pomdp$ is a function $\stratMDP\colon\obsFun(\histSet{\pomdp})\to\dist{\mdpActionSpace}$ such that for all $\hist\in\histSet{\pomdp}$, $\supp{\stratMDP(\obsFun(\hist))}\subseteq\mdpActionSpace(\last{\hist})$.
To lighten notation, we abusively write $\stratMDP(\hist)$ for $\stratMDP(\obsFun(\hist))$ in the sequel.

A strategy is \textit{pure} if it uses no randomisation, i.e., if it maps all observation histories to Dirac distributions.
We view pure strategies as functions  $\stratMDP\colon\obsFun(\histSet{\pomdp})\to \mdpActionSpace$.
A strategy $\stratMDP$  is \textit{memoryless} if the distribution it assigns to a history depends only on its last state, i.e., if for all histories $\hist, \hist'\in\histSet{\pomdp}$, $\last{\hist} = \last{\hist'}$ implies $\stratMDP(\hist) = \stratMDP(\hist')$.
Memoryless (resp.~pure memoryless) strategies can be viewed as functions $\obsSpace\to \dist{\mdpActionSpace}$ (resp.~$\obsSpace\to\mdpActionSpace$).
We let $\stratClassAll{\pomdp}$ and $\stratClassPure{\pomdp}$ respectively denote the set of all strategies of $\pomdp$ and the set of pure strategies of $\pomdp$.

A play $\play = \mdpState_0\mdpAction_0\mdpState_1\ldots$ is \textit{consistent} with a strategy $\stratMDP$ if for all $\indexPosition\in\IN$, we have $\mdpAction_\indexPosition\in\supp{\stratMDP(\playPrefix{\play}{\indexPosition})}$.
An \textit{outcome} of a strategy is any play consistent with the strategy.
Consistency of a history with a strategy is defined similarly.

We now define the distribution over plays induced by a strategy from an initial state.
First, we specify the relevant sigma-algebra.
A \textit{cylinder set} is a set of continuations of a history; given $\hist = \mdpState_0\mdpAction_0\mdpState_1\ldots\mdpState_\indexLast\in\histSet{\pomdp}$, we let $\cyl{\hist} = \{\play\in\playSet{\pomdp}\mid \playPrefix{\play}{\indexLast} = \hist\}$.
We let $\pomdpSigmaAlgebra$ denote the sigma-algebra generated by the cylinder sets.
Let $\stratMDP\in\stratClassAll{\pomdp}$ and $\mdpState\in\mdpStateSpace$.
We define $\proba_{\mdpState}^{\stratMDP}\in\dist{\playSet{\pomdp}, \pomdpSigmaAlgebra}$ over cylinder sets as follows.
For $\hist = \mdpState_0\mdpAction_0\mdpState_1\ldots\mdpState_\indexLast\in\histSet{\pomdp}$, if $\mdpState_0 = \mdpState$, we let
\[\proba_{\mdpState}^{\stratMDP}(\cyl{\hist}) = \prod_{\indexPosition= 0}^{\indexLast-1}\stratMDP(\playPrefix{\hist}{\indexPosition})(\mdpAction_\indexPosition)\cdot \mdpTrans(\mdpState_\indexPosition, \mdpAction_\indexPosition)(\mdpState_{\indexPosition+1}),\]
and otherwise, if $\mdpState_0\neq\mdpState$, we let $\proba_{\mdpState}^{\stratMDP}(\cyl{\hist}) = 0$.
This distribution extends in a unique fashion to $\pomdpSigmaAlgebra$ by the Ionescu-Tulcea extension theorem~\cite[Thm.~8.24]{Kallenberg2021}.

\subparagraph*{Mixed strategies.}
There exists an alternative definition of randomised strategies: a \textit{mixed strategy} is a distribution over pure strategies.
Intuitively, when playing following a mixed strategy, a pure strategy is selected at random at the start of the process and then the play proceeds according to this pure strategy.

We first introduce a sigma-algebra on $\stratClassPure{\pomdp}$.
We view $\stratClassPure{\pomdp}$ as (a subset of) the product $\mdpActionSpace^{\obsFun(\histSet{\pomdp})}$.
We let $\stratSigmaAlgebra$ be the Borel sigma-algebra of $\stratClassPure{\pomdp}$ seen as the previous product equipped with the product topology, where each copy of $\mdpActionSpace$ is endowed with the discrete topology.
A \textit{mixed strategy} is a distribution in $\dist{\stratClassPure{\pomdp}, \stratSigmaAlgebra}$.
The distribution over $\playSet{\pomdp}$ induced by a mixed strategy $\mixedStrat$ from an initial state $\mdpState\in\mdpStateSpace$ is denoted by $\proba_{\mdpState}^\mixedStrat$ and is defined by
\begin{equation}\label{equation:prelim:mixed distribution}
  \proba_{\mdpState}^\mixedStrat(\objective) = \int_{\stratBMDP\in\stratClassPure{\pomdp}}\proba_{\mdpState}^{\stratBMDP}(\objective)\ud\mixedStrat(\stratBMDP)
\end{equation}
for all measurable $\objective\in\pomdpSigmaAlgebra$.
This definition assumes that the mapping $\stratClassPure{\pomdp}\to\ccInt{0}{1}\colon\stratBMDP\mapsto\proba_{\mdpState}^\stratBMDP(\objective)$ is measurable for all $\mdpState\in\mdpStateSpace$ and $\objective\in\pomdpSigmaAlgebra$.
A proof of this property is provided in Appendix~\ref{appendix:prelim:mixed}.

\subparagraph*{Kuhn's theorem.}
We briefly comment on the expressiveness of behavioural and mixed strategies.
We say that two randomised strategies are \textit{outcome-equivalent} if, from all initial states, they induce the same distributions over plays.
Kuhn's theorem~\cite{Aumann64} states that in the perfect recall setting, for all mixed strategies, there exists an outcome-equivalent behavioural strategy and vice-versa.
Perfect recall holds if decision makers never forget their prior knowledge and can observe their actions.
In our setting, it follows from our definition of POMDPs (actions are observable) and strategies (they take in account the previous actions).

The derivation of mixed strategies from behavioural strategies is non-trivial: mixed strategies that are equivalent to some behavioural strategies may have to randomise over an uncountable set of pure strategies.
Furthermore, even when the behavioural strategy uses finite memory (i.e., has a finite automaton-based representation), an equivalent mixed strategy may have to randomise over infinite-memory strategies~\cite{DBLP:journals/iandc/MainR24}.
In particular, the class of finite-support mixed strategies, which is at the core of our results, is a strict subclass of the class of mixed strategies.

Due to the equivalent expressiveness of both strategy models, mixed and behavioural strategies can be used interchangeably.
For the sake of conciseness, unless otherwise stated, by strategy, we mean a behavioural strategy.

\section{Payoffs and multi-objective POMDPs}\label{section:payoffs}
Payoffs are functions that quantify the quality of plays.
Definitions are given in Section~\ref{section:payoffs:definitions}.
Section~\ref{section:payoffs:pure and randomised} discusses the expectation of a payoff under a mixed strategy.
Some classical payoffs are provided in Section~\ref{section:payoffs:classical}.
Continuous payoffs, which are the main object of study of Section~\ref{section:continuous}, are introduced in Section~\ref{section:payoffs:continuous}.
Section~\ref{section:payoffs:multi} presents notation for POMDPs with multiple payoffs.

We fix a POMDP $\pomdp = \pomdpTuple$ for the whole section.

\subsection{Payoff functions and objectives}\label{section:payoffs:definitions}
A \textit{payoff function} (or payoff for short) is a measurable function $\payoff\colon \playSet{\pomdp}\to\IRbar$.
An \textit{objective} is a measurable subset of the set of plays, i.e., an element of $\pomdpSigmaAlgebra$.
An objective represents a qualitative specification, e.g., reaching a target set of states.
To each objective $\objective$,  we associate the payoff $\indic{\objective}$.

Let $\payoff$ be a payoff, $\stratMDP\in\stratClassAll{\pomdp}$ be a strategy and $\mdpState\in\mdpStateSpace$ be an initial state.
The $\proba^{\stratMDP}_\mdpState$-integral of $\payoff$ is only formally defined whenever $\payoff$ is non-negative, non-positive or $\proba^{\stratMDP}_\mdpState$-integrable.
If $\payoff$ is such a payoff, we let $\expectancy_{\mdpState}^{\stratMDP}(\payoff) = \int_{\play\in\playSet{\pomdp}}\payoff(\play)\ud\proba^{\stratMDP}_{\mdpState}(\play)$; $\expectancy_{\mdpState}^{\stratMDP}(\payoff)$ is the expected payoff of the strategy $\stratMDP$ from $\mdpState$ (for $\payoff$).
We also generalise the notion of expected payoff to a broader class of payoff functions as follows.
Let $\payoff^+ = \max(\payoff, 0)$ and $\payoff^- = \max(-\payoff, 0)$ denote the non-negative and non-positive parts of $\payoff$.
We say that $\payoff$ has an \textit{unambiguous $\proba^{\stratMDP}_\mdpState$-integral} if $\expectancy^{\stratMDP}_\mdpState(\payoff^+)\in\IR$ or $\expectancy^{\stratMDP}_\mdpState(\payoff^-)\in\IR$.
If $\payoff$ has an unambiguous $\proba^{\stratMDP}_\mdpState$-integral, we abuse notation and let $\expectancy^{\stratMDP}_\mdpState(\payoff) = \expectancy^{\stratMDP}_\mdpState(\payoff^+) - \expectancy^{\stratMDP}_\mdpState(\payoff^-)$.
We reserve the notation $\expectancy$ for the expectation of payoffs, and use integrals for other probability spaces.

In the sequel, we only consider payoffs for which the expected payoff is (unambiguously) defined for all strategies from all initial states.
We say that a payoff $\payoff$ is \textit{universally unambiguously integrable} if $\payoff$ has an unambiguous $\proba^{\stratMDP}_{\mdpState}$-integral for all $\stratMDP\in\stratClassAll{\pomdp}$ and $\mdpState\in\mdpStateSpace$.
In particular, all non-negative and non-positive payoffs are universally unambiguously integrable.
In this work, we focus on universally unambiguously integrable payoffs: all payoffs considered from this point on are assumed to be universally unambiguously integrable.

Payoffs that are integrable no matter the strategy and initial state are of particular interest in the sequel.
We say that a payoff $\payoff$ is \textit{universally integrable} if it is $\proba^{\stratMDP}_{\mdpState}$-integrable for all $\stratMDP\in\stratClassAll{\pomdp}$ and $\mdpState\in\mdpStateSpace$.
All bounded functions (in particular indicators) are universally integrable.

Given a payoff function $\payoff$, we say that a strategy $\stratMDP$ is \textit{optimal} from a state $\mdpState\in\mdpStateSpace$ if for all strategies $\stratBMDP$, we have $\expectancy_{\mdpState}^\stratMDP(\payoff) \geq \expectancy_{\mdpState}^\stratBMDP(\payoff)$.
Playing optimally with respect to an indicator thus equates to maximising the probability of the underlying objective.
In spite of this maximisation-oriented viewpoint for optimisation, we can capture minimisation by considering $-\payoff$ in lieu of $\payoff$: $\payoff$ is universally unambiguously integrable if and only if $-\payoff$ is, and minimising a payoff $\payoff$ is equivalent to maximising the payoff $-\payoff$.

\subsection{Relating payoffs of randomised and pure strategies}\label{section:payoffs:pure and randomised}
One of our main goals is to show that a simpler class of randomised strategies, i.e., finite-support mixed strategies, suffices when considering POMDPs with multiple payoffs.
To prove these results, we rely on a generalisation of Equation~\eqref{equation:prelim:mixed distribution}.
Equation~\eqref{equation:prelim:mixed distribution} states for all $\mdpState\in\mdpStateSpace$, the probability of an objective $\objective$ under a mixed strategy $\mixedStrat$ from $\mdpState$ is $\int_{\stratBMDP\in\stratClassPure{\pomdp}}\proba_{\mdpState}^{\stratBMDP}(\objective)\ud\mixedStrat(\stratBMDP)$.
We generalise this to the expectation of a payoff under a mixed strategy.
We prove this by considering payoffs of increasing complexity, starting from the indicator of objectives, analogously to the construction of the Lebesgue integral.
In the following statement, we impose restrictions on $\payoff$ that ensure that we obtain a well-defined integral.

\begin{restatable}{lemma}{lemExpectancyPureIntegral}\label{lem:expectancy:pure integral}
  Let $\mixedStrat$ be a mixed strategy, $\mdpState\in\mdpStateSpace$ and $\payoff$ be a universally unambiguously integrable payoff.
  If $\inf_{\stratBMDP}\expectancy^{\stratBMDP}_\mdpState(\payoff)\geq 0$ or $\sup_{\stratBMDP}\expectancy^{\stratBMDP}_\mdpState(\payoff)\leq 0$ or $\payoff$ is $\proba^{\mixedStrat}_\mdpState$-integrable, then the mapping $\stratClassPure{\pomdp}\to\IRbar\colon \stratBMDP\mapsto\expectancy^{\stratBMDP}_\mdpState(\payoff)$ is measurable and
  \[
    \expectancy^\mixedStrat_\mdpState(\payoff) = \int_{\stratBMDP\in\stratClassPure{\pomdp}}\expectancy^{\stratBMDP}_\mdpState(\payoff)\ud\mixedStrat(\stratBMDP).
  \]
\end{restatable}
\begin{proof}
  Fix a state $\mdpState\in\mdpStateSpace$ and let $\payoff$ be a universally unambiguously integrable payoff.
  Throughout this proof, we use $\stratBMDP$ to (implicitly) denote pure strategies.
  We show the result for payoffs of increasing complexity.
  First, we prove it for indicators of objectives.
  Second, we show that it also holds for non-negative simple functions (i.e., linear combinations of indicators) by linearity of the integral.
  Third, we deal with non-negative payoffs with the monotone convergence theorem.
  Fourth, we consider $\proba^{\mixedStrat}_{\mdpState}$-integrable payoffs.
  Finally, we close the proof by considering universally unambiguously integrable payoffs such that $\inf_{\stratBMDP}\expectancy^{\stratBMDP}_\mdpState(\payoff)\geq 0$ or $\sup_{\stratBMDP}\expectancy^{\stratBMDP}_\mdpState(\payoff)\leq 0$.

  If $\payoff$ is the indicator of an objective, the result follows by definition of the distribution induced by a mixed strategy.
  
  For the second step of the argument, we assume that $\payoff$ is a non-negative simple function.
  Let $\scalar_1, \ldots, \scalar_\indexSequence\geq 0$ and $\objective_1, \ldots, \objective_\indexSequence\subseteq\playSet{\pomdp}$ be objectives such that $\payoff =\sum_{\indexPayoff=1}^\indexSequence\scalar_\indexPayoff\indic{\objective_\indexPayoff}$.
  The function $\stratClassPure{\pomdp}\to\IR\colon\stratBMDP\mapsto\expectancy^{\stratBMDP}_\mdpState(\payoff)$ is measurable: it is a non-negative linear combination of measurable functions by the above.
  It follows from the result for indicators applied for all $1\leq\indexPayoff\leq\indexSequence$ and the linearity of the Lebesgue integral that
  \begin{align*}
    \expectancy^{\mixedStrat}_\mdpState(\payoff)
    & =  \sum_{\indexPayoff=1}^\indexSequence 
      \scalar_\indexPayoff\proba^{\mixedStrat}_\mdpState(\objective_\indexPayoff)
      \\
    & = \sum_{\indexPayoff=1}^\indexSequence\scalar_\indexPayoff
      \int_{\stratBMDP\in\stratClassPure{\pomdp}}\proba^{\stratBMDP}_{\mdpState}(\objective_\indexPayoff)\ud\mixedStrat(\stratBMDP)
    \\
& = \int_{\stratBMDP\in\stratClassPure{\pomdp}}\expectancy^{\stratBMDP}_{\mdpState}(\payoff)\ud\mixedStrat(\stratBMDP).
  \end{align*}

  Third, we assume that $\payoff$ is a non-negative measurable function.
  Let $(\payoff_\indexSequence)_{\indexSequence\in\IN}$ be a sequence of measurable simple functions increasing to $\payoff$ (i.e., for all plays $\play\in\playSet{\pomdp}$, $\payoff_{\indexSequence}(\play) \leq \payoff_{\indexSequence+1}(\play)$ and $\lim_{\indexSequence\to\infty}\payoff_{\indexSequence}(\play) = \payoff(\play)$).
  By the monotone convergence theorem and the previous point on simple functions, we have
  \begin{equation}\label{equation:lemma:pure integral:2}
    \expectancy^{\mixedStrat}_\mdpState(\payoff)
    = \lim_{\indexSequence\to\infty}\expectancy^{\mixedStrat}_\mdpState(\payoff_\indexSequence)
    = \lim_{\indexSequence\to\infty}
      \int_{\stratBMDP\in\stratClassPure{\pomdp}}\expectancy^{\stratBMDP}_\mdpState(\payoff_\indexSequence)\ud\mixedStrat(\stratBMDP).
  \end{equation}
  For all pure strategies $\stratBMDP$, the monotone convergence theorem implies that $\lim_{\indexSequence\to\infty}\expectancy^{\stratBMDP}_\mdpState(\payoff_\indexSequence) = \expectancy^{\stratBMDP}_\mdpState(\payoff)$.
  Therefore, the sequence of functions $(\stratBMDP\mapsto\expectancy^{\stratBMDP}_\mdpState(\payoff_\indexSequence))_{\indexSequence\in\IN}$ over $\stratClassPure{\pomdp}$ increases (i.e., is non-decreasing and converges pointwise)  to $\stratBMDP\mapsto\expectancy^{\stratBMDP}_\mdpState(\payoff)$, implying that this function is measurable.
  The monotone convergence theorem allows us to exchange the limit and integral in the rightmost term of Equation~\eqref{equation:lemma:pure integral:2}, and implies that:
  \[\expectancy^{\mixedStrat}_\mdpState(\payoff) =
    \int_{\stratBMDP\in\stratClassPure{\pomdp}}\lim_{\indexSequence\to\infty}
    \expectancy^{\stratBMDP}_\mdpState(\payoff_\indexSequence)\ud\mixedStrat(\stratBMDP) =
    \int_{\stratBMDP\in\stratClassPure{\pomdp}}
    \expectancy^{\stratBMDP}_\mdpState(\payoff)\ud\mixedStrat(\stratBMDP)
    .\]

  We introduce some notation for the two last cases.
  We let $\payoff^+ = \max(\payoff, 0)$ and $\payoff^- = \max(-\payoff, 0)$ denote the non-negative and non-positive parts of $\payoff$ respectively; we have $\payoff = \payoff^+ - \payoff^-$.
    From the above, we obtain that for all universally unambiguously integrable payoffs, the function $\stratBMDP\mapsto\expectancy^{\stratBMDP}_\mdpState(\payoff)$ over $\stratClassPure{\pomdp}$ is measurable; it is the difference of the measurable non-negative functions $\stratBMDP\mapsto\expectancy^{\stratBMDP}_\mdpState(\payoff^+)$ and $\stratBMDP\mapsto\expectancy^{\stratBMDP}_\mdpState(\payoff^-)$.
    
  For the second-to-last case, we assume that $\payoff$ is $\proba^{\mixedStrat}_\mdpState$-integrable.
  We prove that the mappings $\stratBMDP\mapsto\expectancy^{\stratBMDP}_\mdpState(\payoff^+)$ and $\stratBMDP\mapsto\expectancy^{\stratBMDP}_\mdpState(\payoff^-)$ are $\mixedStrat$-integrable.
  We proceed by bounding these functions by a $\mixedStrat$-integrable function.
  For all $\stratBMDP\in\stratClassPure{\pomdp}$, we have
  \(\expectancy^{\stratBMDP}_\mdpState(\payoff^+), \expectancy^{\stratBMDP}_\mdpState(\payoff^-)\leq
  \expectancy^{\stratBMDP}_\mdpState(|\payoff|)
  \).
  By the above, we have $\int_{\stratBMDP\in\stratClassPure{\pomdp}}\expectancy^{\stratBMDP}_\mdpState(|\payoff|)\ud\mixedStrat(\stratBMDP) = \expectancy^{\mixedStrat}_\mdpState(|\payoff|)$, which is a real number since $\payoff$ is $\proba^{\mixedStrat}_\mdpState$-integrable.
  We have shown that $\stratBMDP\mapsto\expectancy^{\stratBMDP}_\mdpState(\left|\payoff\right|)$ is $\mixedStrat$-integrable, which implies that $\stratBMDP\mapsto\expectancy^{\stratBMDP}_\mdpState(\payoff^+)$ and $\stratBMDP\mapsto\expectancy^{\stratBMDP}_\mdpState(\payoff^-)$ also are.
  It follows that $\stratBMDP\mapsto\expectancy^{\stratBMDP}_\mdpState(\payoff)$ is $\mixedStrat$-integrable.

  By definition, we have $\expectancy^{\mixedStrat}_{\mdpState}(\payoff) = \expectancy^{\mixedStrat}_{\mdpState}(\payoff^+) - \expectancy^{\mixedStrat}_{\mdpState}(\payoff^-)$ and $\expectancy^{\stratBMDP}_{\mdpState}(\payoff) = \expectancy^{\stratBMDP}_{\mdpState}(\payoff^+) - \expectancy^{\stratBMDP}_{\mdpState}(\payoff^-)$ for all strategies $\stratBMDP\in\stratClassPure{\pomdp}$.
  Combining this with the linearity of the Lebesgue integral and the result for non-negative payoffs yields the following sequence of equalities:
  \begin{align*}
    \expectancy^{\mixedStrat}_\mdpState(\payoff)
    & = \expectancy^{\mixedStrat}_\mdpState(\payoff^+) -
      \expectancy^{\mixedStrat}_\mdpState(\payoff^-)
      \\
    & = \int_{\stratBMDP\in\stratClassPure{\pomdp}}\expectancy^{\stratBMDP}_\mdpState
      (\payoff^+)\ud\mixedStrat(\stratBMDP) -
      \int_{\stratBMDP\in\stratClassPure{\pomdp}}\expectancy^{\stratBMDP}_\mdpState
      (\payoff^-)\ud\mixedStrat(\stratBMDP) \\
    & =
      \int_{\stratBMDP\in\stratClassPure{\pomdp}}
      \expectancy^{\stratBMDP}_\mdpState(\payoff^+) -
      \expectancy^{\stratBMDP}_\mdpState(\payoff^-)\ud\mixedStrat(\stratBMDP) \\
    & =
      \int_{\stratBMDP\in\stratClassPure{\pomdp}}\expectancy^{\stratBMDP}_\mdpState
      (\payoff)\ud\mixedStrat(\stratBMDP).
  \end{align*}

  To deal with the last case, we assume that $\inf_{\stratBMDP}\expectancy^\stratBMDP_\mdpState(\payoff)\geq 0$.
  The analogous case $\sup_{\stratBMDP}\expectancy^{\stratBMDP}_\mdpState(\payoff)\leq 0$ can be recovered from the case $\inf_{\stratBMDP}\expectancy^\stratBMDP_\mdpState(\payoff)\geq 0$ by considering $-\payoff$ as the payoff function.
  We assume that $\payoff$ is not $\proba^{\mixedStrat}_\mdpState$-integrable as this case has been examined above.
  It follows that $\expectancy^{\mixedStrat}_\mdpState(\payoff) = +\infty$.
  The integral $\int_{\stratBMDP\in\stratClassPure{\pomdp}}\expectancy^{\stratBMDP}_\mdpState(\payoff)\ud\mixedStrat(\stratBMDP)$ is formally well-defined by the assumption that $\inf_{\stratBMDP}\expectancy^\stratBMDP_\mdpState(\payoff)\geq 0$.
  To end the argument, we must show that this integral is $+\infty$.
  Assume towards a contradiction that this is not the case.
  This implies that $\stratBMDP\mapsto\expectancy^{\stratBMDP}_\mdpState(\payoff)$ is $\mixedStrat$-integrable.
  From the result for non-negative payoffs, we obtain that $\int_{\stratBMDP\in\stratClassPure{\pomdp}}\expectancy^{\stratBMDP}_\mdpState(\payoff^+)\ud\mixedStrat(\stratBMDP) = \expectancy^{\mixedStrat}_\mdpState(\payoff^+) = +\infty$ and $\int_{\stratBMDP\in\stratClassPure{\pomdp}}\expectancy^{\stratBMDP}_\mdpState(\payoff^-)\ud\mixedStrat(\stratBMDP) = \expectancy^{\mixedStrat}_\mdpState(\payoff^-)\in\IR$.
  By linearity of the Lebesgue integral (for $\mixedStrat$-integrable payoffs), we obtain that
  \begin{align*}
    \expectancy^{\mixedStrat}_{\mdpState}(\payoff^+)
    & =\int_{\stratBMDP\in\stratClassPure{\pomdp}}
    \expectancy^{\stratBMDP}_\mdpState(\payoff^+)\ud\mixedStrat(\stratBMDP)
    \\
    & =
      \int_{\stratBMDP\in\stratClassPure{\pomdp}}(\expectancy^{\stratBMDP}_\mdpState(\payoff) +
      \expectancy^{\stratBMDP}_\mdpState(\payoff^-))\ud\mixedStrat(\stratBMDP) \\
    & = 
    \int_{\stratBMDP\in\stratClassPure{\pomdp}}
    \expectancy^{\stratBMDP}_\mdpState(\payoff)\ud\mixedStrat(\stratBMDP) +
    \int_{\stratBMDP\in\stratClassPure{\pomdp}}
    \expectancy^{\stratBMDP}_\mdpState(\payoff^-)\ud\mixedStrat(\stratBMDP).
  \end{align*}
  This is a contradiction: on the one hand, we have $\expectancy^{\mixedStrat}_{\mdpState}(\payoff^+) = +\infty$ and, on the other hand, the sum in the last term is a real number.
  This ends the argument for the case $\inf_{\stratBMDP}\expectancy^\stratBMDP_\mdpState(\payoff)\geq 0$.
\end{proof}

We highlight two major consequences of Lemma~\ref{lem:expectancy:pure integral}.
On the one hand, for all strategies whose expected payoff is real, there exists a pure strategy with a greater expected payoff.
On the other hand, if there exists a strategy with an infinite expected payoff,  then there are pure strategies with arbitrarily large expected payoffs in absolute value.
We note that even if a randomised strategy has an infinite expectation, there need not exist a pure strategy with infinite expectation.
This is analogous to the fact that real-valued random variables can have an infinite expectation.

Using Lemma~\ref{lem:expectancy:pure integral}, we can obtain a characterisation of universally integrable payoffs.
Let $\stratClass\in\{\stratClassAll{\pomdp}, \stratClassPure{\pomdp}\}$.
A payoff $\payoff$ is universally integrable if and only if for all $\mdpState\in\mdpStateSpace$, $\sup_{\stratMDP\in\stratClass}\expectancy^{\stratMDP}_\mdpState(|\payoff|)$ is real.
The non-trivial part of the proof is showing that the definition of universally integrable and the property when the supremum ranges over pure strategies both imply the property with the supremum ranging over all strategies.
We show the contrapositive of both implications.
We assume that for some $\mdpState\in\mdpStateSpace$, $\sup_{\stratMDP\in\stratClassAll{\pomdp}}\expectancy^{\stratMDP}_\mdpState(|\payoff|)=+\infty$.
Lemma~\ref{lem:expectancy:pure integral} then implies that there are pure strategies $\stratBMDP$ with arbitrarily large $\expectancy^{\stratBMDP}_\mdpState(|\payoff|)$, which implies the validity of one of the implications.
For the other, we mix these pure strategies to construct a mixed strategy $\mixedStrat$ such that $\expectancy^{\mixedStrat}_\mdpState(|\payoff|)=+\infty$.

\begin{restatable}{lemma}{lemUICharacterisation}\label{lem:ui:characterisation}
  Let $\payoff$ be a payoff.
  Let $\mdpState\in\mdpStateSpace$.
  The following assertions are equivalent.
  \begin{enumerate}
  \item $\payoff$ is $\proba^{\stratMDP}_\mdpState$-integrable for all $\stratMDP\in\stratClassAll{\pomdp}$.\label{item:ui:characterisation:1}
  \item We have $\sup\{\expectancy^{\stratMDP}_{\mdpState}(|\payoff|)\mid \stratMDP\in\stratClassAll{\pomdp}\}\in\IR$.\label{item:ui:characterisation:2}
  \item We have $\sup\{\expectancy^{\stratMDP}_{\mdpState}(|\payoff|)\mid \stratMDP\in\stratClassPure{\pomdp}\}\in\IR$.\label{item:ui:characterisation:3}
  \end{enumerate}
  In particular, $\payoff$ is universally integrable if and only if Item~\ref{item:ui:characterisation:2} (resp.~\ref{item:ui:characterisation:3}) holds for all $\mdpState\in\mdpStateSpace$.
\end{restatable}
\begin{proof}
    Item~\ref{item:ui:characterisation:2} directly implies the other two items.
  We now show that the other two items imply Item~\ref{item:ui:characterisation:2} via the contrapositive of these implications.

  Assume that Item~\ref{item:ui:characterisation:2} does not hold, i.e., that $\sup\{\expectancy^{\stratMDP}_{\mdpState}(|\payoff|)\mid \stratMDP\in\stratClassAll{\mdp}\} = +\infty$.
  If there exists a pure strategy $\stratMDP$ such that $\expectancy^{\stratMDP}_{\mdpState}(|\payoff|) = +\infty$, the negations of Item~\ref{item:ui:characterisation:1} and Item~\ref{item:ui:characterisation:3} follow directly.
  In the remainder of the proof, we assume that this is not the case.
  
  First, we show that Item~\ref{item:ui:characterisation:3} does not hold.
  By Lemma~\ref{lem:expectancy:pure integral} (and Kuhn's theorem), for all strategies $\stratMDP\in\stratClassAll{\mdp}$, if $\expectancy^{\stratMDP}_{\mdpState}(|\payoff|)\in\IR$, there exists a pure strategy $\stratBMDP$ such that $\expectancy^{\stratBMDP}_{\mdpState}(|\payoff|)\geq\expectancy^{\stratMDP}_{\mdpState}(|\payoff|)$ and, otherwise, if $\expectancy^{\stratMDP}_{\mdpState}(|\payoff|)=+\infty$, then for all $M\in\IR$, there exists a pure strategy $\stratBMDP$ such that $\expectancy^{\stratBMDP}_{\mdpState}(|\payoff|)\geq M$.
  In particular, Item~\ref{item:ui:characterisation:3} does not hold.

  We now construct a strategy $\stratMDP$ such that $\expectancy^{\stratMDP}_{\mdpState}(|\payoff|)=+\infty$ from the pure strategies above.
  For all $\indexLast\in\IN$, let $\stratBMDP_\indexLast$ be a pure strategy such that $\expectancy^{\stratBMDP_\indexLast}_{\mdpState}(|\payoff|)\geq 2^\indexLast$.
  Let $\mixedStrat$ be the mixed strategy that randomises over the set $\{\stratBMDP_\indexLast\mid\indexLast\in\IN\}$ and selects strategy $\stratBMDP_\indexLast$ with probability $\frac{1}{2^{\indexLast+1}}$.
  We obtain that $\expectancy^{\mixedStrat}_{\mdpState}(|\payoff|)=+\infty$.
  This shows that Item~\ref{item:ui:characterisation:1} does not hold.
\end{proof}

Let $\payoff$ be a universally unambiguously integrable payoff.
It may be the case that Lemma~\ref{lem:expectancy:pure integral} cannot be directly applied to $\payoff$, e.g., if $\payoff$ is not universally integrable and strategies with a negative expected payoff coexist with strategies with a positive expected payoff.
Nonetheless, we can show that there exists a constant $\alpha\in\IR$ such that $\payoff+\alpha$ satisfies the assumptions of Lemma~\ref{lem:expectancy:pure integral}.
It suffices to establish that for all initial states $\mdpState$, there either exists a lower or upper bound on the expectation of strategies from $\mdpState$.
We provide a proof that relies on Lemma~\ref{lem:ui:characterisation} below.

\begin{restatable}{lemma}{lemUnambiguousBound}\label{lem:unambiguous:bound}
  Let $\payoff$ be a universally unambiguously integrable payoff function.
  For all $\mdpState\in\mdpStateSpace$, we have $\inf_{\stratMDP\in\stratClassAll{\pomdp}}\expectancy^{\stratMDP}_\mdpState(\payoff)\in\IR$ or $\sup_{\stratMDP\in\stratClassAll{\pomdp}}\expectancy^{\stratMDP}_\mdpState(\payoff)\in\IR$.
\end{restatable}
\begin{proof}
  Let $\mdpState\in\mdpStateSpace$.
  Let $\payoff^+ = \max(\payoff, 0)$ and $\payoff^- = \max(0, -\payoff)$.
  Assume towards a contradiction that $\inf_{\stratMDP\in\stratClassAll{\pomdp}}\expectancy^{\stratMDP}_\mdpState(\payoff)\notin\IR$ and $\sup_{\stratMDP\in\stratClassAll{\pomdp}}\expectancy^{\stratMDP}_\mdpState(\payoff)\notin\IR$.
  Because $\stratClassAll{\pomdp}$ is non-empty, we have $\inf_{\stratMDP\in\stratClassAll{\pomdp}}\expectancy^{\stratMDP}_\mdpState(\payoff)=-\infty$ and $\sup_{\stratMDP\in\stratClassAll{\pomdp}}\expectancy^{\stratMDP}_\mdpState(\payoff)=+\infty$.
  We show that this implies that $\payoff$ is not universally unambiguously integrable, i.e., there exists a strategy $\stratMDP\in\stratClassAll{\pomdp}$ such that $\expectation_{\mdpState}^{\stratMDP}(\payoff^+) = \expectation_{\mdpState}^{\stratMDP}(\payoff^-) = +\infty$.

  We observe that for all $\stratMDP\in\stratClassAll{\pomdp}$, we have $\expectancy^{\stratMDP}_\mdpState(\payoff)\leq\expectancy^{\stratMDP}_\mdpState(\payoff^+)$ and $\expectancy^{\stratMDP}_\mdpState(-\payoff)\leq\expectancy^{\stratMDP}_\mdpState(\payoff^-)$.
  It follows from Lemma~\ref{lem:ui:characterisation} and Kuhn's theorem that there exists a mixed strategy $\mixedStrat_+$ (resp.~$\mixedStrat_-$) such that $\payoff^+$ (resp.~$\payoff^-$) is not $\proba^{\mixedStrat_+}_\mdpState$-integrable (resp.~$\proba^{\mixedStrat_-}_\mdpState$-integrable).
  In particular, we obtain that $\expectancy^{\mixedStrat_+}_\mdpState(\payoff^+) = \expectancy^{\mixedStrat_-}_\mdpState(\payoff^-)=+\infty$.
  The mixed strategy $\mixedStrat = \frac{1}{2}\mixedStrat_+ + \frac{1}{2}\mixedStrat_-$  satisfies $\expectancy^{\mixedStrat}_\mdpState(\payoff^+) = \expectancy^{\mixedStrat}_\mdpState(\payoff^-)=+\infty$.
  This shows that $\payoff$ is not universally unambiguously integrable: $\payoff$ does not have an unambiguous $\proba^{\mixedStrat}_\mdpState$-integral.
\end{proof}

\subsection{Some classical objectives and payoffs}\label{section:payoffs:classical}

We now present some classical objectives and payoffs functions.
First, we define reachability objectives.
A reachability objective requires that a set of target states be visited.
Formally, given a set of target states (target for short) $\target\subseteq\mdpStateSpace$, we define the \textit{reachability objective} $\reach{\target}$ as the set $\{\mdpState_0\mdpAction_0\mdpState_1\mdpAction_1\ldots\in\playSet{\pomdp}\mid\exists\,\indexPosition\in\IN,\,\mdpState_\indexPosition\in\target\}$.

The other payoff functions we introduce are defined in a context where transitions of the POMDP have a numerical weight assigned to them.
Formally, a weight function is a function $\weight\colon\mdpStateSpace\times\mdpActionSpace\to\IR$.
We let $\weight$ denote a weight function.
A \textit{discounted-sum payoff} is defined as an accumulated sum of weights multiplied by powers of a discount factor in $\coInt{0}{1}$.
Formally, given a discount factor $\discFactor\in\coInt{0}{1}$, we let $\discSum{\discFactor}{\weight}\colon\playSet{\pomdp}\to\IR$ be the payoff function defined by $\discSum{\discFactor}{\weight}(\play) = \sum_{\indexPosition=0}^\infty\discFactor^\indexPosition\weight(\mdpState_\indexPosition,\mdpAction_\indexPosition)$ for all plays $\play = \mdpState_0\mdpAction_0\mdpState_1\mdpAction_1\ldots\in\playSet{\pomdp}$.
A discounted-sum payoff is bounded and thus universally integrable.

A \textit{total-reward} payoff corresponds to the accumulated weights along a play with no discounting.
Formally, we let $\totrew{\weight}\colon\playSet{\pomdp}\to\IRbar$ be defined, for all plays $\play = \mdpState_0\mdpAction_0\mdpState_1\mdpAction_1\ldots$, by $\liminf_{\indexLast\to\infty}\sum_{\indexPosition=0}^\indexLast\weight(\mdpState_\indexPosition,\mdpAction_\indexPosition)$.
We use a limit-inferior to ensure that this payoff is well-defined for all plays, as the series of weights along a play need not converge.

A \textit{shortest-path} payoff can be seen as a quantitative variant of reachability and is defined as the accumulated sum of weights up to the first visit of a set of targets.
Formally, given a target $\target$, we let $\spath{\target}{\weight}\colon\playSet{\pomdp}\to\IRbar$ be defined by, for all plays $\play = \mdpState_0\mdpAction_0\mdpState_1\mdpAction_1\ldots$, $\spath{\target}{\weight}(\play) = +\infty$ if $\play\notin\reach{\target}$ and, otherwise, $\spath{\target}{\weight}(\play) = \sum_{\indexPosition=0}^{\indexLast-1}\weight(\mdpState_\indexPosition, \mdpAction_\indexPosition)$ where $\indexLast = \min\{\indexPosition\in\IN\mid\mdpState_\indexPosition\in\target\}$.
Traditionally, the goal is to minimise the shortest-path payoff (i.e., it is a cost function rather than a payoff), hence the infinite payoff whenever the target is not visited.

All total-reward and shortest-path payoffs we consider are built on a non-negative weight function and are thus universally unambiguously integrable.

\subsection{Continuous payoffs}\label{section:payoffs:continuous}
The set of plays of $\pomdp$ can be equipped with a natural topology that is metrisable (and also compact if $\pomdp$ is finite).
This allows us to define continuous and uniformly continuous payoffs.
We defer a formal presentation of the topology of $\playSet{\pomdp}$ to Section~\ref{appendix:topology:plays}, and instead provide direct definitions here.

Intuitively, a payoff is continuous if plays with a long common prefix have a close payoff.
Let $\payoff\colon\playSet{\pomdp}\to\IRbar$ be a payoff and let $\play\in\playSet{\pomdp}$.
If $\payoff(\play)\in\IR$, then $\payoff$ is continuous at $\play$ if and only if for all $\varepsilon > 0$, there exists $\indexPosition\in\IN$ such that for all $\play'\in\cyl{\playPrefix{\play}{\indexPosition}}$, $|\payoff(\play)-\payoff(\play')| < \varepsilon$.
If $\payoff(\play) = +\infty$ (resp.~$-\infty$), then $\payoff$ is continuous at $\play$ if and only if for all $M\in\IR$, there exists $\indexPosition\in\IN$ such that for all plays $\play'\in\cyl{\playPrefix{\play}{\indexPosition}}$, $\payoff(\play')\geq M$ (resp.~$\payoff(\play')\leq -M$).
The payoff $\payoff$ is \textit{continuous} if it is continuous at all plays.

If $\payoff$ is real-valued, then $\payoff$ is \textit{uniformly continuous} if for all $\varepsilon > 0$, there exists $\indexPosition\in\IN$ such that for all plays $\play$, $\play'\in\playSet{\pomdp}$, $\playPrefix{\play}{\indexPosition}= \playPrefix{\play'}{\indexPosition}$ implies that $|\payoff(\play) - \payoff(\play')| < \varepsilon$.
If $\pomdp$ is finite, then $\playSet{\pomdp}$ is compact and all continuous real-valued payoffs of $\pomdp$ are uniformly continuous.

Discounted-sum payoffs are uniformly continuous whenever they are built on a weight function that is bounded in absolute value.
Any shortest-path function built on a weight function $\weight\geq\varepsilon$ for some $\varepsilon > 0$ is continuous.
We prove these claims and provide characterisation in finite POMDPs of continuous indicators and continuous prefix-independent payoffs in Appendix~\ref{appendix:continuous payoffs}.

\subsection{Multiple payoffs}\label{section:payoffs:multi}
We now present terminology and notation for multi-objective POMDPs, i.e., POMDPs with multiple payoffs.
Let $\numObj\in\IN_0$.
We summarise $\numObj$ payoffs $f_1, \ldots, f_\numObj\colon\playSet{\pomdp}\to\IRbar$ as a multi-dimensional payoff $\payoffTuple\colon \playSet{\pomdp}\to\IRbar^\numObj$, and write $\payoffTuple = (\payoff_\indexPayoff)_{1\leq\indexPayoff\leq\numObj}$.

Let $\payoffTuple = (\payoff_\indexPayoff)_{1\leq\indexPayoff\leq\numObj}$ and let $\mdpState\in\mdpStateSpace$.
We say that $\payoffTuple$ is universally (resp.~unambiguously) integrable whenever $\payoff_\indexPayoff$ is universally (resp.~unambiguously) integrable for all $1\leq\indexPayoff\leq\numObj$.
We now assume that $\payoffTuple$ is universally unambiguously integrable (recall that we focus on such payoffs).
Given a set of strategies $\stratClass\subseteq\stratClassAll{\pomdp}$, we let $\paySetClass{\payoffTuple}{\mdpState}{\stratClass} = \{\expectancy^{\stratMDP}_{\mdpState}(\payoffTuple)\mid\stratMDP\in\stratClass\}$ denote the set of expected payoffs of strategies in $\stratClass$ from $\mdpState$.
We let $\paySet{\payoffTuple}{\mdpState}$ and $\paySetPure{\payoffTuple}{\mdpState}$ be shorthand for $\paySetClass{\payoffTuple}{\mdpState}{\stratClassAll{\pomdp}}$ and $\paySetClass{\payoffTuple}{\mdpState}{\stratClassPure{\pomdp}}$ respectively.
We refer to elements of $\paySetPure{\payoffTuple}{\mdpState}$ as pure expected payoffs.
A set of expected payoffs need not have a maximum for the component-wise order, e.g., there can be several Pareto-optimal payoffs.

We now claim that convex combinations of elements of $\paySet{\payoffTuple}{\mdpState}$ are well-defined under the convention that $0\cdot(+\infty) = 0\cdot (-\infty) = 0$.
This convention essentially boils down to ignoring terms of convex combinations that have a zero coefficient.
These convex combinations are well-defined because $\paySet{\payoffTuple}{\mdpState}\subseteq \prod_{\indexPayoff=1}^\numObj \ccInt{\inf_{\stratMDP\in\stratClassAll{\pomdp}}\expectancy^{\stratMDP}_\mdpState(\payoff_\indexPayoff)}{\sup_{\stratMDP\in\stratClassAll{\pomdp}}\expectancy^{\stratMDP}_\mdpState(\payoff_\indexPayoff)}$.
By Lemma~\ref{lem:unambiguous:bound}, each interval in this product has a finite (upper or lower) bound and thus no terms of the form $+\infty - \infty$ occur in any component of convex combinations of expected payoffs.
It follows that any convex combination of elements of $\paySet{\payoffTuple}{\mdpState}$ is well-defined and is an element of $\paySet{\payoffTuple}{\mdpState}$.
Furthermore, $\paySet{\payoffTuple}{\mdpState}\cap\IR^\numObj$ is a convex set.
The last two properties can be established directly using mixed strategies.

In multi-objective optimisation, the goal is to ensure a given threshold on each dimension.
This is formalised by the notion of \textit{achievable vectors}.
A vector $\payoffVect\in\IRbar^\numObj$ is achievable (from $\mdpState$) if there exists a strategy $\stratMDP$ such that $\payoffVect\leq\expectancy^{\stratMDP}_{\mdpState}(\payoffTuple)$.
In this case, we say that $\stratMDP$ witnesses that $\payoffVect$ is achievable.
For any class of strategies $\stratClass\subseteq\stratClassAll{\pomdp}$, we let $\achSetClass{\payoffVect}{\mdpState}{\stratClass} = \down{\paySetClass{\payoffTuple}{\mdpState}{\stratClass}}$ denote the set of vectors such that a strategy of $\stratClass$ witnesses that they are achievable.
We define $\achSet{\payoffTuple}{\mdpState}$ and $\achSetPure{\payoffTuple}{\mdpState}$ as above.

We also consider the lexicographic optimisation of multiple objectives.
For this case, we define an analogue of optimal strategies from the one-dimensional setting.
A strategy $\stratMDP$ is \textit{lexicographically optimal} if $\expectancy^{\stratMDP}_{\mdpState}(\payoffTuple)$ is the lexicographic maximum of $\paySet{\payoffTuple}{\mdpState}$.

\section{An introductory example}\label{section:running}
The goal of this section is to illustrate that expected payoff sets in multi-objective MDPs may be complex.
We provide a two-dimensional example illustrating that sets of expected payoffs are not necessarily polytopes, even when the payoffs are universally integrable.
We introduce our example and comment on several of its properties in Section~\ref{section:running:overview}.
We formally prove that these properties hold in Section~\ref{section:running:proofs}.

\subsection{Example overview}\label{section:running:overview}

We consider the MDP $\mdp$ depicted in Figure~\ref{figure:running:mdp} and let $\weight$ denote the two-dimensional weight function from the illustration.
On this MDP, we consider the two-dimensional payoff $\payoffTuple = (\payoff_1, \payoff_2)$ given by the discounted-sum payoffs $\payoff_1=\discSum{3/4}{\weight_1}$ and $\payoff_2=\discSum{1/2}{\weight_2}$.
MDPs with several discounted-sum payoffs with different discount factors have previously been studied in~\cite{DBLP:conf/lpar/ChatterjeeFW13}.

\begin{figure}
  \begin{subfigure}[t]{0.48\textwidth}
    \centering
    \scalebox{0.85}{
      \begin{tikzpicture}[node distance=9mm]
        \node[state, align=center, initial left] (s0) {$\mdpState_0$};
        \node[state, above = of s0] (s1) {$\mdpState_1$};
        \node[state, below = of s0] (s2) {$\mdpState_2$};
        \node[state, right = of s0] (s3) {$\mdpState_3$};
        
        \path[->] (s0) edge node[left] {$\mdpActionC$} node[right] {$(0, 1)$} (s1);
        \path[->] (s0) edge node[left] {$\mdpAction$} node[right] {$(1, 1)$} (s2);
        \path[->] (s0) edge node[above] {$\mdpActionB$} node[below] {$(2, 0)$} (s3);
        \path[->] (s2) edge[bend right] node[align=center,below right,yshift=3mm] {$\mdpActionB$\\$(1, 0)$} (s3);
        \path[->] (s1) edge[loop left] node[align=center,left] {$\mdpAction$\\$(0, 1)$} (s1);
        \path[->] (s2) edge[loop left] node[align=center,left] {$\mdpAction$\\$(0, 1)$} (s2);
        \path[->] (s3) edge[loop right] node[align=center,right] {$\mdpAction$\\$(1, 0)$} (s3);
      \end{tikzpicture}}
    \caption{An MDP with deterministic transitions.
    Pairs next to actions represent two-dimensional weights.}\label{figure:running:mdp}
  \end{subfigure}\hfill \begin{subfigure}[t]{0.48\textwidth}
    \centering
    \scalebox{1}{
      \begin{tikzpicture}[scale=1.3]
        \draw[-stealth] (-0.25,0) -- (3,0) node[below] {$\expectancy(\payoff_1)$}; 
        \draw[-stealth] (0,-0.5) -- (0,2.5) node[left] {$\expectancy(\payoff_2)$};
\coordinate (x) at (0.5, 0);
        \coordinate (y) at (0, 1);
        \node[xshift=-5] at (y) {$1$};
        \node[yshift=-10] at (x) {$1$};
        \node[stochastics] at (x) (qx){};
        \node[stochastics] at (y) (qy){};
        
        \foreach \n in {0,...,10}
        {
          \coordinate (c\n) at ({0.5+(3*(3/4)^(\n-1))/2},{(2 - (1/2)^(\n-1))});
          \node[draw,red,circle,fill,inner sep=0.05mm,minimum size=0.4mm] at (c\n) (q\n){};
        }
        \coordinate (cinf) at (0.5,2);
        \coordinate (cnew) at (0,2);
\draw[red] (cnew) --(c0) -- (c1) -- (c2) -- (c3) -- (c4) -- (c5) -- (c6) -- (c7) -- (c8) -- (c9) -- (c10) -- (cinf) -- cycle;
        
        \begin{scope}[on background layer]
          \fill[red!20] (cnew) --(c0) -- (c1) -- (c2) -- (c3) -- (c4) -- (c5) -- (c6) -- (c7) -- (c8) -- (c9) -- (c10) -- (cinf) -- cycle;
        \end{scope}
        
\draw[dotted, orange] (-0.25, 1.8125) -- (2.5, 2.5);
        \draw[dotted, orange] (-0.25, 2.1875) -- (3, 1.375);
        
        \node[draw,blue,circle,fill,inner sep=0.05mm,minimum size=0.6mm] at (cinf) (qfin){};
        \node[blue,yshift=3mm] at (cinf) {$(1, 2)$};
        \node[draw,red,circle,fill,inner sep=0.05mm,minimum size=0.6mm] at (cnew) (qnew){};
        \node[draw,red,circle,fill,inner sep=0.05mm,minimum size=0.6mm] at (c0) (q0){};
        
      \end{tikzpicture}
    }
    \caption{The set of expected payoffs for the MDP of Figure~\ref{figure:running:mdp} for the payoff $\payoff_1 = \discSum{3/4}{\weight_1}$ and $\payoff_2 = \discSum{1/2}{\weight_2}$.}\label{figure:running:payoff set}
  \end{subfigure}
  \caption{An MDP with a two-dimensional discounted-sum payoff $\payoffTuple$ such that $\corners{\paySet{\payoffTuple}{\mdpState_0}}$ is infinite.}
\end{figure}

Due to the absence of randomness in transitions, the expected payoff of any pure strategy from $\mdpState_0$ is the payoff of a play from $\mdpState_0$.
Therefore, we obtain that
\[{\paySetPure{\payoffTuple}{\mdpState_0}} =
  \{(0, 2), (1, 2)\}\cup\left\{\left(1+\frac{3^\indexLast}{4^{\indexLast-1}}, 2-\frac{1}{2^{\indexLast-1}}\right)\mid\indexLast\in\IN\right\}.\]
On the one hand, the payoffs $(0, 2)$ and $(1, 2)$ are obtained by moving from $\mdpState_0$ to $\mdpState_1$ and $\mdpState_2$ respectively and looping in these states forever.
On the other hand, for all $\indexLast\in\IN$, the payoff $\left(1+\frac{3^\indexLast}{4^{\indexLast-1}}, 2-\frac{1}{2^{\indexLast-1}}\right)$ is obtained by spending $\indexLast$ rounds in $\mdpState_2$ then moving to $\mdpState_3$; for $\indexLast=0$, we move from $\mdpState_0$ to $\mdpState_3$ directly.

We approximately illustrate $\paySet{\payoffTuple}{\mdpState_0}$ in Figure~\ref{figure:running:payoff set}.
This illustration is based on the equality $\paySet{\payoffTuple}{\mdpState_0}=\convex{\paySetPure{\payoffTuple}{\mdpState_0}}$.
This equality follows from Thm.~\ref{thm:mixing:exact}, which states that if $\payoffTuple$ is universally integrable, then this equality holds.
In the interest of keeping this section self-contained, we outline a direct argument for this simple example that does not rely on Theorem~\ref{thm:mixing:exact}.
Because $\mdp$ has countably many plays, the expected payoff of any (randomised) strategy can be seen as an infinite convex combination of payoffs of plays of $\mdp$.
Any such infinite convex combination can be shown to be in the closure of $\convex{\paySetPure{\payoffTuple}{\mdpState_0}}$.
Furthermore, $\convex{\paySetPure{\payoffTuple}{\mdpState_0}}$ is closed because the convex hull of a compact set is compact and $\paySetPure{\payoffTuple}{\mdpState_0}$ is compact.
This shows that $\paySet{\payoffTuple}{\mdpState_0}=\convex{\paySetPure{\payoffTuple}{\mdpState_0}}$ and justifies the illustration.

Since $\paySet{\payoffTuple}{\mdpState_0}=\convex{\paySetPure{\payoffTuple}{\mdpState_0}}$, we conclude that all extreme points of $\paySet{\payoffTuple}{\mdpState_0}$ can be obtained by using pure strategies.
Indeed, any vector of $\paySet{\payoffTuple}{\mdpState_0}$ that cannot be obtained by a pure strategy is a convex combination of the expected payoffs of pure strategies, and thus is not extreme.
There exists another means of showing that a point can be obtained by a pure strategy.
This technique involves reducing to a single-dimensional setting by composing a linear form with $\payoffTuple$.
However, this approach does not work for $(1, 2)\in\paySet{\payoffTuple}{\mdpState_0}$.

First, let us explain how to establish that a point can be obtained by a pure strategy via a reduction to a single-dimensional payoff.
Let $\payoffVect\in\paySet{\payoffTuple}{\mdpState_0}$.
To prove that $\payoffVect\in\paySetPure{\payoffTuple}{\mdpState_0}$, one considers a linear form $\linForm$ such that $\linForm(\payoffVect)>\linForm(\payoffVectB)$ for all $\payoffVectB\in\paySet{\payoffTuple}{\mdpState_0}$.
The existence of such a form is equivalent to the existence of a hyperplane $\hplane$ supporting $\paySet{\payoffTuple}{\mdpState_0}$ at $\payoffVect$ such that $\paySet{\payoffTuple}{\mdpState_0}\cap\hplane=\{\payoffVect\}$.
By linearity of the expectation, $\linForm\circ\payoffTuple$ is universally integrable and any strategy $\stratMDP$ that maximises the expectation of $\linForm\circ\payoffTuple$ from $\mdpState_0$ must satisfy $\expectancy^{\stratMDP}_{\mdpState_0}(\payoffTuple) = \payoffVect$ (such a strategy exists because $\payoffVect\in\paySet{\payoffTuple}{\mdpState_0}$).
Lemma~\ref{lem:expectancy:pure integral} then implies that $\payoffVect$ can be obtained with a pure strategy.
This argument can be used whenever the set of expected payoffs is a compact polytope to prove that all vertices can be attained by pure strategies.

We now demonstrate that this argument is not applicable to show that $(1, 2)\in\corners{\paySet{\payoffTuple}{\mdpState_0}}$ is the payoff of a pure strategy.
We observe that the only hyperplane that supports $\paySet{\payoffTuple}{\mdpState_0}$ at $(1, 2)$ is the line $\hplane$ carrying the segment $\ccInt{(0, 2)}{(1, 2)}$.
This is highlighted by the orange dotted lines in Figure~\ref{figure:running:payoff set}: angling $\hplane$ in any way yields a line that is not a supporting hyperplane.
Furthermore, this hyperplane is not suitable for the above argument because $\hplane\cap\paySet{\payoffTuple}{\mdpState_0} = \ccInt{(0, 2)}{(1, 2)}$.
This implies that there is no hyperplane supporting $\paySet{\payoffTuple}{\mdpState_0}$ at $(1, 2)$ that does not contain other points of $\paySet{\payoffTuple}{\mdpState_0}$.
Therefore, we cannot show that $(1, 2)\in\paySetPure{\payoffTuple}{\mdpState_0}$ by following the argument outlined above.
In Section~\ref{section:achievable}, we reason on MDPs with lexicographic optimisation to prove Theorem~\ref{thm:mixing:exact}, and as a corollary, obtain that all extreme points of a set of expected payoffs can be attained by a pure strategy.

We conclude from the above that $\paySet{\payoffTuple}{\mdpState_0}$ is not a polytope.
Indeed, the vertices of a polytope can be isolated with supporting hyperplanes.
Thus, if $\paySet{\payoffTuple}{\mdpState_0}$ were a polytope, we would have been able to conclude that $(1, 2)$ is the payoff of a pure strategy via the described reduction to a single dimension.
This implies that $\paySet{\payoffTuple}{\mdpState_0}$ has infinitely many extreme points.
In fact, we can show that $\corners{\paySet{\payoffTuple}{\mdpState_0}}=\paySetPure{\payoffTuple}{\mdpState_0}$ and that all of these payoffs aside from $(0, 2)$ are Pareto-optimal.
This illustrates that even for classical payoffs, the expected payoff and achievable sets can be quite complex.

Finally, we comment on the strategy complexity required to obtain certain expected payoffs.
More precisely, we briefly discuss memory requirements.
Memory requirements in games and MDPs are thoroughly studied (see Section~\ref{section:intro}) and constitute a central measure of complexity for strategies in the literature.
All but three extreme points of $\paySet{\payoffTuple}{\mdpState_0}$ are obtained by moving from $\mdpState_0$ to $\mdpState_2$ and looping there finitely many times before moving to $\mdpState_3$.
In other words, these extreme points require pure strategies that count up to some arbitrarily large number.
In fact, in this case, we can only obtain these payoffs by using these specific pure strategies.
Intuitively, the expected payoff of a randomised strategy that induces more than one play is a non-trivial convex combination of payoffs of several plays, and therefore not in $\corners{\paySet{\payoffTuple}{\mdpState_0}}$.
This implies that some expected payoffs need strategies with arbitrarily large albeit finite memory to be obtained in this instance.

\subsection{Proofs and details}\label{section:running:proofs}
We now prove all statements that were made in Section~\ref{section:running:overview}.
First, we prove that the description of $\paySetPure{\payoffTuple}{\mdpState_0}$ given above is accurate.
Second, we show that $\paySet{\payoffTuple}{\mdpState_0}=\convex{\paySetPure{\payoffTuple}{\mdpState_0}}$ in the context of our example.
When then prove that we cannot conclude that $(1, 2)\in\corners{\paySet{\payoffTuple}{\mdpState_0}}$ by reducing to a single-dimensional payoff via a linear form.
Next, we establish that all pure payoffs are extreme points and that all of these points except $(0, 2)$ are Pareto-optimal.
Finally, we close the section by proving that all payoffs of pure strategies can \textit{only} be obtained by playing without randomisation, and comment on the consequences in terms of strategy complexity.

Throughout this section, $\mdp$ refers to the MDP of Figure~\ref{figure:running:mdp}.

\subsubsection{Determining the set of payoffs of pure strategies}\label{section:running:pure payoffs}

We prove that the description of $\paySetPure{\payoffTuple}{\mdpState_0}$ given in Section~\ref{section:running:overview} is correct.
The argument is based on the fact that there are no randomised transitions in the MDP we consider.
Therefore, the payoff of a pure strategy from $\mdpState_0$ is the payoff of a single play.
In the following proof, we directly compute the payoff of each play from $\mdpState_0$.
\begin{lemma}\label{lemma:running:pure:description}
  We have ${\paySetPure{\payoffTuple}{\mdpState_0}} =
  \{(0, 2), (1, 2)\}\cup\left\{\left(1+\frac{3^\indexLast}{4^{\indexLast-1}}, 2-\frac{1}{2^{\indexLast-1}}\right)\mid\indexLast\in\IN\right\}.$
\end{lemma}
\begin{proof}
  There are no randomised transitions in $\mdp$.
  For this reason, any pure strategy induces a single play in $\mdp$ from any starting state.
  We compute the payoff of all plays of $\mdp$ from $\mdpState_0$ to obtain the desired result.

  We first consider the three plays $\mdpState_0\mdpActionC(\mdpState_1\mdpAction)^\omega$, $\mdpState_0\mdpAction(\mdpState_2\mdpAction)^\omega$ and $\mdpState_0\mdpActionB(\mdpState_3\mdpAction)^\omega$ that never leave their second state once it is reached.
  By definition of discounted-sum payoff functions, we have $\payoffTuple(\mdpState_0\mdpActionC(\mdpState_1\mdpAction)^\omega) = (0, \sum_{\indexPosition=0}^\infty\frac{1}{2^\indexPosition}) = (0, 2)$, $\payoffTuple(\mdpState_0\mdpAction(\mdpState_2\mdpAction)^\omega) = (1, \sum_{\indexPosition=0}^\infty\frac{1}{2^\indexPosition}) = (1, 2)$ and $\payoffTuple(\mdpState_0\mdpActionB(\mdpState_3\mdpAction)^\omega) = (1+\sum_{\indexPosition=0}^\infty(\frac{3}{4})^\indexPosition, 0) = (5, 0) = (1+\frac{3^0}{4^{0-1}}, 2-\frac{1}{2^{0-1}})$.

  It remains to deal with the plays that move from $\mdpState_0$ to $\mdpState_2$ and then eventually move to $\mdpState_3$.
  It suffices to show that for all $\indexLast\geq 1$, we have $\payoffTuple(\mdpState_0(\mdpAction\mdpState_2)^\indexLast\mdpActionB(\mdpState_3\mdpAction)^\omega) = (1+\frac{3^\indexLast}{4^{\indexLast-1}}, 2-\frac{1}{2^{\indexLast-1}})$.
  Let $\indexLast\geq 1$.
  We obtain, by definition of discounted-sum payoff functions, that
  \begin{equation*}
    \payoff_1(\mdpState_0(\mdpAction\mdpState_2)^\indexLast\mdpActionB(\mdpState_3\mdpAction)^\omega)
    =
      1+\sum_{\indexPosition=\indexLast}^{\infty}\frac{3^\indexPosition}{4^{\indexPosition}} =
     1+\frac{3^\indexLast}{4^{\indexLast}}\cdot\sum_{\indexPosition=0}^{\infty}\frac{3^\indexPosition}{4^{\indexPosition}} =
   1+\frac{3^\indexLast}{4^{\indexLast-1}},
 \end{equation*}
 and
   \begin{equation*}
    \payoff_2(\mdpState_0(\mdpAction\mdpState_2)^\indexLast\mdpActionB(\mdpState_3\mdpAction)^\omega)
    =
    \sum_{\indexPosition=0}^{\indexLast-1}\frac{1}{2^{\indexPosition}} =
    \frac{1-\frac{1}{2^{\indexLast}}}{1-\frac{1}{2}} =
    2-\frac{1}{2^{\indexLast-1}}.
  \end{equation*}
  This proves the required equality to end the proof.
\end{proof}

\subsubsection{Convex combinations of pure payoffs suffice}\label{section:running:proofs:mixing}
Our goal is to prove that $\paySet{\payoffTuple}{\mdpState_0}=\convex{\paySetPure{\payoffTuple}{\mdpState_0}}$.
The inclusion $\convex{\paySetPure{\payoffTuple}{\mdpState_0}}\subseteq\paySet{\payoffTuple}{\mdpState_0}$ is direct by convexity of sets of expected payoffs (see Section~\ref{section:payoffs:multi}).
For the other inclusion, we break the proof down into three arguments like in the sketch above.
First, we show that $\paySetPure{\payoffTuple}{\mdpState_0}$ is closed.
\begin{lemma}\label{lemma:running:pure:closed}
  The set $\paySetPure{\payoffTuple}{\mdpState_0}\subseteq\IR^2$ is closed.
\end{lemma}
\begin{proof}
  Let $\payoffVect = (\payoffComp_1, \payoffComp_2)\in\IR^2$ such that there exists a sequence $(\payoffVectB^{(\indexSequence)})_{\indexSequence\in\IN}\subseteq\paySetPure{\payoffTuple}{\mdpState}$ such that $\lim_{\indexSequence\to\infty}\payoffVectB^{(\indexSequence)}=\payoffVect$.
  We show that $\payoffVect\in\paySetPure{\payoffTuple}{\mdpState}$.
  We observe that necessarily, $\payoffVect\in\ccInt{0}{5}\times\ccInt{0}{2}$ because $\paySetPure{\payoffTuple}{\mdpState}\subseteq\ccInt{0}{5}\times\ccInt{0}{2}$ (this inclusion follows from Lemma~\ref{lemma:running:pure:description}).

  We distinguish three cases.
  First, we assume that $\payoffVectB^{(\indexSequence)}=(0, 2)$ infinitely often.
  In this case, the constant sequence $((0, 2))_{\indexSequence\in\IN}$ is a subsequence of $(\payoffVectB^{(\indexSequence)})_{\indexSequence\in\IN}$.
  Since a convergent sequence has the same limit as its subsequences, this implies that $\payoffVect=(0, 2)$.
  We assume, for the remainder of the proof, that $(0, 2)$ does not occur in the sequence $(\payoffVectB^{(\indexSequence)})_{\indexSequence\in\IN}$.
  
  Second, we assume that $\payoffComp_2\neq 2$.
  In this case, there exists $\indexLast\in\IN$ such that $\payoffComp_2 \leq 2 - \frac{1}{2^{\indexLast-1}}$.
  This, along with Lemma~\ref{lemma:running:pure:description}, implies that there is a suffix of $(\payoffVectB^{(\indexSequence)})_{\indexSequence\in\IN}$ that is included in the finite subset $\left\{\left(1+\frac{3^\indexPosition}{4^{\indexPosition-1}}, 2-\frac{1}{2^{\indexPosition-1}}\right)\mid\indexPosition\leq\indexLast\right\}$ of $\paySetPure{\payoffTuple}{\mdpState}$.
  We conclude that $\payoffVect\in\paySetPure{\payoffTuple}{\mdpState}$ as all finite sets are closed.

  Finally, we assume that $\payoffComp_2 = 2$ and show that $\payoffVect = (1, 2)$.
  If there are infinitely many $\indexSequence\in\IN$ such that $\payoffVectB^{(\indexSequence)}=(1, 2)$, then $\payoffVect = (1, 2)$ (refer to the first case).
  We therefore assume that $(1, 2)$ does not occur in $(\payoffVectB^{(\indexSequence)})_{\indexSequence\in\IN}$.
  For all $\indexSequence\in\IN$, we let $\indexPosition_\indexSequence\in\IN$ such that $\payoffVectB^{(\indexSequence)} = \left(1+\frac{3^{\indexPosition_\indexSequence}}{4^{\indexPosition_\indexSequence-1}}, 2-\frac{1}{2^{\indexPosition_\indexSequence-1}}\right)$.
  From $\payoffComp_2 = 2$, we obtain that $\lim_{\indexSequence\to\infty}\frac{1}{2^{\indexPosition_\indexSequence-1}} = 0$.
  It follows that
  \[\payoffComp_1 =
    \lim_{\indexSequence\to\infty}3\cdot
    \left(\frac{3}{4}\right)^{\indexPosition_\indexSequence-1} =
    \lim_{\indexSequence\to\infty}3\cdot
    \left(\frac{1}{2}\right)^{(\indexPosition_\indexSequence-1)\cdot\log_2(4/3)} = 0.\]
  This implies that $\payoffVect = (1, 2)\in\paySetPure{\payoffTuple}{\mdpState}$.
  This ends the proof that $\paySetPure{\payoffTuple}{\mdpState}$ is closed.
\end{proof}

We now prove a general result that implies that $\convex{\paySetPure{\payoffTuple}{\mdpState_0}}$ is closed.
We show that the convex hull of any compact subset of $\IR^\numObj$ is also compact.
Lemma~\ref{lemma:running:pure:closed} allows us to apply it to $\paySetPure{\payoffTuple}{\mdpState_0}$.
Recall that a subset of $\IR^\numObj$ is compact if and only if it is closed and bounded.

\begin{lemma}\label{lemma:convex hull of compact}
  Let $\numObj\geq 1$.
  Let $D\subseteq\IR^\numObj$.
  If $D$ is compact, then $\convex{D}$ is also compact.
\end{lemma}
\begin{proof}
  Assume that $D$ is compact.
  We assume that $D$ is non-empty, as otherwise the result is direct.

  We first show that $\convex{D}$ is bounded.
  Let $\payoffVect\in\convex{D}$.
  Let $\scalar_1$, \ldots, $\scalar_{\indexSequence}$ be convex combination coefficients and let $\payoffVectB_1$, \ldots, $\payoffVectB_{\indexSequence}\in D$ such that $\payoffVect = \sum_{\indexSequenceB=1}^{\indexSequence}\scalar_\indexSequenceB\payoffVectB_\indexSequenceB$.
  By triangulation, we obtain that $\|\payoffVect\|_2 \leq \sum_{\indexSequenceB=1}^{\indexSequence}\scalar_\indexSequenceB\|\payoffVectB_\indexSequenceB\|_2 \leq \sup\{\|\payoffVectB\|_2\mid\payoffVectB\in D\}\in\IR$, where the second inequality is a consequence of $\sum_{\indexSequenceB=1}^{\indexSequence}\scalar_\indexSequenceB=1$.
  It follows that $\convex{D}$ is bounded.

  We now show that $\convex{D}$ is closed.
  Let $\payoffVect\in\IR^\numObj$ such that there exists a sequence $(\payoffVect^{(\indexSequence)})_{\indexSequence\in\IN}\subseteq\convex{D}$ such that $\lim_{\indexSequence\to\infty}\payoffVect^{(\indexSequence)} = \payoffVect$.
  By Carath\'{e}odory's theorem for convex hulls (Theorem~\ref{theorem:caratheodory:convex}), all elements of the sequence $(\payoffVect^{(\indexSequence)})_{\indexSequence\in\IN}$ are a convex combination of no more than $\numObj+1$ elements of $D$.
  For all $\indexSequence\in\IN$, let $\scalar_1^{(\indexSequence)}$, \ldots, $\scalar_{\numObj+1}^{(\indexSequence)}$ be convex combination coefficients and let $\payoffVectB_1^{(\indexSequence)}$, \ldots, $\payoffVectB_{\numObj+1}^{(\indexSequence)}$ such that $\payoffVect^{(\indexSequence)} = \sum_{\indexPayoff=1}^{\numObj+1}\scalar_\indexPayoff^{(\indexSequence)}\payoffVectB_\indexPayoff^{(\indexSequence)}$.
  By compactness of $\ccInt{0}{1}$ and $D$, we obtain an increasing sequence of natural numbers $(\indexSequence_\indexSequenceB)_{\indexSequenceB\in\IN}$, convex combination coefficients $\scalar_1$, \ldots, $\scalar_{\numObj+1}$ and $\payoffVectB_1$, \ldots, $\payoffVectB_{\numObj+1}\in D$ such that for all $1\leq\indexPayoff\leq\numObj+1$, $\lim_{\indexSequenceB\to\infty}\scalar_{\indexPayoff}^{(\indexSequence_\indexSequenceB)} = \scalar_\indexPayoff$ and $\lim_{\indexSequenceB\to\infty}\payoffVectB_{\indexPayoff}^{(\indexSequence_\indexSequenceB)} = \payoffVectB_\indexPayoff$.
  It follows (from the uniqueness of the limit) that $\payoffVect = \sum_{\indexPayoff=1}^{\numObj+1}\scalar_\indexPayoff\payoffVectB_\indexPayoff$.
  This shows that $\payoffVect\in\convex{D}$ and ends the proof that $\convex{D}$ is closed.
\end{proof}

Finally, we show that $\paySet{\payoffTuple}{\mdpState_0}\subseteq\convex{\paySetPure{\payoffTuple}{\mdpState_0}}$ by establishing that all elements of $\paySet{\payoffTuple}{\mdpState_0}$ are included in the closure of $\convex{\paySetPure{\payoffTuple}{\mdpState_0}}$.
We use the argument based on infinite convex combinations that was mentioned in Section~\ref{section:running:overview}.

\begin{lemma}\label{lemma:running:finite mixing}
  We have $\paySet{\payoffTuple}{\mdpState_0} = \convex{\paySetPure{\payoffTuple}{\mdpState_0}}$.
\end{lemma}
\begin{proof}
  By Lemma~\ref{lemma:running:pure:closed} and Lemma~\ref{lemma:convex hull of compact}, $\convex{\paySetPure{\payoffTuple}{\mdpState_0}}$ is closed.
  To end the proof, it thus suffices to show that $\paySet{\payoffTuple}{\mdpState_0}\subseteq\closure{\convex{\paySetPure{\payoffTuple}{\mdpState_0}}}$.
  We recall that there are countably many plays in $\mdp$.
  Let $(\play_\indexPosition)_{\indexPosition\in\IN}$ be an enumeration of the set of plays of $\mdp$ from $\mdpState_0$.
  Let $\stratMDP\in\stratClassAll{\mdp}$ and let $\payoffVect=\expectancy^{\stratMDP}_{\mdpState_0}(\payoffTuple)$.
  We can write
  \begin{equation}\label{equation:running:finite mixing}
    \payoffVect = \sum_{\indexPosition=0}^\infty\proba^{\stratMDP}_{\mdpState_0}(\{\play_\indexPosition\})\cdot\payoffTuple(\play_\indexPosition).
  \end{equation}

  Next, we observe that $\paySetPure{\payoffTuple}{\mdpState_0}$ is the set of payoffs of plays from $\mdpState_0$ in $\mdp$ due to the absence of random transitions.
  Therefore, informally, Equation~\eqref{equation:running:finite mixing} implies that $\payoffVect$ is an infinite convex combination of elements of $\paySetPure{\payoffTuple}{\mdpState_0}$.
  To end the proof, it suffices to show that we can approach this infinite convex combination with finite convex combinations of elements of $\paySetPure{\payoffTuple}{\mdpState_0}$.
  For all $\indexLast\in\IN$, let
  \[\payoffVectB^{(\indexLast)} =
    \left(
      \proba^{\stratMDP}_{\mdpState_0}(\{\play_0\}) +
      \sum_{\indexPosition=\indexLast+1}^\infty\proba^{\stratMDP}_{\mdpState_0}(\{\play_\indexPosition\})
    \right)\cdot
    \payoffTuple(\play_0) +
    \sum_{\indexPosition=1}^\indexLast\proba^{\stratMDP}_{\mdpState_0}(\{\play_\indexPosition\})\cdot
    \payoffTuple(\play_\indexPosition).
  \]
  For all $\indexLast\in\IN$, $\payoffVectB^{(\indexLast)}$ is derived from the partial sum of the series in Equation~\eqref{equation:running:finite mixing} up to index $\indexLast$ by increasing the coefficient of $\payoffTuple(\play_0)$ enough to obtain a convex combination of pure payoffs.
  For all $\indexLast\in\IN$, we have $\payoffVectB^{(\indexLast)}\in\convex{\paySetPure{\payoffTuple}{\mdpState_0}}$ and $\payoffVect = \lim_{\indexLast\to\infty}\payoffVectB^{(\indexLast)}$.
  This shows that $\payoffVect\in\closure{\convex{\paySetPure{\payoffTuple}{\mdpState_0}}}$ and ends the proof.
\end{proof}

\subsubsection{One dimension is not enough to show that extreme points are pure}
We formalise the intuition given on Figure~\ref{figure:running:payoff set} that the only hyperplane supporting $\paySet{\payoffTuple}{\mdpState_0}$ at $(1, 2)$ is the line carrying the segment $\ccInt{(0, 2)}{(1, 2)}$.
We reach a contradiction by assuming that there is a linear form $\linForm$ such that $\linForm((1, 2)) > \linForm((0, 2))$ and $\linForm((1, 2))\geq\linForm(\payoffVectB)$ for all $\payoffVectB\in\paySet{\payoffTuple}{\mdpState_0}$.
\begin{lemma}
  Let $\payoffVect = (1, 2)$.
  The unique supporting hyperplane of $\paySet{\payoffTuple}{\mdpState_0}$ containing $\payoffVect$ is the line defined by the equation $y=2$.
\end{lemma}
\begin{proof}
  We assume towards a contradiction that there exists a linear form $\linForm\colon\IR^2\to\IR$ such that for all $\payoffVectB\in\paySet{\payoffTuple}{\mdpState_0}$, $\linForm(\payoffVect) \geq\linForm(\payoffVectB)$ and $\linForm(\payoffVect) \neq\linForm((0, 2))$.
  Let $\scalar, \scalarB\in\IR$ such that for all $\vect=(\vectComp_1, \vectComp_2)\in\IR^2$, $\linForm(\vect) = \scalar\vectComp_1+\scalarB\vectComp_2$.

  We observe that $\linForm(\payoffVect) > \linForm((0, 2))$ is equivalent to $\scalar > 0$.
  We also have, for all $\indexPosition\in\IN$, $\linForm(\payoffVect) \geq \linForm((1+\frac{3^\indexPosition}{4^{\indexPosition-1}}, 2 - \frac{1}{2^{\indexPosition-1}}))$, i.e., $\scalarB\geq\frac{3^\indexPosition}{2^{\indexPosition-1}}\cdot\scalar$.
  Together, these statements imply that $\scalarB$ must be greater than all real numbers, which is a contradiction.
\end{proof}

We have shown that we cannot necessarily isolate extreme points of sets of expected payoffs by intersecting them with a single hyperplane.

\subsubsection{Extreme points and Pareto-optimality}
We now show that $\corners{\paySet{\payoffTuple}{\mdpState_0}} = \paySetPure{\payoffTuple}{\mdpState_0}$ in the context of this example (this property does not hold in full generality).
It follows from $\paySet{\payoffTuple}{\mdpState_0} = \convex{\paySetPure{\payoffTuple}{\mdpState_0}}$ that all extreme points of $\paySet{\payoffTuple}{\mdpState_0}$ are the payoff of a pure strategy.
It remains to show that $\paySetPure{\payoffTuple}{\mdpState_0}\subseteq\corners{\paySet{\payoffTuple}{\mdpState_0}}$.
We first observe that $(0, 2)\in\corners{\paySet{\payoffTuple}{\mdpState_0}}$ because it is with the least first component among all elements of $\paySetPure{\payoffTuple}{\mdpState_0}$.
The main difficulty lies with the elements of $\paySetPure{\payoffTuple}{\mdpState_0}\setminus\{(0, 2)\}$.

We handle the remaining points with a two-part argument.
First, we show that any non-extreme element of the boundary of $\paySet{\payoffTuple}{\mdpState_0}$ is a convex combination of two extreme points.
This implies that all Pareto-optimal elements of $\paySet{\payoffTuple}{\mdpState_0}$ are convex combinations of no more than two vectors.
Second, we prove that for all vectors $\payoffVect\in\paySetPure{\payoffTuple}{\mdpState_0}\setminus\{(0, 2)\}$, $\payoffVect$ is either incomparable or strictly greater (with respect to the component-wise ordering) to convex combinations of any two vectors of $\paySet{\payoffTuple}{\mdpState_0}\setminus\{\payoffVect\}$.
We prove this by reasoning on a strictly concave real function whose graph includes $\paySetPure{\payoffTuple}{\mdpState_0}\setminus\{(0, 2)\}$.
The graphical intuition is as follows: any segment joining two points of the graph of a strictly concave function is beneath the curve, so any points on the segment that are comparable to $\payoffVect$ must be smaller.
This approach also yields that all vectors of $\paySetPure{\payoffTuple}{\mdpState_0}\setminus\{(0, 2)\}$ are Pareto-optimal elements of $\paySetPure{\payoffTuple}{\mdpState_0}$.

We now prove a generalisation of the first property formulated above: given a compact set $D\subseteq\IR^2$, any vector in the boundary of $\convex{D}$ is a convex combination of no more than two elements of $D$.
This can be seen as a refinement of Carathéodory's theorem for convex hulls (Theorem~\ref{theorem:caratheodory:convex}) when considering the boundary of the convex hull of a compact set in a two-dimensional setting.
\begin{lemma}\label{lemma:running:convex:border}
  Let $D\subseteq\IR^2$ be compact.
  For all $\payoffVect\in\border{\convex{D}}$, either $\payoffVect\in\corners{\convex{D}}$ or $\payoffVect$ is a convex combination of two vectors of $D\setminus\{\payoffVect\}$.
\end{lemma}
\begin{proof}
  If $D$ is empty or a singleton set, the result is direct.
  We thus assume that $D$ has at least two elements.
  Let $\payoffVect\in\border{\convex{D}}$.
  Because $D$ is compact, $\convex{D}$ is closed (Lemma~\ref{lemma:convex hull of compact}) and thus $\border{\convex{D}}\subseteq\convex{D}$.
  We let $\payoffVectB^{(1)}$, \ldots, $\payoffVectB^{(\indexSequence)}\in D$ and $\scalar_1$, \ldots, $\scalar_\indexSequence\in\ocInt{0}{1}$ be non-zero convex combination coefficients such that $\payoffVect = \sum_{\indexSequenceB=1}^\indexSequence\scalar_\indexSequence\payoffVectB^{(\indexSequenceB)}$.

  We first show that the vectors $\payoffVect$, $\payoffVectB^{(1)}$, \ldots, $\payoffVectB^{(\indexSequence)}$ lie on a single line.
  If $\spanAff{D}$ is a line, then this is direct.
  We thus assume that $\spanAff{D}$ is not a line.
  We obtain that $\spanAff{D}=\IR^2$: it contains a line because $D$ has at least two elements, and therefore its dimension must be two.
  This implies that $\relInt{D} = \interior{D}$, and thus that $\payoffVect\notin\relInt{D}$.
  By the supporting hyperplane theorem (Theorem~\ref{thm:hyperplane:supporting}), there exists a non-zero linear form $\linForm$ such that for all $\payoffVectB\in\convex{D}$, $\linForm(\payoffVect)\geq\linForm(\payoffVectB)$.
  It follows that for all $1\leq\indexSequenceB\leq\indexSequence$, we have $\linForm(\payoffVectB^{(\indexSequenceB)}) = \linForm(\payoffVect)$ (by linearity, because $\scalar^{(\indexSequenceB)}\neq 0$).
  In other words, $\payoffVect$ and the $\payoffVectB^{(\indexSequenceB)}$ lie on the line $(\linForm)^{-1}(\linForm(\payoffVect))$.
  
  We have shown that $\payoffVect$ and the $\payoffVectB^{(\indexSequenceB)}$ lie on a single line.
  There exists a non-zero vector $\vect\in\IR^2$ such that for all $1\leq\indexSequenceB\leq\indexSequence$, there exists $\scalarB_{\indexSequenceB}\in\IR$ such that $\payoffVectB^{(\indexSequenceB)} = \payoffVect + \scalarB_{\indexSequenceB}\cdot\vect$.
  There must exist $1\leq{\indexSequenceB},{\indexSequenceB'}\leq\indexSequence$ such that $\scalarB_{\indexSequenceB}\geq0$ and $\scalarB_{\indexSequenceB'}\leq0$.
  We conclude that $\payoffVect\in\ccInt{\payoffVectB^{(\indexSequenceB)}}{\payoffVectB^{(\indexSequenceB')}}$.
  We have shown that $\payoffVect$ is a convex combination of at most two elements of $D$.
\end{proof}

We now introduce the strictly concave function $\mathcal{F}\colon\ccInt{1}{5}\to\IR$ used in the remainder of our argument.
For all $x\in\ccInt{1}{5}$, we let
\[\mathcal{F}(x) = 2 - \left(\frac{x-1}{3}\right)^{\log_{4/3}(2)} \]
We observe that $\mathcal{F}$ is well-defined in $0$ because $\log_{4/3}(2) > 0$.
We recall that $\mathcal{F}$ is strictly concave if and only if  for all $x, x'\in\ccInt{1}{5}$ such that $x<x'$ and all $\scalarB\in\ooInt{0}{1}$, we have $\scalarB\mathcal{F}(x) + (1-\scalarB)\mathcal{F}(x')<\mathcal{F}(\scalarB x + (1-\scalarB)x')$.
We now show that $\mathcal{F}$ is decreasing and strictly concave.

\begin{lemma}\label{lemma:running:function:concave}
  The function $\mathcal{F}$ is decreasing and strictly concave.
\end{lemma}
\begin{proof}
  To prove that $\mathcal{F}$ is decreasing (resp.~strictly concave), it suffices to show that it is differentiable over $\ooInt{1}{5}$ and its derivative $\mathcal{F}'$ is strictly negative (resp.~decreasing).
  For the strict concavity of $\mathcal{F}$, in practice, we show that the second derivative $\mathcal{F}''$ of $\mathcal{F}$ (defined over $\ooInt{1}{5}$) is negative.
  For all $x\in\ooInt{1}{5}$, we have
  \[\mathcal{F}'(x) = - \log_{4/3}(2) \cdot \left(\frac{x-1}{3}\right)^{\log_{4/3}(2)-1}\]
  and
  \[\mathcal{F}''(x) = - \log_{4/3}(2) \cdot \left(\log_{4/3}(2) - 1\right) \cdot \left(\frac{x-1}{3}\right)^{\log_{4/3}(2)-2}.\]
  To prove that $\mathcal{F}'$ and $\mathcal{F}''$ are negative over $\ooInt{1}{5}$, it suffices to show that $\log_{4/3}(2)>0$ and  $\log_{4/3}(2) - 1>0$.
  The second inequality implies the first and can be shown to be equivalent to $2 > \frac{4}{3}$.
  This ends the proof that $\mathcal{F}$ is decreasing and strictly concave.
\end{proof}

Next, we prove that $\paySetPure{\payoffTuple}{\mdpState_0}\setminus\{(0, 2)\}$ is included in the graph of $\mathcal{F}$.
\begin{lemma}\label{lemma:running:graph:inclusion}
  The graph of $\mathcal{F}$ includes $\paySetPure{\payoffTuple}{\mdpState_0}\setminus\{(0, 2)\}$, i.e., for all $(x, y)\in\paySetPure{\payoffTuple}{\mdpState_0}$, $\mathcal{F}(x) = y$.
\end{lemma}
\begin{proof}
  First, we observe that $\mathcal{F}(1) = 2$, i.e., $(1, 2)$ is in the graph of $\mathcal{F}$.
  It remains to show that for all $\indexPosition\in\IN$, we have
  \[\mathcal{F}\left(1+\frac{3^\indexPosition}{4^{\indexPosition-1}}\right) = 2-\frac{1}{2^{\indexPosition-1}}.\]
  We let $\indexPosition\in\IN$.
  We obtain the following equalities:
  \begin{align*}
    \mathcal{F}\left(1+\frac{3^\indexPosition}{4^{\indexPosition-1}}\right)
    & = 2 - \left(\frac{3^{\indexPosition-1}}{4^{\indexPosition-1}}\right)^{\log_{4/3}(2)}
    \\
    & = 2 - \left(\frac{3}{4}\right)^{(\indexPosition-1)\cdot\log_{4/3}(2)} \\
    & = 2 - \left(\left(\frac{4}{3}\right)^{\log_{4/3}(2)}\right)^{- (\indexPosition-1)} \\
    & = 2 - \frac{1}{2^{\indexPosition-1}}.
  \end{align*}
  We have shown that all elements of $\paySetPure{\payoffTuple}{\mdpState_0}\setminus\{(0, 2)\}$ are in the graph of $\mathcal{F}$.
\end{proof}

We now prove that all vectors $\payoffVect\in\paySetPure{\payoffTuple}{\mdpState_0}\setminus\{(0, 2)\}$ are not smaller than convex combinations of two elements of $\paySetPure{\payoffTuple}{\mdpState_0}\setminus\{\payoffVect\}$.
\begin{lemma}\label{lemma:running:two-convex}
  Let $\payoffVect\in\paySetPure{\payoffTuple}{\mdpState_0}\setminus\{(0, 2)\}$.
  For all $\payoffVectB_1, \payoffVectB_2\in\paySetPure{\payoffTuple}{\mdpState_0}\setminus\{\payoffVect\}$ and all $\scalar\in\ccInt{0}{1}$, $\payoffVect$ is incomparable to or strictly greater than $\scalar\cdot\payoffVectB_1+(1-\scalar)\cdot\payoffVectB_2$.
\end{lemma}
\begin{proof}
  Let $\payoffVect = (\payoffComp_1, \payoffComp_2)$.
  Let $\payoffVectB_1, \payoffVectB_2\in\paySetPure{\payoffTuple}{\mdpState_0}\setminus\{\payoffVect\}$.
  We fix $\scalar\in\ooInt{0}{1}$; the above statement for $\scalar\in\{0, 1\}$ is covered by the cases $\payoffVectB_1=\payoffVectB_2$.
  We let $\payoffVectB = \scalar\cdot\payoffVectB_1+(1-\scalar)\cdot\payoffVectB_2$.
  In the first part of this proof, we assume that $\payoffVectB_1\neq (0, 2)$ and $\payoffVectB_2\neq(0, 2)$ and discuss the case when this assumption is lifted at the end of the proof.

  By Lemma~\ref{lemma:running:graph:inclusion}, we have $\payoffComp_2=\mathcal{F}(\payoffComp_1)$ and we can write $\payoffVectB_1 = (x_1, \mathcal{F}(x_1))$ and $\payoffVectB_2 = (x_2, \mathcal{F}(x_2))$ where $x_1$ and $x_2$ are the first components of $\payoffVectB_1$ and $\payoffVectB_2$ respectively.
  We assume without loss of generality that $x_1\leq x_2$.
  All elements of $\paySetPure{\payoffTuple}{\mdpState_0}$ have different $x$-components, hence  $x_1\neq\payoffComp_1$ and $x_2\neq\payoffComp_1$.

  We first assume that $x_1\leq x_2< \payoffComp_1$.
  In this case, by Lemma~\ref{lemma:running:function:concave} ($\mathcal{F}$ is strictly decreasing), we have $\mathcal{F}(x_1) \geq \mathcal{F}(x_2) > \payoffComp_1$.
  It follows that $\payoffVect$ and $\payoffVectB$ are incomparable: the first component of $\payoffVect$ is greater than that of $\payoffVectB$ but the second component of $\payoffVect$ is smaller than that of $\payoffVectB$.
  In the case that $\payoffComp_1 <x_1\leq x_2$, we obtain in a similar fashion that $\payoffVect$ and $\payoffVectB$ are incomparable.

  We now assume that $x_1 < \payoffComp_1 < x_2$.
  We let $\scalarB\in\ooInt{0}{1}$ such that $\payoffComp_1 = \scalarB\cdot x_1+(1-\scalarB) x_2$, which exists because $\payoffComp_1\in\ooInt{x_1}{x_2}$.
  If $\payoffVectB$ is not comparable to $\payoffVect$, there is nothing to show.
  We assume that these two vectors are comparable and show that $\payoffVectB < \payoffVect$.
  We proceed by contradiction and assume that $\payoffVectB \geq \payoffVect$.
  
  We consider the linear form $\linForm\colon\IR^2\to\IR$ defined by $\linForm(\vect) = - \frac{\mathcal{F}(x_2) - \mathcal{F}(x_1)}{x_2-x_1}\cdot\vectComp_1 + \vectComp_2$ for all $\vect=(\vectComp_1, \vectComp_2)\in\IR^2$.
  We have $\linForm(\payoffVectB_1)=\linForm(\payoffVectB_2)$ and thus $(\linForm)^{-1}(\linForm(\payoffVectB_1))$ is the line carrying the segment $\ccInt{\payoffVectB_1}{\payoffVectB_2}$.
  Because $\mathcal{F}$ is decreasing (Lemma~\ref{lemma:running:function:concave}), we have $\frac{\mathcal{F}(x_2) - \mathcal{F}(x_1)}{x_2-x_1} < 0$.
  This implies that $\linForm$ is increasing in the sense that for all $\vect_1, \vect_2\in\IR^2$, $\vect_1 < \vect_2$ implies that $\linForm(\vect_1) < \linForm(\vect_2)$.

  Let $\vect = (\payoffComp_1, \scalarB\cdot \mathcal{F}(x_1)+(1-\scalarB) \mathcal{F}(x_2))$.
  We have $\vect < \payoffVect$ by strict concavity of $\mathcal{F}$ (Lemma~\ref{lemma:running:function:concave}).
  Furthermore, we have $\linForm(\vect) = \linForm(\payoffVectB)$ because both of these vectors are in the segment $\ccInt{\payoffVectB_1}{\payoffVectB_2}$.
  We obtain, because $\linForm$ is increasing and $\vect <\payoffVect < \payoffVectB$, that $\linForm(\vect) < \linForm(\payoffVect) < \linForm(\payoffVectB) = \linForm(\vect)$.
  This is a contradiction.
  This ends the proof in the case that $\payoffVectB_1\neq(0, 2)$ and $\payoffVectB_2\neq(0, 2)$.

  Next, we assume that $\payoffVectB_1=\payoffVectB_2 = (0, 2)$.
  We obtain that $\payoffVectB = (0, 2)$, which is smaller than $(1, 2)$ and incomparable to all elements of $\paySetPure{\payoffTuple}{\mdpState_0}\setminus\{(1, 2), (0, 2)\}$.
  Finally, we assume that only one of $\payoffVectB_1$ or $\payoffVectB_2$ are $(0, 2)$.
  We assume without loss of generality that $\payoffVectB_1=(0, 2)$.
  If $\payoffVectB > \payoffVect$, then we would have $\scalar\cdot(1, 2) +(1-\scalar)\cdot\payoffVectB_2 > \payoffVectB > \payoffVect$, which would contradict the first part of the proof.
\end{proof}

With Lemma~\ref{lemma:running:convex:border} and Lemma~\ref{lemma:running:two-convex}, we directly obtain that all elements of $\paySetPure{\payoffTuple}{\mdpState_0}\setminus\{(0, 2)\}$ are Pareto-optimal elements of $\paySet{\payoffTuple}{\mdpState_0}$: if they were not Pareto-optimal, then they would be dominated by an element of the boundary of $\paySet{\payoffTuple}{\mdpState_0}$ which would be the convex combination of two vectors in $\paySetPure{\payoffTuple}{\mdpState_0}$.
\begin{lemma}\label{lemma:running:pareto}
  Let $\payoffVect\in\paySetPure{\payoffTuple}{\mdpState_0}\setminus\{(0, 2)\}$.
  Then $\payoffVect$ is a Pareto-optimal element of $\paySet{\payoffTuple}{\mdpState_0}$.
\end{lemma}
\begin{proof}
  Assume towards a contradiction that there exists $\payoffVectB\in\paySet{\payoffTuple}{\mdpState_0}$ such that $\payoffVect < \payoffVectB$.
  We show that there exists $\payoffVectB'\in\border{\paySet{\payoffTuple}{\mdpState_0}}$  such that $\payoffVectB' > \payoffVect$.
  This yields a contradiction with Lemma~\ref{lemma:running:convex:border} and Lemma~\ref{lemma:running:two-convex}.

  If $\payoffVectB\in\border{\paySet{\payoffTuple}{\mdpState_0}}$, then we let $\payoffVectB'=\payoffVectB$.
  We thus assume that $\payoffVectB\in\interior{\paySet{\payoffTuple}{\mdpState_0}}$.
  Let $\scalar=\sup\{\scalarB\geq 0\mid \payoffVectB + \scalarB\oneVect\in\paySet{\payoffTuple}{\mdpState_0}\}$.
  We have $\scalar\in\IR$ because $\payoffVectB\in\paySet{\payoffTuple}{\mdpState_0}$ (i.e., $0\in\{\scalarB\geq 0\mid \payoffVectB + \scalarB\oneVect\in\paySet{\payoffTuple}{\mdpState_0}\}$ thus $\scalar > -\infty$) and $\paySet{\payoffTuple}{\mdpState_0}$ is bounded (thus $\scalar < +\infty)$.
  Furthermore, because $\paySet{\payoffTuple}{\mdpState_0}$ is closed, we have $\payoffVectB + \scalar\oneVect\in\border{\paySet{\payoffTuple}{\mdpState_0}}$.
  We obtain the announced contradiction by letting $\payoffVectB' = \payoffVectB + \scalar\oneVect > \payoffVect$.
\end{proof}

We conclude this section by showing that the set of extreme points of $\paySet{\payoffTuple}{\mdpState_0}$ is the set of pure payoffs.
\begin{lemma}\label{lemma:running:pure extreme}
  We have $\corners{\paySet{\payoffTuple}{\mdpState_0}} =  \paySetPure{\payoffTuple}{\mdpState_0}$.
\end{lemma}
\begin{proof}
  First, we show that $\corners{\paySet{\payoffTuple}{\mdpState_0}} \subseteq \paySetPure{\payoffTuple}{\mdpState_0}$.
  It suffices to show that all vectors $\payoffVect\in\paySet{\payoffTuple}{\mdpState_0}\setminus\paySetPure{\payoffTuple}{\mdpState_0}$ are not extreme points of $\paySet{\payoffTuple}{\mdpState_0}$.
  Let $\payoffVect\in\paySet{\payoffTuple}{\mdpState_0}\setminus\paySetPure{\payoffTuple}{\mdpState_0}$.
  By Lemma~\ref{lemma:running:finite mixing}, we have $\payoffVect\in\paySet{\payoffTuple}{\mdpState_0} = \convex{\paySetPure{\payoffTuple}{\mdpState_0}}$.
  This implies that $\payoffVect\in\convex{\paySet{\payoffTuple}{\mdpState_0}\setminus\{\payoffVect\}}$ because $\payoffVect\notin\paySetPure{\payoffTuple}{\mdpState_0}$, and thus $\payoffVect\notin\corners{\paySet{\payoffTuple}{\mdpState_0}}$.

  Conversely, let $\payoffVect\in\paySetPure{\payoffTuple}{\mdpState_0}$.
  First, we assume that $\payoffVect = (0, 2)$.
  The first coordinate of all other vectors of $\paySetPure{\payoffTuple}{\mdpState_0}$ is greater than or equal to $1$.
  This therefore also applies to any convex combination of vectors in $\paySetPure{\payoffTuple}{\mdpState_0}\setminus\{(0, 2)\}$.
  It follows that $\payoffVect\in\corners{\paySet{\payoffTuple}{\mdpState_0}}$.
  Second, we assume that $\payoffVect\neq(0, 2)$.
  Assume towards a contradiction that $\payoffVect\notin\corners{\paySet{\payoffTuple}{\mdpState_0}}$.
  By Lemma~\ref{lemma:running:pareto}, $\payoffVect$ is a Pareto-optimal element of $\paySet{\payoffTuple}{\mdpState_0}$, and thus lies on the boundary of this set.
  We obtain that $\payoffVect$ is the convex combination of two elements of $\paySetPure{\payoffTuple}{\mdpState_0}\setminus\{\payoffVect\}$ by Lemma~\ref{lemma:running:convex:border} (which is applicable here by Lemmas~\ref{lemma:running:finite mixing} and~\ref{lemma:running:pure:closed}) and the assumption that $\payoffVect\notin\corners{\paySet{\payoffTuple}{\mdpState_0}}$.
  This yields a contradiction with Lemma~\ref{lemma:running:two-convex}, which states that there is no convex combination of two vectors of $\paySetPure{\payoffTuple}{\mdpState_0}\setminus\{\payoffVect\}$ that is greater than or equal to $\payoffVect$.
\end{proof}

\subsubsection{Memory cannot be traded for randomness}
We prove that the only way to obtain an expected payoff in $\paySetPure{\payoffTuple}{\mdpState_0}$ is through a strategy that induces a single play from $\mdpState_0$ (i.e., intuitively, a strategy that is pure in practice).

\begin{lemma}\label{lemma:running:strategies}
  Let $\stratMDP\in\stratClassAll{\mdp}$ be a strategy such that at least two plays starting in $\mdpState_0$ are consistent with $\stratMDP$.
  Then $\expectancy^{\stratMDP}_{\mdpState_0}(\payoffTuple)\notin\paySetPure{\payoffTuple}{\mdpState}$.
\end{lemma}
\begin{proof}
  We consider the following enumeration of the set of plays of $\mdp$ that start in $\mdpState_0$.
  We let $\play_{-2}=\mdpState_0\mdpActionC(\mdpState_1\mdpAction)^\omega$, $\play_{-1}=\mdpState_0\mdpAction(\mdpState_2\mdpAction)^\omega$, and, for all $\indexLast\in\IN$, we let $\play_{\indexLast}=\mdpState_0(\mdpAction\mdpState_2)^\indexLast\mdpActionB(\mdpState_3\mdpAction)^\omega$.
  We have $\payoffTuple(\play_{-2}) = (0, 2)$, $\payoffTuple(\play_{-1}) = (1, 2)$ and, for all $\indexLast\in\IN$, $\payoffTuple(\play_{\indexLast}) = \left(1+\frac{3^{\indexLast}}{4^{\indexLast-1}}, 2-\frac{1}{2^{\indexLast-1}}\right)$ (refer to the proof of Lemma~\ref{lemma:running:pure:description} for the relevant computations).
  Due to the absence of randomised transitions in $\mdp$, $\paySetPure{\payoffTuple}{\mdpState_0}$ is the set of payoffs of the plays starting in $\mdpState_0$.
  Therefore, to end the proof, we must show that for all $\indexLast\geq-2$, $\expectancy^{\stratMDP}_{\mdpState_0}(\payoffTuple)\neq\payoffTuple(\play_\indexLast)$.
  
  For all $\indexLast\in\{-1, -2\}\cup\IN$, let $\scalar_\indexLast = \proba^{\stratMDP}_{\mdpState_0}(\{\play_\indexLast\})$.
  Through this notation, we obtain that
  \begin{equation}\label{equation:running:strategies}
    \expectancy^{\stratMDP}_{\mdpState_0}(\payoffTuple) = \sum_{\indexLast\geq-2}\scalar_\indexLast\payoffTuple(\play_\indexLast).
  \end{equation}

  First, we show that $\expectancy^{\stratMDP}_{\mdpState_0}(\payoffTuple)\notin\{(0, 2), (1, 2)\}$.
  If there exists $\indexLast\in\IN$ such that $\scalar_\indexLast >0$, then we have $\expectancy^{\stratMDP}_{\mdpState_0}(\payoff_2) < 2$, which implies that $\expectancy^{\stratMDP}_{\mdpState_0}(\payoffTuple)\notin\{(0, 2), (1, 2)\}$.
  Next, assume that $\scalar_\indexLast=0$ for all $\indexLast\in\IN$.
  Then $\scalar_{-2}$ and $\scalar_{-1}$ must sum to one.
  Furthermore, since $\stratMDP$ has at least two outcomes, both $\scalar_{-2}$ and $\scalar_{-1}$ must be non-zero.
  It follows that $\expectancy^{\stratMDP}_{\mdpState_0}(\payoff_1) \in\ooInt{0}{1}$.
  We conclude that $\expectancy^{\stratMDP}_{\mdpState_0}(\payoffTuple)\notin\{(0, 2), (1, 2)\}$ in this case as well.

  We now fix $\indexLast\in\IN$ and show that $\expectancy^{\stratMDP}_{\mdpState_0}(\payoffTuple) \neq\payoffTuple(\play_\indexLast)$.
  We proceed by contradiction.
  Assume towards a contradiction that $\expectancy^{\stratMDP}_{\mdpState_0}(\payoffTuple)=\payoffTuple(\play_\indexLast)$.
  Our goal is to contradict Lemma~\ref{lemma:running:pure extreme}, i.e., to show that $\payoffTuple(\play_\indexLast)$ is not an extreme point of $\paySet{\payoffTuple}{\mdpState_0}$.
  We do so in two steps.
  First, we show that $\expectancy^{\stratMDP}_{\mdpState_0}\in\closure{\convex{\paySetPure{\payoffTuple}{\mdpState_0}\setminus\{\payoffTuple(\play_\indexLast)\}}}$. Next, we prove that $\convex{\paySetPure{\payoffTuple}{\mdpState_0}\setminus\{\payoffTuple(\play_\indexLast)\}}$ is closed.
  Together, these statements imply that $\payoffTuple(\play_\indexLast)\in\convex{\paySetPure{\payoffTuple}{\mdpState_0}\setminus\{\payoffTuple(\play_\indexLast)\}}$, which is the sought contradiction.

  We now establish that $\expectancy^{\stratMDP}_{\mdpState_0}\in\closure{\convex{\paySetPure{\payoffTuple}{\mdpState_0}\setminus\{\payoffTuple(\play_\indexLast)\}}}$.
  It follows from $\expectancy^{\stratMDP}_{\mdpState_0}(\payoffTuple) = \payoffTuple(\play_{\indexLast})$ and Equation~\eqref{equation:running:strategies} that $\expectancy^{\stratMDP}_{\mdpState_0}(\payoffTuple) = \sum_{\indexPosition\neq\indexLast}\frac{\scalar_{\indexPosition}}{1 - \scalar_\indexLast}\payoffTuple(\play_{\indexPosition})$.
  By an argument analogous to the one in the proof of Lemma~\ref{lemma:running:finite mixing}, we conclude from this last equality that $\expectancy^{\stratMDP}_{\mdpState_0}(\payoffTuple)\in\closure{\convex{\paySetPure{\payoffTuple}{\mdpState_0}\setminus\{\payoffTuple(\play_\indexLast)\}}}$.

  It remains to show that $\convex{\paySetPure{\payoffTuple}{\mdpState_0}\setminus\{\payoffTuple(\play_\indexLast)\}}$ is closed.  
  The vector $\payoffTuple(\play_\indexLast)$ is an isolated point of $\paySetPure{\payoffTuple}{\mdpState_0}$: all other elements of $\paySetPure{\payoffTuple}{\mdpState_0}$ are at distance at least $2^{-\indexLast}$ of $\payoffTuple(\play_\indexLast)$.
  Therefore, $\paySetPure{\payoffTuple}{\mdpState_0}\setminus\{\payoffTuple(\play_\indexLast)\}$ is closed (when removing an isolated point from a closed set, the resulting set is still closed), and thus $\convex{\paySetPure{\payoffTuple}{\mdpState_0}\setminus\{\payoffTuple(\play_\indexLast)\}}$ is closed by Lemma~\ref{lemma:convex hull of compact}.
\end{proof}

We close this section by commenting on the significance of Lemma~\ref{lemma:running:strategies} in terms of strategy complexity with respect to memory requirements.
When we only consider pure strategies, some Pareto-optimal payoffs may require strategies with an arbitrarily large memory, in the sense that we may need to count to some high (but finite) counter value to enact a given number of loops in $\mdpState_2$ to obtain certain expected payoffs.
Lemma~\ref{lemma:running:strategies} implies that we cannot substitute this counting by randomisation, even when only considering Pareto-optimal payoffs.
In particular, for this example, we obtain that although all expected payoffs can be obtained with strategies that only count up to a finite number of steps (i.e., to count loops in $\mdpState_2$), we cannot bound this number of steps uniformly for all expected payoff vectors.

\section{Lexicographic POMDPs}\label{section:lexico}
In this section, we consider lexicographic optimisation of multiple payoff functions in POMDPs.
We first show that for universally integrable payoffs, randomness does not provide additional power for lexicographic optimisation.
We then show that this property does not generalise to the whole class of universally unambiguously integrable payoffs, i.e., randomness may be required for lexicographically optimal strategies in general.

We fix a POMDP $\pomdp = \pomdpTuple$ and a $\numObj$-dimensional payoff function $\payoffTuple = (\payoff_\indexPayoff)_{1\leq\indexPayoff\leq\numObj}$.
We assume that $\payoffTuple$ is universally integrable and prove that for all initial states $\mdpState\in\mdpStateSpace$ and all strategies $\stratMDP$, there exists a pure strategy $\stratBMDP$ such that $\expectancy^{\stratMDP}_\mdpState(\payoffTuple)\leLex\expectancy^{\stratBMDP}_\mdpState(\payoffTuple)$.
The crux of the proof is showing that the Lebesgue integral is compatible with the lexicographic order over $\IRbar^\numObj$.
Once this is shown, we assume towards a contradiction that there is no suitable pure strategy $\stratBMDP$.
By Kuhn's theorem and Lemma~\ref{lem:expectancy:pure integral}, we can write $\expectancy^{\stratMDP}_\mdpState(\payoffTuple)$ as an integral over pure expected payoffs and then reach the contradiction that $\expectancy^{\stratMDP}_\mdpState(\payoffTuple)\lLex\expectancy^{\stratMDP}_\mdpState(\payoffTuple)$.
We formalise this argument below.

\begin{restatable}{theorem}{thmLexicoPure}\label{thm:lexico:pure}
  Assume that $\payoffTuple$ is universally integrable.
  Let $\stratMDP$ be a strategy and $\mdpState\in\mdpStateSpace$.
  There exists a pure strategy $\stratBMDP$ such that $\expectancy^{\stratMDP}_{\mdpState}(\payoffTuple) \leLex \expectancy^{\stratBMDP}_{\mdpState}(\payoffTuple)$.
\end{restatable}
\begin{proof}
  Let $\mixedStrat$ be a mixed strategy that is outcome-equivalent to $\stratMDP$ (whose existence follows from Kuhn's theorem).
  To prove the theorem, we reason on the $\mixedStrat$-integral of random variables of $(\stratClassPure{\pomdp}, \stratSigmaAlgebra)$ through Lemma~\ref{lem:expectancy:pure integral}.
  For any real or multivariate random variable $\rvB$ over $\stratClassPure{\pomdp}$, we write $\int\rvB\ud\mixedStrat$ for $\int_{\stratBMDP\in\stratClassPure{\pomdp}}\rvB(\stratBMDP)\ud\mixedStrat(\stratBMDP)$ to lighten notation.

  For all $1\leq\indexPayoff\leq\numObj$, we consider the real random variable  $\rv_\indexPayoff\colon \stratBMDP\mapsto \expectancy^{\stratBMDP}_\mdpState(\payoff_\indexPayoff)$ over $\stratClassPure{\pomdp}$.
  We let $\rvVect = (\rv_1, \ldots, \rv_\numObj)$.
  Because $\payoffTuple$ is universally integrable, $\rvVect$ is $\mixedStrat$-integrable and, by Lemma~\ref{lem:expectancy:pure integral}, we have $\expectation_{\mdpState}^{\mixedStrat}(\payoffTuple) = \int\rvVect\ud\mixedStrat$.

  Let $\rvBVect = (\rvB_1, \ldots, \rvB_\numObj)$ be an integrable multi-variate real random variable over $\stratClassPure{\pomdp}$.
  We first show that if $\rvVect\lLex\rvBVect$, then $\expectancy^{\stratMDP}_\mdpState(\payoffTuple) = \int\rvVect\ud\mixedStrat\lLex\int\rvBVect\ud\mixedStrat$.
  We use this claim below to prove the theorem.
  
  Assume that $\rvVect\lLex\rvBVect$.
  We partition $\stratClassPure{\pomdp}$ as follows.
  For all $1\leq\indexPayoff\leq\numObj$, we let
  \[
    \probaEvent_\indexPayoff = \left\{\stratBMDP\in\stratClassPure{\pomdp}\mid \rv_\indexPayoff(\stratBMDP) < \rvB_\indexPayoff(\stratBMDP)\text{ and } \forall \indexPayoff' <\indexPayoff,\, \rv_{\indexPayoff'}(\stratBMDP) = \rvB_{\indexPayoff'}(\stratBMDP)\right\}.\]
  Intuitively, $\probaEvent_\indexPayoff$ is the set of elements such that the strict lexicographic ordering of their respective images by $\rvVect$ and $\rvBVect$ is witnessed in component $\indexPayoff$.
  The sets $\probaEvent_1$, \ldots, $\probaEvent_\numObj$ partition $\stratClassPure{\pomdp}$ because $\rvVect\lLex\rvBVect$.
  It follows that, for $\rvCVect\in\{\rvVect, \rvBVect\}$, we have $\int\rvCVect\ud\mixedStrat = \sum_{\indexPayoff = 1}^\numObj\int\rvCVect\cdot\indic{\probaEvent_\indexPayoff}\ud\mixedStrat$.
  
  Let $\indexPayoff^\star = \min\{1\leq\indexPayoff\leq\numObj\mid\mixedStrat(\probaEvent_\indexPayoff) > 0\}$.
  To obtain that $\int\rvVect\ud\mixedStrat\lLex\int\rvBVect\ud\mixedStrat$, we show that, for $\indexPayoff < \indexPayoff^\star$, we have $\int\rv_\indexPayoff\ud\mixedStrat = \int\rvB_\indexPayoff\ud\mixedStrat$ and $\int\rv_{\indexPayoff^\star}\ud\mixedStrat < \int\rvB_{\indexPayoff^\star}\ud\mixedStrat$.
  To prove these relations, we formulate two observations.
  First, we observe that for all $1\leq\indexPayoff<\indexPayoff'\leq\numObj$, $\rv_\indexPayoff$ and $\rvB_\indexPayoff$ agree over $\probaEvent_{\indexPayoff'}$ (by definition), and thus we have
  \begin{equation}\label{equation:lexicographic pure:1}
    \int\rv_\indexPayoff\cdot\indic{\probaEvent_{\indexPayoff'}}\ud\mixedStrat = \int\rvB_\indexPayoff\cdot\indic{\probaEvent_{\indexPayoff'}}\ud\mixedStrat.
  \end{equation}
  Second, since $\mixedStrat(\probaEvent_{\indexPayoff'}) = 0$ for all $\indexPayoff'<\indexPayoff^\star$, it follows that for all $1\leq\indexPayoff\leq\numObj$ and all $\rvC_\indexPayoff\in\{\rv_\indexPayoff, \rvB_\indexPayoff\}$ that
  \begin{equation}\label{equation:lexicographic pure:2}
    \int\rvC_\indexPayoff\ud\mixedStrat=
    \sum_{\indexPayoff' =\indexPayoff^\star}^\numObj
    \int\rvC_\indexPayoff\cdot\indic{\probaEvent_{\indexPayoff'}}\ud\mixedStrat.
  \end{equation}

  Let $1\leq\indexPayoff< \indexPayoff^\star$.
  By combining Equations~\eqref{equation:lexicographic pure:1} and~\eqref{equation:lexicographic pure:2}, we obtain that
  \[\int\rv_\indexPayoff\ud\mixedStrat=
    \sum_{\indexPayoff' =\indexPayoff^\star}^\numObj
    \int\rv_\indexPayoff\cdot\indic{\probaEvent_{\indexPayoff'}}\ud\mixedStrat =
    \sum_{\indexPayoff' =\indexPayoff^\star}^\numObj
    \int\rvB_\indexPayoff\cdot\indic{\probaEvent_{\indexPayoff'}}\ud\mixedStrat =
    \int\rvB_\indexPayoff\ud\mixedStrat.
  \]
  To end the proof that $\int\rvVect\ud\mixedStrat\lLex\int\rvBVect\ud\mixedStrat$, it remains to show that $\int\rv_{\indexPayoff^\star}\ud\mixedStrat<\int\rvB_{\indexPayoff^\star}\ud\mixedStrat$.
  This inequality is equivalent to $\int\rv_{\indexPayoff^\star}\cdot\indic{\probaEvent_{\indexPayoff^\star}}\ud\mixedStrat
  < \int\rvB_{\indexPayoff^\star}\cdot\indic{\probaEvent_{\indexPayoff^\star}}\ud\mixedStrat$ by Equations~\eqref{equation:lexicographic pure:1} and~\eqref{equation:lexicographic pure:2}.
  This inequality holds by compatibility of the Lebesgue integral with the order of $\IR$ (recall that $\rv_{\indexPayoff^\star}(\stratBMDP) < \rvB_{\indexPayoff^\star}(\stratBMDP)$ for all $\stratBMDP\in\probaEvent_{\indexPayoff^\star}$) and because $\mixedStrat(\probaEvent_{\indexPayoff^\star}) > 0$.
  We have shown that $\int\rvVect\ud\mixedStrat\lLex\int\rvBVect\ud\mixedStrat$ whenever $\rvVect\lLex\rvBVect$ holds.

  Now assume that for all pure strategies $\stratBMDP$, we have $\expectancy^{\stratBMDP}_\mdpState(\payoffTuple) \lLex \expectancy^{\stratMDP}_\mdpState(\payoffTuple)$ towards a contradiction.
  It follows that $\rvVect(\stratBMDP) \lLex \expectancy^{\stratMDP}_\mdpState(\payoffTuple)$ for all $\stratBMDP\in\stratClassPure{\pomdp}$.
  From the above, we obtain that $\expectancy^{\stratMDP}_\mdpState(\payoffTuple) = \int\rvVect\ud\mixedStrat\lLex \expectancy^{\stratMDP}_\mdpState(\payoffTuple)$, which is a contradiction.
\end{proof}

A direct corollary of Theorem~\ref{thm:lexico:pure} is the following.
\begin{corollary}
  Assume that $\payoffTuple$ is universally integrable.
  For all $\mdpState\in\mdpStateSpace$, if there exists a lexicographically optimal strategy from $\mdpState$, then there exists a pure lexicographically optimal strategy.
\end{corollary}

We now present an example that shows that the conclusion of Theorem~\ref{thm:lexico:pure} does not hold for all universally unambiguously integrable payoffs already in finite MDPs.

    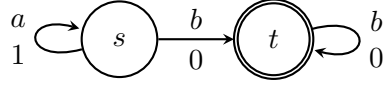
\begin{figure}
    \centering
    \begin{tikzpicture}
      \node[state, align=center] (s0) {$\mdpState$};
      \node[state, accepting, right = of s0] (s1) {$\mdpStateB$};
      \path[->] (s0) edge node[above,align=center] {$\mdpActionB$} node[below,align=center] {$0$} (s1);
      \path[->] (s0) edge[loop left] node[left,align=center] {$\mdpAction$\\$1$} (s0);
      \path[->] (s1) edge[loop right] node[right,align=center] {$\mdpActionB$\\$0$} (s1);

    \end{tikzpicture}
    \caption{An MDP with deterministic transitions. The doubly circled state denotes a target for a reachability objective. }\label{fig:mixing:approx:1}
  \end{figure}

\begin{example}\label{ex:lexico:general}
    We consider the MDP $\mdp$ depicted in Figure~\ref{fig:mixing:approx:1}.
  Let $\weight$ denote the illustrated weight function.
  We consider the two-dimensional payoff function $\payoffTuple  = (\payoff_1, \payoff_2)$ such that $\payoff_1 = \indic{\reach{\{\mdpStateB\}}}$ and $\payoff_2 = \totrew{\weight}$.
  We prove that randomisation is necessary to play lexicographically optimally from $\mdpState$.
  
  Since transitions of $\mdp$ are deterministic, $\paySetPure{\payoffTuple}{\mdpState}$ is the set of payoffs of plays from $\mdpState$.
  We introduce notation for these plays: for all $\indexLast\in\IN$, let $\play_\indexLast = (\mdpState\mdpAction)^\indexLast\mdpState(\mdpActionB\mdpStateB)^\omega$ and let $\play_\infty= (\mdpState\mdpAction)^\omega$.
  It follows that $\paySetPure{\payoffTuple}{\mdpState} = \{\payoffTuple(\play_\indexLast)\mid\indexLast\in\IN\cup\{\infty\}\} = \{(1, \indexLast)\mid\indexLast\in\IN\}\cup\{(0, +\infty)\}$.
  In particular, no pure strategy has an expected payoff of $(1, +\infty)$, which is the greatest that could occur with $\payoffTuple$.

  However, there exists a randomised strategy whose expected payoff from $\mdpState$ is $(1, +\infty)$.
  For each $\indexLast\in\IN$, let $\stratBMDP_\indexLast$ be a pure strategy whose outcome from $\mdpState$ is $\play_\indexLast$.
  We consider the mixed strategy $\mixedStrat$ such that, for all $\indexLast\in\IN$, $\mixedStrat$ assigns probability $2^{-(\indexLast+1)}$ to $\stratBMDP_{2^\indexLast}$.
  It follows that $\expectation_{\mdpState}^{\mixedStrat}(\payoffTuple) = \sum_{\indexLast=0}^\infty\mixedStrat(\stratBMDP_{2^{\indexLast}})\cdot\payoffTuple(\play_{2^\indexLast}) = (1, \infty)$, i.e., $\mixedStrat$ is lexicographically optimal from $\mdpState$.
  \hfill$\lhd$
\end{example}

\section{The structure of expected payoff sets}\label{section:achievable}
In this section, we study the relationship between the set of expected payoffs of pure strategies and of general strategies.
In Section~\ref{section:achievable:ui}, we prove that for universally integrable payoffs, the set of expected payoffs (for general strategies) is the convex hull of the set of expected payoffs of pure strategies.
In other words, \textit{mixed strategies with a finite support can match the expected payoff of any strategy} when considering universally integrable payoffs.
We provide a variant of this result in Section~\ref{section:achievable:general} for the more general universally unambiguously integrable payoffs: we show that any expected payoff can be approximated (in the sense of the topology of $\IRbar^\numObj$) by the expected payoff of a finite-support mixed strategy.
Finally, in Section~\ref{section:achievable:support}, we provide bounds on support sizes of finite-support mixed strategies analogous to Carath\'{e}odory's theorem for convex hulls (Theorem~\ref{theorem:caratheodory:convex}).

We fix, for this whole section, a POMDP $\pomdp = \pomdpTuple$, $\numObj\in\IN_0$ and a $\numObj$-dimensional payoff function $\payoffTuple = (\payoff_\indexPayoff)_{1\leq\indexPayoff\leq\numObj}$.

\subsection{Universally integrable payoffs}\label{section:achievable:ui}

Throughout this section, we assume that $\payoffTuple$ is universally integrable.
Our goal is to show that, for all $\mdpState\in\mdpStateSpace$, $\paySet{\payoffTuple}{\mdpState} = \convex{\paySetPure{\payoffTuple}{\mdpState}}$.
We first provide an overview of the proof, and then provide a formal proof.

\paragraph{Proof overview.}
Let $\mdpState\in\mdpStateSpace$.
By convexity of $\paySet{\payoffTuple}{\mdpState}$, the inclusion $\convex{\paySetPure{\payoffTuple}{\mdpState}}\subseteq\paySet{\payoffTuple}{\mdpState}$ holds.
Therefore, the main difficulty of the proof is to prove the other inclusion.
Let $\stratMDP\in\stratClassAll{\pomdp}$ and $\payoffVect = \expectation_{\mdpState}^{\stratMDP}(\payoffTuple)$.
We show that $\payoffVect$ is a convex combination of expected payoffs of pure strategies that are lexicographically optimal from $\mdpState$ for the composition of a linear map $\linMap_\payoffVect$ with $\payoffTuple$ for which $\stratMDP$ is also lexicographically optimal.

The first step of the proof is to construct a linear map $\linMap_\payoffVect\colon\IR^\numObj\to\IR^{\numObj'}$ with $\numObj'\leq \numObj$ such that (i)~$\linMap_\payoffVect(\payoffVect)$ is the lexicographic maximum of $\linMap_\payoffVect(\paySet{\payoffTuple}{\mdpState})$ and (ii)~$\payoffVect\in\relInt{\paySet{\payoffTuple}{\mdpState}\cap\vectSpace}$ where $\vectSpace = \linMap_\payoffVect^{-1}(\linMap_\payoffVect(\payoffVect))$ denotes the set of vectors that share their image by $\linMap_\payoffVect$ with $\payoffVect$.
Property~(ii) implies that we cannot find a mapping satisfying~(i) such that the set of vectors sharing their image with $\payoffVect$ is smaller (with respect to set inclusion).

The mapping $\linMap_\payoffVect$ is constructed as follows (Examples~\ref{ex:mixing:exact:1} and~\ref{ex:mixing:exact:2} below illustrate this construction).
If $\payoffVect$ is in the relative interior of $\paySet{\payoffTuple}{\mdpState}$, we let $\linMap_\payoffVect$ be the zero-valued linear form.
Otherwise, by the supporting hyperplane theorem (Theorem~\ref{thm:hyperplane:supporting}), there exists a linear form $\linForm_1$ such that $\linForm_1(\payoffVect)\geq\linForm_1(\payoffVectB)$ for all $\payoffVectB\in\paySet{\payoffTuple}{\mdpState}$.
Let $\hplane_1 = (\linForm_1)^{-1}(\linForm_1(\payoffVect))$ denote the supporting hyperplane given by $\linForm_1$.
We check whether $\payoffVect$ is in the relative interior of $\paySet{\payoffVect}{\mdpState}\cap\hplane_1$.
If it is the case, we obtain the desired linear mapping by letting $\linMap_\payoffVect=\linForm_1$.
Otherwise, we continue: there is a linear form $\linForm_2$ describing a hyperplane $\hplane_2$ that supports $\paySet{\payoffVect}{\mdpState}\cap\hplane_1$ at $\payoffVect$.
We choose $\linForm_2$ such that $\linForm_1$ and $\linForm_2$ have distinct kernels, i.e., $\hplane_1\neq\hplane_2$.
To ensure that the kernels are distinct, we construct $\linForm_2$ as an extension to $\IR^\numObj$ of a non-zero linear form over $\ker{\linForm_1}$.
By induction, we continue constructing linear forms with pairwise distinct kernels (i.e., defining pairwise distinct supporting hyperplanes) until $\payoffVect\in\relInt{\paySet{\payoffTuple}{\mdpState}\cap\bigcap_{1\leq\indexPayoff\leq\numObj'}\hplane_\indexPayoff}$.
When this condition is satisfied, we define, for all $\vect\in\IR^\numObj$, $\linMap_\payoffVect(\vect)=(\linForm_1(\vect), \ldots, \linForm_{\numObj'}(\vect))$.
The invocation of the supporting hyperplane theorem at each iteration guarantees that $\linMap_\payoffVect(\payoffVect)$ is a lexicographic maximum of $\linMap_\payoffVect(\paySet{\payoffTuple}{\mdpState})$.
We remark that, when the stopping condition is fulfilled, the supporting hyperplane theorem is no longer applicable.

Let $\vectSpace = \linMap_\payoffVect^{-1}(\linMap_\payoffVect(\payoffVect))$.
By construction, $\payoffVect\in\relInt{\vectSpace\cap\paySet{\payoffTuple}{\mdpState}}$.
To conclude, we establish that $\relInt{\vectSpace\cap\paySet{\payoffTuple}{\mdpState}}=\relInt{\vectSpace\cap\convex{\paySetPure{\payoffTuple}{\mdpState}}}$.
This suffices because the latter set is a subset of $\convex{\paySetPure{\payoffTuple}{\mdpState}}$.
This is equivalent to showing that $\closure{\vectSpace\cap\paySet{\payoffTuple}{\mdpState}}=\closure{\vectSpace\cap\convex{\paySetPure{\payoffTuple}{\mdpState}}}$: convex subsets of $\IR^\numObj$ have the same relative interior as their closure~\cite[Thm.~6.3]{DBLP:books/degruyter/Rockafellar70}.

The inclusion $\convex{\paySetPure{\payoffTuple}{\mdpState}}\subseteq\paySet{\payoffTuple}{\mdpState}$ implies that $\closure{\vectSpace\cap\convex{\paySetPure{\payoffTuple}{\mdpState}}}\subseteq\closure{\vectSpace\cap\paySet{\payoffTuple}{\mdpState}}$.
For the other inclusion, we need only show that $\vectSpace\cap\paySet{\payoffTuple}{\mdpState}\subseteq\closure{\vectSpace\cap\convex{\paySetPure{\payoffTuple}{\mdpState}}}$.
We assume towards a contradiction that this is not the case and let $\payoffVectB\in\vectSpace\cap\paySet{\payoffTuple}{\mdpState}\setminus\closure{\vectSpace\cap\convex{\paySetPure{\payoffTuple}{\mdpState}}}$.
Using the hyperplane separation theorem (Theorem~\ref{thm:hyperplane:separation}), we obtain $\linForm_\star$ such that $\linForm_\star(\payoffVectB) > \linForm_\star(\vect)$ for all $\vect\in\closure{\vectSpace\cap\convex{\paySetPure{\payoffTuple}{\mdpState}}}$.
Because $\payoffVectB\in\paySet{\payoffTuple}{\mdpState}$, Theorem~\ref{thm:lexico:pure} implies that there exists a pure strategy $\stratBMDP$ such that $(\linMap_\payoffVect(\payoffVectB), \linForm_\star(\payoffVectB))\leLex(\linMap_\payoffVect(\expectancy^{\stratBMDP}_\mdpState(\payoffTuple)), \linForm_\star(\expectancy^{\stratBMDP}_\mdpState(\payoffTuple)))$. Furthermore, $\payoffVectB\in\vectSpace$ implies that $\linMap_\payoffVect(\payoffVectB)$ is the lexicographic maximum of $\linMap_\payoffVect(\paySet{\payoffTuple}{\mdpState})$, and thus so is $\linMap_\payoffVect(\expectancy^{\stratBMDP}_\mdpState(\payoffTuple))$, i.e., $\expectancy^{\stratBMDP}_\mdpState(\payoffTuple)\in\vectSpace\cap\paySetPure{\payoffTuple}{\mdpState}$.
It follows that $\linForm_\star(\payoffVectB)\leq\linForm_\star(\expectancy^{\stratBMDP}_\mdpState(\payoffTuple))$, which is contradictory.
This closes the argument that $\closure{\vectSpace\cap\paySet{\payoffTuple}{\mdpState}}=\closure{\vectSpace\cap\convex{\paySetPure{\payoffTuple}{\mdpState}}}$, which implies that $\payoffVect\in\convex{\paySetPure{\payoffTuple}{\mdpState}}$, ending the sketch.

We complement the sketch above with two examples that illustrate the construction of the mapping $\linMap_\payoffVect$.
In the first example, we select $\payoffVect$ as an extreme point of the set of expected payoffs.
In this case, the constructed linear mapping $\linMap_\payoffVect$ is such that $\payoffVect$ is the unique vector $\payoffVectB\in\IR^\numObj$ such that $\linMap_\payoffVect(\payoffVectB)$ is the lexicographic maximum of $\linMap_\payoffVect(\paySet{\payoffTuple}{\mdpState})$.
In our second example, $\paySet{\payoffTuple}{\mdpState}$ is not closed and we choose $\payoffVect$ as a non-extreme point such that no pure strategy has expected payoff $\payoffVect$.
In this case, we observe that the uniqueness property of the first example cannot be obtained no matter the constructed $\linMap_\payoffVect$.

\begin{example}[Extreme point]\label{ex:mixing:exact:1}
  We consider the MDP depicted in Figure~\ref{fig:mixing:exact:1:mdp}.
  Throughout this example, we (implicitly) consider $\mdpState_0$ as the initial state.
    We study indicators of reachability objectives: we let $\payoffTuple = (\indic{\reach{\target_1}}, \indic{\reach{\target_2}})$ where $\target_1 = \{\mdpState_1, \mdpState_4\}$ and $\target_2 = \{\mdpState_2, \mdpState_4\}$.
  The set $\paySet{\payoffTuple}{\mdpState_0}$ is depicted in Figure~\ref{fig:mixing:exact:1:graphs}: it is the convex hull of the expected payoffs of the pure strategies that play one of the actions $\mdpAction$, $\mdpActionB$ or $\mdpActionC$ in $\mdpState_0$.
  We focus on the extreme point $\payoffVect = (\frac{3}{4}, \frac{3}{4})$ in the remainder of this example.

  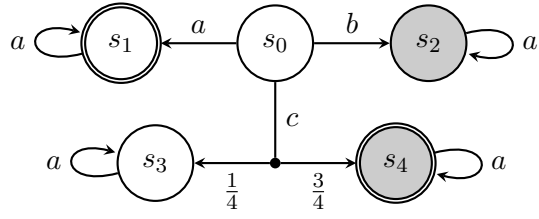
\begin{figure}
    \centering
    \begin{tikzpicture}
      \node[state, align=center] (s0) {$\mdpState_0$};
      \node[state, accepting, left = of s0] (s1) {$\mdpState_1$};
      \node[state, fill=black!20, right = of s0] (s2) {$\mdpState_2$};
      \node[stochasticc, below = of s0] (s0a) {};
      \node[state, left = of s0a] (s3) {$\mdpState_3$};
      \node[state, fill=black!20, accepting, right = of s0a] (s4) {$\mdpState_4$};

      \path[->] (s0) edge node[above,align=center] {$\mdpAction$} (s1);
      \path[->] (s0) edge node[above,align=center] {$\mdpActionB$} (s2);
      \path[-] (s0) edge node[right,align=center] {$\mdpActionC$} (s0a);
      \path[->] (s0a) edge node[below,align=center] {$\frac{1}{4}$} (s3);
      \path[->] (s0a) edge node[below,align=center] {$\frac{3}{4}$} (s4);
      \path[->] (s1) edge[loop left] node[left,align=center] {$\mdpAction$} (s1);
      \path[->] (s2) edge[loop right] node[right,align=center] {$\mdpAction$} (s2);
      \path[->] (s3) edge[loop left] node[left,align=center] {$\mdpAction$} (s3);
      \path[->] (s4) edge[loop right] node[right,align=center] {$\mdpAction$} (s4);

    \end{tikzpicture}
    \caption{An MDP. The doubly circled and filled states respectively highlight the targets of the reachability objectives $\reach{\{\mdpState_1, \mdpState_4\}}$ and $\reach{\{\mdpState_2, \mdpState_4\}}$.}\label{fig:mixing:exact:1:mdp}
  \end{figure}
  
  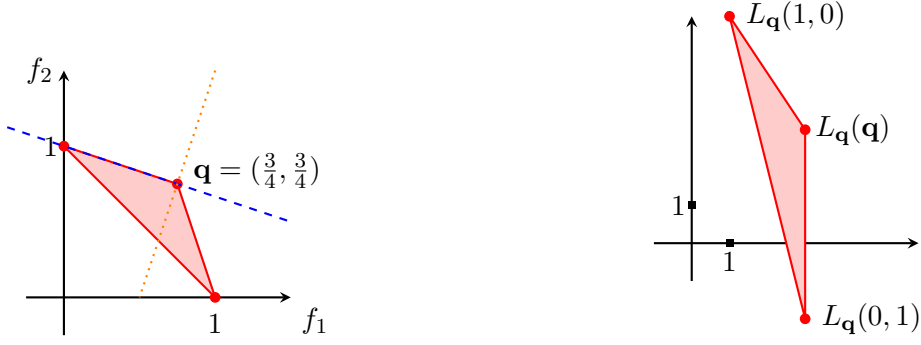
\begin{figure}
    \centering
    \begin{subfigure}[t]{0.45\textwidth}
      \centering
      \begin{tikzpicture}
        \draw[-stealth] (-0.5,0) -- (3,0) node[below right] {$\payoff_1$}; 
        \draw[-stealth] (0,-0.5) -- (0,3) node[left] {$\payoff_2$};
        \coordinate (1x) at (2,0);
        \coordinate (1y) at (0,2);
        \coordinate (m) at (1.5,1.5);
        \coordinate (l1) at (-0.75,2.25);
        \coordinate (l2) at (3,1);
        \coordinate (p1) at (2, 3);
        \coordinate (p2) at (1, 0);
        
        \fill[red!20] (1x) --(1y) -- (m) -- (1x);
        \draw[red] (1x) --(1y) -- (m) -- (1x);
        
        \node[stochasticc,red] at (1y) (qy){};
        \node[xshift=-5] at (1y) {$1$};
        \node[stochasticc,red] at (1x) (qx){};
        \node[yshift=-10] at (1x) {$1$};
        \node[stochasticc,red] at (m) (qxy) {};
        \node[xshift=30,yshift=5] at (m) {$\payoffVect=(\frac{3}{4}, \frac{3}{4})$};
        
\draw[-, thick, dashed, color=blue]  (l1) -- (l2);
        \draw[-, thick, dotted, color=orange]  (p1) -- (p2);        
\end{tikzpicture}
      \caption{
        The blue dashed line ($x+3y=3$) is a hyperplane supporting $\paySet{\payoffTuple}{\mdpState}$ at $\payoffVect$.
        The orange dotted line ($6x-2y=3$) is a hyperplane obtained from an extension of the linear form defining the hyperplane $\{\payoffVect\}$ of the blue line.
      }\label{fig:mixing:exact:1:graph:a}
    \end{subfigure}\hfill \begin{subfigure}[t]{0.45\textwidth}
      \centering
      \begin{tikzpicture}
        \draw[-stealth] (-0.5,0) -- (3,0);\draw[-stealth] (0,-0.5) -- (0,3);\coordinate (x) at (0.5, 0);
        \coordinate (y) at (0, 0.5);
        \coordinate (ix) at (0.5,3); \coordinate (iy) at (1.5,-1); \coordinate (m) at (1.5,1.5); 

        \fill[red!20] (iy) --(ix) -- (m) -- (iy);
        \draw[red] (iy) --(ix) -- (m) -- (iy);
        
        \node[stochasticc,red] at (iy) (qy){};
        \node[xshift=25] at (iy) {$\linMap_\payoffVect(0, 1)$};
        \node[stochasticc,red] at (ix) (qx){};
        \node[xshift=25] at (ix) {$\linMap_\payoffVect(1, 0)$};
        \node[stochasticc,red] at (m) (qxy) {};
        \node[xshift=18] at (m) {$\linMap_\payoffVect(\payoffVect)$};

        \node[stochastics] at (x) {};
        \node[stochastics] at (y) {};
        \node[yshift=-7] at (x) {$1$};
        \node[xshift=-5] at (y) {$1$};
      \end{tikzpicture}
      \caption{The image of the set on the left by the linear mapping $\linMap_{\payoffVect}\colon(\vectComp_1, \vectComp_2)\mapsto(\vectComp_1+3\vectComp_2, 6\vectComp_1-2\vectComp_2)$ obtained through the equations of the hyperplanes on the left. Remark that the image of $\payoffVect$ is lexicographically optimal.}\label{fig:mixing:exact:1:graph:b}
    \end{subfigure}
    \caption{The set of expected payoffs for Example~\ref{ex:mixing:exact:1} and its image by a linear mapping.}\label{fig:mixing:exact:1:graphs}
  \end{figure}

  The construction of $\linMap_\payoffVect$ is in two steps.
First, we consider the linear form $\linForm_1$ defined by, for all $\vect\in\IR^\numObj$, $\linForm_1(\vect) = \vectComp_1 + 3\vectComp_2$.
  The hyperplane $\hplane = (\linForm_1)^{-1}(3)$ supports $\paySet{\payoffTuple}{\mdpState_0}$ at $\payoffVect$ and is depicted by the blue dashed line in Figure~\ref{fig:mixing:exact:1:graph:a}.
  It is not satisfactory to set $\linMap_\payoffVect=\linForm_1$: $\payoffVect$ is an endpoint of the segment $\ccInt{(0, 1)}{\payoffVect}=\paySet{\payoffTuple}{\mdpState_0}\cap\hplane$, and is therefore not in its relative interior.

  We recall that for any linear form $\linFormB$ of $\ker{\linForm_1}$ (i.e., the vector space corresponding to $\hplane$), there exists $\vect\in\ker{\linForm_1}$ such that $\linFormB(\vectB) = \scalarProd{\vectB}{\vect}$ for all $\vectB\in\ker{\linForm_1}$.
  Since any non-zero linear form of $\ker{\linForm_1}$ is bijective, all of them induce a hyperplane of $\hplane$ supporting $\paySet{\payoffTuple}{\mdpState}\cap\hplane$ at $\payoffVect$.
  We proceed with the linear form $\linForm_2\colon\IR^2\to\IR$ defined by $\linForm_2(\vect) = 6\vectComp_1 - 2 \vectComp_2$ for all $\vect = (\vectComp_1, \vectComp_2)\in\IR^2$ (derived from the vector $\vectB=(6, -2)\in\ker{\linForm_1}$).
  Observe that $\ker{\linForm_1}\cap\ker{\linForm_2}=\{\zeroVect\}$.

  We define $\linMap_{\payoffVect}(\vect) = (\linForm_1(\vect), \linForm_2(\vect))$.
  Since $\linMap_\payoffVect$ is bijective, $\linMap_\payoffVect^{-1}(\linMap_\payoffVect(\payoffVect))$ is a singleton set.
  Therefore, $\payoffVect$ is in the relative interior of $\linMap_\payoffVect^{-1}(\linMap_\payoffVect(\payoffVect))\cap\paySet{\payoffTuple}{\mdpState_0}$.
  By linearity of the expectation, it holds that $\paySet{\linMap_{\payoffVect}\circ\payoffTuple}{\mdpState_0} = \linMap_{\payoffVect}(\paySet{\payoffTuple}{\mdpState_0})$.
  This set is illustrated in Figure~\ref{fig:mixing:exact:1:graph:b}; it is easy to check that $\linMap_\payoffVect(\payoffVect)$ is the lexicographic maximum of this set.

  In this case, $\linMap_\payoffVect$ allows us to deduce that there exists a pure strategy $\stratMDP$ such that $\expectancy^{\stratMDP}_{\mdpState_0}(\payoffTuple) = \payoffVect$.
  On the one hand, by Theorem~\ref{thm:lexico:pure}, there exists a pure strategy $\stratMDP$ that is lexicographically optimal from $\mdpState_0$ for $\linMap_\payoffVect\circ\payoffTuple$.
  On the other hand, the only payoff vector $\payoffVectB\in\paySet{\payoffTuple}{\mdpState_0}$ such that $\linMap_\payoffVect(\payoffVectB)$ is the lexicographic maximum of $\linMap_\payoffVect(\paySet{\payoffTuple}{\mdpState_0})$ is $\payoffVect$.
  This implies that $\payoffVect$ is the payoff of the pure strategy $\stratMDP$, and thus $\payoffVect\in\paySetPure{\payoffTuple}{\mdpState_0}$.\hfill$\lhd$
\end{example}
\begin{remark}[The necessity of induction]\label{remark:ex:mixing:exact:1}
  In Example~\ref{ex:mixing:exact:1}, it is possible to isolate $\payoffVect$ by a supporting hyperplane of $\paySet{\payoffTuple}{\mdpState_0}$.
  We could thus choose $\linMap_\payoffVect$ as a linear form and bypass the induction step of the construction of $\linMap_\payoffVect$ here.
  We present it for the sake of illustration, as we cannot use linear forms to isolate extreme points in general (see Section~\ref{section:running}).
  \hfill$\lhd$
\end{remark}

The construction outlined in Example~\ref{ex:mixing:exact:1} can be generalised to show that all extreme points of the set of expected payoffs of $\payoffTuple$ from a given state are the expected payoff of a pure strategy.
However, this property is not sufficient to show that all expected payoffs are convex combinations of pure expected payoffs.
In particular, the following example highlights that some expected payoffs are not convex combinations of extreme points of the set of expected payoffs.

\begin{example}[Non-extreme point]\label{ex:mixing:exact:2}
  We consider the MDP depicted in Figure~\ref{fig:mixing:exact:2:mdp}.
  We assume that state $\mdpState$ is the initial state throughout this example.
  Let $\weight = (\weight_1,\weight_2)$ denote the two-dimensional weight function given on the illustration.
  We consider a two-dimensional payoff function $\payoffTuple$.
  The payoff of a play, for each dimension, is zero if $\mdpStateB$ is not visited and, otherwise, its payoff is given by a discounted-sum payoff.
  We formalise this as the product of a discounted-sum payoff and an indicator.
  Therefore, formally, $\payoffTuple = (\payoff_1, \payoff_2)$ is such that, for $\indexPayoff\in\{1, 2\}$, $\payoff_\indexPayoff = \indic{\reach{\{\mdpStateB\}}}\cdot\discSum{\frac{3}{4}}{\weight_\indexPayoff}$.
  We observe that, by definition of $\weight$, $\payoff_2 = \discSum{\frac{3}{4}}{\weight_2}$.
  \begin{figure}
    \centering
    \begin{tikzpicture}
      \node[state, align=center] (s0) {$\mdpState$};
      \node[state, accepting, right = of s0] (s1) {$\mdpStateB$};

      \path[->] (s0) edge node[above,align=center] {$\mdpAction$} node[below,align=center] {$(0, 1)$} (s1);
      \path[->] (s0) edge[loop left] node[left,align=center] {$\mdpActionB$\\$(1, 0)$} (s0);
      \path[->] (s1) edge[loop right] node[right,align=center] {$\mdpAction$\\$(0, 1)$} (s1);
    \end{tikzpicture}
    \caption{An MDP with deterministic transitions.
      The pairs beneath actions represent a two-dimensional weight function.
      State $\mdpStateB$ is doubly circled because it is a target.}\label{fig:mixing:exact:2:mdp}
  \end{figure}
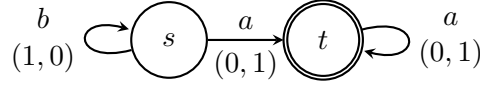
  
  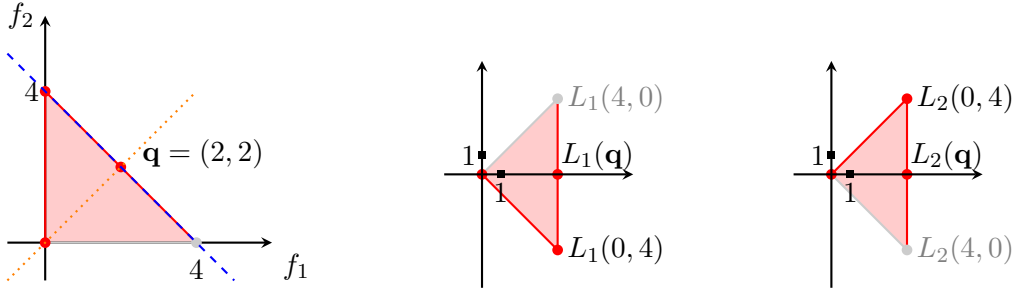
\begin{figure}
    \centering
    \begin{subfigure}[t]{0.38\textwidth}
      \centering
      \begin{tikzpicture}
        \coordinate (og) at (0,0);
        \coordinate (1x) at (2,0);
        \coordinate (1y) at (0,2);
        \coordinate (m) at (1,1);
        \coordinate (l1) at (-0.5,2.5);
        \coordinate (l2) at (2.5,-0.5);
        \coordinate (p1) at (-0.5, -0.5);
        \coordinate (p2) at (2, 2);

        \fill[red!20] (1x) --(1y) -- (og) -- (1x);
        \draw[-stealth] (-0.5,0) -- (3,0) node[below right] {$\payoff_1$}; 
        \draw[-stealth] (0,-0.5) -- (0,3) node[left] {$\payoff_2$};
        \draw[red,thick] (og) -- (1y) --(1x);
        \draw[thick,black!20] (og) -- (1x);

        \node[stochasticc,red] at (og) (qog) {};
        \node[stochasticc,red] at (1y) (qy){};
        \node[xshift=-5] at (1y) {$4$};
        \node[stochasticc,black!20] at (1x) (qx){};
        \node[yshift=-10] at (1x) {$4$};
        \node[stochasticc,red] at (m) (qxy) {};
        \node[xshift=31,yshift=5] at (m) {$\payoffVect=(2, 2)$};

        \draw[-, thick, dashed, color=blue]  (l1) -- (l2);
        \draw[-, thick, dotted, color=orange]  (p1) -- (p2);
      \end{tikzpicture}
      \caption{The blue dashed line ($x+y = 4$) is a hyperplane supporting $\paySet{\payoffTuple}{\mdpState}$ at $\payoffVect$. The orange dotted line ($x-y = 0$) is an orthogonal hyperplane included for reference for the adjacent figures.}\label{fig:mixing:exact:2:graph:a}
    \end{subfigure}\hfill \begin{subfigure}[t]{0.29\textwidth}
      \centering
      \begin{tikzpicture}
        \coordinate (og) at (0,0);
        \coordinate (x) at (0.25, 0);
        \coordinate (y) at (0, 0.25);
        \coordinate (1x) at (1,1);
        \coordinate (1y) at (1,-1);
        \coordinate (m) at (1,0);
        
        \fill[red!20] (1x) --(1y) -- (og) -- (1x);
        \draw[thick,red] (og) -- (1y) --(1x);
        \draw[thick,black!20] (og) -- (1x);

        \node[stochasticc,red] at (og) (qog) {};
        \node[stochasticc,red] at (1y) (qy){};
        \node[xshift=22] at (1y) {$\linMap_1(0, 4)$};
        \node[stochasticc,black!20] at (1x) (qx){};
        \node[black!50, xshift=22] at (1x) {$\linMap_1(4, 0)$};
        \node[stochasticc,red] at (m) (qxy) {};
        \node[xshift=15,yshift=7] at (m) {$\linMap_1(\payoffVect)$};
        
        \node[stochastics] at (x) {};
        \node[stochastics] at (y) {};
        \node[yshift=-7] at (x) {$1$};
        \node[xshift=-5] at (y) {$1$};

        \draw[-stealth] (-0.5,0) -- (2,0);\draw[-stealth] (0,-1.5) -- (0,1.5);\end{tikzpicture}
      \caption{Image of the payoff set in Figure~\ref{fig:mixing:exact:2:graph:a} by the linear mapping $\linMap_1$ such that $(\vectComp_1,\vectComp_2)\mapsto(\vectComp_1+\vectComp_2, \vectComp_1-\vectComp_2)$.}\label{fig:mixing:exact:2:graph:b}
    \end{subfigure}\hfill \begin{subfigure}[t]{0.29\textwidth}
      \centering
      \begin{tikzpicture}
        \coordinate (og) at (0,0);
        \coordinate (x) at (0.25, 0);
        \coordinate (y) at (0, 0.25);
        \coordinate (1x) at (1,-1);
        \coordinate (1y) at (1,1);
        \coordinate (m) at (1,0);

        \fill[red!20] (1x) --(1y) -- (og) -- (1x);
        \draw[thick,red] (og) -- (1y) --(1x);
        \draw[thick,black!20] (og) -- (1x);

        \node[stochasticc,red] at (og) (qog) {};
        \node[stochasticc,red] at (1y) (qy){};
        \node[xshift=22] at (1y) {$\linMap_2(0, 4)$};
        \node[stochasticc,black!20] at (1x) (qx){};
        \node[black!50, xshift=22] at (1x) {$\linMap_2(4, 0)$};
        \node[stochasticc,red] at (m) (qxy) {};
        \node[xshift=15,yshift=7] at (m) {$\linMap_2(\payoffVect)$};
        
        \node[stochastics] at (x) {};
        \node[stochastics] at (y) {};
        \node[yshift=-7] at (x) {$1$};
        \node[xshift=-5] at (y) {$1$};

        \draw[-stealth] (-0.5,0) -- (2,0);\draw[-stealth] (0,-1.5) -- (0,1.5);\end{tikzpicture}
      \caption{Image of the payoff set in Figure~\ref{fig:mixing:exact:2:graph:a} by the linear mapping $\linMap_2$ such that $(\vectComp_1,\vectComp_2)\mapsto(\vectComp_1+\vectComp_2, \vectComp_2-\vectComp_1)$.}\label{fig:mixing:exact:2:graph:c}
    \end{subfigure}
    \caption{The set of expected payoffs for Example~\ref{ex:mixing:exact:2} and its image by two (related) linear functions. The segment $\ocInt{(0, 0)}{(4, 0)}$ in grey does not intersect $\paySet{\payoffTuple}{\mdpState}$. Its image is similarly coloured in the two other illustrations.}\label{fig:mixing:exact:2:graphs}
  \end{figure}

  The set $\paySet{\payoffTuple}{\mdpState}$ is illustrated in Figure~\ref{fig:mixing:exact:2:graph:a}.
  Any vector in $\paySet{\payoffTuple}{\mdpState}$ is a convex combination of $\zeroVect$ and a vector in the segment $\coInt{(0, 4)}{(4, 0)}$.
  In particular, no strategy has an expected payoff of $(4, 0)$ from $\mdpState$.
  We can derive $\paySet{\payoffTuple}{\mdpState}$ from $\paySetPure{\payoffTuple}{\mdpState} = \{\zeroVect\}\cup\{(4-\frac{3^\indexPosition}{4^{\indexPosition-1}}, \frac{3^\indexPosition}{4^{\indexPosition-1}})\mid\indexPosition\in\IN\}$.
  To obtain $\paySetPure{\payoffTuple}{\mdpState}$, we note that any pure strategy in this MDP induces a single play from $\mdpState$, because all transitions are deterministic.
  On the one hand, we can obtain the payoff $\zeroVect$ with the play $(\mdpState\mdpActionB)^\omega$ (the payoff is zero on the first dimension because $\mdpStateB$ is not visited).
  On the other hand, for all $\indexPosition\in\IN$, we have $\payoffTuple((\mdpState\mdpActionB)^\indexPosition\mdpState(\mdpAction\mdpStateB)^\omega) = \left(4- \frac{3^\indexPosition}{4^{\indexPosition-1}}, \frac{3^\indexPosition}{4^{\indexPosition-1}}\right)$.
  
  We consider the payoff vector $\payoffVect = (2, 2)$ and construct $\linMap_\payoffVect$.
  We remark that the vector $\payoffVect$ is not a convex combination of extreme points of $\paySet{\payoffTuple}{\mdpState}$.
  Therefore, is not possible to conclude that $\payoffVect\in\convex{\paySetPure{\payoffTuple}{\mdpState}}$ by adapting the argument of Example~\ref{ex:mixing:exact:1} to deal with all extreme points.
  The only hyperplane $\hplane$ that support $\paySet{\payoffTuple}{\mdpState}$ at $\payoffVect$ is the line depicted in blue in Figure~\ref{fig:mixing:exact:2:graph:a}.
  We let $\linForm_1\colon\IR^2\to\IR$ be the linear form defined by $\linForm_2(\vect) = \vectComp_1+\vectComp_2$ for all $\vect=(\vectComp_1, \vectComp_2)\in\IR^2$.
  We have $\hplane = (\linForm_1)^{-1}(4)$.
  We observe (via Figure~\ref{fig:mixing:exact:2:graph:a}) that $\payoffVect$ is in the relative interior of $\paySet{\payoffTuple}{\mdpState}\cap\hplane$.
  We define $\linMap_\payoffVect=\linForm_1$.

  To close this example, we provide an argument based on $\linMap_\payoffVect^{-1}(\linMap_\payoffVect(\payoffVect))$ being a line to show that $\payoffVect$ is a convex combination of expected payoffs of pure strategies.
While this argument differs from the general proof provided below, it can be generalised to show that $\payoffVect\in\convex{\paySetPure{\payoffTuple}{\mdpState}}$ whenever $\linMap_\payoffVect^{-1}(\linMap_\payoffVect(\payoffVect))$ is a line.
  This argument consists in showing that there are payoffs of pure strategies on either side of $\payoffVect$ on the line segment $\paySet{\payoffTuple}{\mdpState}\cap\hplane = \coInt{(0, 4)}{(4, 0)}$.
  This implies that $\payoffVect\in\convex{\paySetPure{\payoffTuple}{\mdpState}}$.

  We fix a direction vector $\vect_\hplane=(1, -1)$ of $\hplane$.
  For all vectors $\payoffVectB$ of $\paySet{\payoffTuple}{\mdpState}$ in the segment $\coInt{\payoffVect}{(4, 0)}$ (resp.~$\ccInt{\payoffVect}{(0, 4)}$), we have $\scalarProd{\payoffVectB}{\vect_\hplane}\geq\scalarProd{\payoffVect}{\vect_\hplane}$ (resp.~$\scalarProd{\payoffVectB}{-\vect_\hplane}\geq\scalarProd{\payoffVect}{-\vect_\hplane}$).
  Consider the linear mappings $\linMap_1\colon\vectB\mapsto(\linForm_1(\vectB), \scalarProd{\vectB}{\vect_\hplane})$ and $\linMap_2\colon\vectB\mapsto(\linForm_1(\vectB), \scalarProd{\vectB}{-\vect_\hplane})$ over $\IR^2$.
  We illustrate the image of $\paySet{\payoffTuple}{\mdpState}$ by $\linMap_1$ and $\linMap_2$ respectively in Figure~\ref{fig:mixing:exact:2:graph:b} and Figure~\ref{fig:mixing:exact:2:graph:c}.

  Theorem~\ref{thm:lexico:pure} implies that, for $i=1, 2$, there exists a pure strategy $\stratMDP_i$ such that $\linMap_i(\payoffVectB_i)\geLex\linMap_i(\payoffVect)$ where $\payoffVectB_i = \expectancy^{\stratMDP_i}_{\mdpState}(\payoffTuple)$.
  We obtain, by definition of $\linMap_i$, that $\linForm_1(\payoffVectB_i)=\linForm_1(\payoffVect)$ for $i=1, 2$, as $\linForm_1$ supports $\paySet{\payoffTuple}{\mdpState}$ at $\payoffVect$.
  Therefore, $\scalar_1\coloneqq \scalarProd{\payoffVectB_1}{\vect_\hplane}\geq 0$ and $\scalar_2\coloneqq \scalarProd{\payoffVectB_2}{\vect_\hplane} \leq 0$.
  Furthermore, for $i=1, 2$, it holds that $\payoffVectB_i = \payoffVect + \frac{\scalar_i}{\|\vect_\hplane\|_2}\vect_\hplane$ because ($\frac{1}{\|\payoffVect\|_2}\payoffVect, \frac{1}{\|\vect_\hplane\|_2}\vect_\hplane)$ is an orthonormal basis of $\IR^2$, $\scalarProd{\payoffVectB_i}{\frac{1}{\|\payoffVect\|_2}\payoffVect} = \frac{2}{\|\payoffVect\|_2}\cdot\linForm_1(\payoffVectB_i) = 2\sqrt{2} = \|\payoffVect\|_2$ (as $\linForm_1(\payoffVectB_i)=\linForm_1(\payoffVect)$) and $\scalarProd{\payoffVectB_i}{\frac{1}{\|\vect_\hplane\|_2}\vect_\hplane} = \frac{\scalar_i}{\|\vect_\hplane\|_2}$ (by definition of $\scalar_i$).
Thus, $\payoffVect\in\ccInt{\payoffVectB_1}{\payoffVectB_2}\subseteq\convex{\paySetPure{\payoffTuple}{\mdpState}}$.
  \hfill$\lhd$
\end{example}

\paragraph{Formal statement and proof.}
We formally state the main theorem of this section and prove it.

\begin{restatable}{theorem}{thmMixingExact}\label{thm:mixing:exact}
  Assume that $\payoffTuple$ is universally integrable.
  For all $\mdpState\in\mdpStateSpace$, we have $\paySet{\payoffTuple}{\mdpState} = \convex{\paySetPure{\payoffTuple}{\mdpState}}$.
  In other words, the expected payoff of any strategy is also the expected payoff of a finite-support mixed strategy.
\end{restatable}
\begin{proof}
  Throughout this proof, we assume that for all $1\leq\indexPayoff\leq\numObj$, $\payoff_\indexPayoff$ is a real-valued payoff.
  This is without loss of generality: these payoffs are universally integrable,  and thus are $\proba^{\stratMDP}_{\mdpState}$-almost-surely real-valued for all $\stratMDP\in\stratClassAll{\mdp}$ and $\mdpState\in\mdpStateSpace$.
  
  It is sufficient to show that $\paySet{\payoffTuple}{\mdpState}\subseteq\convex{\paySetPure{\payoffTuple}{\mdpState}}$.
  Let $\payoffVect\in\paySet{\payoffTuple}{\mdpState}$.
  We construct the linear mapping $\linMap_\payoffVect$ as explained in the sketch, i.e., such that $\linMap_\payoffVect(\payoffVect)$ is the lexicographic maximum of $\linMap_\payoffVect(\paySet{\payoffTuple}{\mdpState})$ and $\payoffVect\in\relInt{\linMap_\payoffVect^{-1}(\linMap_\payoffVect(\payoffVect))\cap\paySet{\payoffTuple}{\mdpState}}$.

  Let $D = \paySet{\payoffTuple}{\mdpState} - \payoffVect$.
  We observe that $\linMap_\payoffVect$ satisfies the conditions above if and only if $\linMap_\payoffVect(\zeroVect)$ is the lexicographic maximum of $\linMap_\payoffVect(D)$ and $\zeroVect$ is in the relative interior of $D\cap\ker{\linMap_\payoffVect}$.
  We construct $\linMap_\payoffVect$ by working with $D$ and $\zeroVect$ instead of $\paySet{\payoffTuple}{\mdpState}$ and $\payoffVect$.
  This allows us to work with vector sub-spaces instead of affine spaces, overall simplifying the presentation.

  We let $\linFormB_0\colon\IR^\numObj\to\IR$ be the constant zero function.
  If $\zeroVect\in\relInt{D}$, we let $\linMap_\payoffVect=\linFormB_0$.
  This function satisfies the desired properties.
  We now assume that $\zeroVect\notin\relInt{D}$.
  We inductively define a sequence of non-zero linear forms $\linFormB_1, \ldots, \linFormB_{\numObj'}$ such that $\linFormB_\indexPayoff\colon\ker{\linFormB_{\indexPayoff-1}}\to\IR$  for all $1\leq\indexPayoff\leq\numObj'$.
  Next, for all $1\leq\indexPayoff\leq\numObj'$, we extend $\linFormB_\indexPayoff$ to a form $\linForm_\indexPayoff\colon\IR^\numObj\to\IR$.
  Finally, we define the mapping $\linMap_\payoffVect$ as $\linMap_\payoffVect(\vect) = (\linForm_1(\vect), \ldots, \linForm_{\numObj'}(\vect))$ for all $\vect\in\IR^\numObj$ and show that it satisfies the desired properties.

  Let $\indexPayoff\geq 1$.
  By induction, assume that $\linFormB_{\indexPayoff-1}$ is defined (this is the case even for $\indexPayoff=1$).
  We distinguish two cases.
  If $\zeroVect\in\relInt{\ker{\linFormB_{\indexPayoff-1}}\cap D}$, we stop the construction.
  We remark that if $\indexPayoff = \numObj+1$, then $\ker{\linFormB_{\indexPayoff-1}}$ is a singleton set and we are necessarily in this case (i.e., $\numObj'\leq\numObj$).
  Now, assume that $\zeroVect\notin\relInt{\ker{\linFormB_{\indexPayoff-1}}\cap D}$.
  The supporting hyperplane theorem (Theorem~\ref{thm:hyperplane:supporting}) implies that there exists a linear form $\linFormB_\indexPayoff\colon\ker{\linFormB_{\indexPayoff-1}}\to\IR$ such that for all $\payoffVectB\in D\cap\ker{\linFormB_{\indexPayoff-1}}$, we have $\linFormB_\indexPayoff(\payoffVectB)\leq 0 = \linFormB_\indexPayoff(\zeroVect)$.
  This allows us to continue with the induction.

  Assume that the procedure above has provided linear forms $\linFormB_1$,\ldots,$\linFormB_{\numObj'}$.
  We now extend them to $\IR^\numObj$.
  Let $1\leq\indexPayoff\leq\numObj'$.
  There exists $\vectB_\indexPayoff\in\ker{\linFormB_{\indexPayoff-1}}\subseteq\IR^\numObj$ such that for all $\vect\in\ker{\linFormB_{\indexPayoff-1}}$, we have $\linFormB_\indexPayoff(\vect) = \scalarProd{\vect}{\vectB_\indexPayoff}$.
  We define $\linForm_\indexPayoff\colon\IR^\numObj\to\IR$ by, for all $\vect\in\IR^\numObj$, $\linForm_\indexPayoff(\vect) = \scalarProd{\vect}{\vectB_\indexPayoff}$.
  We let $\linMap_\payoffVect\colon\IR^\numObj\to\IR^{\numObj'}$ be such that $\linMap_\payoffVect(\vect) = (\linForm_1(\vect), \ldots, \linForm_{\numObj'}(\vect))$ for all $\vect\in\IR^\numObj$.

  We now show that $\linMap_\payoffVect$ satisfies the required properties.
  By construction, for all $\payoffVectB\in D$ and all $1\leq\indexPayoff\leq\numObj'$, if $\linForm_{\indexPayoff'}(\payoffVectB) = 0$ for all $1\leq\indexPayoff'\leq\indexPayoff-1$, then necessarily $\payoffVectB\in\ker{\linFormB_{\indexPayoff-1}}$, and thus $\linForm_{\indexPayoff}(\payoffVectB)\leq 0$.
  This implies that for all $\payoffVectB\in D$, $\linMap_\payoffVect(\payoffVectB)\leLex\linMap_\payoffVect(\zeroVect)$.
  This shows that $\linMap_\payoffVect(\zeroVect)$ is the lexicographic maximum of $\linMap_\payoffVect(D)$.
  Next, we show that $\zeroVect\in\relInt{D\cap\ker{\linMap_\payoffVect}}$.
  This follows from the stopping condition in the construction of $\linMap_\payoffVect$ and the equality $\ker{\linMap_\payoffVect} =\bigcap_{1\leq\indexPayoff\leq\numObj'}\ker{\linForm_{\indexPayoff}} = \ker{\linFormB_{\numObj'}}$ (the second equality follows from $\linForm_1=\linFormB_1$ and $\ker{\linFormB_1}\supsetneq\ldots\supsetneq \ker{\linFormB_{\numObj'}}$).

  Let $\vectSpace = \linMap_\payoffVect^{-1}(\linMap_\payoffVect(\payoffVect))$.
  We have shown that $\payoffVect\in\relInt{\vectSpace\cap\paySet{\payoffTuple}{\mdpState}}$.
  It suffices to show that $\relInt{\vectSpace\cap\paySet{\payoffTuple}{\mdpState}} = \relInt{\vectSpace\cap\convex{\paySetPure{\payoffTuple}{\mdpState}}}$ to conclude that $\payoffVect\in\convex{\paySetPure{\payoffTuple}{\mdpState}}$.
  Since all convex subsets of $\IR^\numObj$ have the same relative interior as their closure~\cite[Theorem 6.3]{DBLP:books/degruyter/Rockafellar70}, the equality of the relative interiors stated before is implied by the relation $\closure{\vectSpace\cap\paySet{\payoffTuple}{\mdpState}} = \closure{\vectSpace\cap\convex{\paySetPure{\payoffTuple}{\mdpState}}}$.
  To end the proof, we show this equality of closures.
  The inclusion $\closure{\vectSpace\cap\convex{\paySetPure{\payoffTuple}{\mdpState}}}\subseteq\closure{\vectSpace\cap\paySet{\payoffTuple}{\mdpState}}$ is direct.
  
  For the other inclusion, it suffices to show that  $\vectSpace\cap\paySet{\payoffTuple}{\mdpState}\subseteq\closure{\vectSpace\cap\convex{\paySetPure{\payoffTuple}{\mdpState}}}$.
  Let $\payoffVectB\in{\vectSpace\cap\paySet{\payoffTuple}{\mdpState}}$.
  Assume, by contradiction, that $\payoffVectB\notin\closure{\vectSpace\cap\convex{\paySetPure{\payoffTuple}{\mdpState}}}$.
  By the hyperplane separation theorem (Theorem~\ref{thm:hyperplane:separation}), there exists a linear form $\linForm$ over $\IR^\numObj$ such that for all $\payoffVectB'\in\vectSpace\cap\convex{\paySetPure{\payoffTuple}{\mdpState}}$, we have $\linForm(\payoffVectB)>\linForm(\payoffVectB')$.
  Let $\linMapB_\payoffVect\colon\IR^\numObj\to\IR^{\numObj'+1}$ such that for all $\vect\in\IR^\numObj$, we have $\linMapB_\payoffVect(\vect) = (\linMap_\payoffVect(\vect), \linForm(\vect))$.

  Let $\stratMDP$ such that $\payoffVectB = \expectancy^{\stratMDP}_{\mdpState}(\payoffTuple)$.
  By Theorem~\ref{thm:lexico:pure}, there exists a pure strategy $\stratBMDP$ such that $\expectancy^{\stratBMDP}_{\mdpState}(\linMapB_\payoffVect\circ\payoffTuple)\geLex \expectancy^{\stratMDP}_{\mdpState}(\linMapB_\payoffVect\circ\payoffTuple) = \linMapB_\payoffVect(\payoffVectB)$.
  We have $\expectancy^\stratBMDP_{\mdpState}(\payoffTuple)\in\vectSpace$ because $\linMap_\payoffVect(\expectancy^{\stratBMDP}_{\mdpState}(\payoffTuple)) = \expectancy^{\stratBMDP}_{\mdpState}(\linMap_\payoffVect\circ\payoffTuple)\geLex\linMap_\payoffVect(\payoffVectB) = \linMap_\payoffVect(\payoffVect)$ and $\linMap_\payoffVect(\payoffVect)$ is lexicographically optimal in $\linMap_\payoffVect(\paySet{\payoffTuple}{\mdpState})$.
  It follows that $\expectancy^{\stratBMDP}_{\mdpState}(\linForm\circ\payoffTuple) =\linForm(\expectancy^{\stratBMDP}_{\mdpState}(\payoffTuple)) \geq\linForm(\payoffVectB)$.
  This is a contradiction with $\linForm$ defining a strongly separating hyperplane.
\end{proof}

We now formulate two corollaries of Theorem~\ref{thm:mixing:exact}.
The first one relates to extreme points of payoffs sets.
\begin{restatable}{corollary}{corMixingExactExtreme}\label{cor:mixing:exact:extreme}
  Assume that $\payoffTuple$ is universally integrable.
  For all $\mdpState\in\mdpStateSpace$, $\corners{\paySet{\payoffTuple}{\mdpState}}\subseteq\paySetPure{\payoffTuple}{\mdpState}$, i.e., all extreme points of $\paySet{\payoffTuple}{\mdpState}$ are payoffs of pure strategies.
\end{restatable}
\begin{proof}
    By Theorem~\ref{thm:mixing:exact}, we have $\paySet{\payoffTuple}{\mdpState} = \convex{\paySetPure{\payoffTuple}{\mdpState}}$.
  All extreme points of $\convex{\paySetPure{\payoffTuple}{\mdpState}}$ must be in $\paySetPure{\payoffTuple}{\mdpState}$ by definition of the convex hull.
\end{proof}

Second, we establish that, for all $\mdpState\in\mdpStateSpace$, $\paySet{\payoffTuple}{\mdpState}$ is closed whenever $\paySetPure{\payoffTuple}{\mdpState}$ is closed.
We note that, in general, the convex hull of a closed set need not be closed.
However, the convex hull of a compact subset of $\IR^\numObj$ is closed (Lemma~\ref{lemma:convex hull of compact}).
Since the set of expected payoffs of a universally integrable payoff function is bounded, it is thus compact, implying the claimed property.

\begin{restatable}{corollary}{corMixingExactCompact}\label{cor:mixing:exact:compact}
  Assume that $\payoffTuple$ is universally integrable.
  For all $\mdpState\in\mdpStateSpace$, if $\paySetPure{\payoffTuple}{\mdpState}$ is closed, then $\paySet{\payoffTuple}{\mdpState}$ is compact.
\end{restatable}
\begin{proof}
  Let $\mdpState\in\mdpStateSpace$ such that $\paySetPure{\payoffTuple}{\mdpState}$ is closed.
  By Lemma~\ref{lem:ui:characterisation}, $\payoffTuple$ is universally integrable if and only if $\paySetPure{\payoffTuple}{\mdpState}$ is bounded.
  It follows that $\paySetPure{\payoffTuple}{\mdpState}$ is compact.
  Theorem~\ref{thm:mixing:exact} ensures that $\paySet{\payoffTuple}{\mdpState} = \convex{\paySetPure{\payoffTuple}{\mdpState}}$, and thus $\paySet{\payoffTuple}{\mdpState}$ is compact by Lemma~\ref{lemma:convex hull of compact}.
\end{proof}

We close this section by showing that Theorem~\ref{thm:mixing:exact} does not generalise to universally unambiguously integrable payoffs.

\begin{example}[Example~\ref{ex:lexico:general} continued]\label{ex:mixing:approx}
  We consider the MDP $\mdp$ depicted in Figure~\ref{fig:mixing:approx:1} and the payoff function $\payoffTuple = (\indic{\reach{\{\mdpStateB\}}}, \totrew{\weight})$ where $\weight$ is the weight function of Figure~\ref{fig:mixing:approx:1}.
  In Example~\ref{ex:lexico:general}, we have shown that there exists a randomised strategy $\stratMDP$ such that $\expectancy^{\stratMDP}_\mdpState(\payoffTuple) = (1, +\infty)$ and that $\paySetPure{\payoffTuple}{\mdpState} = \{(0, +\infty)\}\cup \{(1, \indexPosition)\mid\indexPosition\in\IN\}$.
  
  We show that $(1, +\infty)\notin\convex{\paySetPure{\payoffTuple}{\mdpState}}$.
  On the one hand, convex combinations of these vectors that give a non-zero coefficient to $(0, +\infty)$ have a first component is not equal to $1$.
  On the other hand, convex combinations that assign a zero coefficient to $(0, +\infty)$ have a finite second component.
  We obtain that $(1, +\infty)\notin\convex{\paySetPure{\payoffTuple}{\mdpState}}$.

  Although we cannot have a payoff of $(1, +\infty)$ with finite-support mixed strategies, we can approximate it with such strategies.
  We generalise this observation in the next section.
  \hfill$\lhd$
\end{example}

\subsection{Universally unambiguously integrable payoffs}\label{section:achievable:general}
We now relax the assumption that $\payoffTuple$ is universally integrable and assume that $\payoffTuple$ is universally unambiguously integrable.
We formulate an approximate variant of Theorem~\ref{thm:mixing:exact}: from a given state, any expected payoff of a strategy can be approached by convex combinations of expected payoffs of pure strategies (in the sense of limits in $\IRbar^\numObj$).

We provide a proof sketch, followed by the formal statement of the result and its proof.

\paragraph{Proof overview.}
Fix $\mdpState\in\mdpStateSpace$ and a mixed strategy $\mixedStrat$.
The goal is to show that all neighbourhoods of $\expectancy^{\mixedStrat}_\mdpState(\payoffTuple)$ intersect $\convex{\paySetPure{\payoffTuple}{\mdpState}}$.
In other words, we must prove that for all $\varepsilon > 0$ and all $M\in\IR$, there exist finitely many pure strategies $\stratBMDP_1$, \ldots, $\stratBMDP_{\indexSequence}$ and convex combination coefficients $\scalar_1, \ldots, \scalar_{\indexSequence}\in\ccInt{0}{1}$ and, for all $1\leq\indexPayoff\leq\numObj$:
\begin{itemize}
\item if $\expectancy^{\mixedStrat}_{\mdpState}(\payoff_\indexPayoff) = + \infty$, then $\sum_{\indexSequenceB=1}^{\indexSequence}\scalar_\indexSequenceB\expectancy^{\stratBMDP_\indexSequenceB}_{\mdpState}(\payoff_\indexPayoff)\geq M$,
\item if $\expectancy^{\mixedStrat}_{\mdpState}(\payoff_\indexPayoff) = - \infty$, then $\sum_{\indexSequenceB=1}^{\indexSequence}\scalar_\indexSequenceB\expectancy^{\stratBMDP_\indexSequenceB}_{\mdpState}(\payoff_\indexPayoff)\leq -M$ and,
\item otherwise, $\sum_{\indexSequenceB=1}^{\indexSequence}\scalar_\indexSequenceB\expectancy^{\stratBMDP_\indexSequenceB}_{\mdpState}(\payoff_\indexPayoff)\geq \expectancy^{\mixedStrat}_{\mdpState}(\payoff_\indexPayoff) - \varepsilon$.
\end{itemize}

We fix $\varepsilon >0$ and $M\in\IR$.
The proof is based on the integral formulation of $\expectancy^{\mixedStrat}_\mdpState(\payoffTuple)$ from Lemma~\ref{lem:expectancy:pure integral}. Even though Lemma~\ref{lem:expectancy:pure integral} is not applicable to all payoffs, Lemma~\ref{lem:unambiguous:bound} implies that there exists a vector $\vect$ such that the payoff $\payoffTuple+\vect$ satisfies the assumptions of Lemma~\ref{lem:expectancy:pure integral}.
We can then recover the result for the original payoff using the linearity of the expectation.
We thus assume without loss of generality that Lemma~\ref{lem:expectancy:pure integral} applies to all payoffs $\payoff_1$, \ldots, $\payoff_\numObj$.
For all $1\leq\indexPayoff\leq\numObj$, we let $\rv_\indexPayoff\colon\stratClassPure{\pomdp}\to\IRbar\colon \stratBMDP\mapsto \expectancy^{\stratBMDP}_\mdpState(\payoff_\indexPayoff)$.
We let $\rvVect = (\rv_1, \ldots, \rv_\numObj)$.
By Lemma~\ref{lem:expectancy:pure integral}, we have $\expectation_{\mdpState}^{\mixedStrat}(\payoffTuple) = \int_{\stratBMDP\in\stratClassPure{\pomdp}}\rvVect(\stratBMDP)\ud\mixedStrat(\stratBMDP)$.
We sketch the proof when the $\rv_j$ are ($\mixedStrat$-almost-surely) real-valued functions.
We comment on the generalisation at the end of the sketch.

The broad idea is as follows.
First, we approximate $\rvVect$ with a multivariate random variable $\rvBVect = (\rvB_1, \ldots, \rvB_\numObj)$ over $\stratClassPure{\pomdp}$.
We then approximate the integral of $\rvBVect$ with a convex combination $\sum_{\indexSequenceB=0}^\indexSequence\scalar_\indexSequenceB\rvBVect(\stratBMDP_\indexSequenceB)\in\convex{\image{\rvBVect}}$.
Finally, we derive the convex combination $\sum_{\indexSequenceB=0}^\indexSequence\scalar_\indexSequenceB\rvVect(\stratBMDP_\indexSequenceB)$ from the previous one.
The successive approximations above ensure that the last convex combination respects the claims of the theorem.

\begin{figure}\centering
  \begin{tikzpicture}[yscale=0.6]
    \draw[very thin,color=gray] (-0.1,-0.1) grid (4.1,8.1);
    \draw[-stealth] (-0.5,0) -- (4.5,0);
    \draw[-stealth] (0,-0.5) -- (0,8);

    \node[left] at (0, 1) {$\frac{1}{4}$};
    \node[left] at (0, 2) {$\frac{1}{2}$};
    \node[left] at (0, 3) {$\frac{3}{4}$};
    \node[left] at (0, 4) {$1$};
    \node[left] at (0, 5) {$\frac{5}{4}$};
    \node[left] at (0, 6) {$\frac{3}{2}$};
    \node[left] at (0, 7) {$\frac{7}{4}$};
    \node[left] at (0, 8) {$2$};

    \node[below] at (4, 0) {$1$};

    \draw[domain=0:4,color=blue,smooth] plot (\x,\x * \x * \x - 6 * \x * \x + 8* \x + 4);
    \draw[domain=0:0.1391941468883,color=red] plot (\x, 4);
    \draw[domain=0.1391941468883:0.3248691294334,color=red] plot (\x, 5);
    \draw[domain=0.3248691294334:0.697224362268,color=red] plot (\x, 6);
    \draw[domain=0.697224362268:1,color=red] plot (\x, 7);
    \draw[domain=1:1.4608111271891,color=red] plot (\x, 6);
    \draw[domain=1.4608111271891:1.7458983116349,color=red] plot (\x, 5);
    \draw[domain=1.7458983116349:2,color=red] plot (\x, 4);
    \draw[domain=2:2.2541016883651,color=red] plot (\x, 3);
    \draw[domain=2.2541016883651:2.5391888728109,color=red] plot (\x, 2);
    \draw[domain=2.5391888728109:3,color=red] plot (\x, 1);
    \draw[domain=3:3.302775637732,color=red,ultra thick] plot (\x, 0);
    \draw[domain=3.302775637732:3.6751308705666,color=red] plot (\x, 1);
    \draw[domain=3.6751308705666:3.8608058531117,color=red] plot (\x, 2);
    \draw[domain=3.8608058531117:4,color=red] plot (\x, 3);

  \end{tikzpicture}
  \vspace{-3mm}
  \caption{Illustration of the approach used to construct the approximation $\rvBVect$ of $\rvVect$ adapted to a function over $\ccInt{0}{1}$.
    We round the blue function down to the closest multiple of $\frac{1}{4}$ to obtain the red function.
    This yields a linear combination of indicators that is $\frac{1}{4}$-close in all points to the function in blue.
  }\label{figure:approximation scheme}
\end{figure}
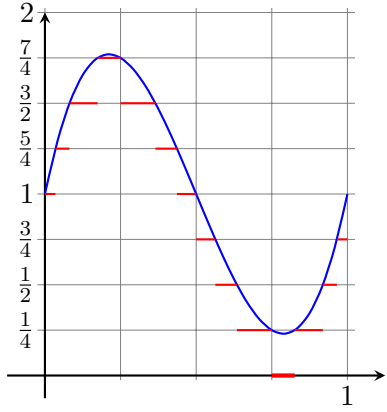

We now expand this idea.
First, we construct $\rvBVect$ such that $\rvVect - \frac{\varepsilon}{3}\oneVect\leq\rvBVect\leq\rvVect+\frac{\varepsilon}{3}\oneVect$.
It follows that, for all $1\leq\indexPayoff\leq\numObj$, $\int_{\stratBMDP\in\stratClassPure{\pomdp}}\rvB_\indexPayoff\ud\mixedStrat(\stratBMDP)$ is equal to $\expectancy^{\mixedStrat}_\mdpState(\payoff_\indexPayoff)$ whenever $\expectancy^{\mixedStrat}_\mdpState(\payoff_\indexPayoff)\in\{-\infty, +\infty\}$ and otherwise is $\frac{\varepsilon}{3}$-close.
Intuitively, we construct $\rvBVect$ as an infinite linear combination of indicators, following the rounding idea illustrated in Figure~\ref{figure:approximation scheme} (where the rounding precision depends on $\varepsilon$).
Its integral is thus (informally) an infinite convex combination of images of $\rvBVect$: there are sequences $(\scalarB_\indexSequenceB)_{\indexSequenceB\in\IN}$ and $(\stratBMDP_\indexSequenceB)_{\indexSequenceB\in\IN}$ respectively of coefficients and elements of $\stratClassPure{\pomdp}$ such that $\sum_{\indexSequenceB=0}^\infty\beta_\indexSequenceB=1$ and $\int_{\stratBMDP\in\stratClassPure{\pomdp}}\rvBVect\ud\mixedStrat(\stratBMDP) = \sum_{\indexSequenceB\in\IN}\scalarB_\indexSequenceB\rvBVect(\stratBMDP_\indexSequenceB)$.
We derive a sequence $(\payoffVectB^{(\indexSequence)})_{\indexSequence\in\IN}$ in $\convex{\image{\rvBVect(\stratBMDP_\indexSequenceB)}}$ that converges to $\int_{\stratBMDP\in\stratClassPure{\pomdp}}\rvBVect\ud\mixedStrat(\stratBMDP)$ from this series: we let
\[\payoffVectB^{(\indexSequence)} = \sum_{\indexSequenceB=0}^\indexSequence\scalarB_\indexSequenceB\rvBVect(\stratBMDP_\indexSequenceB) + \left(1-\sum_{\indexSequenceB=0}^\indexSequence\scalarB_\indexSequenceB\right)\rvBVect(x_0)\] for all $\indexSequence\in\IN$.

Fix $\indexSequence\in\IN$ large enough such that, for all $1\leq\indexPayoff\leq\numObj$, component $\indexPayoff$ of $\payoffVectB^{(\indexSequence)}$ is $\frac{\varepsilon}{3}$-close to $\int_{\stratBMDP\in\stratClassPure{\pomdp}}\rvB_\indexPayoff\ud\mixedStrat(\stratBMDP)$ if it is a real number or greater than $M + \varepsilon$ in absolute value otherwise.
The convex combination $\payoffVect = \sum_{\indexSequenceB=0}^\indexSequence\scalarB_\indexSequenceB\expectancy^{\stratBMDP_\indexSequenceB}_\mdpState(\payoffTuple) + (1-\sum_{\indexSequenceB=0}^\indexSequence\scalarB_\indexSequenceB)\expectancy^{\stratBMDP_{x_0}}_\mdpState(\payoffTuple)$ is a satisfactory convex combination with respect to the claim of the theorem.
Let $1\leq\indexPayoff\leq\numObj$.
If $\expectancy^{\mixedStrat}_\mdpState(\payoff_\indexPayoff) = +\infty$, we obtain that $\payoffComp_\indexPayoff\geq M$.
Similarly, if $\expectancy^{\mixedStrat}_\mdpState(\payoff_\indexPayoff) = -\infty$, we obtain that $\payoffComp_\indexPayoff\leq -M$.
Otherwise, we have that $\payoffComp_\indexPayoff$ is $\varepsilon$-close to $\expectancy^{\mixedStrat}_\mdpState(\payoff_\indexPayoff)$.

We now briefly discuss the case where some $\rv_\indexPayoff$ is not $\mixedStrat$-almost-surely real-valued.
For the sake of illustration, we assume that this only applies to $\indexPayoff=\numObj$ and that $\rv_\numObj\geq 0$.
Therefore, we have $\expectancy^{\mixedStrat}_\mdpState(\payoff_\numObj) = +\infty$ and there exists some $\stratBMDP\in\stratClassPure{\pomdp}$ such that $\expectancy^{\stratBMDP}_\mdpState(\payoff_\numObj) = +\infty$ and $\expectancy^{\stratBMDP}_\mdpState(\payoff_\indexPayoff)\in\IR$ for all $1\leq\indexPayoff\leq\numObj-1$.
Let $\stratBMDP_1, \ldots, \stratBMDP_\indexSequence$ and $\scalar_1, \ldots, \scalar_\indexSequence$ given by the theorem for $(\payoff_1, \ldots, \payoff_{\numObj-1})$, $\frac{\varepsilon}{2}$ and $M+\varepsilon$.
Let $\payoffVect=\sum_{\indexSequenceB=1}^\indexSequence\scalar_\indexSequenceB\expectancy^{\stratBMDP_\indexSequenceB}_\mdpState(\payoffTuple)$.
By choosing $\eta\in\ooInt{0}{1}$ such that all components of $\expectancy^{\stratBMDP}_\mdpState(\payoffTuple)$ other than the last have absolute value no more than $\frac{\varepsilon}{2}$ and the finite components of $(1-\eta)\payoffVect$ are $\frac{\varepsilon}{2}$-close to the corresponding components of $\payoffVect$, we obtain a suitable convex combination in the form of $\eta\expectancy^{\stratBMDP}_\mdpState(\payoffTuple) + (1-\eta)\payoffVect$.

\paragraph{Formal statement and proof.}
We state our theorem below and prove it.

\begin{restatable}{theorem}{thmMixingApprox}\label{thm:mixing:approx}
    Assume that $\payoffTuple$ is universally unambiguously integrable.
  Let $\mdpState\in\mdpStateSpace$.
  We have $\closure{\paySet{\payoffTuple}{\mdpState}} = \closure{\convex{\paySetPure{\payoffTuple}{\mdpState}}}$.
  In particular, for all strategies $\stratMDP$, all $\varepsilon > 0$ and all $M\in\IR$, there exist finitely many pure strategies $\stratBMDP_1$, \ldots, $\stratBMDP_{\indexSequence}$ and convex combination coefficients $\scalar_1, \ldots, \scalar_{\indexSequence}\in\ccInt{0}{1}$ such that for all $1\leq\indexPayoff\leq\numObj$:
  \begin{itemize}
  \item if $\expectancy^{\stratMDP}_{\mdpState}(\payoff_\indexPayoff) = + \infty$, then $\sum_{\indexSequenceB=1}^{\indexSequence}\scalar_\indexSequenceB\expectancy^{\stratBMDP_\indexSequenceB}_{\mdpState}(\payoff_\indexPayoff)\geq M$,
  \item if $\expectancy^{\stratMDP}_{\mdpState}(\payoff_\indexPayoff) = - \infty$, then $\sum_{\indexSequenceB=1}^{\indexSequence}\scalar_\indexSequenceB\expectancy^{\stratBMDP_\indexSequenceB}_{\mdpState}(\payoff_\indexPayoff)\leq -M$, and,
  \item otherwise, if $\expectancy^{\stratMDP}_{\mdpState}(\payoff_\indexPayoff)\in\IR$, $\expectancy^{\stratMDP}_{\mdpState}(\payoff_\indexPayoff) - \varepsilon\leq\sum_{\indexSequenceB=1}^{\indexSequence}\scalar_\indexSequenceB\expectancy^{\stratBMDP_\indexSequenceB}_{\mdpState}(\payoff_\indexPayoff)\leq \expectancy^{\stratMDP}_{\mdpState}(\payoff_\indexPayoff) + \varepsilon$.
  \end{itemize}
\end{restatable}
\begin{proof}
    The inclusion $\closure{\convex{\paySetPure{\payoffTuple}{\mdpState}}}\subseteq\closure{\paySet{\payoffTuple}{\mdpState}}$ follows from the convexity of $\paySet{\payoffTuple}{\mdpState}$.
  For other inclusion, it suffices to show that $\paySet{\payoffTuple}{\mdpState}\subseteq\closure{\convex{\paySetPure{\payoffTuple}{\mdpState}}}$.
  This inclusion is equivalent to the last property of the theorem statement.
  
  Let $\stratMDP$ be a strategy, $\varepsilon > 0$ and $M\in\IR$.
  Let $\mixedStrat$ be a mixed strategy that is outcome-equivalent to $\stratMDP$ (whose existence follows from Kuhn's theorem).
  To prove the theorem, we reason on the $\mixedStrat$-integral of random variables of $(\stratClassPure{\pomdp}, \stratSigmaAlgebra)$.
  For any real or multivariate random variable $\rvB$ over $\stratClassPure{\pomdp}$, we write $\int\rvB\ud\mixedStrat$ for $\int_{\stratBMDP\in\stratClassPure{\pomdp}}\rvB(\stratBMDP)\ud\mixedStrat(\stratBMDP)$ to lighten notation.

  We make two assumptions without loss of generality.
  We defer the proof that these assumptions are without loss of generality to the end of the proof.
  First, we assume that for all $1\leq\indexPayoff\leq\numObj$, either $\expectancy^{\stratBMDP}_{\mdpState}(\payoff_\indexPayoff)\geq 0$ for all strategies $\stratBMDP$ or $\expectancy^{\stratBMDP}_{\mdpState}(\payoff_\indexPayoff)\leq 0$ for all strategies $\stratBMDP$.
  This assumption guarantees that $\expectancy^{\stratMDP}_\mdpState(\payoffTuple) = \int_{\stratBMDP\in\stratClassPure{\pomdp}}\expectancy^{\stratBMDP}_\mdpState(\payoffTuple)\ud\mixedStrat(\stratBMDP)$ by Lemma~\ref{lem:expectancy:pure integral}.
  Second, we assume that for all $1\leq\indexPayoff\leq\numObj$, the set  $\{\stratBMDP\in\stratClassPure{\pomdp}\mid\expectancy^{\stratBMDP}_\mdpState(\payoff_\indexPayoff)\notin\IR\}$ has $\mixedStrat$-measure zero.

  For all $1\leq\indexPayoff\leq\numObj$, we consider the random variable  $\rv_\indexPayoff\colon \stratBMDP\mapsto \expectancy^{\stratBMDP}_\mdpState(\payoff_\indexPayoff)$ over $\stratClassPure{\pomdp}$.
  We let $\rvVect = (\rv_1, \ldots, \rv_\numObj)$.
  The first assumption above implies that for all $1\leq\indexPayoff\leq\numObj$, $\rv_\indexPayoff$ is a non-negative or non-positive random variable.
  The second assumption implies that $\rvVect$ is almost-surely $\IR^\numObj$-valued.  
  We thus interpret the random variable $\rvVect$ as a function $\stratClassPure{\pomdp}\to\IR^\numObj$.

  We now prove the result under the assumptions above.
  The first step of the proof is to construct a random variable $\rvBVect = (\rvB_1, \ldots, \rvB_\numObj)\colon\stratClassPure{\pomdp}\to\IR^\numObj$ which approximates $\rvVect$.
  To ensure that the integral of $\rvBVect$ is an infinite convex combination of images of $\rvBVect$ (in the sense of the proof overview), we define $\rvBVect$ as a function similar to a simple function.
  More precisely, we define $\rvBVect$ as a series of indicators multiplied by coefficients.

  We fix $k\in\IN$ such that $\frac{1}{2^k}\leq\frac{\varepsilon}{3}$.
  We define $\rvBVect$ component by component first, and introduce its series form later.
  Let $1\leq\indexPayoff\leq\numObj$.
  We generalise the construction illustrated in Figure~\ref{figure:approximation scheme}.
  For all $k\in\IN$, we let $\rvB_{\indexPayoff}$ be the random variable over $\stratClassPure{\pomdp}$ such that
  \[\rvB_{\indexPayoff}=
    \sum_{\ell=0}^{\infty}
    \frac{\ell}{2^k}\cdot\indic{\coInt{\frac{\ell}{2^k}}{\frac{\ell+1}{2^k}}}(\rv_\indexPayoff)\]
  if $\rv_\indexPayoff$ is non-negative, i.e., we round $\rv_\indexPayoff$ down to the closest multiple of $\frac{1}{2^k}$, and, otherwise,
  \[\rvB_{\indexPayoff}=
    \sum_{\ell=0}^{\infty}
    \frac{-\ell}{2^k}\cdot\indic{\ocInt{\frac{-\ell-1}{2^k}}{\frac{-\ell}{2^k}}}(\rv_\indexPayoff),\]
  i.e., we round $\rv_\indexPayoff$ up to the closest multiple of $\frac{1}{2^k}$.
  We have $\rv_\indexPayoff-\frac{1}{2^k}\leq \rvB_\indexPayoff\leq \rv_\indexPayoff+\frac{1}{2^k}$.
  This implies that $\rv_\indexPayoff-\frac{\varepsilon}{3}\leq \rvB_\indexPayoff^{(k)}\leq \rv_\indexPayoff+\frac{\varepsilon}{3}$.
  In particular, $\rvB_\indexPayoff$ is integrable if and only if $\rv_\indexPayoff$ is and, and, if both are integrable:
  \begin{equation}\label{equation:mixing:approx:b:4}
    \expectancy^{\stratMDP}_\mdpState(\payoff_\indexPayoff) -\frac{\varepsilon}{3}\leq    
    \int\rvB_\indexPayoff\ud\mixedStrat\leq
    \expectancy^{\stratMDP}_\mdpState(\payoff_\indexPayoff)+
    \frac{\varepsilon}{3}.
  \end{equation}
  We also have $\expectancy^{\stratMDP}_\mdpState(\payoff_\indexPayoff) = +\infty$ (resp.~$-\infty$) if and only if $\int\rvB_\indexPayoff\ud\mixedStrat = +\infty$ (resp.~$-\infty$).

  Now that we have shown that $\rvBVect$ approximates $\rvVect$, we prove that the integral of $\rvBVect$ can be written as an infinite convex combination of elements of $\image{\rvBVect}$.
  To this end, we rewrite $\rvBVect$ in the series form mentioned above.

  Let $1\leq\indexPayoff\leq\numObj$.
  If $\rv_\indexPayoff\geq 0$, we define, for all $\ell\in\IN$, $I_\indexPayoff(\ell) = \coInt{\frac{\ell}{2^k}}{\frac{\ell+1}{2^k}}$ and $\vectComp_{\indexPayoff}(\ell) = \frac{\ell}{2^k}$.
  Otherwise, if $\rv_\indexPayoff\geq 0$ does not hold, we define, for all $\ell\in\IN$, $I_\indexPayoff(\ell) = \ocInt{\frac{-\ell-1}{2^k}}{\frac{-\ell}{2^k}}$ and $\vectComp_{\indexPayoff}(\ell) = \frac{-\ell}{2^k}$.
  We fix an enumeration $(\bar{\ell}^{(\indexSequenceB)})_{\indexSequenceB\in\IN}$ of $\IN^\numObj$.
  For all $\indexSequenceB\in\IN$, we let $\bar{\ell}^{(\indexSequenceB)} = (\ell_1^{(\indexSequenceB)}, \ldots, \ell_\numObj^{(\indexSequenceB)})\in\IN^\numObj$.
  We define $B_\indexSequenceB = \prod_{\indexPayoff=1}^\numObj I_\indexPayoff(\ell_\indexPayoff^{(\indexSequenceB)})$ and $\vect_\indexSequenceB = (\vectComp_\indexPayoff(\ell_\indexPayoff^{(\indexSequenceB)}))_{1\leq\indexPayoff\leq\numObj}$.
  We can rewrite $\rvBVect$ as follows, using this notation:
  \[\rvBVect = \sum_{\indexSequenceB\in\IN}\vect_\indexSequenceB\indic{B_\indexSequenceB}(\rvVect).\]
  Through this, we obtain that the integral of $\rvBVect$ is an infinite convex combination: the monotone convergence theorem ensures that
  \[\int\rvBVect\ud\mixedStrat =
    \sum_{\indexSequenceB\in\IN}
    \vect_\indexSequenceB
    \mixedStrat
    \left(\rvVect^{-1}(B_\indexSequenceB)\right).
  \]

  Next, to determine the coefficients and pure strategies we seek, we define a sequence in $\convex{\image{\rvBVect}}$ converging to $\int\rvBVect\ud\mixedStrat$.
  For all $\indexSequenceB\in\IN$, we let $\scalarB_\indexSequenceB = \mixedStrat\left(\rvVect^{-1}(B_\indexSequenceB)\right)$ and let $\stratBMDP_\indexSequenceB\in\stratClassPure{\pomdp}$ such that $\vect_\indexSequenceB=\rvBVect(\stratBMDP_\indexSequenceB)$ if $\scalarB_\indexSequenceB\neq 0$, and $\stratBMDP_\indexSequenceB$ is left arbitrary otherwise.
  We consider the sequence $(\payoffVectB^{(\indexSequence)})_{\indexSequence\in\IN}$ defined by, for all $\indexSequence\in\IN$,
  \[\payoffVectB^{(\indexSequence)} = \sum_{\indexSequenceB=1}^\indexSequence\scalarB_\indexSequenceB\rvBVect(\stratBMDP_\indexSequenceB) + \left(1 - \sum_{\indexSequenceB=1}^\indexSequence\scalarB_\indexSequenceB\right)\rvBVect(\stratBMDP_0).\]
  We have $\lim_{\indexSequence\to\infty}\payoffVectB^{(\indexSequence)} = \int\rvBVect\ud\mixedStrat$ and for all $\indexSequence\in\IN$, $\payoffVectB^{(\indexSequence)}\in\convex{\image{\rvBVect}}$.
  For all $\indexSequence\in\IN$, we let $\payoffVectB^{(\indexSequence)} = (\payoffCompB_1^{(\indexSequence)}, \ldots, \payoffCompB_\numObj^{(\indexSequence)})$.
  
  We now fix $\indexSequence$ such that, for all $1\leq\indexPayoff\leq\numObj$, $\int\rvB_\indexPayoff\ud\mixedStrat\in\IR$ implies that
  \begin{equation}\label{equation:mixing:approx:b:1}
    -\frac{\varepsilon}{3} \leq  \payoffCompB_\indexPayoff^{(\indexSequence)} - \int\rvB_\indexPayoff\ud\mixedStrat \leq \frac{\varepsilon}{3},
  \end{equation} $\int\rvB_\indexPayoff\ud\mixedStrat=+\infty$ implies that $\payoffCompB_\indexPayoff^{(\indexSequence)}\geq M+\varepsilon$ and $\int\rvB_\indexPayoff\ud\mixedStrat=-\infty$ implies that $\payoffCompB_\indexPayoff^{(\indexSequence)}\leq -M-\varepsilon$.
  We set $\scalar_0 = 1 - \sum_{\indexSequenceB=1}^\indexSequence\scalarB_\indexSequenceB$ and for $1\leq\indexSequenceB\leq\indexSequence$, $\scalar_\indexSequenceB = \scalarB_\indexSequenceB$.
  We remark that $\payoffVectB^{(\indexSequence)} = \sum_{\indexSequenceB=0}^\indexSequence\scalar_\indexSequenceB\rvBVect(\stratBMDP_\indexSequenceB)$.
  We show that the convex combination $\payoffVect= \sum_{\indexSequenceB=0}^\indexSequence\scalar_\indexSequenceB\rvVect(\stratBMDP_\indexSequenceB) = \sum_{\indexSequenceB=0}^\indexSequence\scalar_\indexSequenceB\expectancy^{\stratBMDP_{\indexSequenceB}}_\mdpState(\payoffTuple)$ satisfies the claim of the theorem.
  
  We write $\payoffVect = (\payoffComp_1, \ldots, \payoffComp_\numObj)$.
  It follows from the inequalities $\rvB_\indexPayoff-\frac{\varepsilon}{3}\leq \rv_\indexPayoff\leq \rvB_\indexPayoff+\frac{\varepsilon}{3}$ for all $1\leq\indexPayoff\leq\numObj$ and the definitions of $\payoffVect$ and $\payoffVectB^{(\indexSequence)}$ that
  \begin{equation}\label{equation:mixing:approx:b:2}
    \payoffVectB^{(\indexSequence)} - \frac{\varepsilon}{3}\oneVect\leq \payoffVect \leq\payoffVectB^{(\indexSequence)}+\frac{\varepsilon}{3}\oneVect.
  \end{equation}

  Let $1\leq\indexPayoff\leq\numObj$.
  First, assume that $\expectancy^{\stratMDP}_\mdpState(\payoff_\indexPayoff) = +\infty$ (i.e., $\int\rvB_\indexPayoff\ud\mixedStrat=+\infty$).
  In this case, it follows from Equation~\eqref{equation:mixing:approx:b:2} and $\payoffCompB_\indexPayoff^{(\indexSequence)}\geq M + \varepsilon$ that 
  \[
    \payoffComp_\indexPayoff = 
    \sum_{\indexSequenceB=0}^\indexSequence\scalar_\indexSequenceB
    \expectancy^{\stratBMDP_\indexSequenceB}_\mdpState(\payoff_\indexPayoff)\geq
    \payoffCompB_\indexPayoff^{(\indexSequence)} - \frac{\varepsilon}{3} \geq
    M + \frac{2\varepsilon}{3}\geq
    M.
  \]
  The argument for the case $\expectancy^{\stratMDP}_\mdpState(\payoff_\indexPayoff) = -\infty$ follows an analogous reasoning and is omitted.
  
  We now assume that $\expectancy^{\stratMDP}_\mdpState(\payoff_\indexPayoff) \in\IR$ (i.e., $\int\rvB_\indexPayoff\ud\mixedStrat\in\IR$).
  By applying Equations~\eqref{equation:mixing:approx:b:2},~\eqref{equation:mixing:approx:b:1} and~\eqref{equation:mixing:approx:b:4} in succession, we obtain that
  \begin{align*}
    \payoffComp_\indexPayoff
    & \leq \payoffCompB_\indexPayoff^{(\indexSequence)} + \frac{\varepsilon}{3} \\
    & \leq \int\rvB_\indexPayoff\ud\mixedStrat + \frac{2\varepsilon}{3} \\
    & \leq \expectancy^{\stratMDP}_\mdpState(\payoff_\indexPayoff) + \varepsilon.
  \end{align*}
  A similar succession of inequalities (referring to the same equations), yields $\payoffComp_\indexPayoff\geq \expectancy^{\stratMDP}_\mdpState(\payoff_\indexPayoff) - \varepsilon$.
  This ends the argument that $\payoffVect$ satisfies the conditions outlined in the statement of the theorem.

  To end the proof, it remains to show that the assumptions made above are without loss of generality.
  We recall them first:
  \begin{enumerate}
  \item for all $1\leq\indexPayoff\leq\numObj$, $\expectancy^{\stratBMDP}_{\mdpState}(\payoff_\indexPayoff)\geq 0$ for all $\stratBMDP\in\stratClassPure{\pomdp}$ or $\expectancy^{\stratBMDP}_{\mdpState}(\payoff_\indexPayoff)\leq 0$ for all $\stratBMDP\in\stratClassPure{\pomdp}$;\label{item:mixing:approx:assumption:1}
  \item for all $1\leq\indexPayoff\leq\numObj$, $\mixedStrat(\{\stratBMDP\in\stratClassPure{\pomdp}\mid\rv_\indexPayoff(\stratBMDP)\notin\IR\}) = 0$.\label{item:mixing:approx:assumption:2}
  \end{enumerate}
  In the above, we have shown that the claim of the theorem holds with Assumptions~\ref{item:mixing:approx:assumption:1} and~\ref{item:mixing:approx:assumption:2}.
  In the following, we first show that the theorem with both Assumptions~\ref{item:mixing:approx:assumption:1} and~\ref{item:mixing:approx:assumption:2} implies the theorem with only Assumption~\ref{item:mixing:approx:assumption:2}.
  After this, we show that the theorem with Assumption~\ref{item:mixing:approx:assumption:2} implies the theorem with neither additional assumption.

  Assume that Assumption~\ref{item:mixing:approx:assumption:2} holds and let us show that the claim of the theorem holds.
  To obtain the result for $\payoffTuple$, we derive a payoff $\payoffTupleB = (\payoffB_1, \ldots, \payoffB_\numObj)$ from $\payoffTuple$ such that $\payoffTupleB$ satisfies the conditions outlined in Assumptions~\ref{item:mixing:approx:assumption:1} and~\ref{item:mixing:approx:assumption:2}, so that we can apply the variant of the theorem with Assumptions~\ref{item:mixing:approx:assumption:1} and~\ref{item:mixing:approx:assumption:2} to $\payoffTupleB$.
  
  Let $1\leq\indexPayoff\leq\numObj$.
  If $\inf_{\stratBMDP\in\stratClassAll{\pomdp}}\expectancy^{\stratBMDP}_\mdpState(\payoff_\indexPayoff)\in\IR$, we let $\payoffB_\indexPayoff = \payoff_\indexPayoff - \inf_{\stratBMDP\in\stratClassAll{\pomdp}}\expectancy^{\stratBMDP}_\mdpState(\payoff_\indexPayoff)$.
  We obtain that $\expectancy^{\stratBMDP}_\mdpState(\payoffB_\indexPayoff)\geq 0$ for all strategies $\stratBMDP$.
  Indeed, for all strategies $\stratBMDP$, this follows by linearity of $\expectancy$ if $\payoff_\indexPayoff$ is $\proba^{\stratBMDP}_\mdpState$-integrable (which is equivalent to $\payoffB_\indexPayoff$ being $\proba^{\stratBMDP}_\mdpState$-integrable) and otherwise the non-negative parts of $\payoff_\indexPayoff$ and $\payoffB_\indexPayoff$ are $|\inf_{\stratBMDP'\in\stratClassAll{\pomdp}}\expectancy^{\stratBMDP'}_\mdpState(\payoff_\indexPayoff)|$-close to one another, thus share their infinite integral and we obtain $\expectancy^{\stratBMDP}_\mdpState(\payoffB_\indexPayoff) = \expectancy^{\stratBMDP}_\mdpState(\payoff_\indexPayoff)=+\infty\geq 0$.
  Otherwise, by Lemma~\ref{lem:unambiguous:bound}, we have $\sup_{\stratBMDP\in\stratClassAll{\pomdp}}\expectancy^{\stratBMDP}_\mdpState(\payoff_\indexPayoff)\in\IR$ and we let $\payoffB_\indexPayoff = \payoff_\indexPayoff - \sup_{\stratBMDP\in\stratClassAll{\pomdp}}\expectancy^{\stratBMDP}_\mdpState(\payoff_\indexPayoff)$.
  By adapting the argument of the previous case, we obtain that for all strategies $\stratBMDP$, we have $\expectancy^{\stratBMDP}_\mdpState(\payoffB_\indexPayoff)\leq 0$ and the equivalence $\expectancy^{\stratBMDP}_\mdpState(\payoffB_\indexPayoff)=-\infty$ if and only if $\expectancy^{\stratBMDP}_\mdpState(\payoff_\indexPayoff)=-\infty$.

  Let $M' = M + \max_{1\leq\indexPayoff\leq\numObj}|\scalarC_\indexPayoff|$ where $\scalarC_\indexPayoff$ is the constant such that $\payoffB_\indexPayoff = \payoff_\indexPayoff - \scalarC_\indexPayoff$ for all $1\leq\indexPayoff\leq\numObj$.
  We let $\stratBMDP_1$, \ldots, $\stratBMDP_\indexSequence$ be the pure strategies and $\scalar_1$, \ldots, $\scalar_\indexSequence\in\ccInt{0}{1}$ be the coefficients given by the theorem with Assumptions~\ref{item:mixing:approx:assumption:1} and~\ref{item:mixing:approx:assumption:2} for $\payoffTupleB$, $\stratMDP$, $\varepsilon$ and $M'$.
  Let $1\leq\indexPayoff\leq\numObj$.
  For all $1\leq\indexSequenceB\leq\indexSequence$, $\payoff_\indexPayoff$ is $\proba^{\stratBMDP_\indexSequenceB}_\mdpState$-integrable (by the remaining additional assumption).
  We thus have $\sum_{\indexSequenceB=1}^\indexSequence\scalar_\indexSequenceB\expectancy^{\stratBMDP_\indexSequenceB}(\payoff_\indexPayoff) = \scalarC_\indexPayoff + \sum_{\indexSequenceB=1}^\indexSequence\scalar_\indexSequenceB\expectancy^{\stratBMDP_\indexSequenceB}(\payoffB_\indexPayoff)$.
  If $\expectancy^{\stratMDP}_\mdpState(\payoff_\indexPayoff)\in\IR$, i.e., $\payoff_\indexPayoff$ is $\proba^{\stratMDP}_\mdpState$-integrable, then so is $\payoffB_\indexPayoff$ and we directly obtain $\expectancy^{\stratMDP}_{\mdpState}(\payoff_\indexPayoff) - \varepsilon\leq\sum_{\indexSequenceB=1}^{\indexSequence}\scalar_\indexSequenceB\expectancy^{\stratBMDP_\indexSequenceB}_{\mdpState}(\payoff_\indexPayoff)\leq \expectancy^{\stratMDP}_{\mdpState}(\payoff_\indexPayoff) + \varepsilon$ from the similar inequality for $\payoffB_\indexPayoff$.
  Next, assume that $\expectancy^{\stratMDP}_\mdpState(\payoff_\indexPayoff)=+\infty$.
  By the above, this implies that $\expectancy^{\stratMDP}_\mdpState(\payoffB_\indexPayoff)=+\infty$.
  We obtain (from the application of the theorem to $\payoffB_\indexPayoff$) that $\sum_{\indexSequenceB=1}^{\indexSequence}\scalar_\indexSequenceB\expectancy^{\stratBMDP_\indexSequenceB}_{\mdpState}(\payoff_\indexPayoff)\geq\scalarC_\indexPayoff + M'\geq M$.
  Finally, if $\expectancy^{\stratMDP}_\mdpState(\payoffB_\indexPayoff)=-\infty$, we obtain that $\sum_{\indexSequenceB=1}^{\indexSequence}\scalar_\indexSequenceB\expectancy^{\stratBMDP_\indexSequenceB}_{\mdpState}(\payoff_\indexPayoff)\leq\scalarC_\indexPayoff - M'\leq -M$ in the same way as the previous case.

  It remains to show that the theorem with Assumption~\ref{item:mixing:approx:assumption:2} implies the theorem with no assumption.
  For convenience of notation, we assume that there exists $1\leq\numObj'\leq\numObj$ such that, for all $1\leq\indexPayoff\leq\numObj'$, $\mixedStrat(\{\stratBMDP\in\stratClassPure{\pomdp}\mid\expectancy^{\stratBMDP}_\mdpState(\payoff_\indexPayoff)\notin\IR\}) > 0$ and, for all $\numObj'+1\leq\indexPayoff\leq\numObj$ and $\mixedStrat(\{\stratBMDP\in\stratClassPure{\pomdp}\mid\expectancy^{\stratBMDP}_\mdpState(\payoff_\indexPayoff)\in\IR\}) = 1$.

  For all $1\leq\indexPayoff\leq\numObj'$, $\expectancy^{\stratMDP}_\mdpState(\payoff_\indexPayoff)$ is infinite and there exists $\stratBMDP_\indexPayoff\in\stratClassPure{\pomdp}$ such that $\expectancy^{\stratBMDP_{\indexPayoff}}_\mdpState(\payoff_\indexPayoff) = \expectancy^{\stratMDP}_\mdpState(\payoff_\indexPayoff)$ (because $\mixedStrat(\{\stratBMDP\in\stratClassPure{\pomdp}\mid\expectancy^{\stratBMDP}_\mdpState(\payoff_\indexPayoff)\notin\IR\})>0$) and, for all $\numObj'+1\leq\indexPayoff\leq\numObj$, $\expectancy^{\stratBMDP}_\mdpState(\payoff_\indexPayoff)\in\IR$ (because $\mixedStrat(\{\stratBMDP\in\stratClassPure{\pomdp}\mid\expectancy^{\stratBMDP}_\mdpState(\payoff_{\indexPayoff'})\in\IR\}) = 1$).
  If $\numObj' = \numObj$, we conclude using the convex combination $\sum_{\indexPayoff=1}^{\numObj}\frac{1}{\numObj}\expectancy^{\stratBMDP_\indexPayoff}_\mdpState(\payoffTuple) = \expectancy^{\stratMDP}_\mdpState(\payoffTuple)$.
  Therefore, we assume that $\numObj' < \numObj$.
  The payoff $(\payoff_{\numObj'+1}, \ldots, \payoff_\numObj)$ satisfies Assumption~\ref{item:mixing:approx:assumption:2}.
  We apply the theorem with Assumption~\ref{item:mixing:approx:assumption:2} to $(\payoff_{\numObj'+1}, \ldots, \payoff_\numObj)$, strategy $\stratMDP$, $\frac{\varepsilon}{3}$ and $M+\varepsilon$ to obtain pure strategies $\stratBMDP_{\numObj'+1}$, \ldots, $\stratBMDP_{\indexSequence}$ and convex combination coefficients $\scalarB_{\numObj'+1}$, \ldots, $\scalarB_\indexSequence$ that satisfy the implications given in the last property of the statement of the theorem.

  We now define convex combination coefficients $\scalar_1$, \ldots, $\scalar_\indexSequence$ to obtain the claim of the theorem for $\payoffTuple$.
  Fix $\eta\in\ooInt{0}{1}$ such that (i)~for all $1\leq\indexPayoff\leq\numObj'$, the real components of $\eta\expectancy^{\stratBMDP_\indexPayoff}_\mdpState(\payoffTuple)$ are no more than $\frac{\varepsilon}{3}$ in absolute value and (ii)~$(1-\eta)\sum_{\indexSequenceB=\numObj'+1}^\indexSequence\scalarB_\indexSequenceB\expectancy^{\stratBMDP_{\indexSequenceB}}_\mdpState((\payoff_{\numObj'+1}, \ldots, \payoff_\numObj))$ is $\frac{\varepsilon}{3}$-close to $\sum_{\indexSequenceB=\numObj'+1}^\indexSequence\scalarB_\indexSequenceB\expectancy^{\stratBMDP_{\indexSequenceB}}_\mdpState((\payoff_{\numObj'+1}, \ldots, \payoff_\numObj))$.
  For all $1\leq\indexSequenceB\leq\numObj'$, we set $\scalar_\indexSequenceB = \frac{\eta}{\numObj'}$ and for $\numObj'+ 1\leq\indexSequenceB\leq\indexSequence$, we set $\scalar_\indexSequenceB=(1-\eta)\scalarB_\indexSequenceB$.
  It follows from $\eta\in\ooInt{0}{1}$ that $1-\eta > 0$ and $\sum_{\indexSequenceB=1}^{\indexSequence}\scalar_\indexSequenceB = 1$.

  We show that the pure strategies $\stratBMDP_1$, \ldots, $\stratBMDP_\indexSequenceB$ and coefficients $\scalar_1$, \ldots, $\scalar_\indexSequence$ are witnesses to the implications in the theorem statement.
  Let $1\leq\indexPayoff\leq\numObj$.
  If $\indexPayoff\leq\numObj'$, it follows from $\eta >0$ that $\sum_{\indexSequenceB=1}^\indexSequence\scalar_\indexSequenceB\expectancy^{\stratBMDP_\indexSequenceB}_\mdpState(\payoff_\indexPayoff) = \eta \expectancy^{\stratBMDP_\indexPayoff}_\mdpState(\payoff_\indexPayoff) = \expectancy^{\stratMDP}_\mdpState(\payoff_\indexPayoff)\in\{-\infty, +\infty\}$, and the required inequality is trivially satisfied.
  We assume from here that $\indexPayoff\geq\numObj'+1$.
  It follows from properties (i) and (ii) above that $\sum_{\indexSequenceB=1}^\indexSequence\scalar_\indexSequenceB\expectancy^{\stratBMDP_\indexSequenceB}_\mdpState(\payoff_\indexPayoff)$ is $\frac{2\cdot\varepsilon}{3}$-close to $\sum_{\indexSequenceB=\numObj'+1}^\indexSequence\scalarB_\indexSequenceB\expectancy^{\stratBMDP_\indexSequenceB}_\mdpState(\payoff_\indexPayoff)$.
  Assume that $\expectancy^{\stratMDP}_\mdpState(\payoff_\indexPayoff) = +\infty$.
  Then, we have
  \[
    \sum_{\indexSequenceB=1}^\indexSequence\scalar_\indexSequenceB\expectancy^{\stratBMDP_\indexSequenceB}_\mdpState(\payoff_\indexPayoff) \geq
    \sum_{\indexSequenceB=\numObj'+1}^\indexSequence\scalarB_\indexSequenceB\expectancy^{\stratBMDP_\indexSequenceB}_\mdpState(\payoff_\indexPayoff) - \frac{2\cdot\varepsilon}{3} \geq
    M + \varepsilon - \frac{2\cdot\varepsilon}{3}\geq M.
  \]
  We obtain that $\expectancy^{\stratMDP}_\mdpState(\payoff_\indexPayoff) = -\infty$ implies $\sum_{\indexSequenceB=1}^\indexSequence\scalar_\indexSequenceB\expectancy^{\stratBMDP_\indexSequenceB}_\mdpState(\payoff_\indexPayoff) \leq -M$ similarly.
  Finally, assume that $\expectancy^{\stratMDP}_\mdpState(\payoff_\indexPayoff)\in\IR$.
  We recall that $|\sum_{\indexSequenceB=\numObj'+1}^\indexSequence\scalarB_\indexSequenceB\expectancy^{\stratBMDP_\indexSequenceB}_\mdpState(\payoff_\indexPayoff)-\expectancy^{\stratMDP}_\mdpState(\payoff_\indexPayoff)|\leq\frac{\varepsilon}{3}$ by definition of the strategies $\stratBMDP_{\numObj'+1}$, \ldots, $\stratBMDP_{\indexSequence}$ and coefficients $\scalarB_{\numObj'+1}$, \ldots, $\scalarB_{\indexSequence}$.
  We obtain, by the triangular inequality, that
  \begin{align*}
    \left|
    \sum_{\indexSequenceB=1}^\indexSequence
    \scalar_\indexSequenceB
    \expectancy^{\stratBMDP_\indexSequenceB}_\mdpState(\payoff_\indexPayoff)
    - \expectancy^{\stratMDP}_\mdpState(\payoff_\indexPayoff)\right|
    & \leq \left|\sum_{\indexSequenceB=1}^\indexSequence
      \scalar_\indexSequenceB
      \expectancy^{\stratBMDP_\indexSequenceB}_\mdpState(\payoff_\indexPayoff) -
      \sum_{\indexSequenceB=\numObj'+1}^\indexSequence
      \scalarB_\indexSequenceB
      \expectancy^{\stratBMDP_\indexSequenceB}_\mdpState(\payoff_\indexPayoff)\right| \\
    & + \left|
      \sum_{\indexSequenceB=\numObj'+1}^\indexSequence
      \scalarB_\indexSequenceB
      \expectancy^{\stratBMDP_\indexSequenceB}_\mdpState(\payoff_\indexPayoff) -
      \expectancy^{\stratMDP}_\mdpState(\payoff_\indexPayoff)
      \right| \\
    & \leq\frac{2\cdot\varepsilon}{3} + \frac{\varepsilon}{3} = \varepsilon.
  \end{align*}
  This ends the proof that theorem with Assumption~\ref{item:mixing:approx:assumption:2} implies the theorem without any additional assumption.
\end{proof}

\subsection{Bounding the support of mixed strategies}\label{section:achievable:support}

Theorems~\ref{thm:mixing:exact} and~\ref{thm:mixing:approx} state that it suffices to mix finitely many pure strategies to respectively match or approximate the payoff of any strategy.
We provide bounds on the number of pure strategies to mix in this section, in the same vein as Carathéodory's theorem for convex hulls (Theorem~\ref{theorem:caratheodory:convex}).
We show that the expected payoff of a finite-support mixed strategy can be obtained exactly by mixing no more than $\numObj+1$ strategies that are in its support and that a greater or equal payoff can be obtained by mixing no more than $\numObj$ of these strategies.

Let $\mdpState\in\mdpStateSpace$, $\stratMDP_1, \ldots, \stratMDP_\indexSequence$ be pure strategies and $\scalar_1$, \ldots, $\scalar_\indexSequence\in\ccInt{0}{1}$ be convex combination coefficients.
Let $\payoffVect =\payoffVectVerbose= \sum_{\indexSequenceB=1}^\indexSequence\scalar_\indexSequenceB\expectancy^{\stratMDP_\indexSequenceB}_\mdpState(\payoffTuple)$.
First, let us discuss the case when $\payoffVect\in\IR^\numObj$, i.e., when all considered pure strategies have a finite expected payoff on all dimensions.
In this case, the first bound is direct by Carathéodory's theorem for convex hulls.

For the second bound, we observe that $\payoffVect$ is an element of the compact convex polytope $D =\convex{\{\expectancy^{\stratMDP_\indexSequenceB}_\mdpState(\payoffTuple)\mid 1\leq\indexSequenceB\leq\indexSequence\}}$.
In particular, $\payoffVect$ is dominated by an element $\payoffVectB$ lying on a proper face of $D$ (i.e., the intersection of $D$ and a hyperplane).
For instance, we can consider $\payoffVect + \scalarB\cdot\oneVect$ for $\scalarB=\max\{\scalarC\geq 0\mid \payoffVect + \scalarC\oneVect\in D\}$.
Because proper faces have dimension no more than $\numObj-1$, Carathéodory's theorem implies that $\payoffVectB$ is a convex combination of no more than $\numObj$ vectors of the form $\expectancy^{\stratMDP_\indexSequenceB}_\mdpState(\payoffTuple)$, taken among those lying on the considered proper face.

Assume now that $\payoffVect\notin\IR^\numObj$.
In this case, Carathéodory's theorem does not apply directly.
The idea is to reduce ourselves to the previous case.
For each $1\leq\indexPayoff\leq\numObj$, if $\payoffComp_\indexPayoff$ is infinite, there is $1\leq\indexSequenceB\leq\indexSequence$ such that $\expectancy^{\stratMDP_\indexSequenceB}_\mdpState(\payoff_\indexSequenceB)=\payoffComp_\indexPayoff$.
For each infinite component of $\payoffVect$, we fix $\scalar_\indexSequenceB\expectancy^{\stratMDP_\indexSequenceB}_\mdpState(\payoffTuple)$ in the convex combination for one such $\indexSequenceB$.
We then obtain the sought bounds from the above; we reason on the real-valued components of $\payoffVect$ in the non-fixed part of the convex combination (after normalising its coefficients to sum to one).

We formalise our theorem and the previous proof sketch below.
\begin{restatable}{theorem}{thmMixingSupportAll}\label{thm:mixing:support:all}
  Assume that $\payoffTuple$ is universally unambiguously integrable.
  Let $\mdpState\in\mdpStateSpace$, $\stratMDP_1$, \ldots, $\stratMDP_\indexSequence$ be pure strategies and $\scalar_1$, \ldots, $\scalar_\indexSequence\in\ccInt{0}{1}$ be convex combination coefficients.
  There exist convex combination coefficients $\scalarB_1$, \ldots, $\scalarB_\indexSequence\in\ccInt{0}{1}$ with $|\{1\leq\indexSequenceB\leq\indexSequence\mid \scalarB_\indexSequenceB\neq 0\}|\leq\numObj+1$ and convex combination coefficients $\scalarC_1$, \ldots, $\scalarC_\indexSequence\in\ccInt{0}{1}$ with $|\{1\leq\indexSequenceB\leq\indexSequence\mid \scalarC_\indexSequenceB\neq 0\}|\leq\numObj$ such that
  \[\sum_{\indexSequenceB=1}^\indexSequence\scalar_\indexSequenceB\expectancy^{\stratMDP_\indexSequenceB}_\mdpState(\payoffTuple) =
    \sum_{\indexSequenceB=1}^\indexSequence\scalarB_\indexSequenceB\expectancy^{\stratMDP_\indexSequenceB}_\mdpState(\payoffTuple)\leq
    \sum_{\indexSequenceB=1}^\indexSequence\scalarC_\indexSequenceB\expectancy^{\stratMDP_\indexSequenceB}_\mdpState(\payoffTuple).\]
\end{restatable}
\begin{proof}
    For all $1\leq\indexSequenceB\leq\indexSequence$, let $\payoffVectB^{(\indexSequenceB)} = (\payoffCompB^{(\indexSequenceB)}_\indexPayoff)_{1\leq\indexPayoff\leq\numObj} = \expectancy^{\stratMDP_\indexSequenceB}_\mdpState(\payoffTuple)$ and let $\payoffVect=\payoffVectVerbose=\sum_{\indexSequenceB=1}^\indexSequence\scalar_\indexSequenceB\payoffVectB^{(\indexSequenceB)}$.
  We assume that for all $1\leq\indexSequenceB\leq\indexSequence$, $\scalar_\indexSequenceB\neq 0$.

  First, we assume that $\payoffVect\in\IR^\numObj$.
  The existence of coefficients $\scalarB_1$, \ldots, $\scalarB_\indexSequence$  obeying the required conditions is direct by Carathéodory's theorem for convex hulls (Theorem~\ref{theorem:caratheodory:convex}).
  We let $D = \convex{\{\payoffVectB^{(\indexSequenceB)}\mid1\leq\indexSequenceB\leq\indexSequence\}}$, which is a compact set (see Lemma~\ref{lemma:convex hull of compact}).  
  
  We define $\payoffVectB = \payoffVect + \scalarB\cdot\oneVect$ for $\scalarB=\sup\{\scalarC\geq 0\mid \payoffVect + \scalarC\cdot\oneVect\in D\}$.
  We remark that $\scalarB$ is a real number.
  On the one hand, $0\in\{\scalarC\geq 0\mid \payoffVect + \scalarC\cdot\oneVect\in D\}$, and thus $\scalarB\neq-\infty$.
  On the other hand, $D$ is bounded, therefore $\scalarB\neq+\infty$.
  Furthermore, $\payoffVectB\in D$.
  By convexity of $D$, $\payoffVect+\scalarC\cdot\oneVect\in D$ for all $0\leq\scalarC<\scalarB$.
  It follows that $\payoffVectB\in\closure{D}=D$.
  
  We have $\payoffVect\leq\payoffVectB$.
  To end the case $\payoffVect\in\IR^\numObj$, we show that there exists a hyperplane $\hplane$ such that $\payoffVectB$ is a convex combination of the vectors $\payoffVectB^{(\indexSequenceB)}$ that lie in $\hplane$.
  This suffices, because Carathéodory's theorem then ensures that $\payoffVectB$ is a convex combination of no more than $\numObj$ vectors among the $\payoffVectB^{(\indexSequenceB)}$.

  If $D$ is included in a hyperplane, there is nothing to show.
  We thus assume that $D$ is not included in a hyperplane and obtain a hyperplane using the supporting hyperplane theorem (Theorem~\ref{thm:hyperplane:supporting}).
  By construction, $\payoffVectB\notin\interior{D}=\relInt{D}$: for all $\scalarC>0$, $\payoffVectB+\scalarC\cdot\oneVect\notin D$.
  Therefore, there exists a linear form $\linForm\colon\IR^\numObj\to\IR$ such that for all $\vect\in D$, $\linForm(\vect)\leq\linForm(\payoffVectB)$.
  We claim that $\payoffVectB$ is a convex combination of the $\payoffVectB^{(\indexSequenceB)}$ that lie in the hyperplane $(\linForm)^{-1}(\linForm(\payoffVectB))$.
  Write $\payoffVectB$ as a convex combination $\sum_{\indexSequenceB=1}^\indexSequence\scalar'_\indexSequenceB\payoffVectB^{(\indexSequenceB)}$.
  We observe that for all $1\leq\indexSequenceB\leq\indexSequence$, $\scalar'_\indexSequenceB\neq 0$ implies $\linForm(\payoffVectB^{(\indexSequenceB)})=\linForm(\payoffVectB)$, as otherwise we would obtain that $\linForm(\payoffVectB) < \linForm(\payoffVectB)$.
  This ends the proof in the case $\payoffVect\in\IR^\numObj$.

  We now assume that some components of $\payoffVect$ are infinite.
  For convenience of notation, we assume that there is $1\leq\numObj'\leq\numObj$ such that, for all $1\leq\indexPayoff\leq\numObj'$, $\payoffComp_\indexPayoff\in\{-\infty,+\infty\}$ and, for all $\numObj'+1\leq\indexPayoff\leq\numObj$, $\payoffComp_\indexPayoff\in\IR$.
  For all $1\leq\indexPayoff\leq\numObj'$, let $1\leq\indexSequenceB_\indexPayoff\leq\indexSequence$ such that $\payoffCompB^{(\indexSequenceB_\indexPayoff)}_\indexPayoff = \payoffComp_\indexPayoff$.
  If $\numObj'=\numObj$, we have $\payoffVect = \frac{1}{\numObj}\sum_{\indexPayoff=1}^\numObj\payoffVectB^{(\indexSequenceB_\indexPayoff)}$.
  We now assume that $\numObj'<\numObj$.
  
  Let $I_\infty = \{\indexSequenceB_\indexPayoff\mid 1\leq\indexPayoff\leq\numObj'\}$, $I_\IR = \{1, \ldots, \indexSequence\}\setminus I_\infty$.
  We have $|I_\infty|\leq\numObj'\leq\numObj$.
  In particular, if $I_\IR$ is empty, the sought result is direct.
  We assume that $I_\IR\neq\emptyset$.
  Let $\scalar_\IR = \sum_{\indexSequenceB\in I_\IR}\scalar_\indexSequenceB>0$, and let $\proj{>\numObj'}\colon(\IRbar)^\numObj\to(\IRbar)^{\numObj-\numObj'}$ denote the projection of a vector onto its $\numObj - \numObj'$ last components.
  We apply the result for vectors in $\IR^{\numObj-\numObj'}$ to $\proj{>\numObj'}(\sum_{\indexSequenceB\in{I_\IR}}\frac{\scalar_\indexSequenceB}{\scalar_\IR}\payoffVectB^{(\indexSequenceB)})$ with respect to the payoff function $(\payoff_{\numObj'+1}, \ldots, \payoff_\numObj)$.
  It follows that there exist convex combination coefficients $(\scalarB'_\indexSequenceB)_{\indexSequenceB\in I_\IR}$ of which at most $\numObj-\numObj'+1$ are positive and convex combination coefficients $(\scalarC'_\indexSequenceB)_{\indexSequenceB\in I_\IR}$ of which at most $\numObj-\numObj'$ are positive such that
  \[\proj{>\numObj'}\left(
      \sum_{\indexSequenceB\in{I_\IR}}\frac{\scalar_\indexSequenceB}{\scalar_\IR}
      \payoffVectB^{(\indexSequenceB)}
    \right) =
    \sum_{\indexSequenceB\in{I_\IR}}\scalarB'_\indexSequenceB
    \proj{>\numObj'}(\payoffVectB^{(\indexSequenceB)})\leq
    \sum_{\indexSequenceB\in{I_\IR}}\scalarC'_\indexSequenceB
    \proj{>\numObj'}(\payoffVectB^{(\indexSequenceB)}).
\]
  We conclude by observing that
  \begin{align*}
    \payoffVect
    & =
      \sum_{\indexSequenceB\in I_\infty}
      \scalar_\indexSequenceB\payoffVectB^{(\indexSequenceB)} +
      \scalar_\IR\cdot\sum_{\indexSequenceB\in I_\IR}
      \frac{\scalar_\indexSequenceB}{\scalar_\IR}\payoffVectB^{(\indexSequenceB)} \\
    & =
      \sum_{\indexSequenceB\in I_\infty}
      \scalar_\indexSequenceB\payoffVectB^{(\indexSequenceB)} +
      \sum_{\indexSequenceB\in I_\IR}
      \scalar_\IR\scalarB'_\indexSequenceB\payoffVectB^{(\indexSequenceB)} \\
    & \leq
      \sum_{\indexSequenceB\in I_\infty}
      \scalar_\indexSequenceB\payoffVectB^{(\indexSequenceB)} +
      \sum_{\indexSequenceB\in I_\IR}
      \scalar_\IR\scalarC'_\indexSequenceB\payoffVectB^{(\indexSequenceB)},
  \end{align*}
  i.e., for all $\indexSequenceB\in I_\infty$, we let $\scalarB_\indexSequenceB =  \scalarC_\indexSequenceB = \scalar_\indexSequenceB$, and for all $\indexSequenceB\in I_\IR$, we let $\scalarB_\indexSequenceB = \scalar_\IR\cdot\scalarB'_\indexSequenceB$ and $\scalarC_\indexSequenceB = \scalar_\IR\cdot\scalarC'_\indexSequenceB$.
\end{proof}

We now highlight two corollaries of Theorem~\ref{thm:mixing:support:all}.
First, we obtain bounds when dealing with universally integrable payoffs that follow from Theorem~\ref{thm:mixing:exact}.

\begin{restatable}{corollary}{corMixingSupportExact}\label{cor:mixing:support:exact}
  Assume that $\payoffTuple$ is universally integrable.
  Let $\mdpState\in\mdpStateSpace$.
  \begin{itemize}
  \item For all $\payoffVect\in\paySet{\payoffTuple}{\mdpState}$, $\payoffVect$ is a convex combination of at most $\numObj+1$ elements of $\paySetPure{\payoffTuple}{\mdpState}$.
  \item For all $\payoffVect\in\achSet{\payoffTuple}{\mdpState}$, there exists a convex combination $\payoffVectB$ of at most $\numObj$ elements of $\paySetPure{\payoffTuple}{\mdpState}$ such that $\payoffVect\leq\payoffVectB$.
  \end{itemize}
\end{restatable}
\begin{proof}
  By Theorem~\ref{thm:mixing:exact}, $\paySet{\payoffTuple}{\mdpState} = \convex{\paySetPure{\payoffTuple}{\mdpState}}$.
  Furthermore, we recall that $\payoffVect\in\achSet{\payoffTuple}{\mdpState}$ if and only if there exists a strategy $\stratMDP$ such that $\payoffVect\leq\expectancy^{\stratMDP}_\mdpState(\payoffTuple)$.
  We obtain both claims of the corollary directly by Thm.~\ref{thm:mixing:support:all}.
\end{proof}

Example~\ref{ex:mixing:approx} implies that Corollary~\ref{cor:mixing:support:exact} does not directly extend to universally unambiguously integrable payoffs.
Nonetheless, we can identify a class of vectors that can be achieved by mixing no more than $\numObj$ strategies: the vectors in $\interior{\achSet{\payoffTuple}{\mdpState}\cap\IR^\numObj}$.
These are vectors $\payoffVect\in\achSet{\payoffTuple}{\mdpState}\cap\IR^\numObj$ that can be improved in all dimensions simultaneously, i.e., such that $\payoffVect + \varepsilon\oneVect\in\achSet{\payoffTuple}{\mdpState}$ for some $\varepsilon > 0$.
To prove our result, intuitively, we approximate, via Theorem~\ref{thm:mixing:approx}, the payoff of a strategy achieving $\payoffVect + \varepsilon\oneVect$ with a finite-support mixed strategy then invoke Theorem~\ref{thm:mixing:support:all}.

\begin{restatable}{corollary}{corMixingSupportInterior}\label{cor:mixing:support:interior}
  Assume that $\payoffTuple$ is universally unambiguously integrable.
  Let $\mdpState\in\mdpStateSpace$.
  For $\payoffVect\in\interior{\achSet{\payoffTuple}{\mdpState}\cap\IR^\numObj}$, there exists a convex combination $\payoffVectB$ of at most $\numObj$ elements of $\paySetPure{\payoffTuple}{\mdpState}$ such that $\payoffVect\leq\payoffVectB$.  
\end{restatable}
\begin{proof}
      Let $\payoffVect=\payoffVectVerbose\in\interior{\achSet{\payoffTuple}{\mdpState}\cap\IR^\numObj}$.
  By definition of the interior of a subset of $\IR^\numObj$, there exists $\varepsilon > 0$ such that $\payoffVect + \varepsilon\oneVect\in\achSet{\payoffTuple}{\mdpState}$.
  Therefore, there exists a strategy $\stratMDP$ such that, for all $1\leq\indexPayoff\leq\numObj$, $\payoffComp_\indexPayoff < \expectancy^{\stratMDP}_\mdpState(\payoff_\indexPayoff)$.
  For all $1\leq\indexPayoff\leq\numObj$, we have $\expectancy^{\stratMDP}_\mdpState(\payoff_\indexPayoff)\neq-\infty$ because $\expectancy^{\stratMDP}_\mdpState(\payoff_\indexPayoff) > \payoffComp_\indexPayoff\in\IR$.
  Let $\eta = \min(\{1\}\cup\{\expectancy^{\stratMDP}_\mdpState(\payoff_\indexPayoff) -\payoffComp_\indexPayoff\mid \expectancy^{\stratMDP}_\mdpState(\payoff_\indexPayoff)\in\IR,\, 1\leq\indexPayoff\leq\numObj\})$ and $M = \max(\{1\}\cup\{\payoffComp_\indexPayoff+1 \mid \expectancy^{\stratMDP}_\mdpState(\payoff_\indexPayoff)=+\infty,\, 1\leq\indexPayoff\leq\numObj\})$.

  Theorem~\ref{thm:mixing:approx} implies that there exist pure strategies $\stratBMDP_1$, \ldots, $\stratBMDP_\indexSequence$ and convex combination coefficients $\scalar_1$, \ldots, $\scalar_\indexSequence\in\ccInt{0}{1}$ such that if $\expectancy^{\stratMDP}_\mdpState(\payoff_\indexPayoff)= +\infty$, then $\sum_{\indexSequenceB=1}^\indexSequence\scalar_\indexSequenceB\expectancy^{\stratBMDP_\indexSequenceB}_\mdpState(\payoff_\indexPayoff)\geq M$ and, otherwise, $\sum_{\indexSequenceB=1}^\indexSequence\scalar_\indexSequenceB\expectancy^{\stratBMDP_\indexSequenceB}_\mdpState(\payoff_\indexPayoff) - \expectancy^{\stratMDP}_\mdpState(\payoff_\indexPayoff)\geq-\eta$.
  By Theorem~\ref{thm:mixing:support:all}, we can assume that $\indexSequence\leq\numObj$.
  Let $\payoffVectB = \payoffVectBVerbose = \sum_{\indexSequenceB=1}^\indexSequence\scalar_\indexSequenceB\expectancy^{\stratBMDP_\indexSequenceB}_\mdpState(\payoffTuple)$.

  We show that $\payoffVect\leq \payoffVectB$.
  Let $1\leq\indexPayoff\leq\numObj$.
  First, assume that $\expectancy^{\stratMDP}_\mdpState(\payoff_\indexPayoff) = +\infty$.
  In this case, $M\geq\payoffComp_\indexPayoff$, and we obtain
  $\payoffComp_\indexPayoff\leq M \leq
    \payoffCompB_\indexPayoff.$
  Second, assume that $\expectancy^{\stratMDP}_\mdpState(\payoff_\indexPayoff)\in\IR$.
  In that case, we have $\eta \leq \expectancy^{\stratMDP}_\mdpState(\payoff_\indexPayoff) - \payoffComp_\indexPayoff$, and therefore
  $\payoffComp_\indexPayoff \leq \expectancy^{\stratMDP}_\mdpState(\payoff_\indexPayoff) - \eta \leq
  \payoffCompB_\indexPayoff$.
\end{proof}

\section{Continuous payoffs}\label{section:continuous}
In this section, we study the topological properties of sets of expected payoff in \textit{finite POMDPs} when all payoffs are continuous.
The main result of this section applies to payoffs whose square is universally integrable, or \textit{universally square integrable} for short. We prove that the set of expected payoffs and of achievable vectors are closed for such payoff functions.
In general, the set of expected payoffs need not be closed when considering (universally unambiguously integrable) continuous multi-dimensional payoffs.
Our result for universally square integrable continuous payoffs is a generalisation of~\cite[Lem.~2]{DBLP:conf/lpar/ChatterjeeFW13} for the special case where all payoffs are discounted-sum payoffs.

The Cauchy-Schwarz inequality (e.g.,~\cite[Thm.~1.5.2.]{Dur19}) guarantees that any universally square integrable payoff is universally integrable.
In finite POMDPs, this class subsumes real-valued continuous payoffs because such payoffs are bounded (since the set of plays is compact) and universally integrable continuous shortest-path payoffs (see Appendix~\ref{appendix:square:shortest path}).

We divide this section in three parts.
In Section~\ref{section:continuous:topology}, we introduce a topology on the space of randomised strategies such that the resulting topological space is metrisable and compact to formalise a notion of convergence of strategies.
In Section~\ref{section:continuous:square integrable}, we formulate our result for continuous universally square integrable payoffs.
Section~\ref{section:continuous:infinite} provides examples illustrating that the results of Section~\ref{section:continuous:square integrable} do not hold for continuous payoffs that are not universally integrable.
We leave open whether our results extend to all universally integrable continuous payoffs.
For this section, we fix a finite POMDP $\pomdp=\pomdpTuple$ and a $\numObj$-dimensional continuous payoff $\payoffTuple = (\payoff_\indexPayoff)_{1\leq\indexPayoff\leq\numObj}$.

\subsection{A topology on the space of strategies}\label{section:continuous:topology}

In this section, we define a topology on the space of strategies to formalise a notion of convergence for sequences of behavioural strategies.
We view $\stratClassAll{\pomdp}$ as a subset of the product $\dist{\mdpActionSpace}^{\obsFun(\histSet{\pomdp})}$ equipped with the product topology.
We endow each copy of $\dist{\mdpActionSpace}$ in this product with the topology given by the metric (induced by the Euclidean norm) $\distMetric(\mu, \mu') = \sqrt{\sum_{\mdpAction\in\mdpActionSpace(\mdpState)}|\mu(\mdpAction)-\mu'(\mdpAction)|^2}$ for all $\mu$, $\mu'\in\dist{\mdpActionSpace}$.

We can show that $\stratClassAll{\pomdp}$ is a compact metrisable topological space (it is a closed subset of $\dist{\mdpActionSpace}^{\obsFun(\histSet{\pomdp})}$).
Compactness follows from the assumption that $\pomdp$ is finite: $\dist{\mdpActionSpace}$ is homeomorphic to a closed subset of $\ccInt{0}{1}^{|\mdpActionSpace|}$.
We do not define a metric over strategies, as it is not necessary in the sequel.
Instead, we recall that a sequence of strategies $(\stratMDP^{(\indexSequence)})_{\indexSequence\in\IN}$ converges to a strategy $\stratMDP$ if and only if, for all $\hist\in\histSet{\pomdp}$, $(\stratMDP^{(\indexSequence)}(\hist))_{\indexSequence\in\IN}$ converges to $\stratMDP(\hist)$.

We now formulate a result intuitively stating that close strategies induce distributions that assign similar probabilities to cylinders of histories of bounded length.
We first require the following technical lemma.
\begin{restatable}{lemma}{lemProdDiffIneq}\label{lem:prod:diff:inequality}
  Let $\scalar_1, \ldots, \scalar_\indexSequence, \scalarB_1, \ldots, \scalarB_\indexSequence\in\ccInt{0}{1}$.
  Then
  \[\left|\prod_{\indexSequenceB = 1}^\indexSequence \scalar_\indexSequenceB
      - \prod_{\indexSequenceB = 1}^\indexSequence \scalarB_\indexSequenceB\right|\leq
    \sum_{\indexSequenceB=1}^\indexSequence \left|\scalar_\indexSequenceB - \scalarB_\indexSequenceB\right|.\] 
\end{restatable}
\begin{proof}
  We proceed by induction.
  The claim trivially holds for $\indexSequence=1$.
  To lighten notation, we prove the case $\indexSequence=2$ separately first.
  We perform the general induction step below by building on this simpler case.
  We have:
  \begin{align*}
    |\scalar_1\scalar_2 - \scalarB_1\scalarB_2|
    & = |\frac{1}{2}(\scalar_1-\scalarB_1)(\scalar_2+\scalarB_2) + \frac{1}{2}(\scalar_1+\scalarB_1)(\scalar_2-\scalarB_2)| \\
    & \leq |\scalar_1 - \scalarB_1|\frac{|\scalar_2+\scalarB_2|}{2} + |\scalar_2 - \scalarB_2|\frac{|\scalar_1+\scalarB_1|}{2} \\
    & \leq |\scalar_1 - \scalarB_1| + |\scalar_2 -\scalarB_2|.
  \end{align*}
  The second line is obtained by triangulation and the last line is obtained by using the fact that $\scalar_1, \scalar_2, \scalarB_1, \scalarB_2\in \ccInt{0}{1}$.

  We now perform the general induction step.
  We assume that the result holds for some $\indexSequence\in\IN_0$ and prove that it holds for $\indexSequence+1$.
  By applying the simpler case proven above and then the induction hypothesis, we obtain that
  \begin{align*}
    \left|\prod_{\indexSequenceB = 1}^{\indexSequence+1} \scalar_\indexSequenceB
    - \prod_{\indexSequenceB = 1}^{\indexSequence+1} \scalarB_\indexSequenceB\right|
    & \leq 
      \left|\prod_{\indexSequenceB = 1}^{\indexSequence} \scalar_\indexSequenceB
      - \prod_{\indexSequenceB = 1}^{\indexSequence} \scalarB_\indexSequenceB\right| + \left|\scalar_{\indexSequence+1}-\scalarB_{\indexSequence+1}\right| \\
    & \leq
    \sum_{\indexSequenceB=1}^{\indexSequence+1} \left|\scalar_\indexSequenceB - \scalarB_\indexSequenceB\right|.
  \end{align*}
\end{proof}

We now state the lemma regarding induced distributions.
\begin{restatable}{lemma}{lemStratDistCloseness}\label{lem:strategy distribution closeness}
  Let $\stratMDP, \stratBMDP$ be two strategies, $\lengthBound\in\IN_0$ and $\eta > 0$.
  Assume that, for all histories $\hist$ that are at most $\lengthBound$ states long,  $\distMetric(\stratMDP(\hist), \stratBMDP(\hist))\leq\frac{\eta}{\lengthBound}$.
  Then, for all histories $\hist$ that are at most $\lengthBound+1$ states long and all $\mdpState\in\mdpStateSpace$, we have $|\proba_{\mdpState}^{\stratMDP}(\cyl{\hist}) - \proba_{\mdpState}^{\stratBMDP}(\cyl{\hist})|\leq \eta$.
\end{restatable}
\begin{proof}
    Let $\hist = \mdpState_0\mdpAction_0\mdpState_1\ldots\mdpState_\indexLast$ with $\indexLast \leq \lengthBound$.
  We only prove the claim for $\mdpState = \mdpState_0$.
  The other case is direct because, for $\mdpState\in\mdpStateSpace\setminus\{\mdpState_0\}$, we have $\proba_{\mdpState}^{\stratMDP}(\cyl{\hist}) = \proba_{\mdpState}^{\stratBMDP}(\cyl{\hist}) = 0$.
  
  The proof follows from the following sequence of inequations.
  \begin{align*}
    |\proba_{\mdpState_0}^{\stratMDP}(\cyl{\hist}) - \proba_{\mdpState_0}^{\stratBMDP}(\cyl{\hist})|
    & = \prod_{\indexPosition=0}^{\indexLast-1} \mdpTrans(\mdpState_\indexPosition, \mdpAction_\indexPosition)(\mdpState_{\indexPosition+1}) \cdot
      \left|\prod_{\indexPosition=0}^{\indexLast-1} \stratMDP(\playPrefix{\hist}{\indexPosition}) - \prod_{\indexPosition=0}^{\indexLast-1} \stratBMDP(\playPrefix{\hist}{\indexPosition})\right| \\
    & \leq \left|\prod_{\indexPosition=0}^{\indexLast-1} \stratMDP(\playPrefix{\hist}{\indexPosition}) - \prod_{\indexPosition=0}^{\indexLast-1} \stratBMDP(\playPrefix{\hist}{\indexPosition})\right| \\
    & \leq \sum_{\indexPosition=0}^{\indexLast-1} |\stratMDP(\playPrefix{\hist}{\indexPosition}) - \stratBMDP(\playPrefix{\hist}{\indexPosition})| \\
    & \leq \lengthBound\cdot \frac{\eta}{\lengthBound} = \eta.
  \end{align*}
  The first line is by definition of probability measures induced by the strategies from $\mdpState_0$.
  The second line uses the fact that transition probabilities are at most $1$.
  The third line is obtained by Lemma~\ref{lem:prod:diff:inequality}.
  The last line follows from the assumption of the lemma and $\indexLast\leq\lengthBound$.
\end{proof}

We close this section by showing that $\stratClassPure{\pomdp}$ is a closed subset of $\stratClassAll{\pomdp}$.
In the following section, we prove that if $\payoffTuple$ is universally square integrable and $\stratClass\subseteq\stratClassAll{\pomdp}$ is closed, then $\paySetClass{\payoffTuple}{\mdpState}{\stratClass}$ is compact.
Therefore, establishing that $\stratClassPure{\pomdp}$ is closed implies that the result is also applicable to $\stratClassPure{\pomdp}$.
The proof boils down to showing that the limit of a converging sequence of Dirac distributions is a Dirac distribution.
This follows from such sequences being ultimately constant.
\begin{restatable}{lemma}{lemClosedPureStrategies}\label{lem:closed:pure strategies}
  The set $\stratClassPure{\pomdp}$ is a closed subset of $\stratClassAll{\pomdp}$.
\end{restatable}
\begin{proof}
    Let $(\stratMDP^{(\indexSequence)})_{\indexSequence\in\IN}$ be a sequence of pure strategies that converges to a strategy $\stratMDP$, i.e., for all $\hist\in\histSet{\pomdp}$, the sequence $(\stratMDP^{(\indexSequence)}(\hist))_{\indexSequence\in\IN}$ converges to $\stratMDP(\hist)$.
  We must show that, for all $\hist\in\histSet{\pomdp}$, $\stratMDP(\hist)$ is a Dirac distribution.

  Let $\hist\in\histSet{\pomdp}$.
  Because $(\stratMDP^{(\indexSequence)}(\hist))_{\indexSequence\in\IN}$ is convergent, it is a Cauchy sequence.
  Thus, there exists some $\indexSequence_0\in\IN$ such that for all $\indexSequence, \indexSequenceB\geq \indexSequence_0$, $\distMetric(\stratMDP^{(\indexSequence)}(\hist), \stratMDP^{(\indexSequenceB)}(\hist)) < 1$.
  It follows from $(\stratMDP^{(\indexSequence)}(\hist))_{\indexSequence\in\IN}$ being a sequence of a Dirac distributions and the definition of $\distMetric$ that the sequence $(\stratMDP^{(\indexSequence)}(\hist))_{\indexSequence\geq\indexSequence_0}$ is constant.
  It follows that $\stratMDP(\hist) = \stratMDP^{(\indexSequence_0)}(\hist)$, and thus $\stratMDP(\hist)$ is a Dirac distribution.
\end{proof}

\subsection{Universally square integrable continuous payoffs}\label{section:continuous:square integrable}

We prove that $\payoffTuple$ is universally square integrable, then for all $\mdpState\in\mdpStateSpace$, $\paySet{\payoffTuple}{\mdpState}$ and $\achSet{\payoffTuple}{\mdpState}$ are closed.
Our proof relies on the following property: if $\payoffTuple$ is universally square integrable, then for all sequences $(\stratMDP^{(\indexSequence)})_{\indexSequence\in\IN}$ of strategies converging to a strategy $\stratMDP$, and for all $\mdpState\in\mdpStateSpace$, $(\expectancy^{\stratMDP^{(\indexSequence)}}_\mdpState(\payoffTuple))_{\indexSequence\in\IN}$ converges to $\expectancy^{\stratMDP}_\mdpState(\payoffTuple)$.
It suffices to prove the convergence on each dimension to obtain this property.
Therefore, we need only consider one-dimensional payoffs for now.

We split the proof into two parts: first, we consider the particular case of real-valued continuous payoffs and then generalise to universally square integrable payoffs.
We remark that we may not assume that a universally integrable continuous payoff is real-valued without loss of generality: changing the payoff of plays that have an infinite payoff to a real number will violate the continuity property.
Therefore, we do not generalise the property for real-valued continuous payoffs through such an assumption.

Let $\payoff\colon\playSet{\pomdp}\to\IR$ be a continuous real-valued payoff.
It follows from $\pomdp$ being finite that $\payoff$ is uniformly continuous.
Recall that $\payoff$ is uniformly continuous if for all $\varepsilon>0$, there exists $\indexPosition\in\IN$ such that for all plays $\play$, $\play'$, $\playPrefix{\play}{\indexPosition}=\playPrefix{\play'}{\indexPosition}$ implies that $|\payoff(\play)-\payoff(\play')|<\varepsilon$.
In particular, for all $\varepsilon > 0$, $\payoff$ can be $\varepsilon$-approximated by a linear combination of indicator functions of cylinders of histories of a fixed length.
This provides a means to $\varepsilon$-approximate $\expectancy^{\stratBMDP}_\mdpState(\payoff)$ for any strategy $\stratBMDP$ as a linear combination of probabilities of cylinders of histories of bounded length.
Since Lemma~\ref{lem:strategy distribution closeness} guarantees that the distributions over plays induced by strategies that are close assign similar probabilities to such cylinders, through the approximations above, we can obtain that $\expectancy^{\stratMDP^{(\indexSequence)}}_\mdpState(\payoffTuple)$ is $3\varepsilon$-close $\expectancy^{\stratMDP}_\mdpState(\payoffTuple)$ for all $\mdpState\in\mdpStateSpace$ for large values of $\indexSequence$.
We formalise this argument below.

\begin{restatable}{theorem}{thmContinuousConvergenceReal}\label{thm:continuous:convergence:real}
  Let $\mdpState\in\mdpStateSpace$.
  Assume that $\pomdp$ is finite, $\payoffTuple\colon\playSet{\pomdp}\to\IR^\numObj$ and $\payoffTuple$ is continuous.
  Then the function $\stratClassAll{\pomdp}\to\IR^\numObj\colon\stratMDP\mapsto\expectancy^{\stratMDP}_\mdpState(\payoffTuple)$ is continuous.
  In other words, for all sequences $(\stratMDP^{(\indexSequence)})_{\indexSequence\in\IN}$ of strategies that converge to a strategy $\stratMDP$, $\lim_{\indexSequence\to\infty}\expectancy^{\stratMDP^{(\indexSequence)}}_\mdpState(\payoffTuple)=\expectancy^{\stratMDP}_\mdpState(\payoffTuple)$.
\end{restatable}
\begin{proof}
    It suffices to prove the theorem in the case $\numObj = 1$ to obtain the general case.
  For this reason, we consider a one-dimensional real-valued continuous payoff $\payoff$ below.
  Let $(\stratMDP^{(\indexSequence)})_{\indexSequence\in\IN}$ be a sequence of strategies converging to a strategy $\stratMDP$.
  We start with some notation.
  By continuity of $\payoff$, $\payoff$ is bounded.
  We let $\|\payoff\|_\infty = \sup_{\play\in\playSet{\pomdp}}|\payoff(\play)|$.
  For any history $\hist = \mdpState_0\mdpAction_0\ldots\mdpState_\indexLast$, we let $|\hist| = \indexLast$  denote the index of the last state of the history.
  We also fix, for all histories $\hist\in\histSet{\pomdp}$, a play $\play(\hist)\in\cyl{\hist}$ which is a continuation of $\hist$.

  We must prove that $(\expectancy^{\stratMDP^{(\indexSequence)}}_\mdpState(\payoff))_{\indexSequence\in\IN}$ converges to $\expectancy^{\stratMDP}_\mdpState(\payoff)$.
  Let $\varepsilon > 0$.
  We assume that $\|\payoff\|_\infty > 0$, as otherwise the result is direct: if $\|\payoff\|_\infty = 0$, then $\expectancy^{\stratMDP^{(\indexSequence)}}_\mdpState(\payoff) = \expectancy^{\stratMDP}_{\mdpState}(\payoff) = 0$ for all $\indexSequence\in\IN$.

  We start by constructing a simple function which $\frac{\varepsilon}{3}$-approximates $\payoff$ by exploiting the uniform continuity of $\payoff$.
  By uniform continuity of $\payoff$, there exists some $\lengthBound\in\IN_0$ such that, for any two plays $\play, \play'\in\playSet{\pomdp}$, if $\playPrefix{\play}{\lengthBound} = \playPrefix{\play'}{\lengthBound}$, then $|\payoff(\play) - \payoff(\play')| \leq \frac{\varepsilon}{3}$.
  It follows that
  \[\left|\payoff - \sum_{|\hist|=\lengthBound}\payoff(\play(\hist))\cdot\indic{\cyl{\hist}}\right|\leq\frac{\varepsilon}{3}.\]
  Since $\payoff$ is universally integrable, it follows that, for all $\stratBMDP\in\stratClassAll{\pomdp}$, 
  \begin{equation}\label{equation:thm:continuous closed:approx}
    \left|\expectancy^{\stratBMDP}_{\mdpState}(\payoff) -
      \sum_{|\hist| = \lengthBound} \payoff(\play(\hist))\cdot
      \proba^{\stratBMDP}_{\mdpState}(\cyl{\hist})\right| \leq
    \frac{\varepsilon}{3}.
  \end{equation}
  
  To end the argument, we now determine $\indexSequence_0$ such that, for all $\indexSequence\geq\indexSequence_0$, we have
  \begin{equation}\label{equation:thm:continuous closed:proba}
    \left|\sum_{|\hist| = \lengthBound}
      \payoff(\play(\hist))\cdot
      \left(
        \proba^{\stratMDP}_{\mdpState}(\cyl{\hist}) -
        \proba^{\stratMDP^{(\indexSequence)}}_{\mdpState}(\cyl{\hist})
      \right)
    \right| \leq\frac{\varepsilon}{3}.
  \end{equation}
  Let $M$ denote the number of histories $\hist$ such that $|\hist| = \lengthBound$.
  Since $\lim_{\indexSequence\to\infty}\stratMDP^{(\indexSequence)} = \stratMDP$, there exists $\indexSequence_0\in\IN$ such that for all histories $\hist$ such that $|\hist|\leq\lengthBound$ (i.e., with at most $\lengthBound+1$ states), we have $\distMetric(\stratMDP^{(\indexSequence)}(\hist), \stratMDP(\hist))\leq\frac{\varepsilon}{3\cdot M\cdot\|\payoff\|_\infty\cdot\lengthBound}$ (here, we use the fact that there are finitely many such histories).
  Equation~\eqref{equation:thm:continuous closed:proba} follows from the triangular inequality and Lemma~\ref{lem:strategy distribution closeness}: we have, for all $\indexSequence\geq\indexSequence_0$,
  \begin{align*}
    \left|\sum_{|\hist| = \lengthBound}
      \payoff(\play(\hist))\cdot
      \left(
        \proba^{\stratMDP}_{\mdpState}(\cyl{\hist}) -
        \proba^{\stratMDP^{(\indexSequence)}}_{\mdpState}(\cyl{\hist})
      \right)
    \right|
    &
      \leq \|\payoff\|_\infty\cdot
      \sum_{|\hist| = \lengthBound}
      \left|
      \proba^{\stratMDP}_{\mdpState}(\cyl{\hist}) -
      \proba^{\stratMDP^{(\indexSequence)}}_{\mdpState}(\cyl{\hist})
      \right|\\
    & \leq
      \|\payoff\|_\infty\cdot
      \sum_{|\hist| = \lengthBound}\frac{\varepsilon}{3\cdot M\cdot\|\payoff\|_\infty}\\
    & \leq\frac{\varepsilon}{3}.
  \end{align*}

  Let $\indexSequence\geq\indexSequence_0$.
  We now show that $|\expectancy^{\stratMDP}_\mdpState(\payoff) - \expectancy^{\stratMDP_n}_\mdpState(\payoff)|\leq\varepsilon$.
  From the triangular inequality, Equation~\eqref{equation:thm:continuous closed:approx} and Equation~\eqref{equation:thm:continuous closed:proba}, we obtain
  \begin{align*}
    \left|\expectancy^{\stratMDP}_\mdpState(\payoff) - \expectancy^{\stratMDP^{(\indexSequence)}}_\mdpState(\payoff)\right|
    & \leq \left|
      \expectancy^{\stratMDP}_\mdpState(\payoff) -
      \sum_{|\hist| = \lengthBound}  \payoff(\play(\hist))\cdot
      \proba^{\stratMDP}_{\mdpState}(\cyl{\hist})\right| \\
    & + \left|
      \sum_{|\hist| = \lengthBound}\payoff(\play(\hist))\cdot
      \left(\proba^{\stratMDP}_{\mdpState}(\cyl{\hist}) -\proba^{\stratMDP^{(\indexSequence)}}_{\mdpState}(\cyl{\hist})\right) \right| \\
    & + \left|
      \sum_{|\hist| = \lengthBound}
      \payoff(\play(\hist))\cdot
      \proba^{\stratMDP^{(\indexSequence)}}_{\mdpState}(\cyl{\hist}) -
      \expectancy^{\stratMDP^{(\indexSequence)}}_\mdpState(\payoff)\right| \\
    & \leq \frac{\varepsilon}{3} + \frac{\varepsilon}{3} + \frac{\varepsilon}{3} = \varepsilon.
  \end{align*}
  We have thus shown that $(\expectancy^{\stratMDP^{(\indexSequence)}}_{\mdpState}(\payoff))_{n\in\IN}$ converges to $\expectancy^{\stratMDP}_{\mdpState}(\payoff)$.
\end{proof}

We now generalise Theorem~\ref{thm:continuous:convergence:real} to universally square integrable payoffs.
We note that the previous proof relies on the boundedness of continuous payoffs and their uniform continuity, and thus is not valid in this more general case.

Similarly to above, it suffices to consider one-dimensional non-negative payoffs; the general case can be recovered by writing payoffs on each dimension as the difference of their non-negative and non-positive parts.
The non-negative and non-positive parts of a continuous payoff are also continuous (see Lemma~\ref{lemma:continuity:max-min} of Appendix~\ref{appendix:prelim:topology}).
We let $\payoff\colon\playSet{\pomdp}\to\IRbar$ be a non-negative continuous square integrable payoff, $\mdpState\in\mdpStateSpace$, let $(\stratMDP^{(\indexSequence)})_{\indexSequence\in\IN}$ be a sequence of strategies that converges to a strategy $\stratMDP$ and $\varepsilon>0$.
Given $M\in\IR$, we abbreviate $\{\payoff(\play)\geq M\mid\play\in\playSet{\pomdp}\}$ by $\{\payoff\geq M\}$, for all strategies $\stratBMDP$, we write $\proba^{\stratBMDP}_\mdpState(\payoff\geq M)$ instead of $\proba^{\stratBMDP}_\mdpState(\{\payoff\geq M\})$ and we let $\min(\payoff, M)$ denote the payoff $\play\mapsto\min\{\payoff(\play), M\}$.

The goal is to show that $|\expectancy^{\stratMDP^{(\indexSequence)}}_\mdpState(\payoff) - \expectancy^{\stratMDP}_\mdpState(\payoff))|\leq\varepsilon$ for all large enough values of $\indexSequence$.
We show that there exists a constant $M\geq 0$ (dependent on $\varepsilon$) such that the real-valued continuous payoff $\min(\payoff, M)$ (whose continuity follows from Lemma~\ref{lemma:continuity:max-min}) satisfies $0\leq\expectancy^{\stratBMDP}_\mdpState(\payoff) - \expectancy^{\stratBMDP}_\mdpState(\min(\payoff, M))\leq\frac{\varepsilon}{3}$ for all strategies $\stratBMDP\in\stratClassAll{\pomdp}$.
By the triangular inequality, $|\expectancy^{\stratMDP^{(\indexSequence)}}_\mdpState(\payoff) - \expectancy^{\stratMDP}_\mdpState(\payoff))|$ is no more than
\begin{align*}
  |\expectancy^{\stratMDP^{(\indexSequence)}}_\mdpState(\payoff) -
  \expectancy^{\stratMDP^{(\indexSequence)}}_\mdpState(\min(\payoff, M))|
  & +
    |\expectancy^{\stratMDP^{(\indexSequence)}}_\mdpState(\min(\payoff, M)) -
    \expectancy^{\stratMDP}_\mdpState(\min(\payoff, M))| \\
  & +
    |\expectancy^{\stratMDP}_\mdpState(\min(\payoff, M)) -
    \expectancy^{\stratMDP}_\mdpState(\payoff)|.
\end{align*}
The first and last terms are no more than $\frac{\varepsilon}{3}$ by choice of $M$, and the second term is smaller than $\frac{\varepsilon}{3}$ for large values of $\indexSequence$ by Theorem~\ref{thm:continuous:convergence:real}.

Thus, the main hurdle of the proof is establishing the existence of a suitable $M$.
The first step consists in showing that the probability of $\payoff$ being large can be made arbitrarily small, in the sense that for all $\eta>0$, there exists $M(\eta)$ such that $\proba^{\stratBMDP}_\mdpState(\payoff\geq M(\eta))\leq\eta$ for all strategies $\stratBMDP$.
The negation of this property would allow the existence of strategies with an arbitrarily large expected payoff from $\mdpState$, contradicting Lemma~\ref{lem:ui:characterisation}.
It then follows from the Cauchy-Schwarz inequality that, for all strategies $\stratBMDP$, $\expectancy^{\stratBMDP}_\mdpState(\payoff-\min(\payoff, M))\leq\sqrt{\expectancy^{\stratBMDP}_\mdpState(\payoff^2)\cdot\eta}$ because $\payoff-\min(\payoff, M) \leq \payoff\cdot\indic{\payoff\geq M}$.
Because $\payoff^2$ is universally integrable, Lemma~\ref{lem:ui:characterisation} guarantees that $\sup_{\stratBMDP}\expectancy^{\stratBMDP}_\mdpState(\payoff^2)$ is finite, i.e., we obtain the desired inequality by choosing $\eta$ small enough.

We provide the details of the above sketch in the following proof.

\begin{restatable}{theorem}{thmContinuousConvergenceSquare}\label{thm:continuous:convergence:square}
  Let $\mdpState\in\mdpStateSpace$.
  Assume that $\pomdp$ is finite and that $\payoffTuple$ is continuous and universally square integrable.
  Then the function $\stratClassAll{\pomdp}\to\IR^\numObj\colon\stratMDP\mapsto\expectancy^{\stratMDP}_\mdpState(\payoffTuple)$ is continuous.
  In other words, for all sequences $(\stratMDP^{(\indexSequence)})_{\indexSequence\in\IN}$ of strategies that converge to a strategy $\stratMDP$, $\lim_{\indexSequence\to\infty}\expectancy^{\stratMDP^{(\indexSequence)}}_\mdpState(\payoffTuple)=\expectancy^{\stratMDP}_\mdpState(\payoffTuple)$.
\end{restatable}
\begin{proof}
    It suffices to prove the theorem in the case $\numObj = 1$ to obtain the general case.
  For this reason, we consider a one-dimensional real-valued continuous payoff $\payoff$.
  Let $(\stratMDP^{(\indexSequence)})_{\indexSequence\in\IN}$ be a sequence of strategies that converges to a strategy $\stratMDP$.
  We assume that $\payoff$ is non-negative.
  We show that this implies the general case at the end of the proof.

  We show that $\lim_{\indexSequence\to\infty}(\expectancy^{\stratMDP^{(\indexSequence)}}_\mdpState(\payoffTuple))_{\indexSequence\in\IN}=\expectancy^{\stratMDP}_\mdpState(\payoffTuple)$ by the standard definition of convergence.
  Let $\varepsilon > 0$.
  Our goal is to determine some $M\geq 0$ and some $\indexSequence_0$ such that for all $\indexSequence\geq\indexSequence_0$, we have
  \begin{align*}
    |\expectancy^{\stratMDP^{(\indexSequence)}}_\mdpState(\payoff) -
    \expectancy^{\stratMDP^{(\indexSequence)}}_\mdpState(\min(\payoff, M))| + &
    |\expectancy^{\stratMDP^{(\indexSequence)}}_\mdpState(\min(\payoff, M)) -
    \expectancy^{\stratMDP}_\mdpState(\min(\payoff, M))| \\ 
    & + |\expectancy^{\stratMDP}_\mdpState(\min(\payoff, M)) -
    \expectancy^{\stratMDP}_\mdpState(\payoff)|\leq\varepsilon
  \end{align*}
  This is sufficient: the sum highlighted above is greater or equal to $|\expectancy^{\stratMDP^{(\indexSequence)}}_\mdpState(\payoff) - \expectancy^{\stratMDP}_\mdpState(\payoff)|$ by the triangular inequality.
  We bound each term of the above sum by $\frac{\varepsilon}{3}$ in the following.
  
  The crux of the proof is determining a bound $M\geq 0$ such that for all strategies $\stratBMDP$, $\expectancy^{\stratBMDP}_\mdpState(\payoff) - \expectancy^{\stratBMDP}_\mdpState(\min(\payoff, M))<\frac{\varepsilon}{3}$.
  To establish this, we show the following property: for all $\eta>0$, there exists a bound $M(\eta)\geq 0$ such that, for all strategies $\stratBMDP$, $\proba^{\stratBMDP}_\mdpState(\payoff\geq M(\eta))\leq\eta$.
  This last property is shown by contradiction.
  Assume that there exists some $\eta > 0$ such that for all $M\geq 0$, there exists a strategy $\stratBMDP_M$ such that $\proba^{\stratBMDP_M}_\mdpState(\payoff\geq M)>\eta$.
  Then, we obtain that for all $M\geq 0$, $\expectancy^{\stratBMDP_M}_\mdpState(\payoff) \geq \expectancy^{\stratBMDP_M}_\mdpState(M\cdot\indic{\{\payoff\geq M\}}) = M\cdot\proba^{\stratBMDP_M}_\mdpState(\payoff\geq M)\geq\eta\cdot M$.
  This contradicts the fact that $\payoff$ is universally integrable (Lemma~\ref{lem:ui:characterisation}).

  Let $\scalar = 1 + \sup_{\stratBMDP\in\stratClassAll{\pomdp}}\expectancy^{\stratBMDP}_\mdpState(\payoff^2)>0$.
  The value $\scalar$ is real by universal square integrability of $\payoff$ (Lemma~\ref{lem:ui:characterisation}).
  For the remainder of the proof, we fix $M\geq 0$ such that for all strategies $\stratBMDP$, we have $\proba^{\stratBMDP}_\mdpState(\payoff\geq M)\leq\frac{\varepsilon^2}{9\cdot\scalar}$.
  We prove that $M$ is the bound sought above.
  First, we observe that $\payoff - \min(\payoff, M) = (\payoff - M)\cdot\indic{\{\payoff\geq M\}}\leq \payoff\cdot\indic{\{\payoff\geq M\}}$.
  Because indicators are equal to their square, they are universally square integrable.
  By applying the Cauchy-Schwarz inequality, we obtain that, for all strategies $\stratBMDP$,
  \begin{equation*}
    \expectancy^{\stratBMDP}_\mdpState(\payoff - \min(\payoff, M)) \leq
    \expectancy^{\stratBMDP}_\mdpState(\payoff\cdot\indic{\{\payoff\geq M\}}) \leq
    \sqrt{\expectancy^{\stratBMDP}_\mdpState(\payoff^2)\cdot
      \expectancy^{\stratBMDP}_\mdpState(\indic{\{\payoff\geq M\}})} \leq \frac{\varepsilon}{3}.
  \end{equation*}

  To conclude (for non-negative payoffs), it remains to show that there exists $\indexSequence_0$ such that for all $\indexSequence\geq\indexSequence_0$, we have $|\expectancy^{\stratMDP^{(\indexSequence)}}_\mdpState(\min(\payoff, M)) - \expectancy^{\stratMDP}_\mdpState(\min(\payoff, M))|\leq\frac{\varepsilon}{3}$.
  To this end, we observe that the payoff $\min(\payoff, M)$ is continuous (Lemma~\ref{lemma:continuity:max-min}).
  Theorem~\ref{thm:continuous:convergence:real} then implies that a suitable $\indexSequence_0$ exists ($\min(\payoff, M)$ is real-valued).

  We have shown that the theorem holds for non-negative continuous universally square integrable payoffs.
  We now generalise to arbitrary continuous universally square integrable payoffs.
  Let $\payoff^+ = \max(\payoff, 0)$ and $\payoff^- = \max(-\payoff, 0)$ denote the non-negative and non-positive parts of $\payoff$.
  These payoffs are continuous by Lemma~\ref{lemma:continuity:max-min}.
  We observe that $\payoff^2 = (\payoff^+)^2 + (\payoff^-)^2$, and, in particular, $\payoff^2 \geq (\payoff^+)^2 + (\payoff^-)^2$.
  It follows that $(\payoff^+)^2$ and $(\payoff^-)^2$ are universally integrable.
  We obtain the claim of the theorem from the above and the following sequence of equations:
  \begin{align*}
    \lim_{\indexSequence\to\infty}\expectancy^{\stratMDP^{(\indexSequence)}}_\mdpState(\payoff)
    & =
      \lim_{\indexSequence\to\infty}\expectancy^{\stratMDP^{(\indexSequence)}}_\mdpState(\payoff^+) -
      \lim_{\indexSequence\to\infty}\expectancy^{\stratMDP^{(\indexSequence)}}_\mdpState(\payoff^-) \\
    & = \expectancy^{\stratMDP}_\mdpState(\payoff^+) -
      \expectancy^{\stratMDP}_\mdpState(\payoff^-) \\
    & = \expectancy^{\stratMDP}_\mdpState(\payoff).
  \end{align*}
\end{proof}

We now prove that for multi-dimensional universally square integrable payoffs that are continuous, given a closed set of strategies, its set of expected payoffs and achievable set from a state are compact.
For sets of expected payoffs, this follows from Theorem~\ref{thm:continuous:convergence:square}: since the image of compact set by a continuous function is compact, the result is direct.
For achievable sets, it follows from the property that the downward-closure of any compact set is compact.
\begin{restatable}{theorem}{thmContinuousClosed}\label{thm:continuous:closed}
  Let $\mdpState\in\mdpStateSpace$ and $\stratClass\subseteq\stratClassAll{\pomdp}$ be a closed set of strategies.
  Assume that $\payoffTuple$ is a continuous universally square integrable payoff.
  Then $\paySetClass{\payoffTuple}{\mdpState}{\stratClass}$ and $\achSetClass{\payoffTuple}{\mdpState}{\stratClass}$ are compact subsets of $\IR^\numObj$ and $\IRbar^\numObj$ respectively.
  In particular, $\paySet{\payoffTuple}{\mdpState}$, $\achSet{\payoffTuple}{\mdpState}$, $\paySetPure{\payoffTuple}{\mdpState}$ and $\achSetPure{\payoffTuple}{\mdpState}$ are compact.
\end{restatable}
\begin{proof}
    Since $\stratClass$ is closed and $\stratClassAll{\pomdp}$ is compact, we obtain that $\stratClass$ is compact.
  It follows from Theorem~\ref{thm:continuous:convergence:real} that $\paySetClass{\payoffTuple}{\mdpState}{\stratClass}$ is the image of $\stratClass$ by a continuous function, and thus is compact.
  The claim regarding $\paySetPure{\payoffTuple}{\mdpState}$ follows from Lemma~\ref{lem:closed:pure strategies}.
  The claims related to achievable sets follow from the property that for all compact $D\subseteq\IRbar^\numObj$, $\down{D}$ is compact (see Lemma~\ref{lem:down:closed} of Appendix~\ref{appendix:down:closure} for a proof) and the fact that $\achSetClass{\payoffTuple}{\mdpState}{\stratClass} = \down{\paySetClass{\payoffTuple}{\mdpState}{\stratClass}}$ by definition.
\end{proof}

Theorem~\ref{thm:continuous:closed} provides a sufficient condition on payoffs which guarantees that we can dominate any expected payoff by a Pareto-optimal expected payoff.
This property does not hold for all universally integrable payoffs in full generality, e.g., in the one-dimensional setting, optimal strategies need not exist.

\subsection{Non-universally integrable continuous payoffs}\label{section:continuous:infinite}

In this section, we present counterexamples to Theorems~\ref{thm:continuous:convergence:square} and~\ref{thm:continuous:closed} when the assumption of universal square integrability is relaxed.
The MDPs used in these counterexamples are depicted in Figure~\ref{figure:continuous:infinite:ex}.

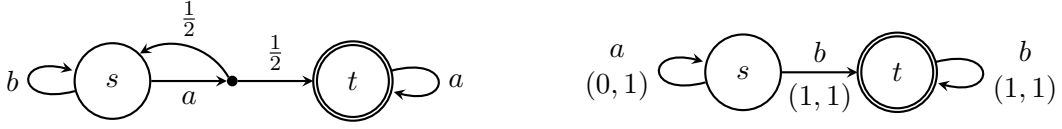
\begin{figure}
  \centering
  \begin{subfigure}[t]{0.48\textwidth}
    \centering
    \begin{tikzpicture}
      \node[state, align=center] (s) {$\mdpState$};
      \node[stochasticc, right=of s] (ss) {};
      \node[state, accepting, right = of ss] (t) {$\mdpStateB$};
      \path[->] (s) edge node[below,align=center] {$\mdpAction$} (ss);
      \path[->] (s) edge[loop left] node[left,align=center] {$\mdpActionB$} (s);
      \path[->] (ss) edge[bend right=45] node[above,align=center] {$\frac{1}{2}$} (s);
      \path[->] (ss) edge node[above,align=center] {$\frac{1}{2}$} (t);
      \path[->] (t) edge[loop right] node[right,align=center] {$\mdpAction$} (t);
    \end{tikzpicture}
    \caption{An MDP with one randomised transition. Weights are omitted and are all $1$.}\label{figure:continuous:infinite:ex:1}
  \end{subfigure}\hfill \begin{subfigure}[t]{0.48\textwidth}
    \centering
    \begin{tikzpicture}
      \node[state, align=center] (s0) {$\mdpState$};
      \node[state, accepting, right = of s0] (s1) {$\mdpStateB$};
      \path[->] (s0) edge node[above,align=center] {$\mdpActionB$} node[below,align=center] {$(1, 1)$} (s1);
      \path[->] (s0) edge[loop left] node[left,align=center] {$\mdpAction$\\$(0, 1)$} (s0);
      \path[->] (s1) edge[loop right] node[right,align=center] {$\mdpActionB$\\$(1, 1)$} (s1);

    \end{tikzpicture}
    \caption{An MDP with two-dimensional weights.}\label{figure:continuous:infinite:ex:2}
  \end{subfigure}
  \caption{MDPs for counter-examples to Theorems~\ref{thm:continuous:convergence:square} (left) and~\ref{thm:continuous:closed} (right) without the universally square integrable assumption.}\label{figure:continuous:infinite:ex}
\end{figure}

For the first example, we provide an MDP with a shortest-path payoff witnessing that Theorem~\ref{thm:continuous:convergence:square} does not generalise to shortest-path payoffs that are not universally (square) integrable, even when only considering pure strategies and when $\paySet{\payoffTuple}{\mdpState}$ is closed.

\begin{example}\label{ex:continuous:infinite:1}
  We consider the MDP $\mdp$ depicted in Figure~\ref{figure:continuous:infinite:ex:1}, the constant weight function $\weight=1$ and the shortest-path function $\spath{\{\mdpStateB\}}{\weight}$.
  We provide a sequence of pure strategies $(\stratMDP^{(\indexSequence)})_{\indexSequence\in\IN}$ converging to a pure strategy $\stratMDP$ such that $\lim_{\indexSequence\to\infty}\expectancy^{\stratMDP^{(\indexSequence)}}_\mdpState(\spath{\{\mdpStateB\}}{\weight})\neq\expectancy^{\stratMDP}_\mdpState(\spath{\{\mdpStateB\}}{\weight})$.
  For all $\indexSequence\in\IN$, we define $\stratMDP^{(\indexSequence)}$ as the pure strategy such that, for all $\hist\in\histSet{\mdp}$, $\stratMDP^{(\indexSequence)}(\hist) = \mdpActionB$ if $\last{\hist} = \mdpState$ and there are at least $\indexSequence$ actions in $\hist$, and otherwise, $\stratMDP^{(\indexSequence)}(\hist) = \mdpAction$.
  Intuitively, when starting in $\mdpState$, $\stratMDP^{(\indexSequence)}$ uses action $\mdpAction$ for the first $\indexSequence$ rounds of the play and then uses $\mdpActionB$ for all subsequent rounds.
  The pure memoryless strategy $\stratMDP$ assigning action $\mdpAction$ to all states is easily checked to be $\lim_{\indexSequence\to\infty}\stratMDP^{(\indexSequence)}$.

  Let $\indexSequence\in\IN$.
  Let $\play_\indexSequence$ denote the play $(\mdpState\mdpAction)^\indexSequence(\mdpState\mdpActionB)^\omega$.
  We have $\proba^{\stratMDP^{(\indexSequence)}}_\mdpState(\{\play_\indexSequence\}) = \frac{1}{2^\indexSequence}$.
  Indeed, for all $\indexLast\in\IN$, the definition of the probability distribution over plays under a strategy implies that
  \[ \proba^{\stratMDP^{(\indexSequence)}}_\mdpState(
    \cyl{(\mdpState\mdpAction)^\indexSequence(\mdpState\mdpActionB)^\indexLast\mdpState}
    ) =
    \proba^{\stratMDP^{(\indexSequence)}}_\mdpState(
    \cyl{(\mdpState\mdpAction)^\indexSequence\mdpState}
    ) = \frac{1}{2^\indexSequence}.
  \]
  It follows from $\spath{\{\mdpStateB\}}{\weight}(\play_\indexSequence) = +\infty$ that $\expectancy^{\stratMDP^{(\indexSequence)}}_\mdpState(\spath{\{\mdpStateB\}}{\weight}) = +\infty$.
  We conclude that $\lim_{\indexSequence\to\infty}\expectancy^{\stratMDP^{(\indexSequence)}}_\mdpState(\spath{\{\mdpStateB\}}{\weight}) = +\infty$.
  
  We now show that $\expectancy^{\stratMDP}_\mdpState(\spath{\{\mdpStateB\}}{\weight})\in\IR$.
  First, we observe that $\proba^{\stratMDP}_\mdpState(\reach{\{\mdpStateB\}}) =1$.
  Therefore, $\expectancy^{\stratMDP}_\mdpState(\spath{\{\mdpStateB\}}{\weight})$ is the unique solution of the equation $x = 1 + \frac{1}{2}x$ (see, e.g.,~\cite{BK08}), i.e., $\expectancy^{\stratMDP}_\mdpState(\spath{\{\mdpStateB\}}{\weight})=2\in\IR$.
  We have shown that $\lim_{\indexSequence\to\infty}\expectancy^{\stratMDP^{(\indexSequence)}}_\mdpState(\spath{\{\mdpStateB\}}{\weight})\neq\expectancy^{\stratMDP}_\mdpState(\spath{\{\mdpStateB\}}{\weight})$.

  We now show that $\paySet{\spath{\{\mdpStateB\}}{\weight}}{\mdpState}$ is closed.
  The memoryless strategy playing $\mdpAction$ is optimal when adopting a minimisation point of view.
  Furthermore, there exist strategies with arbitrarily large but finite expected payoffs from $\mdpState$.
  For instance, for all $\indexSequence\in\IN$, the strategy that plays $\mdpActionB$ for the first $\indexSequence$ rounds in $\mdpState$ and then only uses $\mdpAction$ after ensures a finite payoff greater than $\indexSequence$.
  By convexity of $\paySet{\spath{\{\mdpStateB\}}{\weight}}{\mdpState}\cap\IR$, it follows that $\paySet{\spath{\{\mdpStateB\}}{\weight}}{\mdpState} = \ccInt{2}{+\infty}$ and thus is closed.
  \hfill$\lhd$
\end{example}

We now present a counter-example to Theorem~\ref{thm:continuous:closed} when the considered payoffs are not universally integrable: we provide a continuous payoff such that its set of expected payoffs and achievable sets from a given state are not closed.

\begin{example}\label{ex:continuous:infinite:2}
    We consider the MDP $\mdp$ and the weight function $\weight=(\weight_1, \weight_2)$ depicted in Figure~\ref{figure:continuous:infinite:ex:2}.
  Let $\payoffTuple= (\payoff_1, \payoff_2)$ be the two-dimensional payoff such that $\payoff_1 = \discSum{1/2}{\weight_1}$ and $\payoff_2=\spath{\{\mdpStateB\}}{\weight_2}$.
  To show that $\paySet{\payoffTuple}{\mdpState}$ is not closed, we show that $(2, +\infty)\in\closure{\paySet{\payoffTuple}{\mdpState}}\setminus \paySet{\payoffTuple}{\mdpState}$.
  
  First, we show that $(2, +\infty)\in\closure{\paySet{\payoffTuple}{\mdpState}}$.
  We have $\payoffTuple((\mdpState\mdpAction)^\omega) = (0, \infty)$ and $\payoffTuple(\mdpState(\mdpActionB\mdpStateB)^\omega) = (2, 1)$.
  By convexity of $\paySet{\payoffTuple}{\mdpState}$, we obtain that for all $\eta\in\ooInt{0}{1}$, $(2\eta, +\infty)\in\paySet{\payoffTuple}{\mdpState}$.
  It follows that $(2, +\infty)\in\closure{\paySet{\payoffTuple}{\mdpState}}$.
  
  Next, we argue that $(2, +\infty)\notin\paySet{\payoffTuple}{\mdpState}$.
  We observe that the only play $\play$ from $\mdpState$ such that $\payoff_1(\play) = 2$ is the play $\mdpState(\mdpActionB\mdpStateB)^\omega$.
  All other plays $\play'$ starting in $\mdpState$ are such that $\payoff_1(\play') < 2$.
  Indeed, for all $\indexPosition\in\IN$, we have $\payoff_1((\mdpState\mdpAction)^\indexPosition\mdpState(\mdpActionB\mdpStateB)^\omega) = \sum_{\indexPosition'\geq\indexPosition}\frac{1}{2^{\indexPosition'}} = \frac{1}{2^{\indexPosition-1}} < 2$.
  It follows that, to obtain a payoff of $2$ on the first dimension from $\mdpState$, we require a strategy whose only outcome is $\mdpState(\mdpActionB\mdpStateB)^\omega$.
  This implies that $(2, +\infty)\notin\paySet{\payoffTuple}{\mdpState}$.
  We have shown that $\paySet{\payoffTuple}{\mdpState}$ is not closed.
  \hfill$\lhd$
\end{example}

\bibliographystyle{alpha}
\bibliography{master_references}

\newcommand{\etalchar}[1]{$^{#1}$}
\begin{thebibliography}{CKM{\etalchar{+}}23}

\bibitem[ACK{\etalchar{+}}20]{DBLP:conf/lics/AshokCKWW20}
Pranav Ashok, Krishnendu Chatterjee, Jan Kret{\'{\i}}nsk{\'{y}}, Maximilian
  Weininger, and Tobias Winkler.
\newblock Approximating values of generalized-reachability stochastic games.
\newblock In Holger Hermanns, Lijun Zhang, Naoki Kobayashi, and Dale Miller,
  editors, {\em Proceedings of the 35th Annual {ACM/IEEE} Symposium on Logic in
  Computer Science, {LICS} 2020, Saarbr{\"{u}}cken, Germany, July 8--11, 2020},
  pages 102--115. {ACM}, 2020.

\bibitem[Aum64]{Aumann64}
Robert J~. Aumann.
\newblock Mixed and behavior strategies in infinite extensive games.
\newblock In Melvin Dresher, Lloyd~S. Shapley, and Albert~William Tucker,
  editors, {\em Advances in Game Theory. (AM-52), Volume 52}, pages 627--650.
  Princeton University Press, 1964.

\bibitem[BBC{\etalchar{+}}14]{lmcs:1156}
Tomáš Brázdil, Václav Brožek, Krishnendu Chatterjee, Vojtěch Forejt, and
  Antonín Kučera.
\newblock Markov decision processes with multiple long-run average objectives.
\newblock {\em Logical Methods in Computer Science}, Volume 10, Issue 1, Feb
  2014.

\bibitem[BCM{\etalchar{+}}23]{DBLP:conf/atva/BusattoGastonCMMPR23}
Damien Busatto{-}Gaston, Debraj Chakraborty, Anirban Majumdar, Sayan Mukherjee,
  Guillermo~A. P{\'{e}}rez, and Jean{-}Fran{\c{c}}ois Raskin.
\newblock Bi-objective lexicographic optimization in {Markov} decision
  processes with related objectives.
\newblock In {\'{E}}tienne Andr{\'{e}} and Jun Sun, editors, {\em Proceedings
  (Part {I}) of the 21st Interval Symposium on Automated Technology for
  Verification and Analysis, {ATVA} 2023, Singapore, October 24--27, 2023},
  volume 14215 of {\em Lecture Notes in Computer Science}, pages 203--223.
  Springer, 2023.

\bibitem[BDOR20]{DBLP:journals/lmcs/BrihayeDOR20}
Thomas Brihaye, Florent Delgrange, Youssouf Oualhadj, and Mickael Randour.
\newblock Life is random, time is not: {M}arkov decision processes with window
  objectives.
\newblock {\em Logical Methods in Computer Science}, 16(4), 2020.

\bibitem[BFG{\etalchar{+}}25]{DBLP:conf/aaai/BellyFG00V25}
Marius Belly, Nathana{\"{e}}l Fijalkow, Hugo Gimbert, Florian Horn,
  Guillermo~A. P{\'{e}}rez, and Pierre Vandenhove.
\newblock Revelations: {A} decidable class of {POMDPs} with omega-regular
  objectives.
\newblock In Toby Walsh, Julie Shah, and Zico Kolter, editors, {\em Proceedings
  of the 39th {AAAI} Conference on Artificial Intelligence, Philadelphia, PA,
  {USA}, February 25 -- March 4, 2025}, pages 26454--26462. {AAAI} Press, 2025.

\bibitem[BFRR17]{DBLP:journals/iandc/BruyereFRR17}
V{\'{e}}ronique Bruy{\`{e}}re, Emmanuel Filiot, Mickael Randour, and
  Jean{-}Fran{\c{c}}ois Raskin.
\newblock Meet your expectations with guarantees: Beyond worst-case synthesis
  in quantitative games.
\newblock {\em Information and Computation}, 254:259--295, 2017.

\bibitem[BGG17]{DBLP:journals/jacm/BertrandGG17}
Nathalie Bertrand, Blaise Genest, and Hugo Gimbert.
\newblock Qualitative determinacy and decidability of stochastic games with
  signals.
\newblock {\em Journal of the ACM}, 64(5):33:1--33:48, 2017.

\bibitem[BGMR23]{DBLP:conf/fsttcs/BrihayeGMR23}
Thomas Brihaye, Aline Goeminne, James C.~A. Main, and Mickael Randour.
\newblock Reachability games and friends: {A} journey through the lens of
  memory and complexity (invited talk).
\newblock In Patricia Bouyer and Srikanth Srinivasan, editors, {\em Proceedings
  of the 43rd {IARCS} Annual Conference on Foundations of Software Technology
  and Theoretical Computer Science, {FSTTCS} 2023, {IIIT} Hyderabad, Telangana,
  India, December 18--20, 2023}, volume 284 of {\em LIPIcs}, pages 1:1--1:26.
  Schloss Dagstuhl - Leibniz-Zentrum f{\"{u}}r Informatik, 2023.

\bibitem[BHRR19]{DBLP:conf/concur/BruyereHRR19}
V{\'{e}}ronique Bruy{\`{e}}re, Quentin Hautem, Mickael Randour, and
  Jean{-}Fran{\c{c}}ois Raskin.
\newblock Energy mean-payoff games.
\newblock In Wan~J. Fokkink and Rob van Glabbeek, editors, {\em Proceedings of
  the 30th International Conference on Concurrency Theory, {CONCUR} 2019,
  Amsterdam, the Netherlands, August 27--30, 2019}, volume 140 of {\em LIPIcs},
  pages 21:1--21:17. Schloss Dagstuhl - Leibniz-Zentrum f{\"{u}}r Informatik,
  2019.

\bibitem[Bie87]{bierth1987expected}
Kark-Josef Bierth.
\newblock An expected average reward criterion.
\newblock {\em Stochastic processes and their applications}, 26:123--140, 1987.

\bibitem[BK08]{BK08}
Christel Baier and Joost{-}Pieter Katoen.
\newblock {\em Principles of model checking}.
\newblock {MIT} Press, 2008.

\bibitem[BLO{\etalchar{+}}22]{DBLP:journals/lmcs/BouyerLORV22}
Patricia Bouyer, St{\'{e}}phane {Le Roux}, Youssouf Oualhadj, Mickael Randour,
  and Pierre Vandenhove.
\newblock Games where you can play optimally with arena-independent finite
  memory.
\newblock {\em Logical Methods in Computer Science}, 18(1), 2022.

\bibitem[BMR{\etalchar{+}}18]{DBLP:journals/acta/BouyerMRLL18}
Patricia Bouyer, Nicolas Markey, Mickael Randour, Kim~G. Larsen, and Simon
  Laursen.
\newblock Average-energy games.
\newblock {\em Acta Informatica}, 55(2):91--127, 2018.

\bibitem[BORV23]{DBLP:journals/lmcs/BouyerORV23}
Patricia Bouyer, Youssouf Oualhadj, Mickael Randour, and Pierre Vandenhove.
\newblock Arena-independent finite-memory determinacy in stochastic games.
\newblock {\em Log. Methods Comput. Sci.}, 19(4), 2023.

\bibitem[BRV23]{DBLP:journals/theoretics/BouyerRV23}
Patricia Bouyer, Mickael Randour, and Pierre Vandenhove.
\newblock Characterizing omega-regularity through finite-memory determinacy of
  games on infinite graphs.
\newblock {\em TheoretiCS}, 2, 2023.

\bibitem[CDGH15]{DBLP:journals/iandc/Chatterjee0GH15}
Krishnendu Chatterjee, Laurent Doyen, Hugo Gimbert, and Thomas~A. Henzinger.
\newblock Randomness for free.
\newblock {\em Information and Computation}, 245:3--16, 2015.

\bibitem[CdH04]{DBLP:conf/qest/ChatterjeeAH04}
Krishnendu Chatterjee, Luca {de Alfaro}, and Thomas~A. Henzinger.
\newblock Trading memory for randomness.
\newblock In {\em Proceedings of the 1st International Conference on
  Quantitative Evaluation of Systems, {QEST} 2004, Enschede, The Netherlands,
  27--30 September 2004}, pages 206--217. {IEEE} Computer Society, 2004.

\bibitem[CDH10]{DBLP:conf/mfcs/ChatterjeeDH10}
Krishnendu Chatterjee, Laurent Doyen, and Thomas~A. Henzinger.
\newblock Qualitative analysis of partially-observable {Markov} decision
  processes.
\newblock In Petr Hlinen{\'{y}} and Anton{\'{\i}}n Kucera, editors, {\em
  Proceedings of the 35th International Symposium on Mathematical Foundations
  of Computer Science, {MFCS} 2010, Brno, Czech Republic, August 23--27, 2010},
  volume 6281 of {\em Lecture Notes in Computer Science}, pages 258--269.
  Springer, 2010.

\bibitem[CENR18]{DBLP:conf/ijcai/ChatterjeeE0R18}
Krishnendu Chatterjee, Adri{\'{a}}n Elgy{\"{u}}tt, Petr Novotn{\'{y}}, and Owen
  Rouill{\'{e}}.
\newblock Expectation optimization with probabilistic guarantees in {POMDPs}
  with discounted-sum objectives.
\newblock In J{\'{e}}r{\^{o}}me Lang, editor, {\em Proceedings of the 27th
  International Joint Conference on Artificial Intelligence, {IJCAI} 2018,
  Stockholm, Sweden, July 13--19, 2018}, pages 4692--4699. ijcai.org, 2018.

\bibitem[CFK{\etalchar{+}}13]{DBLP:conf/mfcs/ChenFKSW13}
Taolue Chen, Vojtech Forejt, Marta~Z. Kwiatkowska, Aistis Simaitis, and Clemens
  Wiltsche.
\newblock On stochastic games with multiple objectives.
\newblock In Krishnendu Chatterjee and Jir{\'{\i}} Sgall, editors, {\em
  Proceedings of the 38th International Symposium on Mathematical Foundations
  of Computer Science, {MFCS} 2013, Klosterneuburg, Austria, August 26--30,
  2013}, volume 8087 of {\em Lecture Notes in Computer Science}, pages
  266--277. Springer, 2013.

\bibitem[CFW13]{DBLP:conf/lpar/ChatterjeeFW13}
Krishnendu Chatterjee, Vojtech Forejt, and Dominik Wojtczak.
\newblock Multi-objective discounted reward verification in graphs and {MDPs}.
\newblock In Kenneth~L. McMillan, Aart Middeldorp, and Andrei Voronkov,
  editors, {\em Proceedings of the 19th International Conference on Logic for
  Programming, Artificial Intelligence, and Reasoning, {LPAR}-19, Stellenbosch,
  South Africa, December 14--19, 2013}, volume 8312 of {\em Lecture Notes in
  Computer Science}, pages 228--242. Springer, 2013.

\bibitem[CJH04]{DBLP:conf/soda/ChatterjeeJH04}
Krishnendu Chatterjee, Marcin Jurdzinski, and Thomas~A. Henzinger.
\newblock Quantitative stochastic parity games.
\newblock In J.~Ian Munro, editor, {\em Proceedings of the Fifteenth Annual
  {ACM-SIAM} Symposium on Discrete Algorithms, {SODA} 2004, New Orleans,
  Louisiana, USA, January 11--14, 2004}, pages 121--130. {SIAM}, 2004.

\bibitem[CKK17]{DBLP:journals/lmcs/ChatterjeeKK17}
Krishnendu Chatterjee, Zuzana Kret{\'{\i}}nsk{\'{a}}, and Jan
  Kret{\'{\i}}nsk{\'{y}}.
\newblock Unifying two views on multiple mean-payoff objectives in {Markov}
  decision processes.
\newblock {\em Logical Methods in Computer Science}, 13(2), 2017.

\bibitem[CKM{\etalchar{+}}23]{CKMWW23}
Krishnendu Chatterjee, Joost-Pieter Katoen, Stefanie Mohr, Maximilian
  Weininger, and Tobias Winkler.
\newblock Stochastic games with lexicographic objectives.
\newblock {\em Formal methods in system design}, pages 1--41, 2023.

\bibitem[CNP{\etalchar{+}}17]{DBLP:conf/aaai/Chatterjee0PRZ17}
Krishnendu Chatterjee, Petr Novotn{\'{y}}, Guillermo~A. P{\'{e}}rez,
  Jean{-}Fran{\c{c}}ois Raskin, and Dorde Zikelic.
\newblock Optimizing expectation with guarantees in {POMDPs}.
\newblock In Satinder Singh and Shaul Markovitch, editors, {\em Proceedings of
  the 31st {AAAI} Conference on Artificial Intelligence, San Francisco,
  California, {USA}, February 4--9, 2017}, pages 3725--3732. {AAAI} Press,
  2017.

\bibitem[CO23]{DBLP:conf/icalp/CasaresO23}
Antonio Casares and Pierre Ohlmann.
\newblock Characterising memory in infinite games.
\newblock In Kousha Etessami, Uriel Feige, and Gabriele Puppis, editors, {\em
  Proceedings of the 50th International Colloquium on Automata, Languages, and
  Programming, {ICALP} 2023, July 10--14, 2023, Paderborn, Germany}, volume 261
  of {\em LIPIcs}, pages 122:1--122:18. Schloss Dagstuhl - Leibniz-Zentrum
  f{\"{u}}r Informatik, 2023.

\bibitem[CRR14]{DBLP:journals/acta/ChatterjeeRR14}
Krishnendu Chatterjee, Mickael Randour, and Jean{-}Fran{\c{c}}ois Raskin.
\newblock Strategy synthesis for multi-dimensional quantitative objectives.
\newblock {\em Acta Informatica}, 51(3-4):129--163, 2014.

\bibitem[DKQR20]{DBLP:conf/tacas/DelgrangeKQR20}
Florent Delgrange, Joost{-}Pieter Katoen, Tim Quatmann, and Mickael Randour.
\newblock Simple strategies in multi-objective {MDPs}.
\newblock In Armin Biere and David Parker, editors, {\em Proceedings (Part {I})
  of the 26th International Conference on Tools and Algorithms for the
  Construction and Analysis of Systems, {TACAS} 2020, Held as Part of {ETAPS}
  2020, Dublin, Ireland, April 25--30, 2020}, volume 12078 of {\em Lecture
  Notes in Computer Science}, pages 346--364. Springer, 2020.

\bibitem[Dur19]{Dur19}
Rick Durrett.
\newblock {\em Probability: Theory and Examples}.
\newblock Cambridge Series in Statistical and Probabilistic Mathematics.
  Cambridge University Press, 5th edition, 2019.

\bibitem[EKVY08]{DBLP:journals/lmcs/EtessamiKVY08}
Kousha Etessami, Marta~Z. Kwiatkowska, Moshe~Y. Vardi, and Mihalis Yannakakis.
\newblock Multi-objective model checking of {Markov} decision processes.
\newblock {\em Logical Methods in Computer Science}, 4(4), 2008.

\bibitem[FBB{\etalchar{+}}23]{gog23}
Nathana{\"{e}}l Fijalkow, Nathalie Bertrand, Patricia Bouyer{-}Decitre, Romain
  Brenguier, Arnaud Carayol, John Fearnley, Hugo Gimbert, Florian Horn, Rasmus
  Ibsen{-}Jensen, Nicolas Markey, Benjamin Monmege, Petr Novotn{\'{y}}, Mickael
  Randour, Ocan Sankur, Sylvain Schmitz, Olivier Serre, and Mateusz Skomra.
\newblock Games on graphs.
\newblock {\em CoRR}, abs/2305.10546, 2023.

\bibitem[FH13]{FH13}
Nathana{\"{e}}l Fijalkow and Florian Horn.
\newblock Les jeux d'accessibilit{\'{e}} g{\'{e}}n{\'{e}}ralis{\'{e}}e.
\newblock {\em Technique et Science Informatiques}, 32(9-10):931--949, 2013.

\bibitem[FKP12]{DBLP:conf/atva/ForejtKP12}
Vojtech Forejt, Marta~Z. Kwiatkowska, and David Parker.
\newblock Pareto curves for probabilistic model checking.
\newblock In Supratik Chakraborty and Madhavan Mukund, editors, {\em
  Proceedings of the 10th International Symposium on Automated Technology for
  Verification and Analysis, {ATVA} 2014, Thiruvananthapuram, India, October
  3--6, 2012}, volume 7561 of {\em Lecture Notes in Computer Science}, pages
  317--332. Springer, 2012.

\bibitem[Gim07]{DBLP:conf/stacs/Gimbert07}
Hugo Gimbert.
\newblock Pure stationary optimal strategies in {Markov} decision processes.
\newblock In Wolfgang Thomas and Pascal Weil, editors, {\em Proceedings of the
  24th Annual Symposium on Theoretical Aspects of Computer Science, {STACS}
  2007, Aachen, Germany, February 22--24, 2007}, volume 4393, pages 200--211.
  Springer, 2007.

\bibitem[GK23]{DBLP:journals/ijgt/GimbertK23}
Hugo Gimbert and Edon Kelmendi.
\newblock Submixing and shift-invariant stochastic games.
\newblock {\em International Journal of Game Theory}, 52(4):1179--1214, 2023.

\bibitem[Hor09]{DBLP:conf/stacs/Horn09}
Florian Horn.
\newblock Random fruits on the {Zielonka} tree.
\newblock In Susanne Albers and Jean{-}Yves Marion, editors, {\em Proceedings
  of the 26th International Symposium on Theoretical Aspects of Computer
  Science, {STACS} 2009, Freiburg, Germany, February 26--28, 2009}, volume~3 of
  {\em LIPIcs}, pages 541--552. Schloss Dagstuhl --Leibniz-Zentrum f{\"{u}}r
  Informatik, Germany, 2009.

\bibitem[HPS{\etalchar{+}}21]{DBLP:conf/fm/HahnPSSTW21}
Ernst~Moritz Hahn, Mateo Perez, Sven Schewe, Fabio Somenzi, Ashutosh Trivedi,
  and Dominik Wojtczak.
\newblock Model-free reinforcement learning for lexicographic omega-regular
  objectives.
\newblock In Marieke Huisman, Corina~S. Pasareanu, and Naijun Zhan, editors,
  {\em Proceedings of the 24th International Symposium on Formal Methods, {FM}
  2021, Virtual Event, November 20--26, 2021}, volume 13047 of {\em Lecture
  Notes in Computer Science}, pages 142--159. Springer, 2021.

\bibitem[Kal21]{Kallenberg2021}
Olav Kallenberg.
\newblock {\em Foundations of Modern Probability}.
\newblock Springer International Publishing, Cham, 2021.

\bibitem[Kec95]{kechris1995classical}
A.~Kechris.
\newblock {\em Classical Descriptive Set Theory}.
\newblock Graduate Texts in Mathematics. Springer New York, 1995.

\bibitem[KLC98]{DBLP:journals/ai/KaelblingLC98}
Leslie~Pack Kaelbling, Michael~L. Littman, and Anthony~R. Cassandra.
\newblock Planning and acting in partially observable stochastic domains.
\newblock {\em Artificial Intelligence}, 101(1-2):99--134, 1998.

\bibitem[Kuh53]{Kuhn53}
Harold~W Kuhn.
\newblock Extensive games and the problem of information.
\newblock {\em Contributions to the Theory of Games}, 2(28):193--216, 1953.

\bibitem[LPR18]{DBLP:conf/fsttcs/0001PR18}
St{\'{e}}phane {Le Roux}, Arno Pauly, and Mickael Randour.
\newblock Extending finite-memory determinacy by {B}oolean combination of
  winning conditions.
\newblock In Sumit Ganguly and Paritosh~K. Pandya, editors, {\em Proceedings of
  the 38th {IARCS} Annual Conference on Foundations of Software Technology and
  Theoretical Computer Science, {FSTTCS} 2018, Ahmedabad, India, December
  11--13, 2018}, volume 122 of {\em LIPIcs}, pages 38:1--38:20. Schloss
  Dagstuhl - Leibniz-Zentrum fuer Informatik, 2018.

\bibitem[MPR20]{DBLP:conf/concur/MonmegePR20}
Benjamin Monmege, Julie Parreaux, and Pierre{-}Alain Reynier.
\newblock Reaching your goal optimally by playing at random with no memory.
\newblock In Igor Konnov and Laura Kov{\'{a}}cs, editors, {\em Proceedings of
  the 31st International Conference on Concurrency Theory, {CONCUR} 2020,
  Vienna, Austria, September 1--4, 2020}, volume 171 of {\em LIPIcs}, pages
  26:1--26:21. Schloss Dagstuhl --Leibniz-Zentrum f{\"{u}}r Informatik, 2020.

\bibitem[MR24]{DBLP:journals/iandc/MainR24}
James C.~A. Main and Mickael Randour.
\newblock Different strokes in randomised strategies: Revisiting {Kuhn's}
  theorem under finite-memory assumptions.
\newblock {\em Information and Computation}, 301:105229, 2024.

\bibitem[Mun97]{munkres1997topology}
James~R. Munkres.
\newblock {\em Topology: A First Course}.
\newblock Prentice Hall International, 1997.

\bibitem[QK21]{DBLP:conf/tacas/QuatmannK21}
Tim Quatmann and Joost{-}Pieter Katoen.
\newblock Multi-objective optimization of long-run average and total rewards.
\newblock In Jan~Friso Groote and Kim~Guldstrand Larsen, editors, {\em
  Proceedings (Part {I}) of the 27th International Conference on Tools and
  Algorithms for the Construction and Analysis of Systems, {TACAS} 2021, Held
  as Part of {ETAPS} 2021, Luxemburg City, Luxemburg, March 27--April 1, 2021},
  volume 12651 of {\em Lecture Notes in Computer Science}, pages 230--249.
  Springer, 2021.

\bibitem[Qua23]{DBLP:phd/dnb/Quatmann23}
Tim Quatmann.
\newblock {\em Verification of multi-objective {Markov} models}.
\newblock PhD thesis, {RWTH} Aachen University, Germany, 2023.

\bibitem[Ran13]{rECCS}
Mickael Randour.
\newblock Automated synthesis of reliable and efficient systems through game
  theory: A case study.
\newblock In {\em Proc. of ECCS 2012}, Springer Proceedings in Complexity XVII,
  pages 731--738. Springer, 2013.

\bibitem[RBN22]{DBLP:conf/atal/ReymondBN22}
Mathieu Reymond, Eugenio Bargiacchi, and Ann Now{\'{e}}.
\newblock Pareto conditioned networks.
\newblock In Piotr Faliszewski, Viviana Mascardi, Catherine Pelachaud, and
  Matthew~E. Taylor, editors, {\em Proceedings of the 21st International
  Conference on Autonomous Agents and Multiagent Systems, {AAMAS} 2022,
  Auckland, New Zealand, May 9--13, 2022}, pages 1110--1118. International
  Foundation for Autonomous Agents and Multiagent Systems {(IFAAMAS)}, 2022.

\bibitem[Roc70]{DBLP:books/degruyter/Rockafellar70}
R.~Tyrrell Rockafellar.
\newblock {\em Convex Analysis}.
\newblock Princeton Landmarks in Mathematics and Physics. Princeton University
  Press, 1970.

\bibitem[RRS17]{DBLP:journals/fmsd/RandourRS17}
Mickael Randour, Jean{-}Fran{\c{c}}ois Raskin, and Ocan Sankur.
\newblock Percentile queries in multi-dimensional {Markov} decision processes.
\newblock {\em Formal methods in system design}, 50(2-3):207--248, 2017.

\bibitem[RWO15]{DBLP:conf/ijcai/RoijersWO15}
Diederik~Marijn Roijers, Shimon Whiteson, and Frans~A. Oliehoek.
\newblock Point-based planning for multi-objective {POMDPs}.
\newblock In Qiang Yang and Michael~J. Wooldridge, editors, {\em Proceedings of
  the 24th International Joint Conference on Artificial Intelligence, {IJCAI}
  2015, Buenos Aires, Argentina, July 25--31, 2015}, pages 1666--1672. {AAAI}
  Press, 2015.

\bibitem[SB18]{SuttonB18}
Richard~S. Sutton and Andrew~G. Barto.
\newblock {\em Reinforcement Learning: An Introduction}.
\newblock {MIT} Press, 2018.

\bibitem[Sha53]{Sha53}
Lloyd~S. Shapley.
\newblock Stochastic games.
\newblock {\em Proceedings of the National Academy of Sciences},
  39(10):1095--1100, 1953.

\bibitem[WZ15]{DBLP:conf/ijcai/WrayZ15}
Kyle~Hollins Wray and Shlomo Zilberstein.
\newblock Multi-objective {POMDPs} with lexicographic reward preferences.
\newblock In Qiang Yang and Michael~J. Wooldridge, editors, {\em Proceedings of
  the 24th International Joint Conference on Artificial Intelligence, {IJCAI}
  2015, Buenos Aires, Argentina, July 25--31, 2015}, pages 1719--1725. {AAAI}
  Press, 2015.

\end{thebibliography}

\appendix
\newpage

\section{Additional preliminaries}\label{appendix:prelim}
This section complements Section~\ref{section:prelim}.

\subsection{General notions of topology}\label{appendix:prelim:topology}
The goal of this first section is to recall topological definitions and results mentioned in the body of the paper.
Some definitions provided in this section are relevant for the later Section~\ref{appendix:topology:plays}.
We refer the reader to~\cite{munkres1997topology} for a reference on topology.

\subparagraph*{Topology.}
Let $\topSpace$ be a non-empty set.
A \textit{topology} over $\topSpace$ is a set $\topology\subseteq \subsets{\topSpace}$ of subsets of $\topSpace$ such that (i) $\emptyset$, $\topSpace\in\topology$, (ii) for any family $(\openSet_i)_{i\in I}$ such that $\openSet_i\in\topology$ for all $i\in I$, $\bigcup_{i\in I}\openSet_i\in\topology$ and (iii) if $\openSet$, $\openSet'\in\topology$, then $\openSet\cap \openSet'\in\topology$.
The pair $(X, \topology)$ is called a \textit{topological space}.
Elements of $\topology$ are \textit{open sets}.
A set $\closedSet\subseteq\topSpace$ is \textit{closed} if it is the complement of an open set, i.e., if there exists $\openSet\in\topology$ such that $\closedSet = \topSpace\setminus \openSet$.

We say that $(\topSpace, \topology)$ is a \textit{Hausdorff space} when for any two distinct elements $x$ and $y\in\topSpace$, there exists disjoint open sets $\openSet_x$ and $\openSet_y$ such that $x\in\openSet_x$ and $y\in\openSet_y$.
We assume that all topological spaces below are Hausdorff.

Simple examples of topologies include the discrete topology $\topologyDis = \subsets{\topSpace}$ and the trivial topology $\{\emptyset, \topSpace\}$.
The discrete topology is Hausdorff.
The trivial topology is not Hausdorff whenever $\topSpace$ has at least two elements.

A set $\nhood\subseteq\topSpace$ is a \textit{neighbourhood} of $x\in\topSpace$ if there exists an open set $\openSet_x\in\topology$ such that $x\in\openSet_x\subseteq\nhood$.
A set is open if and only if it is a neighbourhood of all of its elements.
A point $x\in\topSpace$ is an \textit{isolated point} if $\{x\}$ is a neighbourhood of $x$.

Let $\topSpaceB\subseteq\topSpace$.
The \textit{closure} $\closure{\topSpaceB}$ of $\topSpaceB$ is the smallest closed set in which $\topSpaceB$ is included.
An element $x\in\topSpace$ is in $\closure{\topSpaceB}$ if and only if all (open) neighbourhoods of $x$ intersect $\topSpaceB$.
The \textit{interior} $\interior{\topSpaceB}$ of $\topSpaceB$ is the greatest open set that is included in $\topSpaceB$.
An element $x\in\topSpace$ is in $\interior{\topSpaceB}$ if and only if there exists an open neighbourhood $\nhood_x$ of $x$ such that $\nhood_x\subseteq\topSpaceB$.
A set is closed (resp.~open) if and only if it is equal to its closure (resp.~interior).

A \textit{base} of $\topology$ is a set $\topBasis\subseteq\topology$ such that all elements of $\topology$ are (arbitrary) unions of elements of $\topBasis$.
For instance, a topology is a base of itself.
A base of the discrete topology is the set of all singleton sets.
Another example is the usual topology of the extended real line $\IRbar$; a base of this topology is given by the set of intervals
\[
  \{\ooInt{\scalar}{\scalarB}, \coInt{-\infty}{\scalar}, \ocInt{\scalar}{+\infty} \mid \scalar, \scalarB\in\IR,\, \scalar < \scalarB\}.
\]
This topological space is Hausdorff.

Given a non-empty set $\topSpaceB\subseteq\topSpace$, we define the \textit{induced} (or subspace) topology $\topologyInd$ on $\topSpaceB$ as the topology defined by $\topologyInd = \{\openSet\cap\topSpaceB\mid\openSet\in\topology\}$.
An element $x\in\topSpaceB$ is an \textit{isolated point of $\topSpaceB$} if it is an isolated point in $(\topSpaceB, \topologyInd)$.

\subparagraph*{Metric and normed spaces.}
Let $\topSpace$ be a non-empty set.
A \textit{metric} over $\topSpace$ is a function $\metric\colon\topSpace\times\topSpace\to\coInt{0}{+\infty}$ such that, for all $x$, $y$, $z\in\topSpace$, (i) $\metric(x, y) = 0$ if and only if $x = y$, (ii) $\metric(x, y) = \metric(y, x)$ and (iii) $\metric(x, z)\leq \metric(x, y) + \metric(y, z)$.
The last condition is called the \textit{triangle inequality}.
An \textit{open ball} centred in $x\in\topSpace$ of radius $\varepsilon > 0$ is the set $\ball{x}{\varepsilon} = \{y\in\topSpace\mid \metric(x, y) < \varepsilon\}$.
A base of the topology induced by a metric is the set of open balls.
A topological space $(\topSpace, \topology)$ is \textit{metrisable} if there exists a metric that induces $\topology$.
We remark that all metrisable spaces are necessarily Hausdorff.

For instance, the usual topology of $\IR$ is metrisable and is induced by the metric $\metric$ defined by $\metric(x, y) = |x - y|$ for all $x$, $y\in\IR$. The extended real line $\IRbar$ with its usual topology is also metrisable; it is homeomorphic (i.e., topologically isomorphic) to $\ccInt{0}{1}$ with the induced topology inherited from $\IR$.
Topological spaces $\topSpace$ with the discrete topology are also metrisable; the \textit{discrete metric} $\discMetric$ defined by $\discMetric(x, y) = 1$ whenever $x\neq y$ induces the discrete topology (observe that all singleton sets are open balls).

Let $\numObj\in\IN_0$.
Any norm $\|\!\cdot\!\|$ on $\IR^\numObj$ induces a topology via the metric $\metric(\vect, \vectB) = \|\vect - \vectB\|$ for all $\vect$, $\vectB\in\IR^\numObj$.
All norms of $\IR^\numObj$ are \textit{equivalent}, i.e., induce the same topology, which is the usual topology of $\IR^\numObj$. 
For infinite-dimensional spaces, which we do not consider here, some norms may not be equivalent, and can induce different topologies.

\subparagraph*{Convergence.}
Let $(\topSpace, \topology)$ be a Hausdorff topological space.
A sequence $(x_\indexSequence)_{\indexSequence\in\IN}$ of elements of $\topSpace$ is said to converge to $x\in\topSpace$ if for all open neighbourhoods $\openSet_x$ of $x$, there exists $n_0\in\IN$ such that for all $n\geq n_0$, $x_n\in\openSet_x$.
It is equivalent to universally quantify only over open neighbourhoods of $x$ in a fixed base of $\topology$ instead of all open neighbourhoods.
The uniqueness of the limit is guaranteed by the Hausdorff assumption.
We note that in general spaces, sequences can have several limits: e.g., for the trivial topology $\{\emptyset, \topSpace\}$, all sequences converge to all elements of the set.

In a metric space $(\topSpace, \metric)$, this definition of convergence is equivalent to the usual definition recalled hereafter: for all $\varepsilon > 0$, there exists some $\indexSequence_0\in\IN$ such that for all $\indexSequence > \indexSequence_0$, $\metric(x_\indexSequence, x) < \varepsilon$ (i.e., $x_\indexSequence\in\ball{x}{\varepsilon}$).

Convergence of (real) sequences in $\IRbar$ in the above sense is equivalent to the classical definitions for convergence to a real limit, $+\infty$ or $-\infty$.
\subparagraph*{Continuity.}
A function $f\colon (\topSpace, \topology)\to(\topSpaceB, \topologyB)$ is \textit{continuous} at $x\in\topSpace$ if for all neighbourhoods $\nhood_{f(x)}\subseteq\topSpaceB$ of $f(x)$, $f^{-1}(\nhood_{f(x)})$ is a neighbourhood of $x$.
If $\topBasis_\topSpaceB$ is a basis of $(\topSpaceB, \topologyB)$, continuity can be checked by looking only at elements of the basis in the following sense: $f$ is continuous at $x$ if and only if for all $\openSet_{f(x)}\in\topBasis_\topSpaceB$ such that $f(x)\in\openSet_{f(x)}$, $f^{-1}(\openSet_{f(x)})$ is a neighbourhood of $x$.
The function $f$ is said to be continuous if it is continuous at $x$ for all $x\in\topSpace$.

If $f\colon (\topSpace, \metric_\topSpace)\to(\topSpaceB, \metric_\topSpaceB)$ is a function between metric spaces, the definition above is directly equivalent to the usual $\varepsilon$-$\delta$ definition of continuity: $f$ is continuous at $x\in\topSpace$ if for all $\varepsilon>0$, there exists $\delta > 0$ such that for all $x'\in\topSpace$, $\metric_\topSpace(x, x') < \delta$ implies that $\metric_\topSpaceB(f(x), f(x')) < \varepsilon$.

For functions between metric spaces, there exists a stronger variant of continuity, called uniform continuity.
A function $\payoff\colon (\topSpace, \metric_\topSpace)\to (\topSpaceB, \metric_\topSpaceB)$ is \textit{uniformly continuous} if for all $\varepsilon >0$, there exists some $\delta > 0$ such that for all $x, x'\in\topSpace$, $\metric_\topSpace(x, x')<\delta$ implies that $\metric_\topSpaceB(\payoff(x), \payoff(x'))<\varepsilon$.
The difference with continuity is the quantification order.
For continuity, $\delta$ may depend on both $\varepsilon$ and the point at which we check continuity, whereas for uniform continuity, $\delta$ may only depend on $\varepsilon$ and must work for all points.

For instance, the function $\coInt{0}{+\infty}\to\coInt{0}{+\infty}\colon x\mapsto\sqrt{x}$ is uniformly continuous.
This can be shown via the observation that errors at a neighbourhood of some $x\geq 0$ can be bounded independently of $x$, i.e., we have that for all $x\geq 0$ and $|h|\leq x$ (this ensures that $\sqrt{x + h}$ is well-defined), we have $|\sqrt{x+h} - \sqrt{x}|\leq\sqrt{|h|}$.
On the other hand, the function $\IR\to\IR\colon x\mapsto x^2$ is not uniformly continuous.
For any $x\in\IR$ and $h\in\IR$, we have $|(x+h)^2 - x^2| = |h^2 - 2xh|$, and thus, intuitively, we cannot choose $\delta$ independently of $x$ in the definition of continuity because errors depend on $x$ (which cannot be bounded).

For functions from a metric space to another, continuity at $x$ is equivalent to \textit{sequential continuity} at $x$.
A function $f\colon (\topSpace, \topology)\to(\topSpaceB, \topologyB)$ is sequentially continuous at $x\in\topSpace$ if for all sequences $(x_\indexSequence)_{\indexSequence\in\IN}$ that converge to $x$, the sequence $(f(x_\indexSequence))_{\indexSequence\in\IN}$ converges to $f(x)$.

We now prove a result implying that the non-negative and non-positive parts of a continuous function are continuous, as this is used in the main text.
\begin{lemma}\label{lemma:continuity:max-min}
  Let $(\topSpace, \topology)$ be a topological space, $x\in\topSpace$, $f\colon\topSpace\to\IRbar$ be a function that is continuous at $x$ and $M\in\IR$.
  Then the functions $\min(f, M)\colon y\mapsto\min\{f(y), M\}$ and $\max(f, M)\colon y\mapsto\max\{f(y), M\}$ are continuous at $x$.
\end{lemma}
\begin{proof}
  We provide a proof only for $\min(f, M)$ as the argument is analogous for $\max(f, M)$.
  We distinguish three cases.

  First, assume that $f(x) < M$.
  Let $\nhood_{f(x)}$ be a neighbourhood of $f(x)=\min(f, M)(x)$.
  We must show that $\min(f, M)^{-1}(\nhood_{f(x)})$ is a neighbourhood of $x$.
  By continuity of $f$ at $x$, since $\nhood_{f(x)}\cap\coInt{-\infty}{M}$ is a neighbourhood of $f(x)$, $f^{-1}(\nhood_{f(x)}\cap\coInt{-\infty}{M})$ is a neighbourhood of $x$.
  To close the first case, it suffices to establish that $f^{-1}(\nhood_{f(x)}\cap\coInt{-\infty}{M}) \subseteq \min(f, M)^{-1}(\nhood_{f(x)})$.
  This inclusion follows from the fact that for all $y\in\topSpace$, if $f(y) < M$, then $f(y) = \min(f, M)(y)$.
  This ends the proof of the first case.

  Second, assume that $f(x) > M$.
  Let $\nhood_M$ be a neighbourhood of $M = \min(f, M)(x)$.
  By definition of $\min(f, M)$, we obtain that for all $y\in f^{-1}(\ocInt{M}{+\infty})$, $\min(f, M)(y) = M$.
  It follows that $f^{-1}(\ocInt{M}{+\infty})\subseteq\min(f, M)^{-1}(\nhood_M)$.
  By continuity of $f$ at $x$, $f^{-1}(\ocInt{M}{+\infty})$ is a neighbourhood of $x$.
  We conclude from the above that $\min(f, M)^{-1}(\nhood_M)$ is a neighbourhood of $x$.

  Finally, assume that $f(x) = M$ and let $\nhood_M$ be a neighbourhood of $M = \min(f, M)(x)$.
  It suffices to show that $f^{-1}(\nhood_M)\subseteq\min(f, M)^{-1}(\nhood_M)$ by continuity of $f$ at $x$.
  Let $y\in\topSpace$ such that $f(y)\in\nhood_M$.
  If $f(y)\leq M$, then $\min(f, M)(y) = f(y)\in\nhood_M$.
  Otherwise, $\min(f, M)(y) = M\in\nhood_M$.
  This shows the required inclusion and ends the proof.
\end{proof}

\subparagraph*{Compactness.}
A topological space $(\topSpace, \topology)$ is \textit{compact} if for any open cover $(\openSet_i)_{i\in I}$ of $\topSpace$ (i.e., for all $i\in I$, $\openSet_i$ is open and $\bigcup_{i\in I} \openSet_i = \topSpace$), one can extract a finite open cover of $\topSpace$, i.e., there exists $I'\subseteq I$ finite such that $\bigcup_{i\in I'} \openSet_i = \topSpace$.
A subset $\topSpaceB$ of a topological space $(\topSpace, \topology)$ is compact if $(\topSpaceB, \topologyInd)$ is compact, where $\topologyInd$ is the induced topology.
A compact subset of a Hausdorff topological space is closed.
For instance, any finite set with the discrete topology is compact.
Any closed bounded interval of $\IR$ is also compact.
This can be used to show that the extended real line $\IRbar$ is compact (without relying on the fact that $\IRbar$ and $\ccInt{0}{1}$ are homeomorphic).

In metrisable spaces, there is an equivalent characterisation of compactness based on sequences.
A topological space $(\topSpace, \topology)$ is \textit{sequentially compact} if for all sequences $(x_n)_{\indexSequence\in\IN}$ of $\topSpace$, there exists a convergent subsequence.
A metrisable space is compact if and only if it is sequentially compact (e.g.,~\cite[Chap.~3, Thm.~7.4]{munkres1997topology}).

For subsets of $\IR^\numObj$ (and more generally, of finite-dimensional normed vector spaces), there is yet another equivalent formulation of compactness.
A set $D\subseteq\IR^\numObj$ is compact if and only if it is closed and bounded (e.g.,~\cite[Chap.~3, Thm.~6.3]{munkres1997topology}).

Let $(\topSpace, \topology)$ and $(\topSpaceB, \topologyB)$ be topological spaces and let $f\colon \topSpace\to\topSpaceB$.
If $(\topSpace, \topology)$ is compact and $f$ is continuous, then $f(\topSpace)$ is compact.
Furthermore, if $(\topSpace, \metric_\topSpace)$ is a compact metric space and $(\topSpaceB, \metric_{\topSpaceB})$ is a metric space, $f$ is continuous if and only if $f$ is uniformly continuous (e.g.,~\cite[Chap.~3, Thm.~7.3]{munkres1997topology}).

\subparagraph*{Product topology.}
The \textit{product topology} is a topology defined over Cartesian products of topological spaces.
Let $I$ be an arbitrary non-empty set.
For all $i\in I$, let $(\topSpace_i, \topology_i)$ be a topological space.
A base of open sets for the product topology over $\prod_{i\in I}\topSpace_i$ consists of the open sets $\prod_{i\in I}\openSet_i$ where for all $i\in I$, $\openSet_i\in\topology_i$ and $\openSet_i = \topSpace_i$ for all but finitely many $i\in I$.
The product topology is the coarsest topology for which the projections $\prod_{i'\in I}\topSpace_{i'}\to\topSpace_{i}\colon (x_{i'})_{i'\in I}\mapsto x_i$ are continuous for all $i\in I$.

A simple example of the product topology is $\IR^\numObj$; the usual topology of $\IR^\numObj$ corresponds to the product topology of $\numObj$ copies of $\IR$ with its usual topology.

A sequence in a product of topological spaces converges with respect to the product topology if and only if it converges component-wise (see, e.g.,~\cite[Chap.~2, Sect.~8, Exercise~6]{munkres1997topology}).
This is formalised in the following result.
\begin{lemma}\label{lem:componentwise convergence}
  Let $I$ be a non-empty set.
  For all $i\in I$, let $(\topSpace_i, \topology_i)$ be a topological space.
  Let $\prodTopology$ denote the product topology on $\prod_{i\in I}\topSpace_i$.
  Let $(\mathbf{x}^{(\indexSequence)})_{\indexSequence\in\IN}$ be a sequence of elements of $\prod_{i\in I}X_i$ and $\mathbf{x} = (x_i)_{i\in\IN}\in\prod_{i\in I}\topSpace_i$.
  Then $(\mathbf{x}^{(\indexSequence)})_{\indexSequence\in\IN}$ converges to $\mathbf{x}$ if and only if for all $i\in I$, $(x_i^{(\indexSequence)})_{\indexSequence\in\IN}$ converges to $x_i$.
\end{lemma}

A product of metrisable spaces need not be metrisable in general.
However, countable products of metrisable spaces are metrisable.
This can be shown in the same way that~\cite[Chap~2, Thm.~9.5]{munkres1997topology} proves that $\IR^\omega$ with the product topology is metrisable (where the usual topology is assumed on $\IR$).

Finally, we show that countable products of compact metrisable spaces are compact.
In full generality, arbitrary products of compact topological spaces are compact; this result is known as Tychonoff's theorem (e.g.,~\cite[Chap.~5, Thm.~1.1]{munkres1997topology}).
Tychonoff's theorem is equivalent to the axiom of choice.
Because of this, we provide an alternative proof below for the case of countable products.
\begin{theorem}\label{thm:countable:product}
  For all $i\in\IN$, let $(X_i, \topology_i)$ be a sequentially compact topological space.
  Then $(\prod_{i\in\IN}X_i, \prodTopology)$ is sequentially compact, where $\prodTopology$ denotes the product topology.
  In particular, countable products of compact metrisable spaces are compact.
\end{theorem}
\begin{proof}
  The second claim of the theorem follows from the first and the fact that, in metrisable spaces, compactness and sequential compactness are equivalent.
  We thus focus on the first claim of the theorem.
  Let $(\mathbf{x}^{(\indexSequence)})_{\indexSequence\in\IN}$ be a sequence of elements of $\prod_{i\in\IN}\topSpace_i$.
  Our goal is to show that $(\mathbf{x}^{(\indexSequence)})_{\indexSequence\in\IN}$ has a convergent subsequence.

  First, for all $i\in\IN$, we construct a subsequence $(\mathbf{x}^{(\indexSequence)})_{\indexSequence\in I_i}$ of $(\mathbf{x}^{(\indexSequence)})_{\indexSequence\in\IN}$ such that for all $j < i$, $(x_{j}^{(\indexSequence)})_{\indexSequence\in I_i}$ converges in $\topSpace_{j}$.
  We proceed by induction.
  We let $I_0 = \IN$ for the base case.
  For the induction step, we assume that $I_i$ is defined.
  By compactness of $\topSpace_{i+1}$, there exists $I_{i+1}\subseteq I_i$ such that $(x_{i}^{(\indexSequence)})_{\indexSequence\in I_{i+1}}$ converges.
  The induction hypothesis holds by construction because $(\mathbf{x}^{(\indexSequence)})_{\indexSequence\in I_{i+1}}$ is a subsequence of $(\mathbf{x}^{(\indexSequence)})_{\indexSequence\in I_{i}}$.
  
  We now construct a convergent subsequence of $(\mathbf{x}^{(\indexSequence)})_{\indexSequence\in\IN}$.
  Let $\indexSequence_0=\min I_1$ and, for all $i>0$,  let $\indexSequence_i$ be the least element of $I_{i+1}$ strictly greater than $\indexSequence_{i-1}$.
  It is easy to check that $(\mathbf{x}^{(\indexSequence_i)})_{i\in\IN}$ converges via Lemma~\ref{lem:componentwise convergence}.
\end{proof}

\subsection{Properties of the downward closure}\label{appendix:down:closure}
Let $\numObj\in\IN_0$.
We show that the downward closure of a closed subset of $\IRbar^\numObj$ is closed.
We remark below that the statement, in general, does not apply to closed sets of $\IR^\numObj$.

\begin{lemma}\label{lem:down:closed}
  Let $D\subseteq\IRbar^\numObj$ be closed (i.e., compact) with respect to the topology of $\IRbar^\numObj$.
  Then $\down{D}$ is compact with respect to the topology of $\IRbar^\numObj$ and $\down{D}\cap\IR^\numObj$ is closed with respect to the topology of $\IR^\numObj$.
\end{lemma}
\begin{proof}
  Recall that $\IRbar^\numObj$ is a compact metrisable space.
  Thus, it suffices to show that $\down{D}$ is closed with respect to the topology of $\IRbar^\numObj$ to end the proof.

  Let $\payoffVect\in\IR^\numObj$ and $(\payoffVect_n)_{n\in\IN}$ a sequence of elements of $\down{D}$ such that $\payoffVect_n\to\payoffVect$ when $n\to\infty$.
  We must show that $\payoffVect\in \down{D}$.
  The idea of the proof is to bound $\payoffVect$ from above by a vector that is the limit of some sequence of elements of $D$.

  For all $n\in\IN$, we let $\payoffVectB_n\in D$ such that $\payoffVectB_n\geq \payoffVect_n$.
  By (sequential) compactness of $D$, $(\payoffVectB_n)_{n\in\IN}$ has a convergent subsequence.
  Let $\payoffVectB\in D$ denote the limit of one such subsequence.
  It follows that $\payoffVect\leq\payoffVectB$.
  This shows that $\payoffVect\in\down{D}$, and thus $\down{D}$ is a closed subset of $\IRbar^\numObj$.
  Since the topology of $\IR^\numObj$ can be seen as the topology induced on $\IR^\numObj$ by that of $\IRbar^\numObj$, it follows that $\down{D}\cap\IR^\numObj$ is a closed subset of $\IR^\numObj$.
\end{proof}

\begin{remark}[Assumption of Lemma~\ref{lem:down:closed}]
  The subsets of $\IR^\numObj$ that are closed with respect to the topology of $\IRbar^\numObj$ are the compact subsets of $\IR^\numObj$.
  Therefore, the assumption of Lemma~\ref{lem:down:closed} does not apply to all closed subsets of $\IR^\numObj$.

  We provide closed subset of $\IR^2$ the downward-closure of which is not a closed subset of $\IR^2$.
  We consider $D = \{(-\frac{1}{n}, n)\mid n\in\IN_0\}$.
  To see that $D$ is $\IR^2$-closed, consider a convergent sequence of elements of $D$.
  As it is a Cauchy sequence, from some point on, all subsequent elements of the sequence are at a distance of at most $\frac{1}{2}$ from one another.
  Because the distance between two different elements of $D$ is at least $1$, it follows the sequence that is considered is ultimately constant, thus its limit lies in $D$.
  This shows that $D$ is $\IR^2$-closed.

  We now argue that $\down{D}$ and $\down{D}\cap\IR^2$ are not closed: $(0, 0)$ is in the closure of these sets, but not in them.
  To show that $(0, 0)\in\closure{\down{D}}$, we observe that the sequence $((-\frac{1}{n}, 0))_{n\in\IN_0}$ is a sequence of elements of $\down{D}$ that converges to $(0, 0)$.
  On the other hand, we see that for all $n\in\IN_0$, $(0, 0)\leq (-\frac{1}{n}, n)$ does not hold, i.e., $(0, 0)\notin\down{D}$.
  \hfill$\lhd$
\end{remark}

\section{Topology over plays}\label{appendix:plays:topology}
We introduce the usual topology over the set of plays of a POMDP.
We then prove that the definition of a continuous payoff given in Section~\ref{section:payoffs} is equivalent to the topological of continuity.
Finally, we show that the function integrated in Equation~\eqref{equation:prelim:mixed distribution}, which describes the probability distribution induced by a mixed strategy, is measurable.

For the remainder of the section, we fix a POMDP $\pomdp = \pomdpTuple$.
\subsection{Definition}\label{appendix:topology:plays}

We consider $(\mdpStateSpace\mdpActionSpace)^\omega$ equipped with the product topology, where $\mdpStateSpace$ and $\mdpActionSpace$ are both equipped with the discrete topology.
It follows that $(\mdpStateSpace\mdpActionSpace)^\omega$ is metrisable: it is a countable product of metrisable spaces.
If $\pomdp$ is finite, then $(\mdpStateSpace\mdpActionSpace)^\omega$ is also compact (by Theorem~\ref{thm:countable:product}).
The topology of $\playSet{\pomdp}$ is the topology induced on $\playSet{\pomdp}$ by the topology of $(\mdpStateSpace\mdpActionSpace)^\omega$.
A base of the topology of $\playSet{\pomdp}$ is the set of cylinder sets.
We show below that $\playSet{\pomdp}$ is a closed subset of $(\mdpStateSpace\mdpActionSpace)^\omega$, and is therefore compact.

\begin{lemma}\label{lemma:play set closed}
  For any POMDP $\pomdp=\pomdpTuple$, $\playSet{\pomdp}$ is a closed subset of $(\mdpStateSpace\mdpActionSpace)^\omega$.
  In particular, if $\pomdp$ is finite, then $\playSet{\pomdp}$ is a compact space.
\end{lemma}
\begin{proof}
  The first part of the statement implies the second, because all closed subsets of compact spaces are themselves compact~\cite[Chap.~3, Thm.~5.2]{munkres1997topology}.

  For all $\word\in(\mdpStateSpace\mdpActionSpace)^*$, the subset of $(\mdpStateSpace\mdpActionSpace)^\omega$ that contains all continuations of $\word$ is an open set (in fact, this set is in the base of the product topology).
  Furthermore, any $\wordOmega\in(\mdpStateSpace\mdpActionSpace)^\omega\setminus\playSet{\pomdp}$ has a prefix $\word_{\wordOmega}\in (\mdpStateSpace\mdpActionSpace)^*\mdpStateSpace$ that is not a history.
  We obtain that, for all $\wordOmega\in(\mdpStateSpace\mdpActionSpace)^\omega\setminus\playSet{\pomdp}$,  the set of continuations of $\word_\wordOmega$ does not intersect $\playSet{\pomdp}$.
  Therefore, $(\mdpStateSpace\mdpActionSpace)^\omega\setminus\playSet{\pomdp}$ can be written as the union of the sets of continuations of each $\word_\wordOmega$, thus is an open set.
  We have shown that $\playSet{\pomdp}$ is closed.
\end{proof}

\subsection{Continuous payoff functions}

We prove that the definitions of continuous and uniformly continuous payoffs provided in the main text are equivalent to their classical definitions.
For uniform continuity, we require a metric on $\playSet{\pomdp}$.
We consider the metric $\playMetric$ over $\playSet{\pomdp}$ defined by, for all $\play, \play'\in\playSet{\pomdp}$, $\playMetric(\play, \play') = 2^{-\indexLast}$, where $\indexLast\in\IN\cup\{+\infty\}$ is the number of elements in the longest common prefix of $\play$ and $\play'$ (states and actions are counted separately).
This distance is derived from the discrete metric and a construction that shows that countable products of metric spaces are metrisable.

The crux of the argument for both continuity and uniform continuity is that cylinder sets are open balls with respect to $\playMetric$: for all $\hist = \mdpState_0\mdpAction_0\mdpState_1\ldots\mdpAction_{\indexLast-1}\mdpState_\indexLast\in\histSet{\pomdp}$, by definition of $\playMetric$, $\cyl{\hist} = \ball{\play}{2^{-(2\indexLast+1)}}$ for all $\play\in\cyl{\hist}$.
Furthermore, any open neighbourhood of a play contains a cylinder, as all balls include balls with smaller radii.
We provide separate statements for continuity and uniform continuity.
We start with continuity.
\begin{lemma}
  Let $\payoff\colon\playSet{\pomdp}\to\IRbar$ be a payoff function and let $\play\in\playSet{\pomdp}$.
  \begin{itemize}
  \item If $\payoff(\play)\in\IR$, then $\payoff$ is continuous at $\play$ if and only if for all $\varepsilon > 0$, there exists $\indexPosition\in\IN$ such that for all plays $\play'\in\cyl{\playPrefix{\play}{\indexPosition}}$, $|\payoff(\play)-\payoff(\play')| < \varepsilon$.
  \item If $\payoff(\play) = +\infty$ (resp.~$-\infty$), then $\payoff$ is continuous at $\play$ if and only if for all $M\in\IR$, there exists $\indexPosition\in\IN$ such that for all plays $\play'\in\cyl{\playPrefix{\play}{\indexPosition}}$, $\payoff(\play')\geq M$ (resp.~$\payoff(\play')\leq -M$).
  \end{itemize}
\end{lemma}
\begin{proof}
  We proceed by equivalences.
  We first open with general equivalences without assuming anything on $\payoff(\play)$.
  By definition, $\payoff$ is continuous at $\play$ if and only if, for all open neighbourhoods $\nhood_{\payoff(\play)}$ of $\payoff(\play)$, there exists an open neighbourhood $\nhood_{\play}$ of $\play$ such that $\payoff(\nhood_{\play})\subseteq\nhood_{\payoff(\play)}$.
  All open neighbourhoods of $\play$ contain an open ball centred in $\play$ with respect to $\playMetric$, and thus a cylinder set of a prefix of $\play$.
  Thus $\payoff$ is continuous at $\play$ if and only if, for all open neighbourhoods $\nhood_{\payoff(\play)}$ of $\payoff(\play)$, there exists $\indexPosition\in\IN$ such that for all $\play'\in\cyl{\playPrefix{\play}{\indexPosition}}$, $\payoff(\play')\in\nhood_{\payoff(\play)}$.

  We now distinguish cases following whether $\payoff(\play)\in\IR$.
  First, assume that $\payoff(\play)\in\IR$.
  Then, all neighbourhoods of $\payoff(\play)$ contain an interval $\ooInt{\payoff(\play)-\varepsilon}{\payoff(\play)+\varepsilon}$ for some $\varepsilon > 0$.
  The equivalence in the statement of the lemma follows.
  Now, assume that $\payoff(\play) = +\infty$.
  By definition of the topology of $\IRbar$ (see the base specified previously), all neighbourhoods of $+\infty$ contain an interval $\ocInt{M}{\infty}$ for some $M\in\IR$.
  The equivalence in the statement of the lemma follows.
  The argument is similar in the case $\payoff(\play) = -\infty$.
\end{proof}

We now continue with uniform continuity.
\begin{lemma}
  Let $\payoff\colon\playSet{\pomdp}\to\IR$ be a real-valued payoff.
  The payoff $\payoff$ is uniformly continuous if and only if for all $\varepsilon > 0$, there exists $\indexPosition\in\IN$ such that for all plays $\play$, $\play'\in\playSet{\pomdp}$, $\playPrefix{\play}{\indexPosition}= \playPrefix{\play'}{\indexPosition}$ implies that $|\payoff(\play) - \payoff(\play')| < \varepsilon$.
\end{lemma}
\begin{proof}
  The payoff $\payoff$ is uniformly continuous if and only if for all $\varepsilon > 0$, there exists some $\eta > 0$ such that for all $\play$, $\play'\in\playSet{\pomdp}$, $\playMetric(\play, \play') < \eta$ implies that $|\payoff(\play) - \payoff(\play')| < \varepsilon$.
  The desired equivalence can be obtained by exploiting the fact that for all $\indexPosition\in\IN$ and $\play$, $\play'\in\playSet{\pomdp}$, $\playMetric(\play, \play') \leq 2^{-(2\indexPosition+1)}$ if and only if $\playPrefix{\play}{\indexPosition} = \playPrefix{\play'}{\indexPosition}$.
\end{proof}

\subsection{Distributions induced by mixed strategies}\label{appendix:prelim:mixed}
Equation~\eqref{equation:prelim:mixed distribution} defines the distribution induced by a mixed strategy from an initial state.
To show that the integral in this equation is well-defined, we must show that the mapping $\stratClassPure{\pomdp}\to\ccInt{0}{1}\colon\stratBMDP\mapsto\proba_{\mdpState}^\stratBMDP(\objective)$ is measurable for all $\mdpState\in\mdpStateSpace$ and $\objective\in\pomdpSigmaAlgebra$.

Our proof is based on an induction on the Borel hierarchy for measurable subsets of $\playSet{\pomdp}$.
Borel subsets of a metrisable topological space can be arranged in a hierarchy.
In particular, this applies to measurable subsets of $\playSet{\pomdp}$.
We refer the reader to~\cite{kechris1995classical} for an extended exposition on the Borel hierarchy.

Let $\omega_1$ be the first uncountable ordinal.
For $\playSet{\pomdp}$, this hierarchy is as follows.
We let $\Sigma_1^0$ be the open subsets of $\playSet{\pomdp}$.
For each ordinal $1\leq \xi < \omega_1$, we let $\Pi_\xi^0 = \{\playSet{\pomdp}\setminus\openSet\mid\openSet\in\Sigma_\xi^0\}$ be the complements of the sets in $\Sigma_\xi^0$, and if $\xi > 1$, we let
\[\Sigma_\xi^0 = \left\{
    \bigcup_{\indexPosition\in\IN}\openSet_\indexPosition\mid
    \openSet_\indexPosition\in\Pi_{\xi_\indexPosition}, \xi_\indexPosition < \xi, \indexPosition\in\IN
  \right\}\]
be the set of countable unions of sets in $\bigcup_{\xi'<\xi}\Pi_{\xi'}^0$.
Every Borel set is in one of the sets $\Sigma_\xi^0$ for some $\xi < \omega_1$ and for all $1 < \xi'\leq \xi <\omega_1$, we have $\Sigma_{\xi'}^0\subseteq\Sigma_{\xi}^0$ and $\Pi_{\xi'}^0\subseteq\Pi_{\xi}^0$.

Let $\mdpState\in\mdpStateSpace$, $\objective\in\pomdpSigmaAlgebra$ and $P_\objective\colon\stratClassPure{\pomdp}\to\ccInt{0}{1}\colon\stratBMDP\mapsto\proba_{\mdpState}^\stratBMDP(\objective)$.
We show that $P_\objective$ is measurable in three steps.
We first establish it directly for history cylinders.
We extend the result to open subsets $\objective$ of $\playSet{\pomdp}$ by showing that $P_\objective$ can be written as a sum or series of the form $\sum_{\hist\in\histPart}P_{\cyl{\hist}}$ for a countable set of histories $\histPart$.
Finally, we generalise to all measurable sets by induction on the Borel hierarchy, by writing the function as a pointwise limit of measurable functions.
\begin{lemma}\label{lemma:preliminaries:strategy probability measurable}
  Let $\objective\in\pomdpSigmaAlgebra$ and let $\concurState\in\concurStateSpace$.
  The function $P_\objective\colon\stratClassPure{\pomdp}\to\ccInt{0}{1}\colon\stratBMDP\mapsto\proba_{\mdpState}^{\stratBMDP}(\objective)$ is measurable.
\end{lemma}
\begin{proof}
  First, we assume that $\objective$ is a cylinder set.
  Let $\hist\in\histSet{\pomdp}$ such that $\objective = \cyl{\hist}$.
  We assume that $\concurState$ is the first state of $\hist$, as otherwise $P_\objective$ is the constant zero function and the result is direct.
  By definition of distributions over plays for pure strategies, there exists a constant $\alpha$ such that, for all strategies $\stratBMDP\in\stratClassPure{\pomdp}$ such that $\hist$ is consistent (resp.~inconsistent) with $\stratBMDP$, we have $P_\objective(\stratBMDP) = \alpha$ (resp.~$P_\objective(\stratBMDP) = 0$).
  The set of pure strategies with which $\hist$ is consistent is a generator of $\stratSigmaAlgebra$.
  It follows that $P_\objective$ is a linear combination of two indicators of measurable sets, and is therefore measurable.

  We now assume that $\objective$ is an open subset of $\playSet{\pomdp}$.
  Open balls in $\playSet{\pomdp}$ are history cylinders, therefore we can write $\objective = \cyl{\histPart}$ for a countable set of histories $\histPart$.
  We assume that the cylinders of histories in $\histPart$ are pairwise disjoint, as two cylinder sets have a non-empty intersection if and only if one is included in the other.
  It follows that $P_\objective = \sum_{\hist\in\histPart}P_{\cyl{\hist}}$ (by sigma-additivity of $\proba_{\concurState}^{\stratBMDP}$ for all $\stratBMDP\in\stratClassPure{\pomdp}$).
  This shows that $P_\objective$ is the pointwise limit of a sequence of measurable functions and is thus measurable.

  We now show the general case by induction on the Borel hierarchy of $\playSet{\pomdp}$.
  In the previous point, we have shown that $P_\objective$ is measurable if $\objective\in\Sigma_1^0$.
  This is our base case.
  For all ordinals $1\leq \xi < \omega_1$, by showing that $P_\objective$ is measurable for all $\objective\in\Sigma_\xi^0$, we obtain that $P_\objective = 1 - P_{\playSet{\pomdp}\setminus\objective}$ is measurable for all $\objective\in\Pi_\xi^0$.

  Let $1 < \xi < \omega_1$.
  We assume by induction that for all $1\leq \xi' < \xi$ and for all $\objective\in\Pi_{\xi'}^0$, $P_\objective$ is measurable.
  Let $\objective\in\Pi_{\xi}^0$.
  We show that $P_\objective$ is measurable.
  Let $(\objective_\indexPosition)_{\indexPosition\in\IN}$ be a sequence of elements of $\bigcup_{\xi'<\xi}\Pi_{\xi'}^0$ such that $\objective = \bigcup_{\indexPosition\in\IN}\objective_\indexPosition$.
  We let $\objective_{\leq\indexPosition} = \bigcup_{\indexPosition'\leq\indexPosition}\objective_{\indexPosition'}$.
  Since the sequence of sets $(\objective_{\leq\indexPosition})_{\indexPosition\in\IN}$ increases to $\objective$, it follows from the continuity of probability measures that $P_\objective$ is the pointwise limit of $(P_{\objective_{\leq\indexPosition}})_{\indexPosition\in\IN}$.
  Thus, to conclude, it remains to show that for all $\indexPosition\in\IN$, $\objective_{\leq\indexPosition}\in\bigcup_{\xi'<\xi}\Pi_{\xi'}^0$ to invoke the induction hypothesis.
  This property follows from the fact that for each $\xi'<\xi$, $\Pi_{\xi'}^0$ is stable by finite unions~\cite[Prop.~22.1]{kechris1995classical}.
  We have shown that $P_\objective$ is measurable, which ends the inductive argument and the overall proof.
\end{proof}

\section{Continuous payoff functions}\label{appendix:continuous payoffs}
We present examples of continuous payoffs in this section.
First, we show that (a generalisation of) the discounted-sum payoff is continuous.
Second, we show that the shortest-path payoff is continuous whenever the considered weight function is bounded from below by a positive number.
Finally, we provide characterisations in finite POMDPs of the objectives whose indicator is continuous and of prefix-independent payoffs that are continuous.
We fix a POMDP $\pomdp = \pomdpTuple$ for the remainder of the section.

\subsection{Discounted-sum payoff}
In the main text, we defined the discounted-sum payoff with a fixed discount factor.
Certain authors study a variant of this payoff where the discount factor changes at each step of the play depending on the current state.
We define a generalisation of the discounted-sum payoff in which both the discount factor and the weights depend not only on the current state-action pair, but on the history and current action at each step.
We provide sufficient conditions ensuring that this generalisation is well-defined and continuous.

We consider a history-dependent weight function $\weight\colon (\mdpStateSpace\mdpActionSpace)^+\to\IR$ and a history-dependent discount factor function $\discFactor\colon(\mdpStateSpace\mdpActionSpace)^+\to\coInt{0}{1}$.
We assume that $\weight$ is no more than $\weightMax\in\IR$ in absolute value and that $\discFactor$ is bounded away from $1$, i.e., there exists $\discFactor_\star\in\coInt{0}{1}$ such that $\discFactor(\wordOmega) \leq \discFactor_\star$ for all $\wordOmega\in(\mdpStateSpace\mdpActionSpace)^+$.
We define the \textit{generalised discounted-sum payoff} as the function $\discSumG{\discFactor}{\weight}$ defined, for all $\play=\mdpState_0\mdpAction_0\mdpState_1\ldots\in\playSet{\pomdp}$, by
\[\discSumG{\discFactor}{\weight}(\play) =
  \sum_{\indexLast=0}^\infty
  \left(\prod_{\indexPosition=0}^{\indexLast-1}
    \discFactor(\mdpState_0\mdpAction_0\ldots\mdpState_{\indexPosition}\mdpAction_{\indexPosition})\right)
  \weight(\mdpState_0\mdpAction_0\ldots\mdpState_{\indexLast}\mdpAction_{\indexLast}).\]
This function is well-defined for all plays: the defining series is absolutely convergent by the assumptions on $\weight$ and $\discFactor$.
We now prove that $\discSumG{\discFactor}{\weight}$ is continuous.
\begin{lemma}
  Let $\weight\colon (\mdpStateSpace\mdpActionSpace)^+\to\IR$ be a history-dependent weight function such that $\weight$ is no more than $\weightMax\in\IR$ in absolute value and $\discFactor\colon(\mdpStateSpace\mdpActionSpace)^+\to\coInt{0}{1}$ be a history-dependent discount factor function such that there exists $\discFactor_\star\in\coInt{0}{1}$ such that $\discFactor(\wordOmega) \leq \discFactor_\star$ for all $\wordOmega\in(\mdpStateSpace\mdpActionSpace)^+$.
  Then $\discSumG{\discFactor}{\weight}$ is a continuous payoff.
\end{lemma}
\begin{proof}
  Let $\play = \mdpState_0\mdpAction_1\mdpState_1\ldots\in\playSet{\pomdp}$ and $\varepsilon > 0$.
  Let $\indexPosition\in\IN$ such that $\frac{2\cdot\weightMax\cdot\discFactor_\star^\indexPosition}{1-\discFactor_\star} < \varepsilon$ (whose existence is guaranteed by $\discFactor_\star\in\coInt{0}{1}$).
  Let $\play' = \mdpStateB_0\mdpActionB_0\mdpStateB_0\ldots\in\cyl{\playPrefix{\play}{\indexPosition}}$.
  For all $\indexLast\in\IN$, let $\wordOmega_\indexLast = \mdpState_0\mdpAction_0\ldots\mdpState_{\indexLast}\mdpAction_{\indexLast}$ and $\wordOmega_\indexLast' = \mdpStateB_0\mdpActionB_0\ldots\mdpStateB_{\indexLast}\mdpActionB_{\indexLast}$.
  By definition of $\discSumG{\discFactor}{\weight}$, we obtain that
  \begin{align*}
    \big|\discSumG{\discFactor}{\weight}(\play)
    & - \discSumG{\discFactor}{\weight}(\play')\big| \\
    & =
      \left| \sum_{\indexLast=\indexPosition}^\infty
      \left(\left(
      \prod_{\indexSequence=0}^{\indexLast-1}
      \discFactor(\wordOmega_{\indexSequence})
      \right)
      \weight(\wordOmega_{\indexLast}) -
      \left(
      \prod_{\indexSequence=0}^{\indexLast-1}
      \discFactor(\wordOmega'_{\indexSequence})
      \right)
      \weight(\wordOmega'_\indexLast)
      \right)\right| \\
    & \leq
      2 \cdot \sum_{\indexLast=\indexPosition}^\infty\discFactor_\star^\indexLast\cdot\weightMax \\
    & = \frac{2\cdot\weightMax\cdot\discFactor_\star^\indexPosition}{1-\discFactor_\star} \\
    & < \varepsilon.
  \end{align*}
  We have shown that $\discSumG{\discFactor}{\weight}$ is continuous in $\play$.
\end{proof}
\subsection{Shortest-path payoff}

We consider a weight function $\weight\colon\mdpStateSpace\times\mdpActionSpace\to\IR$ and a target $\target\subseteq\mdpStateSpace$.
The shortest-path payoff $\spath{\target}{\weight}$ is continuous over $\reach{\target}$, since the payoff of a play depends only on its prefix prior to the first visit to $\target$.
However, without imposing any conditions on $\weight$, the shortest-path payoff $\spath{\target}{\weight}$ is not necessarily continuous everywhere.

\begin{example}\label{example:continuity:shortest path:ce}
  Consider a two-state MDP $\mdp = (\{\mdpState, \mdpStateB\}, \{\mdpAction\}, \mdpTrans)$ such that $\mdpTrans(\mdpState, \mdpAction)$ is the uniform distribution on $\{\mdpState, \mdpStateB\}$ and $\mdpTrans(\mdpStateB, \mdpAction)(\mdpStateB) = 1$.
  Let $\weight$ be a the constant zero weight function and consider the target $\target = \{\mdpStateB\}$.
  The payoff $\spath{\target}{\weight}$ is not continuous at $\play = (\mdpState\mdpAction)^\omega$.
  Indeed, for all $\indexPosition\in\IN$, the play $\play' = (\mdpState\mdpAction)^{\indexPosition+1}(\mdpStateB\mdpAction)^\omega$ is such that $\playPrefix{\play}{\indexPosition} = \playPrefix{\play'}{\indexPosition}$ and $\spath{\target}{\weight}(\play') = 0$.
  Therefore, for $\spath{\target}{\weight}$ to be continuous,  $0$ would have to be in all neighbourhoods of $\spath{\target}{\weight}(\play) = +\infty$, which is not the case.

  The same example can be used to show that $\indic{\reach{\target}}$ is not continuous in general.
  \hfill$\lhd$
\end{example}

The problem occurring in the previous example is that there is a play with an infinite payoff that can be approached (in the sense of convergence) by plays that have a bounded payoff.
A sufficient condition to avoid this phenomenon is to require that all weights are positive.
In that case, the payoff of a play is no less than the smallest weight multiplied by the number of actions occurring before the first visit to $\target$.
We show that this condition implies the continuity of $\spath{\target}{\weight}$.

\begin{lemma}
  Let $\weight\colon\mdpStateSpace\times\mdpActionSpace\to\IR$ be a weight function and $\target\subseteq\mdpStateSpace$ be a target.
  Assume that there exists $\eta > 0$ such that $\weight(\mdpState, \mdpAction)\geq \eta$ for all $\mdpState\in\mdpStateSpace$ and all $\mdpAction\in\mdpActionSpace$.
  Then $\spath{\target}{\weight}$ is continuous.
\end{lemma}
\begin{proof}
  Let $\play = \mdpState_0\mdpAction_0\mdpState_1\ldots\in\playSet{\pomdp}$.
  First, assume that $\spath{\target}{\weight}(\play)\in\IR$, i.e., $\play\in\reach{\target}$.
  Let $\indexLast\in\IN$ such that $\mdpState_\indexLast\in\target$.
  By definition of $\spath{\target}{\weight}$, for all $\play'\in\cyl{\playPrefix{\play}{\indexLast}}$, we have $\spath{\target}{\weight}(\play) = \spath{\target}{\weight}(\play')$.
  This implies that $\spath{\target}{\weight}$ is continuous in $\play$.

  Now, assume that $\spath{\target}{\weight}(\play) = +\infty$.
  Let $M\in\IR$.
  Let $\indexLast\in\IN$ such that $M\leq\indexLast\cdot\eta$.
  We claim that for all $\play'\in\cyl{\playPrefix{\play}{\indexLast}}$, $\spath{\target}{\weight}(\play')\geq M$.
  Let $\play'\in\cyl{\playPrefix{\play}{\indexLast}}$.
  If $\play'\notin\reach{\target}$, the sought inequality is direct.
  We now assume that $\play'\in\reach{\target}$.
  Because no state of $\target$ occurs in $\playPrefix{\play}{\indexLast}$ and all weights are non-negative, we obtain that
  \[\spath{\target}{\weight}(\play')\geq \sum_{\indexPosition=0}^\indexLast\weight(\mdpState_\indexPosition,\mdpAction_\indexPosition) \geq \indexLast\cdot \eta \geq M.\]
  We have shown that $\spath{\target}{\weight}$ is continuous.
\end{proof}

\subsection{Objectives and indicators}
We now assume that $\pomdp$ is finite and characterise objectives that have a continuous indicator function.
Let $\objective$ be an objective.
The co-domain of an indicator function is $\{0, 1\}$.
This implies that $\indic{\objective}$ is continuous if and only if, for all plays $\play$, there exist $\indexPosition\in\IN$ such that, for all plays $\play'\in\cyl{\playPrefix{\play}{\indexPosition}}$,  we have $\play'\in\objective$ if and only if $\play\in\objective$.
Furthermore it follows from uniform continuity that $\indexPosition$ can be chosen independently of the play in the previous statement.
Therefore, $\indic{\objective}$ is continuous if and only if membership in $\objective$ depends only on a bounded prefix of plays.
We show below that this is also equivalent to $\objective$ being both open and closed.

\begin{lemma}\label{lem:continuous:indicator characterisation}
  Assume that $\pomdp$ is finite.
  Let $\objective$ be an objective.
  The three following statements are equivalent:
  \begin{enumerate}
  \item $\indic{\objective}$ is continuous; \label{item:continuous:indicator characterisation:1}
  \item there exists $\indexPosition\in\IN$ such that for all $\play\in\playSet{\pomdp}$, $\cyl{\playPrefix{\play}{\indexPosition}}$ is either included in or disjoint from $\objective$;\label{item:continuous:indicator characterisation:2}
  \item $\objective$ is open and closed, i.e., there exist finitely many histories $\hist_1, \ldots, \hist_\indexSequence$ such that $\objective=\bigcup_{\indexSequenceB=1}^\indexSequence\cyl{\hist_\indexSequenceB}$. \label{item:continuous:indicator characterisation:3}
  \end{enumerate}
\end{lemma}
\begin{proof}
  We show that Items~\ref{item:continuous:indicator characterisation:1} and~\ref{item:continuous:indicator characterisation:2} are equivalent, then that Items~\ref{item:continuous:indicator characterisation:2} and~\ref{item:continuous:indicator characterisation:3} are equivalent.
  In the interest of this proof being self-contained, we also show that the two properties in Item~\ref{item:continuous:indicator characterisation:3} are equivalent at the end of the proof.
  
  We first show that Items~\ref{item:continuous:indicator characterisation:1} and~\ref{item:continuous:indicator characterisation:2} are equivalent.
  Since $\indic{\objective}$ is real-valued, it is continuous if and only if it is uniformly continuous, i.e., for all $\varepsilon > 0$ there exists $\indexPosition\in\IN$ such that for all plays $\play, \play'\in\playSet{\pomdp}$, if $\playPrefix{\play}{\indexPosition}=\playPrefix{\play'}{\indexPosition}$, then $|\indic{\objective}(\play) - \indic{\objective}(\play')| < \varepsilon$.
  It follows from the co-domain of $\indic{\objective}$ being $\{0, 1\}$ that $\indic{\objective}$ is continuous if and only if there exists $\indexPosition\in\IN$ such that for all plays $\play, \play'\in\playSet{\pomdp}$, if $\play'\in\cyl{\playPrefix{\play}{\indexPosition}}$, then $\indic{\objective}(\play) = \indic{\objective}(\play')$ (the non-trivial direction follows by choosing $\varepsilon = \frac{1}{2}$), i.e., $\cyl{\playPrefix{\play}{\indexPosition}}$ is included in or disjoint from $\objective$.
  This establishes the equivalence of Items~\ref{item:continuous:indicator characterisation:1} and~\ref{item:continuous:indicator characterisation:2}.

  Next, we establish that Items~\ref{item:continuous:indicator characterisation:2} and~\ref{item:continuous:indicator characterisation:3} are equivalent.
  First, assume that Item~\ref{item:continuous:indicator characterisation:2} holds and let $\indexPosition\in\IN$ be given by this property.
  We obtain that $\objective = \bigcup_{\play\in\objective}\cyl{\playPrefix{\play}{\indexPosition}}$.
  There are finitely many histories of the form $\playPrefix{\play}{\indexPosition}$ ($\play\in\playSet{\pomdp}$) because $\mdpStateSpace$ and $\mdpActionSpace$ are finite.
  This shows that Item~\ref{item:continuous:indicator characterisation:3} holds.
  Conversely, assume that Item~\ref{item:continuous:indicator characterisation:3} holds and let $\hist_1$, \ldots, $\hist_\indexSequence$ such that $\objective = \bigcup_{\indexSequenceB=1}^\indexSequence\cyl{\hist_\indexSequenceB}$.
  Item~\ref{item:continuous:indicator characterisation:2} follows by letting $\indexPosition$ be the greatest number of states in the histories $\hist_1$, \ldots, $\hist_\indexSequence$ or zero if there are no histories.

  For the sake of completeness, we close the proof by showing that the two properties given in Item~\ref{item:continuous:indicator characterisation:3} are equivalent.
  First, assume that $\objective$ is open and closed.
  A base of the topology of $\playSet{\pomdp}$ is the set of cylinder sets, therefore $\objective$ is a union of cylinder sets.
  Since $\playSet{\pomdp}$ is compact and $\objective$ is closed, it follows that $\objective$ is compact.
  We conclude that $\objective$ can be written as a finite union of cylinder sets by compactness.

  Conversely, assume that $\objective = \bigcup_{\indexSequenceB=1}^\indexSequence\cyl{\hist_\indexSequenceB}$ for some histories $\hist_1$, \ldots, $\hist_\indexSequence$.
  It suffices to show that $\cyl{\hist}$ is open and closed for all $\hist\in\histSet{\pomdp}$.
  Let $\hist = \mdpState_0\mdpAction_0\ldots\mdpAction_{\indexLast-1}\mdpState_\indexLast\in\histSet{\pomdp}$.
  By definition of the product topology, $\cyl{\hist}$ is open.
  To show that $\cyl{\hist}$ is closed, we consider $\play\in\closure{\cyl{\hist}}$ and show that $\play\in\cyl{\hist}$.
  By definition of closure, the cylinder $\cyl{\playPrefix{\play}{\indexLast}}$ intersects $\cyl{\hist}$, i.e., there exists a play with the prefixes $\playPrefix{\play}{\indexLast}$ and $\hist$.
  It follows that $\playPrefix{\play}{\indexLast} = \hist$, and thus $\play\in\cyl{\hist}$.
  We have shown that $\cyl{\hist}$ is both open and closed, ending the proof.
\end{proof}

\subsection{Prefix-independent payoffs}

Prefix-independent functions assign the same payoff to any two plays that share a common suffix.
Formally, a payoff $\payoff\colon\playSet{\pomdp}\to\IRbar$ is \textit{prefix-independent} if for any play $\play\in\playSet{\pomdp}$ and any $\indexLast\in\IN$, $\payoff(\play) = \payoff(\playSuffix{\play}{\indexLast})$.
We open the section by introducing some useful notation.
Given a history $\hist = \mdpState_0\mdpAction_0\mdpState_1\ldots\mdpAction_{\indexPosition-1}\mdpState_\indexPosition$ and a play $\play = \mdpState_\indexPosition\mdpAction\indexPosition\mdpState_{\indexPosition+1}$ such that the last state of $\hist$ and the first state of $\play$ coincide, we let $\histConcat{\hist}{\play} = \mdpState_0\mdpAction_0\mdpState_1\ldots\mdpAction_{\indexPosition-1}\mdpState_\indexPosition\mdpAction_\indexPosition\mdpState_{\indexPosition+1}\ldots$ denote the play obtained by concatenating $\hist$ and $\play$ without a repetition of $\mdpState_\indexPosition$.
For any play $\play = \mdpState_0\mdpAction_0\mdpState_1\ldots\in\playSet{\pomdp}$, we let $\infSet{\play} = \{\mdpState\in\mdpStateSpace\mid \forall \indexLast\in\IN,\,\exists\,\indexPosition\geq\indexLast,\,\mdpState_\indexPosition=\mdpState\}$ denote the set of states that occur infinitely often in $\play$.

Let $\payoff\colon\playSet{\pomdp}\to\IRbar$ be a prefix-independent payoff.
Assume that $\pomdp$ is finite and that $\payoff$ is continuous.
We claim that the payoff of a play $\play$ is uniquely determined by $\infSet{\play}$.
Let $\play$ and $\play'$ be two plays such that $\infSet{\play} = \infSet{\play'}$.
By continuity of $\payoff$, to prove that $\payoff(\play) = \payoff(\play')$, it suffices to show that in all neighbourhoods of $\play$, there is a play whose payoff is $\payoff(\play')$.
By prefix-independence of $\payoff$, we need only establish that in all neighbourhoods of $\play$, there is a play that shares a common suffix with $\play'$.
Let $\indexPosition\in\IN$ such that $\last{\playPrefix{\play}{\indexPosition}}\in\infSet{\play}$.
We construct a play $\play''$ in $\cyl{\playPrefix{\play}{\indexPosition}}$ such that $\payoff(\play'') = \payoff(\play)$ as follows.
Since $\infSet{\play} = \infSet{\play'}$, there exists $\indexLast\in\IN$ such that the first state of $\playSuffix{\play'}{\indexLast}$ is $\last{\playPrefix{\play}{\indexPosition}}$.
We thus define $\play'' = \histConcat{\playPrefix{\play}{\indexPosition}}{\playSuffix{\play'}{\indexLast}}$.
By prefix-independence of $\payoff$, we obtain that $\payoff(\play'') = \payoff(\playSuffix{\play'}{\indexLast}) = \payoff(\play')$.
This shows that all neighbourhoods of $\play$ contain a play whose payoff is $\payoff(\play')$ and this closes our argument.

This argument above can be adapted to provide a more general property of continuous prefix-independent payoffs: for any two plays $\play$ and $\play'$, if $\infSet{\play}$ is reachable from $\infSet{\play'}$, then both plays have the same payoff.
We show that this property characterises continuous prefix-independent payoffs.

To formalise our characterisation, we introduce some additional terminology and notation.
First, we define Büchi objectives.
Given a set $\target\subseteq\mdpStateSpace$ of target states, the \textit{Büchi objective} for $\target$ requires that states of $\target$ be visited infinitely often.
Formally, given $\target\subseteq\mdpStateSpace$, we define $\buchi{\target} = \{\play\in\playSet{\pomdp}\mid \infSet{\play}\cap\target\neq\emptyset\}$.
Second, we consider the strongly connected components (SCCs) of (the graph induced by) $\mdp$.
An SCC is a maximal set of states $\scc\subseteq\mdpStateSpace$ such that, for all $\mdpState$, $\mdpStateB\in\scc$ there exists a history with at least one action starting in $\mdpState$ and ending in $\mdpStateB$.
An SCC $\scc$ is reachable from an SCC $\scc'$ if there exists a history starting in $\scc'$ and ending in $\scc$.

We obtain the following characterisation.
For the sake of conciseness, we apply the convention $0\cdot (+\infty) = 0\cdot (-\infty) = 0$ below.
\begin{lemma}\label{lem:continuous:prefix-independence}
  Assume that $\pomdp$ is finite.
  Let $\payoff$ be a prefix-independent payoff function.
  Let $\scc_1$, \ldots, $\scc_\numEC$ be the SCCs of $\mdp$.
  Then $\payoff$ is continuous if and only if there exist constants $\scalar_1, \ldots, \scalar_\numEC\in\IRbar$ such that $\payoff = \sum_{\indexEC=1}^\numEC \scalar_\indexEC\cdot\indic{\buchi{\scc_\indexEC}}$ and, for all $1\leq \indexEC, \indexEC'\leq\numEC$, $\scalar_\indexEC = \scalar_{\indexEC'}$ whenever $\scc_{\indexEC'}$ is reachable from $\scc_{\indexEC}$.
\end{lemma}
\begin{proof}
  We first observe that the set of states visited infinitely often along a play is a subset of an SCC.
  In other words, for all plays $\play\in\playSet{\pomdp}$, there exists a unique SCC $\scc$ such that $\play\in\buchi{\scc}$; we say that $\play$ stabilises in $\scc$ for short.
  
  We first assume that $\payoff$ is continuous.
  Let $\play$, $\play'\in\playSet{\pomdp}$.
  We show that if $\play$ stabilises in an SCC $\scc$ and $\play'$ stabilises in an SCC $\scc'$ such that $\scc$ is reachable from $\scc'$, then $\payoff(\play) = \payoff(\play')$.
  Let $\indexLast\in\IN$ such that $\playSuffix{\play}{\indexLast}$ starts in a state of $\scc$.
  For all $\indexPosition\in\IN$, since $\scc$ is reachable from $\scc'$, there exists a play $\play^{(\indexPosition)}\in\cyl{\playPrefix{\play'}{\indexPosition}}$ such that $\playSuffix{\play}{\indexLast}$ is a suffix of $\play^{(\indexPosition)}$.
  By prefix-independence of $\payoff$, for all $\indexPosition\in\IN$, we have $\payoff(\play^{(\indexPosition)}) = \payoff(\play)$.
  Furthermore, all neighbourhoods of $\play'$ contain at least one play $\play^{(\indexPosition)}$ by construction.
  It follows from the continuity of $\payoff$ that $\payoff(\play)$ is in all neighbourhoods of $\payoff(\play')$, i.e., $\payoff(\play) = \payoff(\play')$.

  The previous argument implies that the payoff of a play depends only on the SCC in which the play stabilises.
  For all $1\leq\indexEC\leq\numEC$, let $\scalar_\indexEC$ be the payoff of any play that stabilises in $\scc_\indexEC$.
  From the above, we obtain that $\payoff = \sum_{\indexEC=1}^\numEC \scalar_\indexEC\cdot\indic{\buchi{\scc_\indexEC}}$ and that the coefficients $\scalar_\indexEC$ satisfy the required conditions.
  This ends the proof of the first implication.

  We now let $\scalar_1, \ldots, \scalar_\numEC\in\IRbar$ such that $\payoff = \sum_{\indexEC=1}^\numEC \scalar_\indexEC\cdot\indic{\buchi{\scc_\indexEC}}$ and, for all $1\leq\indexEC,\indexEC'\leq\numEC$, $\scalar_{\indexEC} = \scalar_{\indexEC'}$ whenever $\scc_{\indexEC'}$ is reachable from $\scc_\indexEC$.
  We show that $\payoff$ is continuous.
  Let $\play\in\playSet{\pomdp}$.
  It suffices to show that $\payoff$ is constant over a neighbourhood of $\play$.
  Let $1\leq\indexEC\leq\numEC$ such that $\play$ stabilises in $\scc_\indexEC$ and let $\indexLast\in\IN$ such that all states of $\playSuffix{\play}{\indexLast}$ are in $\scc_\indexEC$.
  Let $\play'\in\cyl{\playSuffix{\play}{\indexLast}}$.
  Then $\play'$ remains in an SCC that is reachable from $\scc_\indexEC$, thus $\payoff(\play') = \scalar_\indexEC = \payoff(\play)$.
  We have shown that $\payoff$ is constant over $\cyl{\playSuffix{\play}{\indexLast}}$, ending the proof of the second implication.
\end{proof}

\section{Universally square integrable shortest-path payoffs}\label{appendix:square:shortest path}
In this appendix, our goal is to prove that, in finite POMDPs, all universally integrable shortest-path payoffs are universally square integrable.
In particular, this implies that Theorem~\ref{thm:continuous:convergence:square} is applicable to universally integrable continuous shortest-path payoffs.
We fix a finite POMDP $\pomdp=\pomdpTuple$ and a target $\target\subseteq\mdpStateSpace$ for this entire section.

We first introduce some notation and terminology.
For all $\indexLast\in\IN$, we let $\reachBounded{\target}{\indexLast} = \{\mdpState_0\mdpAction_0\mdpState_1\ldots\in\playSet{\pomdp}\mid\exists\,\indexPosition\leq\indexLast,\,\mdpState_\indexPosition\in\target\}$ denote the set of plays in which a target is reached in at most $\indexLast$ transitions.
We let $\reachExact{\target}{0} = \reachBounded{\target}{0}$ and, for all $\indexLast\in\IN_0$, we let $\reachExact{\target}{\indexLast} = \reachBounded{\target}{\indexLast}\setminus \reachBounded{\target}{\indexLast-1}$.

A \textit{belief support} $\belief$ is a subset of $\mdpStateSpace$ such that all states in $\belief$ share the same observation.
Intuitively, a belief support indicates the states in which we could possibly be in after a given history.
We define a \textit{belief update} function as follows.
Given a belief support $\belief$, an action $\mdpAction\in\mdpActionSpace$ and an observation $\observation\in\obsSpace$, we let $\beliefUpdate(\belief, \mdpAction, \observation)$ denote the set of states $\mdpState$ such that $\obsFun(\mdpState) = \observation$ and $\mdpState\in\supp{\mdpTrans(\mdpState', \mdpAction)}$ for some $\mdpState'\in\belief$.
We assign a belief support to each history via the function $\beliefHistory$, inductively defined by $\beliefHistory(\mdpState) = \{\mdpState\}$ for all $\mdpState\in\mdpStateSpace$ and $\beliefHistory(\hist) = \beliefUpdate(\beliefHistory(\hist'), \mdpAction, \obsFun(\mdpState))$ for all $\hist = \hist'\mdpAction\mdpState\in\histSet{\pomdp}$.
Belief supports have been used in the literature to qualitatively analyse POMDPs and stochastic games with imperfect information (e.g.,~\cite{DBLP:conf/aaai/BellyFG00V25,DBLP:journals/jacm/BertrandGG17}).

Let $\mdpState\in\mdpStateSpace$ be an initial state.
The goal of this section is to show that the three following properties are equivalent:
\begin{itemize}
\item for all $\stratBMDP\in\stratClassAll{\pomdp}$, $\proba^{\stratBMDP}_\mdpState(\reach{\target})=1$;
\item for all weight functions $\weight$ and $\stratBMDP\in\stratClassAll{\pomdp}$, $\spath{\target}{\weight}$ is $\proba_{\mdpState}^{\stratBMDP}$-integrable;
\item for all weight functions $\weight$ and $\stratBMDP\in\stratClassAll{\pomdp}$, $\spath{\target}{\weight}$ is square integrable with respect to $\proba_{\mdpState}^{\stratBMDP}$.
\end{itemize}
In particular, a shortest-path payoff is universally integrable if and only if it is universally square integrable.

On the one hand, if there exists a strategy $\stratBMDP$ with which $\target$ is not reached almost-surely from $\mdpState$, then the set of plays with an infinite payoff for $\spath{\target}{\weight}$ has non-zero $\proba_{\mdpState}^{\stratBMDP}$-probability, which implies that $\spath{\target}{\weight}$ and its square are not $\proba_{\mdpState}^{\stratBMDP}$-integrable.

To obtain the other implications, it suffices to show that if all strategies ensure that $\target$ is reached almost-surely from $\mdpState$, then $\spath{\target}{1}$ is square integrable with respect to $\proba_{\mdpState}^{\stratBMDP}$ for all $\stratBMDP\in\stratClassAll{\mdp}$.
Indeed, this result for the constant weight function $\weight=1$ implies the result for arbitrary weight functions: there are finitely many different weights as we deal with finite POMDPs, and therefore any shortest-path function $\spath{\target}{\weight}$ is bounded from above in absolute value by a multiple of $\spath{\target}{1}$.

We thus assume that for all $\stratBMDP\in\stratClassAll{\pomdp}$, $\proba^{\stratBMDP}_\mdpState(\reach{\target})=1$.
Let $\stratMDP\in\stratClassAll{\pomdp}$.
Because $\target$ is reached almost-surely from $\mdpState$ under $\stratMDP$, we can write the expectation of $(\spath{\target}{1})^2$ as follows:
\begin{equation*}
  \expectancy^{\stratMDP}_{\mdpState}((\spath{\target}{1})^2) = \sum_{\indexLast\in\IN}\indexLast^2\cdot\proba^{\stratMDP}_{\mdpState}(\reachExact{\target}{\indexLast}).
\end{equation*}
To prove the convergence of this series, we show that $\proba^{\stratMDP}_{\mdpState}(\reachExact{\target}{\indexLast})\in\bigo(\scalar^\indexLast)$ for some $\scalar\in\coInt{0}{1}$.
We observe that for all $\indexLast\in\IN_0$,
\begin{align*}
  \proba^{\stratMDP}_{\mdpState}(\reachExact{\target}{\indexLast})
  & \leq \proba^{\stratMDP}_{\mdpState}(\reach{\target})
    - \proba^{\stratMDP}_{\mdpState}(\reachBounded{\target}{\indexLast-1}) \\
  & = 1 - \proba^{\stratMDP}_{\mdpState}(\reachBounded{\target}{\indexLast-1}).
\end{align*}
It suffices therefore to provide well-chosen lower bounds on the values of (a subsequence of) the non-decreasing sequence $(\proba^{\stratMDP}_{\mdpState}(\reachBounded{\target}{\indexLast}))_{\indexLast\in\IN}$ to prove our result.

Let $\lengthBound = 2^{|\mdpStateSpace|}$ be an upper bound on the number of beliefs.
The first step in our reasoning is to show that if a target is almost-surely reached under all strategies, then for all strategies $\stratBMDP$ and states $\mdpState'$ reachable from $\mdpState$, we have
\begin{equation}\label{equation:square:k-step}
  \proba_{\mdpState'}^{\stratBMDP}(\reachBounded{\target}{\lengthBound}) \geq \eta^\lengthBound,
\end{equation}
where $\eta$ is the smallest positive transition probability in $\pomdp$.
By Lemma~\ref{lem:expectancy:pure integral}, it suffices to show this property for pure strategies to obtain it for all strategies.
If $\stratBMDP$ is pure, Equation~\eqref{equation:square:k-step} holds if and only if there exists a history from $\mdpState'$ to $\target$ with no more than $\lengthBound$ transitions that is consistent with $\stratBMDP$.

To show the above claim, we prove the following: for a given initial state $\mdpState$, if there exists a strategy $\stratBMDP$ such that $\proba_{\mdpState}^{\stratBMDP}(\reachBounded{\target}{\lengthBound}) = 0$ (i.e., such that all histories of length at most $\lengthBound$ that start in $\mdpState$ and are consistent with $\stratBMDP$ do not visit $\target$), then we can construct a strategy $\stratMDP$ that almost-surely avoids $\target$ from $\mdpState$.
We construct a \textit{belief-based} strategy.
Intuitively, the constructed strategy $\stratMDP$ makes decisions based on the belief support of the current history and mimics a decision of $\stratBMDP$ for one of the shortest histories consistent with $\stratBMDP$ to which this belief support is assigned by $\beliefHistory$.
This strategy $\stratMDP$ avoids $\target$ forever when starting from $\mdpState$: if it were to visit a belief intersecting $\target$, then $\stratBMDP$ would be able to reach $\target$ in no more than $\lengthBound$ steps.
We formalise the above ideas in the proof of the following lemma.

\begin{lemma}\label{lemma:square-integ:reach-length}
  Let $\mdpState\in\mdpStateSpace$ and $\lengthBound = 2^{|\mdpStateSpace|}$.
  Assume that there exists $\stratBMDP\in\stratClassAll{\pomdp}$ such that $\proba_{\mdpState}^{\stratBMDP}(\reachBounded{\target}{\lengthBound}) = 0$.
  Then there exists a pure strategy $\stratMDP$ such that $\proba_{\mdpState}^{\stratMDP}(\reach{\target}) = 0$.
\end{lemma}
\begin{proof}
  In the following, we use the following notions.
  For any history $\hist = \mdpState_0\mdpAction_0\ldots\mdpState_{\indexLast}\in\histSet{\pomdp}$, we let $|\hist| = \indexLast$ denote the number of transitions in $\hist$ and refer to $|\hist|$ as the \textit{length} of $\hist$.
  We say that two histories $\hist$ and $\hist'$ are \textit{indistinguishable} if $\obsFun(\hist) = \obsFun(\hist')$.
  Finally, we say that a belief support $\belief$ is \textit{reachable} under a strategy $\stratMDP$ from $\mdpState$ if there is a history $\hist$ starting in $\mdpState$ that is consistent with $\stratMDP$ such that $\beliefHistory(\hist)= \belief$.

  We define $\stratMDP$ as follows.
  For all histories $\hist\in\histSet{\pomdp}$ starting in $\mdpState$, we let $\stratMDP(\hist)$ be an action in $\supp{\stratBMDP(\hist')}$ for some history $\hist'\in\histSet{\pomdp}$ of minimum length that is consistent with $\stratBMDP$ and such that $\beliefHistory(\hist) = \beliefHistory(\hist')$.
  The strategy $\stratMDP$ is a well-defined strategy of $\pomdp$: two indistinguishable histories starting in $\mdpState$ have the same belief support.

  We first show that, for all $\indexLast\in\IN_0$, any belief support that can be reached under $\stratMDP$ from $\mdpState$ in at most $\indexLast$ transitions can also be reached under $\stratBMDP$ in at most $\indexLast$ transitions.
  We prove this by induction on the length of histories starting in $\mdpState$.
  The base case, histories of length zero, is direct: $\mdpState$ is the only such history.
  
  We now assume by induction that the claim holds for all histories of length at most $\indexSequence$.
  Let $\hist = \hist'\mdpAction'\mdpState'\in\histSet{\pomdp}$ be a history of length $\indexSequence+1$.
  By the induction hypothesis (for $\hist'$) and the definition of $\stratMDP$, there must be a history $\hist''$ that is consistent with $\stratBMDP$ of length at most $\indexSequence$ such that $\beliefHistory(\hist'') = \beliefHistory(\hist')$ and $\mdpAction'\in\supp{\stratBMDP(\hist'')}$.
  We may assume that $\last{\hist''} = \last{\hist}$: if a history has a belief support $\belief$, then for all $\mdpState''\in\belief$, there exists an indistinguishable history ending in $\mdpState''$ with the same belief support $\belief$.
  It follows that the history $\hist''\mdpAction'\mdpState'$ is consistent with $\stratBMDP$, has length at most $\indexSequence+1$ and satisfies $\beliefHistory(\hist) = \beliefHistory(\hist'')$.
  This ends the induction argument.

  To conclude, we observe that all belief supports that can be reached under $\stratMDP$ from $\mdpState$ can be reached in at most $\lengthBound-1$ transitions.
  Any belief support that can be reached from $\mdpState$ by playing with $\stratMDP$ can be reached without seeing some belief support twice along the way because $\stratMDP$ makes its decisions based on the current belief support only.
  In particular, any belief support that can be reached from $\mdpState$  under $\stratMDP$ can be reached in no more than $\lengthBound - 1$ transitions from $\mdpState$ under $\stratBMDP$ by the above, and thus no such belief support intersects $\target$ by our assumption on $\stratBMDP$.
  This implies that $\proba_{\mdpState}^{\stratMDP}(\reach{\target}) = 0$.
\end{proof}

Lemma~\ref{lemma:square-integ:reach-length} implies that, under the assumption that the target $\target$ is reached almost-surely from $\mdpState$ under all strategies, then for the probability of reaching $\target$ within $\lengthBound$ steps under a pure strategy is at least $\eta^\lengthBound$ from any state that can be reached from $\mdpState$ (i.e., Equation~\eqref{equation:square:k-step} holds).
This allows us to bound the probability of reaching $\target$ from $\mdpState$ from below in step bounds that are multiples of $\lengthBound$.
More precisely, we now show that, for all $\indexPosition\in\IN$,
\begin{equation}\label{equation:square:kl-inequality}
  \proba_{\mdpState}^{\stratMDP}(\reachBounded{\target}{\indexPosition\lengthBound}) \geq 1 - (1 - \eta^\lengthBound)^{\indexPosition}.
\end{equation}
Intuitively, we obtain Equation~\eqref{equation:square:kl-inequality} by observing that the probability of not reaching $\target$ from $\mdpState$ within $\indexPosition\lengthBound$ transitions is the same as failing coin tosses $\indexPosition$ times in a row such that the coin has a favourable outcome with probability greater than or equal to $\eta^\lengthBound$.
We provide a formal argument below.

\begin{lemma}\label{lemma:square-integ:big-o:period}
  Let $\mdpState\in\mdpStateSpace$, $\lengthBound = 2^{|\mdpStateSpace|}$ and $\eta$ be the minimum transition probability of $\pomdp$.
  Assume that for all $\stratMDP\in\stratClassAll{\mdp}$, $\proba_{\mdpState}^{\stratMDP}(\reach{\target}) = 1$.
  Then for all $\stratMDP\in\stratClassAll{\pomdp}$ and all $\indexPosition\in\IN$, we have
  \begin{equation*}
    \proba_{\mdpState}^{\stratMDP}(\reachBounded{\target}{\indexPosition\lengthBound}) \geq 1 - (1 - \eta^\lengthBound)^{\indexPosition}.
  \end{equation*}
\end{lemma}
\begin{proof}
  We need only prove the result of the lemma for pure strategies; Lemma~\ref{lem:expectancy:pure integral} allows us to extend the result directly to randomised strategies.
  
  We first show that for all pure strategies $\stratMDP\in\stratClassPure{\pomdp}$ and all $\mdpState'\in\mdpStateSpace$ that is reachable from $\mdpState$, we have
  \begin{equation}\label{equation:square:k-step:proof}
    \proba_{\mdpState'}^{\stratMDP}(\reachBounded{\target}{\lengthBound}) \geq \eta^\lengthBound,
  \end{equation}
  i.e., Equation~\eqref{equation:square:k-step} holds.
  Let $\mdpState'\in\mdpStateSpace$ be reachable from $\mdpState$ and let $\stratMDP\in\stratClassPure{\pomdp}$.
  No matter the strategy that is used from $\mdpState'$, $\target$ is reached almost-surely.
  If this were not the case, then there would exist a strategy that does not visit $\target$ almost-surely from $\mdpState$: first ensure that $\mdpState'$ is reached with positive probability and then follow a strategy that does not almost-surely reach $\target$ from $\mdpState'$.
  By Lemma~\ref{lemma:square-integ:reach-length}, there is a history $\hist$ starting in $\mdpState'$ and ending in $\target$ that is consistent with $\stratMDP$ with a most $\lengthBound$ transitions.
  Because $\stratMDP$ is pure, we obtain that
  \[
    \proba_{\mdpState'}^{\stratMDP}(\reachBounded{\target}{\lengthBound}) \geq \proba_{\mdpState'}^{\stratMDP}(\cyl{\hist}) \geq \eta^\lengthBound.
  \]

  We prove the claim of the lemma by exploiting the result shown above.
  Let $\stratMDP\in\stratClassPure{\pomdp}$.
  The proof is by induction on $\indexPosition\in\IN$.
  For $\indexPosition = 0$, the inequality is direct and for $\indexPosition = 1$, it is a special case of the inequality shown in the first part of the proof.
  We now assume by induction that the result holds for $\indexPosition\geq 1$ and show it for $\indexPosition + 1$.

  We first establish the following inequality:
  \begin{equation}\label{equation:square:big-o:induction}
    \proba_{\mdpState}^{\stratMDP}(\reachBounded{\target}{(\indexPosition+1)\lengthBound}\setminus\reachBounded{\target}{\indexPosition\lengthBound}) \geq
    \eta^\lengthBound \cdot \left(1- \proba_{\mdpState}^{\stratMDP}(\reachBounded{\target}{\indexPosition\lengthBound})\right).
  \end{equation}
  For all states $\mdpState'\notin\target$, we let $\objective(\mdpState', \indexPosition)$ be the union of the cylinders of the histories starting in $\mdpState$, ending in $\mdpState'$, with $\indexPosition\lengthBound$ transitions and that do not traverse $\target$.
  We obtain that
  \[
    \reachBounded{\target}{(\indexPosition+1)\lengthBound}\setminus\reachBounded{\target}{\indexPosition\lengthBound} =
    \bigcup_{\mdpState'\notin\target} \reachBounded{\target}{(\indexPosition+1)\lengthBound}\cap\objective(\mdpState', \indexPosition),
  \]
  that is, any element of $\reachBounded{\target}{(\indexPosition+1)\lengthBound}\setminus\reachBounded{\target}{\indexPosition\lengthBound}$ starting in $\mdpState$ can be decomposed into a history in some $\objective(\mdpState', \indexPosition)$ and a play from $\mdpState'$ that visits $\target$ within $\lengthBound$ steps.
  For all $\mdpState'\notin\target$ that is reachable from $\mdpState$, Equation~\eqref{equation:square:k-step:proof} ensures that $\target$ is reached from $\mdpState'$ with probability at least $\eta^\lengthBound$, and therefore,
  \[
    \proba_{\mdpState}^{\stratMDP}(\reachBounded{\target}{(\indexPosition+1)\lengthBound}\cap\objective(\mdpState', \indexPosition)) \geq \eta^\lengthBound\cdot \proba_{\mdpState}^{\stratMDP}(\objective(\mdpState', \indexPosition)).
  \]
  By summing over all $\mdpState'\notin\target$ that are reachable from $\mdpState$, we conclude that
  \begin{align*}
    \proba_{\mdpState}^{\stratMDP}(\reachBounded{\target}{(\indexPosition+1)\lengthBound}\setminus\reachBounded{\target}{\indexPosition\lengthBound})
    & \geq \eta^\lengthBound\cdot \sum_{\mdpState'\notin\target} \proba_{\mdpState}^{\stratMDP}(\objective(\mdpState', \indexPosition))\\
    & = \eta^\lengthBound\cdot \proba_{\mdpState}^{\stratMDP}(\playSet{\pomdp}\setminus\reachBounded{\target}{\indexPosition\lengthBound})\\
    & = \eta^\lengthBound\cdot \left( 1 - \proba_{\mdpState}^{\stratMDP}(\reachBounded{\target}{\indexPosition\lengthBound})\right),
  \end{align*}
  where the second line follows from additivity.

  To end the argument, we use Equation~\eqref{equation:square:big-o:induction} and the induction hypothesis as follows:
  \begin{align*}
    \proba_{\mdpState}^{\stratMDP}(\reachBounded{\target}{(\indexPosition+1)\lengthBound})
    & =
      \proba_{\mdpState}^{\stratMDP}(\reachBounded{\target}{\indexPosition\lengthBound}) + \proba_{\mdpState}^{\stratMDP}(\reachBounded{\target}{(\indexPosition+1)\lengthBound}\setminus\reachBounded{\target}{\indexPosition\lengthBound})\\
    & \geq \proba_{\mdpState}^{\stratMDP}(\reachBounded{\target}{\indexPosition\lengthBound})  + \eta^\lengthBound\cdot \left( 1 - \proba_{\mdpState}^{\stratMDP}(\reachBounded{\target}{\indexPosition\lengthBound})\right)\\
    & = \eta^\lengthBound + (1 - \eta^\lengthBound)\cdot\proba_{\mdpState}^{\stratMDP}(\reachBounded{\target}{\indexPosition\lengthBound})\\
    & \geq \eta^\lengthBound + (1 - \eta^\lengthBound)\cdot\left(1 - (1-\eta^\lengthBound)^\indexPosition\right)\\
    & = 1 - (1 - \eta^\lengthBound)^{\indexPosition+1}.
  \end{align*}
  This ends the induction argument and the overall proof.
\end{proof}

To prove the main result of this section with the approach outlined earlier, it remains to show that $\proba^{\stratMDP}_{\mdpState}(\reachExact{\target}{\indexLast})\in\bigo(\scalar^\indexLast)$ for $\scalar = \sqrt[\lengthBound]{1-\eta^\lengthBound} < 1$.
We use Lemma~\ref{lemma:square-integ:big-o:period} to do so in the following.

\begin{lemma}\label{lemma:square-integ:big-o:global}
  Let $\mdpState\in\mdpStateSpace$, $\lengthBound = 2^{|\mdpStateSpace|}$, $\eta$ be the minimum transition probability of $\pomdp$ and $\scalar = \sqrt[\lengthBound]{1-\eta^\lengthBound}$.
  Assume that for all $\stratMDP\in\stratClassAll{\mdp}$, $\proba_{\mdpState}^{\stratMDP}(\reach{\target}) = 1$.
  Then for all strategies $\stratMDP\in\stratClassAll{\pomdp}$, we have $\proba^{\stratMDP}_{\mdpState}(\reachExact{\target}{\indexLast})\in\bigo(\scalar^\indexLast)$.
\end{lemma}
\begin{proof}
  To establish the result, we show that, for all $\indexLast\in\IN_0$,
  \[
    \proba^{\stratMDP}_{\mdpState}(\reachExact{\target}{\indexLast}) \leq \scalar^{\indexLast - \lengthBound - 1}.
  \]
  
  We recall that the sequence $(\proba^{\stratMDP}_{\mdpState}(\reachBounded{\target}{\indexLast}))_{\indexLast\in\IN}$ is non-decreasing.
  It thus follows from Lemma~\ref{lemma:square-integ:big-o:period} that for all $\indexLast\in\IN$,
  \begin{align*}
    \proba^{\stratMDP}_{\mdpState}(\reachBounded{\target}{\indexLast})
    & \geq \proba^{\stratMDP}_{\mdpState}(\reachBounded{\target}{\lengthBound\cdot\lfloor\frac{\indexLast}{\lengthBound}\rfloor}) \\
    & \geq 1 - (1-\eta^\lengthBound)^{\lfloor\frac{\indexLast}{\lengthBound}\rfloor} \\
    & \geq 1 - (1-\eta^\lengthBound)^{\frac{\indexLast}{\lengthBound}-1}\\
    & = 1 - \scalar^{\indexLast-\lengthBound}.
  \end{align*}

  It follows from the above that, for all $\indexLast\in\IN_0$, we have
  \begin{align*}
    \proba^{\stratMDP}_{\mdpState}(\reachExact{\target}{\indexLast})
    & \leq \proba^{\stratMDP}_{\mdpState}(\reach{\target})
      - \proba^{\stratMDP}_{\mdpState}(\reachBounded{\target}{\indexLast-1}) \\
    & = 1 - \proba^{\stratMDP}_{\mdpState}(\reachBounded{\target}{\indexLast-1})\\
    & \leq \scalar^{\indexLast-1-\lengthBound}.
  \end{align*}
\end{proof}

We can now state the main theorem of this section and provide a short proof based on Lemma~\ref{lemma:square-integ:big-o:global}.

\begin{theorem}
  Let $\target\subseteq\mdpStateSpace$ be a target set and $\mdpState\in\mdpStateSpace$ be an initial state.
  The three following assertions are equivalent:
  \begin{enumerate}
  \item $\proba^{\stratMDP}_{\mdpState}(\reach{\target})=1$ for all strategies $\stratMDP\in\stratClassAll{\pomdp}$;\label{item:theorem:shortest-square:1}
  \item for all weight functions $\weight$ and $\stratMDP\in\stratClassAll{\pomdp}$, $\spath{\target}{\weight}$ is $\proba_{\mdpState}^{\stratMDP}$-integrable;\label{item:theorem:shortest-square:2}
  \item for all weight functions $\weight$ and $\stratMDP\in\stratClassAll{\pomdp}$, $\spath{\target}{\weight}$ is square integrable with respect to $\proba_{\mdpState}^{\stratMDP}$.\label{item:theorem:shortest-square:3}
  \end{enumerate}
  In particular, for any weight function $\weight$, if $\spath{\target}{\weight}$ is universally integrable, then $\spath{\target}{\weight}$ is universally square integrable.
\end{theorem}
\begin{proof}
  If Item~\ref{item:theorem:shortest-square:1} does not hold, then there exists a strategy $\stratMDP$ such that $\proba_{\mdpState}^{\stratMDP}(\spath{\target}{1} = +\infty) > 0$ and it follows that Items~\ref{item:theorem:shortest-square:2} and~\ref{item:theorem:shortest-square:3} do not hold.

  We now assume that Item~\ref{item:theorem:shortest-square:1} holds.
  To end the proof, it suffices to show that Item~\ref{item:theorem:shortest-square:3} holds, as it implies Item~\ref{item:theorem:shortest-square:2}.
  Let $\stratMDP\in\stratClassAll{\pomdp}$.
  We first show that $\spath{\target}{1}$ is square integrable with respect to $\proba_{\mdpState}^{\stratMDP}$ and derive the result for all weight functions from this special case.
  The equality $\proba^\stratMDP_{\mdpState}(\reach{\target})=1$ implies that
  \[
    \expectancy^{\stratMDP}_{\mdpState}((\spath{\target}{1})^2) = \sum_{\indexLast\in\IN}\indexLast^2\cdot\proba^{\stratMDP}_{\mdpState}(\reachExact{\target}{\indexLast}).
  \]
  Lemma~\ref{lemma:square-integ:big-o:global} implies that the above non-negative series is bounded from above by a converging series.
  This shows that $\spath{\target}{1}$ is square integrable with respect to $\proba_{\mdpState}^{\stratMDP}$.

  It remains to establish that for all weight functions $\weight$, the shortest-path payoff $\spath{\target}{\weight}$ is $\proba^{\stratMDP}_{\mdpState}$-square integrable.
  Let $\weightMax=\max_{(\mdpState,\mdpAction)\in\mdpStateSpace\times\mdpActionSpace}|\weight(\mdpState, \mdpAction)|$ ($\weightMax$ is well-defined because $\pomdp$ is finite).
  By the triangular inequality, we obtain that $|\spath{\target}{\weight}|\leq\weightMax\cdot\spath{\target}{1}$.
  We have thus bounded $\spath{\target}{\weight}$ in absolute value by a multiple of a $\proba^{\stratMDP}_{\mdpState}$-square integrable function, thus implying that $\spath{\target}{\weight}$ is also $\proba^{\stratMDP}_{\mdpState}$-square integrable.
\end{proof}

\end{document}